\documentclass[letterpaper,twocolumn,10pt]{article}
\usepackage{usenix}
\usepackage[compact]{titlesec}
\usepackage[numbers]{natbib}
\usepackage[margin=0.75in]{geometry}
\usepackage{enumitem}
\usepackage{mdframed}
\usepackage{morefloats}

\providecommand{\Description}[2][]{}

\newenvironment{ethics}{\section*{Ethical Considerations}
  \phantomsection\addcontentsline{toc}{section}{Ethical Considerations}}{}
\newenvironment{openscience}{\section*{Open Science}
  \phantomsection\addcontentsline{toc}{section}{Open Science}}{}
\newenvironment{ai}{\section*{AI Use}
  \phantomsection\addcontentsline{toc}{section}{AI Use}}{}

\usepackage{graphicx}

\usepackage{amssymb}
\usepackage{amsmath}
\usepackage{amsthm}
\usepackage{verbatim}
\usepackage{algorithm}
\usepackage{algpseudocode}
\usepackage{subcaption}
\usepackage{csquotes}
\usepackage{hyperref}
\usepackage{xfp}
\usepackage[export]{adjustbox}
\usepackage{placeins}
\usepackage{cuted}
\usepackage{mathtools}

\usepackage{tabularx}
\usepackage{booktabs}
\usepackage{xurl}
\newcommand{\obf}[1]{\rule{#1}{0.8em}}

\usepackage{microtype} 
\usepackage{tikz}
\usetikzlibrary{external}
\usepackage{array}
\usepackage{booktabs}
\usepackage{xfp} 
\usepackage{graphicx}

\newcommand{\1}{(1)}
\newcommand{\2}{(2)}
\newcommand{\3}{(3)}
\newcommand{\4}{(4)}
\newcommand{\5}{(5)}

\newcommand{\thead}[1]{\makebox[1.1cm][c]{\parbox{1.1cm}{\centering\scriptsize #1}}}

\newcommand{\diagcell}[2]{%
\begin{tikzpicture}[baseline={([yshift=-0.5ex]current bounding box.center)}]
    \def\w{1.1}  
    \def\h{0.75} 
    \useasboundingbox (0,0) rectangle (\w,\h);
    \clip (0,0) rectangle (\w,\h);

    \fill[red, opacity=#1] (0,0) -- (\w,0) -- (0,\h) -- cycle;
    \fill[blue, opacity={0.5*(1-#2)}] (\w,\h) -- (\w,0) -- (0,\h) -- cycle;

    \node[inner sep=0pt] at ({\w*0.25},{\h*0.25}) {\scriptsize \textcolor{black}{#1}};
    \node[inner sep=0pt] at ({\w*0.75},{\h*0.75}) {\scriptsize \textcolor{black}{#2}};
\end{tikzpicture}%
}

\newcommand{\failDiagCell}[2]{%
\begin{tikzpicture}[baseline={([yshift=-0.5ex]current bounding box.center)}]
    \def\w{1.1}  
    \def\h{0.75} 
    \useasboundingbox (0,0) rectangle (\w,\h);
    \clip (0,0) rectangle (\w,\h);
    \draw[red, line width=1pt] (0.02,0.02) rectangle ({\w-0.02},{\h-0.02});

    \fill[red, opacity=#1] (0,0) -- (\w,0) -- (0,\h) -- cycle;
    \fill[blue, opacity={0.5*(1-#2)}] (\w,\h) -- (\w,0) -- (0,\h) -- cycle;

    \node[inner sep=0pt] at ({\w*0.25},{\h*0.25}) {\scriptsize \textcolor{black}{#1}};
    \node[inner sep=0pt] at ({\w*0.75},{\h*0.75}) {\scriptsize \textcolor{black}{#2}};
\end{tikzpicture}%
}

\newcommand{\diagcellNA}{%
\begin{tikzpicture}[baseline={([yshift=-0.5ex]current bounding box.center)}]
    \def\w{1.1}  
    \def\h{0.75} 
    \useasboundingbox (0,0) rectangle (\w,\h);
    \fill[white] (0,0) rectangle (\w,\h);
\end{tikzpicture}%
}

\newcommand{\paretoDiagCell}[2]{%
\begin{tikzpicture}[baseline={([yshift=-0.5ex]current bounding box.center)}]
    \def\w{1.1}  
    \def\h{0.75} 
    \useasboundingbox (0,0) rectangle (\w,\h);
    \clip (0,0) rectangle (\w,\h);

    \fill[red, opacity=#1] (0,0) -- (\w,0) -- (0,\h) -- cycle;
    \fill[blue, opacity={0.5*(1-#2)}] (\w,\h) -- (\w,0) -- (0,\h) -- cycle;

    \node[inner sep=0pt] at ({\w*0.25},{\h*0.25}) {\scriptsize \textcolor{black}{#1}};
    \node[inner sep=0pt] at ({\w*0.75},{\h*0.75}) {\scriptsize \textcolor{black}{#2}};
    
   \draw[black, line width=1pt] (0.02,0.02) rectangle (\w-0.02,\h-0.02);
\end{tikzpicture}%
}

\newcommand{\diagcellLabeled}[3]{%
\begin{tikzpicture}[baseline={([yshift=-0.5ex]current bounding box.center)}]
    \def\w{1.2}  
    \def\h{0.85} 
    \useasboundingbox (0,0) rectangle (\w,\h);
    \clip (0,0) rectangle (\w,\h);

    \fill[red, opacity=#1] (0,0) -- (\w,0) -- (0,\h) -- cycle;
    \fill[blue, opacity={0.5*(1-#2)}] (\w,\h) -- (\w,0) -- (0,\h) -- cycle;

    \node[anchor=north, inner sep=1pt, text depth=0pt] at ({\w/2}, \h) {\tiny #3};

    \node[inner sep=0pt] at ({\w*0.25},{\h*0.25}) {\scriptsize \textcolor{black}{#1}};
    \node[inner sep=0pt] at ({\w*0.75},{\h*0.50}) {\scriptsize \textcolor{black}{#2}};
\end{tikzpicture}%
}

\newcommand{\paretoDiagcellLabeled}[3]{%
\begin{tikzpicture}[baseline={([yshift=-0.5ex]current bounding box.center)}]
    \def\w{1.2}  
    \def\h{0.85} 
    \useasboundingbox (0,0) rectangle (\w,\h);
    \clip (0,0) rectangle (\w,\h);
    \draw[black, line width=1pt] (0.04,0.04) rectangle (\w,\h);

    \fill[red, opacity=#1] (0,0) -- (\w,0) -- (0,\h) -- cycle;
    \fill[blue, opacity={0.5*(1-#2)}] (\w,\h) -- (\w,0) -- (0,\h) -- cycle;

    \node[anchor=north, inner sep=1pt, text depth=0pt] at ({\w/2}, \h) {\tiny #3};

    \node[inner sep=0pt] at ({\w*0.25},{\h*0.25}) {\scriptsize \textcolor{black}{#1}};
    \node[inner sep=0pt] at ({\w*0.75},{\h*0.50}) {\scriptsize \textcolor{black}{#2}};
\end{tikzpicture}%
}

\usepackage{multirow}
\usepackage{booktabs}
\usepackage{pgfplots}
\pgfplotsset{compat=1.18}
\usepgfplotslibrary{groupplots}

\definecolor{oiBlue}{RGB}{0,114,178}
\definecolor{oiGreen}{RGB}{0,158,115}
\definecolor{oiVermillion}{RGB}{213,94,0}
\definecolor{oiPurple}{RGB}{204,121,167}

\pgfplotsset{
privacyplot/.style={
width=\linewidth,
height=6cm,
xlabel={Model Complexity},
ylabel={Privacy Resolution},
xmin=0.5,
ymin=0,
ymax=0.4,
grid=both,
major grid style={gray!25},
minor grid style={gray!10},
tick style={draw=none},
xtick=\empty,
legend style={at={(0.98,0.98)},anchor=north east},
line width=1pt,
mark size=1.6pt,
nodes near coords,
point meta=explicit symbolic,
every node near coord/.append style={
font=\tiny,
inner sep=1pt
}
}
}

\newtheorem{theorem}{Theorem}

\theoremstyle{definition}
\newtheorem{definition}{Definition}

\newtheorem{example}{Example}

\usepackage{xcolor}

\newtheorem{numquote}{Concepts}
\newenvironment{concept}[1][]
{
  \begin{numquote}[#1]\normalfont\itshape
}
{
  \end{numquote}
}

\newtheorem{numprompt}{Prompt}
\newenvironment{prompt}[1][]
{
  \begin{numprompt}[#1]\normalfont\itshape
}
{
  \end{numprompt}
}

\mdfdefinestyle{arxivstatement}{
  linewidth=2pt,linecolor=gray!30,
  topline=false,bottomline=false,rightline=false,
  innerleftmargin=6pt,innerrightmargin=0pt,
  innertopmargin=0pt,innerbottommargin=0pt,
  skipabove=1pt,skipbelow=1pt
}
\surroundwithmdframed[style=arxivstatement]{theorem}
\surroundwithmdframed[style=arxivstatement]{lemma}
\surroundwithmdframed[style=arxivstatement]{definition}
\surroundwithmdframed[style=arxivstatement]{example}

\hypersetup{
  pdftitle={Contrastive Privacy: A Semantic Approach to Measuring Privacy of AI-Based Sanitization},
  pdfauthor={George Bissias, Eugene Bagdasarian, Brian Neil Levine},
  pdfkeywords={privacy, sanitization}
}
\date{}

\begin{document}

\title{\Large\bfseries Contrastive Privacy: A Semantic Approach to\\
Measuring Privacy of AI-Based Sanitization}

\author{
  {\rmfamily\upshape George Bissias} \quad
  {\rmfamily\upshape Eugene Bagdasarian} \quad
  {\rmfamily\upshape Brian Neil Levine}\vspace{0.1cm} \\
  University of Massachusetts Amherst \\
  {\small \{gbiss, eugene, brian\}@cs.umass.edu}
}

\maketitle

\begin{abstract}
AI-based sanitization can remove concepts from images and text, but privacy evaluation remains largely ad hoc. We propose contrastive privacy, a formal definition that yields a quantitative test with a semantic interpretation. Under formal assumptions, we derive a conditional sufficiency result for a class of sanitized renderings (i.e., media files). We operationalize the definition using imperfect semantic-distance models such as CLIP. The test compares sanitized renderings under audit with both the original and sanitized versions of reference renderings known to contain privacy-relevant properties; if the rendering under audit is semantically closer to the unsanitized reference, then the former might leak private information even after sanitization. Importantly, the test is able to conditionally audit an abstract privacy target without enumerating its constituent properties or requiring per-item ground-truth labels; results remain relative to the chosen models and data.

We evaluate 34 image-sanitization configurations, 15 social-media text models and one entity recognizer, and four off-the-shelf PII sanitizers on Enron emails. The tests detect residual semantic associations in every setting, including all four PII tools even when they sanitize every candidate their detectors return. Two further studies examine sensitivity to the semantic-distance model. For a synthetic identity, fine-tuning reveals target associations missed by the base model; after broader sanitization, the adapted test detects none. Across 94 matched image collections sanitized to conceal Leonardo DiCaprio's identity, three models yield broadly correlated assessments but sometimes disagree on which sanitizations appear most private. These findings support using multiple models and show how contrastive privacy can reveal retained identifiers and identity-revealing context across sanitization methods.
\end{abstract}

\begin{figure}[!t]

\begin{subfigure}[t]{0.4\textwidth}
\centering
\begin{tikzpicture}
    \begin{axis}[
        width=\linewidth,
        height=2.6in, 
        xlabel={Contrastive Privacy Resolution (lower is better)},
        ylabel={Utility (higher is better)},
        label style={font=\footnotesize},
        tick label style={font=\scriptsize},
        grid=both,
        grid style={line width=.1pt, draw=gray!30},
        major grid style={line width=.2pt, draw=gray!50},
        xmin=0.1, xmax=0.6,
        ymin=0.15, ymax=1.0,
        legend pos=south east,
        legend style={font=\scriptsize, cells={align=left}},
        clip=false 
    ]

    \addplot[
        only marks,
        mark=*,
        mark size=1.5pt,
        draw=blue!70!black,
        fill=blue!50,
    ] coordinates {
        (0.31, 0.51) 
        (0.25, 0.44) 
        (0.27, 0.71) 
        (0.30, 0.67) 
        (0.32, 0.70) 
        (0.29, 0.71) 
        (0.29, 0.66) 
        (0.52, 0.60) 
        (0.53, 0.63) 
        (0.45, 0.62) 
        (0.53, 0.61) 
        (0.54, 0.55) 
        (0.31, 0.77) 
        (0.41, 0.63) 
        (0.46, 0.56) 
        (0.39, 0.55) 
        (0.31, 0.59) 
        (0.28, 0.78) 
        (0.30, 0.75) 
        (0.32, 0.78) 
        (0.31, 0.80) 
        (0.33, 0.73) 
        (0.27, 0.67) 
        (0.28, 0.69) 
        (0.22, 0.69) 
        (0.40, 0.21) 
        (0.41, 0.39) 
        (0.23, 0.31) 
        (0.31, 0.49) 
    };

    \addplot[
        only marks,
        mark=square*,
        mark size=1.8pt,
        draw=red!80!black,
        fill=red!60,
    ] coordinates {
        (0.16, 0.54) 
        (0.17, 0.71) 
        (0.21, 0.74) 
        (0.25, 0.79) 
        (0.30, 0.86) 
    };

    \addplot[
        thick,
        dashed,
        color=red!80!black,
        mark=none,
    ] coordinates {
        (0.16, 0.54) 
        (0.17, 0.71) 
        (0.21, 0.74) 
        (0.25, 0.79) 
        (0.30, 0.86) 
    };

    \tikzset{
        every pin edge/.style={draw=black!60, thin, shorten <=2pt},
        every pin/.style={font=\scriptsize, text=black, inner sep=1pt}
    }

    \node[coordinate, pin={-45:iGPT15/OPS46}] at (axis cs:0.16, 0.54) {};
    \node[coordinate, pin={-45:iGPT15/Manual}]   at (axis cs:0.17, 0.7) {};
    \node[coordinate, pin={126:iGEM31f}] at (axis cs:0.21, 0.74) {};
    \node[coordinate, pin={105:FLX2p/Manual}]    at (axis cs:0.25, 0.79) {};
    \node[coordinate, pin={45:FLX2d}]    at (axis cs:0.30, 0.86) {};

    \end{axis}
\end{tikzpicture}

\end{subfigure}
\hspace{0.04\textwidth}
\begin{subfigure}[t]{0.4\textwidth}
\centering

\begin{tikzpicture}
    \begin{axis}[
        width=\linewidth,
        height=2.6in, 
        title style={font=\footnotesize\bfseries, yshift=-1ex},
        xlabel={Contrastive Privacy Resolution (lower is better)},
        ylabel={Utility (higher is better)},
        label style={font=\footnotesize},
        tick label style={font=\scriptsize},
        grid=both,
        grid style={line width=.1pt, draw=gray!30},
        major grid style={line width=.2pt, draw=gray!50},
        xmin=0.1, xmax=0.4,
        ymin=0.6, ymax=0.9,
        legend pos=south east,
        legend style={font=\scriptsize, cells={align=left}},
        clip=false 
    ]

    \addplot[
        only marks,
        mark=*,
        mark size=1.5pt,
        draw=blue!70!black,
        fill=blue!50,
    ] coordinates {
        (0.34, 0.74) 
        (0.22, 0.80) 
        (0.21, 0.79) 
        (0.24, 0.67) 
        (0.19, 0.76) 
        (0.21, 0.79) 
        (0.22, 0.80) 
        (0.23, 0.85) 
        (0.22, 0.79) 
    };

    \addplot[
        only marks,
        mark=square*,
        mark size=1.8pt,
        draw=red!80!black,
        fill=red!60,
    ] coordinates {
        (0.17, 0.74) 
        (0.18, 0.77) 
        (0.19, 0.81) 
        (0.22, 0.82) 
        (0.23, 0.86) 
    };

    \addplot[
        thick,
        dashed,
        color=red!80!black,
        mark=none,
    ] coordinates {
        (0.17, 0.74) 
        (0.18, 0.77) 
        (0.19, 0.81) 
        (0.22, 0.82) 
        (0.23, 0.86) 
    };

    \tikzset{
        every pin edge/.style={draw=black!60, thin, shorten <=2pt},
        every pin/.style={font=\scriptsize, text=black, inner sep=1pt}
    }

    \node[coordinate, pin={-90:GPT4om}] at (axis cs:0.17, 0.74) {};
    \node[coordinate, pin={135:GPT4o}]   at (axis cs:0.18, 0.77) {};
    \node[coordinate, pin={135:GPT54}]   at (axis cs:0.19, 0.81) {};
    \node[coordinate, pin={110:GEM2f}] at (axis cs:0.22, 0.82) {};
    \node[coordinate, pin={0:Haiku-4.5}] at (axis cs:0.23, 0.86) {};

    \end{axis}
\end{tikzpicture}

\end{subfigure}

\caption{Contrastive privacy versus utility (as similarity) for a variety of mechanisms sanitizing: (top) \texttt{the identity of the fast food restaurant} from 49 images; (bottom) \texttt{the movie discussed in the passage} from 49 text passages.
Pareto dominant mechanisms have square markers.
(Note the different axes for each figure.)
}
\Description{Two scatterplots compare privacy resolution with utility. The top plot shows image-sanitization mechanisms for McDonald's data, and the bottom shows text-sanitization mechanisms for Avengers comments. Dashed lines connect each Pareto frontier.}
\label{fig:privacy_performance}
\end{figure}

\section{Introduction}

Public releases of data are commonplace in  many forms, including court documents, military and law enforcement disclosures, acts of Congress, Freedom of Information Act (FOIA) requests, and even social media posts. Also common are failed attempts to privatize these data by way of partial redaction~\cite{bbc2026epstein,lin2015police,deuber2023assessing}. A typical workflow for releasing data involves humans, sometimes professionally trained, poring through media files, redacting sensitive information, and then releasing the remainder. This task is both tedious and delicate; privacy can be easily compromised through phenomena like the soft biometrics attack~\cite{garfinkel2015deidentification}, whereby secondary evidence allows an agent with ancillary knowledge to deanonymize a subject, such as through posture, body shape, or clothing. When the size of the released data is vast, it becomes impractical for humans to enforce strict privacy guidelines. As a result, reports of privacy failures appear frequently in news stories.

A common approach to privacy in the academic literature is to apply techniques from differential privacy, adapted to unstructured data. Yet, in practice, differential privacy is applied infrequently. Indeed, prior work has noted the shortcomings of such approaches~\cite{garrido2023lessons,Cummings:2023,Schneider:2025,li2024visualmixer}. Many have turned instead to machine-assisted privatization using large language models (LLMs) and vision-language models (VLMs)~\cite{monteiro2025imago}. LLMs and VLMs are now able to ``reason'' about privacy vulnerabilities in vast corpora of text and images and, to a more limited extent, video and audio. Several works introduce what we call \emph{semantic sanitization mechanisms} (SSMs), which are privacy mechanisms that operate at a semantic level, identifying specific objects and concepts within a dataset that should be removed to preserve privacy~\cite{monteiro2025imago,chen2025vision,garza2025prvl,wang2026guiguard,dinh2026unsafe2safe}.

Semantic sanitization is appealing because it naturally aligns with privacy as required by law or policy (unlike approaches such as DP). For example, the Seattle Police Department redacts from police-worn body-camera footage ``any identifying marks, including faces, clothes, tattoos, addresses, photos, paperwork, computer screens, names, etc.''~\cite{seattlepd_bodywornvideo}.
And the 
Health Insurance Portability and Accountability Act (HIPAA) lists 16 types of patient identifiers that should be removed before releasing patient data~\cite{hhs_hipaa_deidentification}. 
As AI models become increasingly capable, we anticipate a proliferation of SSMs in the near future. 

A critical limitation of SSMs is the lack of a generally applicable, quantitative, and automated method for evaluating their privacy. The primary challenge is that these mechanisms are as complex as the large models that underlie them, making a statistical model of their privatizing behavior elusive. Recent work~\cite{patwari2024perceptanon,abdulaziz2025evaluation,garza2025prvl} introduces systems that evaluate narrower aspects of SSM privacy. Existing measures commonly rely on cumbersome and error-prone hand-labeling of ground-truth data or combine collections of privacy-suggestive metrics such as anomaly detectors.

\paragraph{Contributions.} In this paper, we introduce \emph{contrastive privacy}, a privacy definition that admits an empirical and largely automated test of sanitized \emph{renderings}, i.e., documents that capture semantically rich concepts, with respect to a well-defined abstract privacy concept whose constituent properties need not be enumerated. The renderings whose sanitized forms are being evaluated constitute the \emph{audited set}. The \emph{reference set} contains renderings known to capture privacy-relevant properties of the concrete concepts selected for sanitization (e.g., face, logo, or title). By comparing each reference before and after sanitization, the test characterizes the removed semantics and checks the audited set for residual associations with them. Each evaluation is relative to a provided \textit{distance mechanism} that captures semantic similarity between two renderings. Candidates for sanitization mechanisms are LLMs for text and image generation models for images. Models like SBERT and CLIP are candidates for distance mechanisms in text and images.

We stress that contrastive privacy is neither limited to text and image modalities, nor is it limited to machine learning models for its mechanisms. Moreover, despite its theoretical underpinnings, our privacy test can readily be applied to real-world privacy problems, which we demonstrate empirically in Section~\ref{sec:experiments}. In fact, for a commonly used family of distance mechanisms, we develop an efficient algorithm for applying the test to finite audited and reference sets.

\paragraph{Media Provenance.} Real images are from Wikimedia Commons; the linked \href{https://github.com/umass-forensics/contrastive-privacy/blob/main/image_licenses.csv}{\texttt{image\_licenses.csv}} file maps each image to its source page and license, with creator credits on the source pages. Displayed sanitizations modify these source images; experimental figures identify the mechanisms used. Reddit comments came from Pushshift~\cite{baumgartner2020pushshift}, and Enron emails from the Klimt--Yang corpus~\cite{klimt2004enron}. AI-generated media and descriptions are disclosed in the AI Use section.

\paragraph{Terminology.} We call the procedure over audited and reference sets the \emph{test} and each ordered comparison of an audited rendering $x$ with a reference rendering $y$ a \emph{trial}. Each trial yields a signed difference and a nonnegative \emph{violation} value; larger violations indicate worse privacy. The \emph{resolution} of the test is the largest violation among all trials. For a candidate resolution, a trial is \emph{satisfied} exactly when its signed difference is smaller than the resolution. The test \emph{passes} at a given candidate resolution if and only if every trial is satisfied. A \emph{witness} for a violation is a concrete trial that exhibits it. A trial attaining the largest violation is a \emph{decisive witness}, since it determines the test's resolution. We describe the practical, model- and corpus-relative test as a \emph{diagnostic}; its result is not an unconditional privacy guarantee.

Within our idealized formal model, we prove conditions under which passing the test makes every rendering in a designated audited set private with respect to the privacy target. The result excludes the degenerate case in which all reference renderings used by the idealized test sanitize identically. Unlike evaluations that check a predetermined list of sensitive properties, our approach audits sanitization with respect to an abstract privacy target without exhaustively enumerating its constituent properties. It instead assumes that the properties selected for sanitization are semantically linked to the target. Given a representative reference proxy, reliable sanitization of those properties, and a distance mechanism sensitive to relevant semantic relationships, a pass indicates that the selected properties suffice for the audited set. Section~\ref{sec:distance_robustness} recommends heterogeneous distance mechanisms for sensitive deployments.

We also demonstrate empirically that the operational test often fails for common SSMs under the evaluated distance mechanisms and on the evaluated datasets. Figure~\ref{fig:privacy_performance} shows privacy resolution, our model-relative measure of privacy, for the evaluated mechanisms, highlighting broad privacy--utility tradeoffs across modalities. Observed resolutions vary by task: privatizing a celebrity identity in images can be easier than privatizing a global brand like McDonald's in images or a movie franchise such as \emph{the Avengers} in text. A separate synthetic-identity study shows that identity-specific fine-tuning causes the test to fail where the base distance mechanism passes; adding clothing to the sanitization concepts returns the adapted test to resolution 0. Finally, each of four off-the-shelf PII sanitizers yields a nonzero resolution on a corpus of 1,841 processed Enron emails selected for mentioning Ken Lay, even with acceptance thresholds that retain every candidate returned by the detectors: for each of the first three, a decisive witness includes a sanitized email that retains his name, while a decisive witness for Cloud DLP includes sanitized emails that omit his name but retain identifying facts. We show how to interpret the observed resolutions using semantic similarities between familiar concepts.

\section{Preliminaries}

In this section, we introduce definitions and vocabulary that help us define the problem and the contrastive privacy approach.
A \emph{rendering} is any representation of something from the real world. For example, a rendering could be an image, a text description, a video, an audio clip, or another \textit{modality}. Let $\mathbb{X}$ be the set of \emph{all possible} renderings of a given modality, and let $\mathbb{P}$ denote the set of all possible real-world \textit{properties} latent in $\mathbb{X}$.\footnote{The set $\mathbb{X}$ could indeed be extremely large. For example, the cardinality of the set of all \texttt{800x600} images with 24-bit pixels is $24^{48 \times 10^4}\!\!$. Thus, we acknowledge upfront that it will typically be impossible to work with $\mathbb{X}$ in its entirety.} When we say that a rendering \emph{captures} a property, we mean that it is possible to algorithmically infer the property from the rendering. For $P \subseteq \mathbb{P}$, we denote by $\mathbb{X}(P)$ the set of renderings that capture $P$. A \emph{concept} $c \subseteq \mathbb{P}$ is any collection of properties, and an \emph{instance} of $c$ is any subset of properties in $c$. By $\mathcal{I}(c)$, we denote the exhaustive, but not necessarily disjoint, set of instances comprising $c$.

For example, the concept $c$, \texttt{smurf}, is a collection of properties \{\texttt{blue skin}, \texttt{red hat}, \texttt{white beard}, \ldots\}. The concept $e$, \texttt{papa smurf}, is a specific instance: a subset of $c$ and an element of $\mathcal{I}(c)$ that may be captured by a rendering.

We differentiate between \emph{natural} and \emph{abstract} instances. A natural instance is any instance whose properties are computationally efficient to enumerate, such as those properties describing \texttt{the face of papa smurf}. An abstract instance is one whose properties cannot be efficiently identified with computational methods. For example, the instance \texttt{anything that can be used to identify papa smurf} is abstract. These properties could include his physical characteristics, the characteristics of his smurf hut, or merely the presence of other smurfs.
Natural and abstract concepts are concepts that comprise natural and abstract instances, respectively. 

\subsection{Threat Model}
\label{sec:threat}

In this paper, we seek to keep private from an adversary a specific abstract instance captured within a set of renderings. To express this formally, let $\widehat{\mathbb{X}}\subseteq\mathbb{X}$ be the set whose sanitized forms we seek to audit, and let $e$ be an abstract instance we seek to sanitize from it. The adversary receives only the sanitized forms of the renderings in $\widehat{\mathbb{X}}$; neither the original audited renderings nor the reference set is released to the adversary. Let $\gamma \subseteq \mathbb{P}$ be a set of natural and abstract properties such that $e \in \mathcal{I}(\gamma)$. We are limited to obfuscating computationally identifiable concepts $c\subseteq \mathbb{P}$, and we seek conditions under which sanitizing $c$ makes every rendering in $\widehat{\mathbb{X}}$ private with respect to $\gamma$. The adversary's goal is to learn what $e$ is from this release.

Our experiments assume a controlled audit: a trusted auditor selects the audited and reference sets and distance mechanism and runs the test on outputs of non-adversarial SSMs. Such an SSM may be inaccurate, and its configuration may be adaptively refined, as in the \texttt{Manual} condition in Section~\ref{sec:personal_privacy}, but it is not designed to pass the test while deliberately retaining target semantics. This models ordinary audits of commercial and locally deployed sanitizers. We leave analysis of adversarial SSMs to future work.

\subsection{Assumptions}

In the sequel, we make the following fundamental assumptions.
\begin{enumerate}
    \item \textbf{Efficiency.} There exist  computationally efficient mechanisms for:
    \begin{enumerate}
        \item identifying and representing the set of properties that comprise a natural instance;
        \item identifying and removing instance $e$ or replacing $e$ by instance $e'$, transforming rendering $x$ into $x'$; and
        \item calculating semantic similarity between renderings.
    \end{enumerate}
    \item \textbf{Inefficiency.} Generally, there is no computationally efficient mechanism for enumerating the properties comprising an abstract concept or an instance of that concept.
    \item \textbf{Representability.} Corresponding to every property $p \in \mathbb{P}$ there exists a rendering in $\mathbb{X}$ that captures $p$.
\end{enumerate}

\section{Contrastive Privacy Formulation}
\label{sec:priv_formulation}

\subsection{Privacy for Complex Data}

Ensuring the privacy of unstructured information is a complex problem. In this paper, we ground our definition of privacy with respect to appropriate information flows following Contextual Integrity (CI)~\cite{nissenbaum2004privacy}. CI postulates that an appropriate information flow conforms to contextual norms, which differ across contexts and are defined through the following parameters: \textit{subject}, \textit{sender}, \textit{recipient}, \textit{information type}, and \textit{transmission principle}. Once the contextual norm is captured, it can be evaluated for appropriateness. This evaluation is usually straightforward once all parameters are accurately identified. CI clearly captures, for example, an information flow in which the police (sender) share an image (information type) of a given individual (subject) with the public (receiver) as prescribed by a court of law (transmission principle).

However, in an unstructured setting, a single rendering might capture a vast number of appropriate information flows. It could contain ancillary information (other individuals or objects) that might be inappropriate to share because doing so could reveal a victim's identity or other PII. Therefore, privacy measures that focus on a single appropriate information flow are inadequate. What is needed is a method designed to test the many semantic connections that may contribute to such information flows. Contrastive privacy is designed for this purpose.

\subsection{Problem Formulation}
\label{sec:problem}

Next, we define privacy for SSMs and prove, under our stated assumptions, conditions under which sanitizing a concept $c$ privatizes the abstract target $\gamma$. We evaluate the test with respect to a privacy \emph{resolution} parameter $\delta \in \mathbb{R}_{\geq 0}$, where $\delta = 0$ imposes the strongest test and the numerical scale is inherited from the distance mechanism. Within any application of the definitions and results below, semantic connectedness, semantic closure, and contrastive privacy are evaluated using the same $\delta$; we suppress this shared dependence in the notation for readability.

\begin{definition}
\label{def:distance}
    A \emph{distance mechanism} $\mathcal{D}(x, y)$ is a deterministic algorithm that accepts renderings $x, y \in \mathbb{X}$ and returns a distance with respect to their semantic similarity. It has the following properties.\\
    \begin{itemize}
        \item Symmetry: $\forall x,y \in \mathbb{X}$, $\mathcal{D}(x, y) = \mathcal{D}(y, x)$.
        \item Proximity: $\forall x\in \mathbb{X}$, $\mathcal{D}(x, x) = 0$.
        \item Nonnegativity: $\forall x,y \in \mathbb{X}$, $\mathcal{D}(x, y) \geq 0$.
    \end{itemize}
    We do not require $\mathcal{D}(x,y) > 0$ when $x \neq y$; mechanisms based on learned embeddings may map distinct renderings to the same representation.
\end{definition}

\begin{definition}
A \emph{privacy mechanism}, $\mathcal{X}: 2^\mathbb{P} \times \mathbb{X} \rightarrow \mathbb{X}$, is any deterministic algorithm that takes a subset of properties and a rendering and returns a modified rendering. When the first argument $\gamma \subseteq \mathbb{P}$ is fixed, we write $\mathcal{X}_\gamma$.
\end{definition}

\begin{definition}
\label{def:obfuscation}
    For any $\gamma \subseteq \mathbb{P}$, $\mathcal{X}$ is a \emph{$\gamma$-privacy mechanism} if $\mathcal{X}_\gamma$ sanitizes from $x\in\mathbb{X}$ every instance $e \in \mathcal{I}(\gamma)$; $e$ can be obfuscated or replaced with another instance $e' \not \in \mathcal{I}(\gamma)$. In either case, we assume idempotence: $\mathcal{X}_\gamma(\mathcal{X}_\gamma(x)) = \mathcal{X}_\gamma(x)$.
\end{definition}

When the renderings $\mathbb{X}$ are images, sanitizing could be as simple as replacing concepts $c$ (e.g., a face) with black pixels. Alternatively, it can be as complex as a mechanism that provides differential privacy for those concepts~\cite{xue2021dp}.
If $\gamma$ is an abstract concept, then $\mathcal{X}_\gamma(x)$ will generally not be computable (according to our Inefficiency Assumption). Nevertheless, the results in this section allow us to reason about $\mathcal{X}_\gamma(x)$ without directly computing it.

\begin{definition}
\label{def:capture}
    We say that a rendering $x \in \mathbb{X}$ \emph{captures} an arbitrary concept $\gamma \subseteq \mathbb{P}$, denoted $x \in \mathbb{X}(\gamma)$, if $x \neq \mathcal{X}_{\gamma}(x)$. Equivalently, $x \in \mathbb{X}(\gamma)$ if and only if $x \in \mathbb{X}(\{q\})$ for some $q \in \gamma$.
\end{definition}

\begin{definition}[Privacy of $x$ with respect to $\gamma$]
\label{def:priv_x}
    For any concept $\gamma \subseteq \mathbb{P}$, rendering $x \in \mathbb{X}$ is \emph{private} with respect to $\gamma$ provided that $x \not \in \mathbb{X}(\gamma)$, i.e., provided that $x = \mathcal{X}_\gamma(x)$.
\end{definition}

\noindent In other words, if sanitizing a concept from an image $x$ does not produce the same image, then the concept is captured in the original image.

With there generally being no efficient way to compute $\gamma$, how can we actually test if $x = \mathcal{X}_\gamma(x)$, i.e., if $x$ is private with respect to $\gamma$? The remainder of this section answers this question by reasoning about the relationship between concepts $c$ that we \emph{can} compute and concepts $\gamma$ that we cannot.

\begin{definition}
\label{def:sem_conn_det}
    Sets of properties $c, d \subseteq \mathbb{P}$ are said to be \emph{semantically connected by $\mathcal{X}$ and $\mathcal{D}$}, denoted $c \sim d$, if for each $x \in \mathbb{X}(c)$ and $y \in \mathbb{X}(d)$,
    \begin{equation}
    \label{eq:conn1}
        \mathcal{D}(x, y) + \delta < \mathcal{D}(\mathcal{X}_{c}(x), y)
        \quad\text{whenever }\mathcal{X}_{c}(x)\neq y,
    \end{equation}
    and
    \begin{equation}
    \label{eq:conn2}
        \mathcal{D}(x, y) + \delta < \mathcal{D}(x, \mathcal{X}_{d}(y))
        \quad\text{whenever }x\neq\mathcal{X}_{d}(y).
    \end{equation}
\end{definition}

See Appendix~\ref{app:sc} for illustrative examples comparing oranges, apples, and iPhones. The intuition behind Definition~\ref{def:sem_conn_det} is that renderings $x$ and $y$ are more similar when concepts $c$ and $d$ are both included than when one is omitted. This implies that there is some semantic similarity between them.

\begin{definition}
\label{def:semantic_closure}
    For any privacy mechanism $\mathcal{X}$, distance mechanism $\mathcal{D}$, and properties $P \subseteq \mathbb{P}$, the \emph{semantic closure of $P$ with respect to $\mathcal{X}$ and $\mathcal{D}$}, denoted $\texttt{cl}(\mathcal{X}, \mathcal{D}, P)$, is the set of all properties that are semantically connected to $P$. Specifically,
    \begin{equation}
        \texttt{cl}(\mathcal{X}, \mathcal{D}, P) = \{q \in \mathbb{P}: P \sim \{q\}\}.
    \end{equation}
\end{definition}

From our definitions, we can test the efficacy of a computationally efficient $c$-privacy mechanism that we use to privatize concepts $\gamma$. We define this as \emph{contrastive privacy}, and apply the test throughout the remainder of this paper.

\begin{definition}[Contrastive privacy of $\mathcal{X}_c$ for $\widehat{\mathbb{X}}$ with respect to $\gamma$]
\label{def:privacy_det}
    Let $\gamma$ be any concept, $c$ a natural concept such that \mbox{$\gamma \subseteq \texttt{cl}(\mathcal{X}, \mathcal{D}, c)$}, and $\widehat{\mathbb{X}}\subseteq\mathbb{X}$ a set of renderings to be audited. We say that $\mathcal{X}_c$ offers \emph{contrastive privacy for $\widehat{\mathbb{X}}$ with respect to $\gamma$} provided that, for every $x \in \widehat{\mathbb{X}}$ and $y \in \mathbb{X}(\gamma \cap c)$,
    \begin{equation}
    \label{eq:priv_def1}
        \mathcal{D}(\mathcal{X}_c(x), y) + \delta > \mathcal{D}(\mathcal{X}_c(x), \mathcal{X}_c(y)).
    \end{equation}
\end{definition}
The test is generally asymmetric: each trial $(x,y)$ lists the \emph{audited rendering} $x$ first and the \emph{reference rendering} $y$ second, and $x$ appears only as $\mathcal{X}_c(x)$. We call the test \emph{symmetric} when its audited and reference sets coincide. Thus, Definition~\ref{def:privacy_det} is symmetric when $\widehat{\mathbb{X}} = \mathbb{X}(\gamma \cap c)$. If $\delta$ is not fixed, satisfaction at $\delta \geq 0$ gives \emph{contrastive privacy at resolution $\delta$}; resolution 0 is ideal.

\paragraph{Violation and Resolution.} For trial $(x,y)$, define
\begin{equation*}
v(x,y)=\max\!\left\{0,\mathcal{D}(\mathcal{X}_c(x),\mathcal{X}_c(y))-\mathcal{D}(\mathcal{X}_c(x),y)\right\}.
\end{equation*}
The trial is satisfied exactly when the difference inside the maximum is less than $\delta$. The test's resolution is
\begin{equation*}
    \sup\!\left\{v(x,y)\;\middle|\;
    \begin{aligned}
        x &\in \widehat{\mathbb{X}},\\[-2pt]
        y &\in \mathbb{X}(\gamma\cap c)
    \end{aligned}
    \right\}.
\end{equation*}
Over finite audited and reference sets, this is a maximum attained by a decisive witness. Because Inequality~\ref{eq:priv_def1} is strict, the resolution is a boundary value (an infimum); any larger $\delta$ satisfies every trial.

Most importantly, Definition~\ref{def:privacy_det} permits an audited set of sanitized renderings to be tested for properties in $\gamma$ without enumerating those properties or computing $\mathcal{X}_\gamma$. It relies on original and sanitized references and assumes $\gamma \subseteq \texttt{cl}(\mathcal{X}, \mathcal{D}, c)$ under the chosen mechanisms and rendering domain. Definition~\ref{def:priv_x}, by contrast, describes the privacy of an individual rendering.

\begin{figure}[!t]
\begin{center}
\begin{subfigure}[t]{0.9\columnwidth}
\centering
\includegraphics[width=0.9\linewidth]{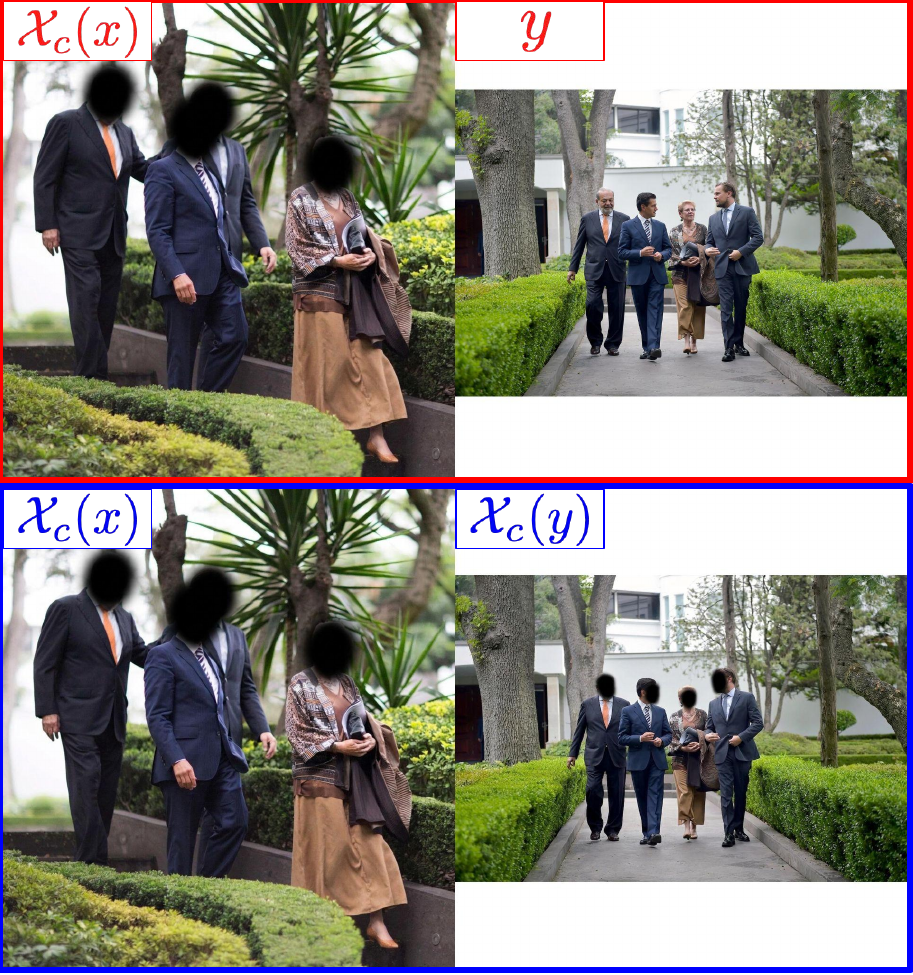} 
\caption{Trial satisfied at $\delta=0$ under \texttt{EVA}.}\label{fig:dicaprio_success}
\end{subfigure}
\begin{subfigure}[t]{0.9\columnwidth}
\centering
\includegraphics[width=0.9\linewidth]{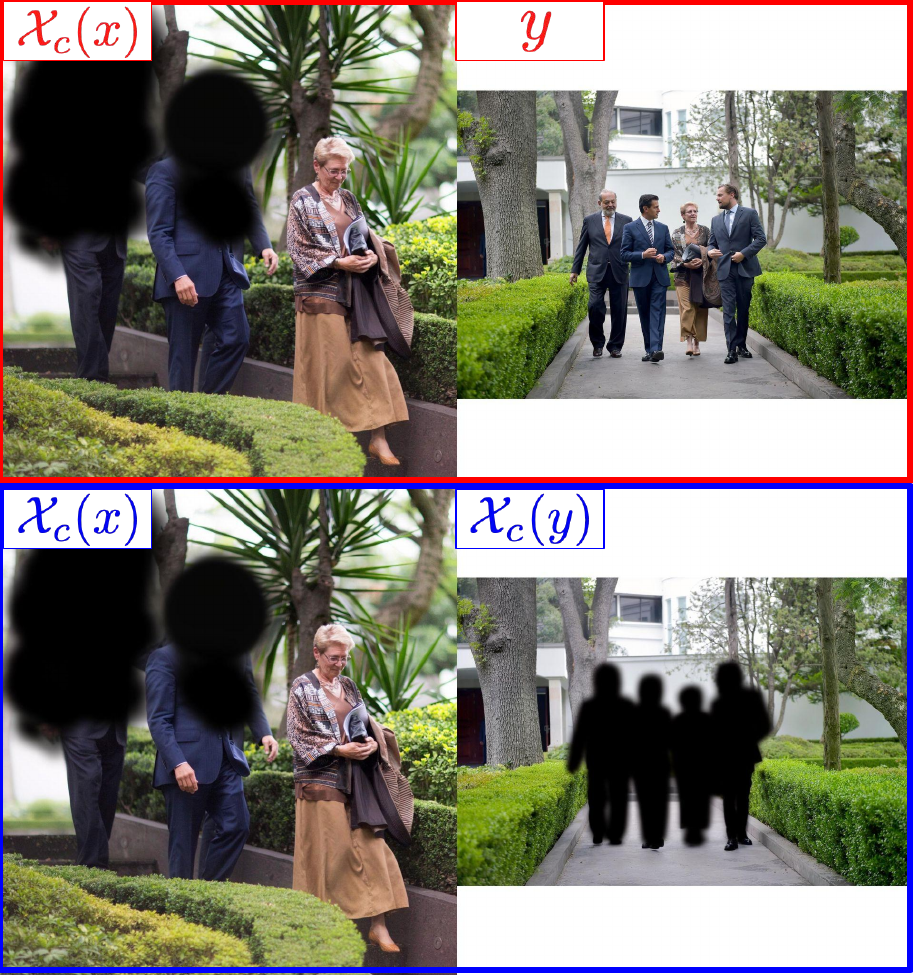}
\caption{Trial with violation 0.10 under \texttt{EVA}. In $\mathcal{X}_c(x)$, a celebrity's face is retained.}\label{fig:dicaprio_failure}
\end{subfigure}
\caption{Illustration of Inequality~\ref{eq:priv_def1}, ${\color{red}\mathcal{D}(\mathcal{X}_c(x), y)} +\delta > {\color{blue}\mathcal{D}(\mathcal{X}_c(x), \mathcal{X}_c(y))}$ (red and blue for clarity), when sanitizing \texttt{celebrities} and measuring distance with \texttt{EVA-CLIP-18B}. At $\delta=0$, trial (a)~is satisfied for \texttt{iGEM31f}, while trial (b)~has violation 0.10 for \texttt{FLX2d/GEM31p}.}
\Description{Two three-image trials compare a sanitized audited image with an original and sanitized reference image. The first trial is satisfied at resolution zero; in the second, the audited sanitization retains a celebrity's face and has violation 0.10.}
\label{fig:dicaprio_privacy_failures}
\end{center}
\end{figure}

\subsection{Application to CLIP / Cosine Similarity}\label{sec:application_to_clip}
To demonstrate the applicability of contrastive privacy, we apply it here to CLIP as a concrete example. (N.B., this example also applies to other embedding models that use the same metric.) For each rendering $x \in \mathbb{X}$, let $\mathcal{E}(x)$ denote its unit-normalized embedding in a space that uses cosine similarity as its metric. 
  
\begin{example}
\label{ex:clip}
   We define 
    \begin{equation}
    \label{eq:clip_distance}
        \mathcal{D}(x, y) \coloneqq 1 - \cos(\mathcal{E}(x), \mathcal{E}(y)).         
    \end{equation}
    The embedding map need not be injective, so distinct renderings may have identical embeddings and distance zero, as permitted by Definition~\ref{def:distance}. In this context, we have the following equivalent expressions.
    \begin{equation}
    \label{eq:clip_test}
        \begin{array}{c}
            \mathcal{D}(\mathcal{X}_c(x), y) + \delta> \mathcal{D}(\mathcal{X}_c(x), \mathcal{X}_c(y)) \\
            1\!-\!\mathcal{E}(\mathcal{X}_c(x)) \cdot \mathcal{E}(y) + \delta> 1\!-\!\mathcal{E}(\mathcal{X}_c(x))\!\cdot\!\mathcal{E}(\mathcal{X}_c(y)) \\
            \mathcal{E}(\mathcal{X}_c(x)) \cdot \mathcal{E}(y) < \mathcal{E}(\mathcal{X}_c(x)) \cdot \mathcal{E}(\mathcal{X}_c(y)) + \delta \\
            \mathcal{E}(\mathcal{X}_c(x)) \cdot \big(\mathcal{E}(y) - \mathcal{E}(\mathcal{X}_c(y))\big) < \delta 
        \end{array}
\end{equation}

\noindent Vector $\mathcal{E}(y) - \mathcal{E}(\mathcal{X}_c(y))$ is a $\mathcal{D}$-specific direction encoding the semantics of the portion of $c$ captured by $y$. Its inner product with $\mathcal{E}(\mathcal{X}_c(x))$ is the signed difference for trial $(x,y)$ and measures how much the sanitized $x$ retains those semantics. If it is below $\delta$ for every audited $x$ and reference $y$, the test detects no such residual encoding at that resolution.

Figure~\ref{fig:dicaprio_privacy_failures} illustrates two trials when sanitizing the concept $c =$ \texttt{celebrities} with different SSMs and using the \texttt{EVA} CLIP distance mechanism. (Each is discussed in detail in Section~\ref{sec:personal_privacy}.) Figure~\ref{fig:dicaprio_success} shows a trial satisfied at $\delta = 0$, whereas Figure~\ref{fig:dicaprio_failure} shows one with positive violation. In each trial, Inequality~\ref{eq:priv_def1} compares the sanitized audited image $\mathcal{X}_c(x)$ on the left with the concept removed from the reference image on the right, as represented by the difference between $y$ and $\mathcal{X}_c(y)$.
\end{example}

\subsection{Sufficiency of \texorpdfstring{$c$}{c} to Privatize \texorpdfstring{$\gamma$}{gamma}}

\begin{definition}
\label{def:diverse_representation}
    A natural concept $c$ is \emph{diversely represented with respect to $\gamma$} by a $c$-privacy mechanism $\mathcal{X}_c$ if there exist $y_1, y_2 \in \mathbb{X}(\gamma\cap c)$ whose sanitized forms remain distinct:
    \begin{equation*}
        \mathcal{X}_c(y_1) \neq \mathcal{X}_c(y_2).
    \end{equation*}
\end{definition}

Suppose that $z \in \widehat{\mathbb{X}}$ but $\mathcal{X}_c(z)$ retains a property $q\in\gamma$. Diverse representation supplies a reference $y \in \mathbb{X}(\gamma\cap c)$ whose sanitized form differs from $\mathcal{X}_c(z)$. Because $q$ lies in the semantic closure of $c$, semantic connectedness makes $\mathcal{X}_c(z)$ closer to $y$ than to $\mathcal{X}_c(y)$, while contrastive privacy requires the opposite. The test can therefore surface audited renderings that remain nonprivate. The following theorem formalizes this intuition for every rendering in the audited set $\widehat{\mathbb{X}}$, without requiring those renderings themselves to capture a property in $\gamma \cap c$.

\begin{theorem}
\label{thm:privacy_implies_no_capture}
    Fix abstract concept $\gamma$, natural concept $c$, and audited set $\widehat{\mathbb{X}}\subseteq\mathbb{X}$ meeting the criteria of Definition~\ref{def:privacy_det}, with $\gamma \cap c \neq \emptyset$. Suppose that $c$ is diversely represented with respect to $\gamma$ by the $c$-privacy mechanism $\mathcal{X}_c$. If $\mathcal{X}_{c}$ offers contrastive privacy for $\widehat{\mathbb{X}}$ with respect to $\gamma$, then $\mathcal{X}_c(z)$ is private with respect to $\gamma$ for every $z \in \widehat{\mathbb{X}}$.
\end{theorem}

\begin{proof}
    Suppose that $s = \mathcal{X}_c(z)$ is not private with respect to $\gamma$ for some $z \in \widehat{\mathbb{X}}$. Then $s \in \mathbb{X}(\{q\})$ for some $q \in \gamma$. Definition~\ref{def:diverse_representation} supplies $y_1,y_2 \in \mathbb{X}(\gamma \cap c)$ with distinct sanitized forms, so one, denoted $y$, satisfies $\mathcal{X}_c(y) \neq s$. Since $y \in \mathbb{X}(c)$ and $q \in \gamma \subseteq \texttt{cl}(\mathcal{X}, \mathcal{D}, c)$, we have $c \sim \{q\}$. Applying Inequality~\ref{eq:conn1} and symmetry of $\mathcal{D}$ gives
    \begin{equation}
        \mathcal{D}(s,y) + \delta < \mathcal{D}(s,\mathcal{X}_c(y)).
    \end{equation}
    However, $z \in \widehat{\mathbb{X}}$ and $y \in \mathbb{X}(\gamma \cap c)$. Therefore, contrastive privacy for $\widehat{\mathbb{X}}$ with respect to $\gamma$ requires
    \begin{equation}
        \mathcal{D}(s,y) + \delta > \mathcal{D}(s,\mathcal{X}_c(y)),
    \end{equation}
    a contradiction.
\end{proof}

Our formal analysis concerns the concept actually and consistently sanitized, which need not equal the full concept $c$ requested of an implementation. Suppose that imperfect communication with an SSM yields a $c'$-privacy mechanism for some fixed natural concept $c' \subseteq c$. If $c'$ meets the hypotheses of Theorem~\ref{thm:privacy_implies_no_capture} for the audited set $\widehat{\mathbb{X}}$, then the theorem applies directly: passing the test implies that $\mathcal{X}_{c'}(z)$ is private with respect to $\gamma$ for every $z \in \widehat{\mathbb{X}}$. Thus, failure to sanitize $c \setminus c'$ does not itself invalidate the rendering-level privacy conclusion. This reasoning does not directly apply when the effective sanitized concept varies by rendering without a fixed core, and a practical result remains conditional on the representativeness of the finite reference set and the fidelity of $\mathcal{X}_{c'}$ and $\mathcal{D}$.

This formulation also separates mechanism failure from concept misspecification. If $c$ contains only properties inside a rendered face, a target-relevant silhouette may remain even when $\mathcal{X}_c$ removes $c$ perfectly. The resulting test failure indicates that the effective concept $c$ should be broadened (for example, by directing the SSM to blur a larger region) rather than that $\mathcal{X}_c$ failed to remove the concept it was given.

\section{Operationalizing Contrastive Privacy}

In this section, we \1~develop a finite, corpus-relative reference test; \2~explain semantic closure and $\delta$; \3~discuss limitations of direct LLM and VLM assessments; \4~give practical considerations; and \5~provide an audit workflow and efficient algorithm.

\subsection{Proxy Sets}
\label{sec:proxy_sets}

Let $\widehat{\mathbb{X}} \subseteq \mathbb{X}$ be the audited set. In practice, a finite \emph{proxy set} $\widetilde{\mathbb{X}} \subseteq \mathbb{X}(\gamma\cap c)$ approximates the full reference domain. Enlarging it adds trials and cannot lower the observed resolution; an incomplete proxy may omit a decisive reference and yield an optimistic result. Every audited rendering is still tested, but absent reference properties remain untested. If $\widehat{\mathbb{X}} \subseteq \mathbb{X}(\gamma \cap c)$ is itself a reasonable proxy, the symmetric choice $\widetilde{\mathbb{X}}=\widehat{\mathbb{X}}$ is desirable: every item fills both roles, and all within-corpus ordered pairs are tested.

\subsection{Semantic Closure and the Privacy Parameter}
\label{sec:privacy_param}

In addition to requiring Inequality~\ref{eq:priv_def1} to be satisfied, Definition~\ref{def:privacy_det} requires $\gamma \subseteq \texttt{cl}(\mathcal{X}, \mathcal{D}, c)$. The algorithm exhausts the audited set but only samples the reference domain, and in practice we assume closure for the chosen mechanism and domain. This is often plausible by construction: $c$ may include a target person's face or a target brand's logo. Each application should justify why every target property is expected to be semantically connected to $c$ and how that premise might fail; an operational pass remains conditional on it.

Increasing $\delta$ makes the test easier to pass but the semantic-connectedness inequalities harder to satisfy, shrinking the semantic closure of $c$. Maintaining $\gamma$ within this closure may therefore require a broader $c$. Independently, a larger $\delta$ reduces sensitivity to semantic similarities. This tradeoff motivates the term privacy \emph{resolution}.

\subsection{Large Language Models Give an Incomplete Privacy Picture}

A natural approach to assessing the privacy of SSMs is to use large language models (LLMs) or vision-language models (VLMs). The idea is to provide an input text string or image that is believed to be obfuscated from anything that reveals a given concept and prompt the model to decide whether privacy has been achieved. This approach is highly effective up to a point. Nevertheless, it ultimately fails to offer a reliable privacy measure because of the stochasticity of the models and the ambiguity and qualitative nature of language itself.

To demonstrate the shortcomings of this approach, consider the problem of deciding whether the image in Figure~\ref{fig:obs_dicaprio} captures the abstract concept \texttt{anything that reveals the identity of Leonardo DiCaprio}. We gave this image and the following two prompts to Gemini 3 Pro.

\begin{figure}
    \centering
    \includegraphics[width=0.8\linewidth]{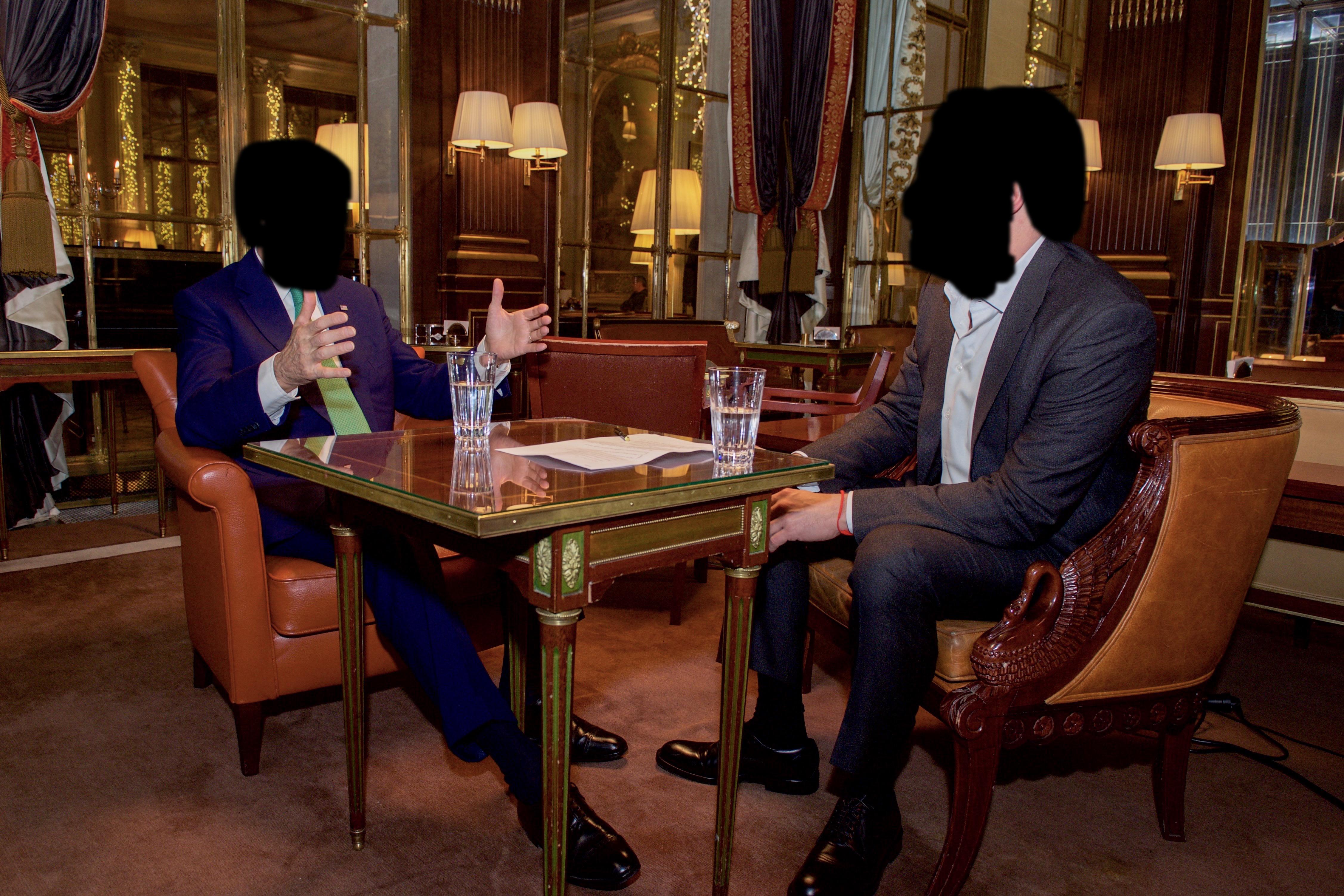}
    \caption{A sanitized image of John Kerry and  Leonardo DiCaprio.}
    \Description{Two men sit across a table in an ornate room. Solid black regions obscure both heads, while their clothing, bodies, and the room remain visible.}
    \label{fig:obs_dicaprio}
\end{figure}

\begin{prompt}
\label{prompt:dicaprio1}
    Is there anything in this image that can identify Leonardo DiCaprio?
\end{prompt}

\begin{prompt}
\label{prompt:dicaprio2}
    Can Leonardo DiCaprio be identified conclusively from this image?
\end{prompt}

To Prompt~\ref{prompt:dicaprio1} Gemini responded ``Based on the visual clues and context of this specific photograph, the person on the right can be identified as Leonardo DiCaprio.'' And to Prompt~\ref{prompt:dicaprio2}, Gemini responded, ``Based on the image provided, no, Leonardo DiCaprio cannot be conclusively identified.'' In both cases, Gemini went on to identify several of the same characteristics, such as the presence of a red bracelet and the likely location, as reasons for (Prompt~\ref{prompt:dicaprio1}) or against (Prompt~\ref{prompt:dicaprio2}) the identifiability of the actor. 
Thus, with only small changes to the prompt, the model's privacy determination changed dramatically. In Prompt 3, we asked the model to identify the two people in the image without offering the answer.

\begin{prompt}
\label{prompt:dicaprio3}
    Prompt 3: Identify the people in this image. 
\end{prompt}

To Prompt~\ref{prompt:dicaprio3}, Gemini responded ``The two people in this image are Nigel Farage (on the left) and Tucker Carlson (on the right).'' The model is capable of finding reasons to justify the presence of Leonardo DiCaprio, \emph{a posteriori}, but it struggles to produce the actor as a candidate \emph{a priori}. 

It is tempting to use Prompt~\ref{prompt:dicaprio3} as a binary test for privacy: if it does not successfully identify the instance in question, then the obfuscated image is private; otherwise, it is not. There are three problems with this approach. First, LLMs and VLMs are inherently stochastic. In our experiments, Gemini 3 Pro returned different guesses for the people in the image when the query was repeated. How should we interpret the privacy of the obfuscation if the model successfully identified Leonardo DiCaprio on the second, third, or fourth attempt?
Second, this approach requires that we know upfront the instance whose privacy we seek. This was the case in our experiment but need not always be. Consider, for example, the task of privatizing a collection of images of celebrities when we do not know which specific celebrities appear in any given image. Third, there is no way to quantify the extent of privacy loss if the obfuscated image is deemed to have compromised privacy.

\subsection{Practical Considerations}

\subsubsection{Concept Ambiguity Penalty}
\label{sec:conc_ambig}

Consider a person named Bob whose privacy, in reality, can be completely protected by obfuscating his face. Obfuscating all faces should privatize the instance of \texttt{Bob}. However, suppose that there exist several images where Bob appears next to another individual whose face has a very large semantic closure, such as Santa Claus. Then each of those images will require more comprehensive sanitization because properties of Santa Claus, other than his face, will be semantically connected with other images of Bob that also capture Santa Claus, such as all Christmas trees, even though Bob is not associated with Christmas. This challenge arises because we have chosen to obfuscate \emph{more} than is actually necessary (all faces, rather than Bob's face). 

This problem arises because the privacy mechanism $\mathcal{X}_c$ is a blunt tool that removes all instances of concept $c$ --- it is not possible for our approach to target the privacy of one instance of an abstract concept $e \in \gamma$ but not another. Semantically rich instances of $c$ in either an audited or reference rendering can determine the resolution. The test is all-or-nothing over the configured sets: every ordered trial must be satisfied, so the combinations most difficult to sanitize determine whether contrastive privacy holds. We refer to this as the \emph{concept ambiguity penalty}.

Approaches to coping with concept ambiguity include:
\begin{enumerate}
     \item Audited renderings that participate in witness trials arising from the concept ambiguity penalty could be flagged and ultimately not released.
    \item For audited renderings $x \in \widehat{\mathbb{X}}$ that participate in such trials, additional properties $d \subseteq \mathbb{P}$ could be obfuscated from $x$ \emph{prior} to applying the privacy test.
    \item The concept $c$ could be \emph{edited} to remove semantically rich instances that are unrelated to the instances of interest.
\end{enumerate}
Approach 1 is the least desirable because it entirely fails to release some renderings.

Approach 2 would involve obfuscating more features in renderings that participate in unsatisfied trials, such as clothing or a silhouette. All audited renderings are obfuscated minimally by $c$, and sometimes by $d$. Thus, any contrastive-privacy conclusion remains tied to the semantic closure of $c$, not $d$.

Approach 3 gives the most desirable result in that it both allows for all images to be released and avoids unnecessarily obfuscating concepts. In our example, it amounts to avoiding the obfuscation of any portion of an image related to Santa Claus. However, the approach also presents an additional technical hurdle: solving the open entity recognition problem. In addition to custom solutions~\cite{hu2023open,wu2020scalable}, LLMs and VLMs are also good at solving this problem in practice. 

\subsubsection{Choosing Concepts \texorpdfstring{$c$}{c}}

The present work gives no insight into \emph{how} to achieve contrastive privacy, that is to say, how to choose $c$ or how to sanitize a rendering $x$ with respect to $c$. This amounts to discovering an effective SSM, which can be difficult. Assuming those details are provided, Theorem~\ref{thm:privacy_implies_no_capture} gives a conditional sufficiency result within the formal model for using a natural concept $c$ to sanitize a semantically related, but distinct, and in principle unknown, abstract concept $\gamma$.

\subsubsection{The Reference Domain Must Be Approximated}

In practice, the reference domain $\mathbb{X}(\gamma\cap c)$ in Definition~\ref{def:privacy_det} must be approximated by a proxy set $\widetilde{\mathbb{X}}$, as defined in Section~\ref{sec:proxy_sets}. The test is exhaustive over the specified finite audited set but corpus-relative with respect to its references, so the result is not an unconditional privacy guarantee. Differential privacy (DP)~\cite{dwork2014algorithmic}, by contrast, gives a formal, parameterized guarantee under its stated neighboring relation and threat-model assumptions. Contrastive privacy is intended for a different setting and supplies a direct semantic interpretation and diagnostic witnesses, but it inherits the blind spots of the mechanisms $\mathcal{D}$ and $\mathcal{X}_c$.

\subsubsection{Blind Spots and Adaptation of \texorpdfstring{$\mathcal{D}$}{D}}
\label{sec:distance_robustness}

We do not evaluate SSMs designed to evade a known or queryable distance mechanism; such adversarial SSMs and corresponding defenses remain future work.

For sensitive deployments, we recommend an ensemble of heterogeneous distance mechanisms, such as embedding models from different families, supplemented by components such as OCR. The test should be run separately for each member. If resolution 0 is required, every member must pass at 0. At nonzero resolutions, each member's resolution should be reported and decisive witnesses should be provided separately rather than averaging across models. Diversity can help identify model-specific blind spots, though shared blind spots may remain. Table~\ref{tab:dicaprio_distance_agreement} provides an initial empirical illustration using three image distance mechanisms.

The distance mechanism $\mathcal{D}$ must also encode target-relevant semantic relationships and may fail when the abstract concept is absent from its training data.
For example, a pretrained vision-language embedding model may not have seen an obscure individual. Fine-tuning on that individual's images may therefore be needed to make $\mathcal{D}$ suitable.
One approach to addressing a model's lack of familiarity with the privacy target is to fine-tune on media containing the target concept with coarse labels, such as \texttt{a picture of John Doe} or \texttt{a picture containing protected health information}, without property-level annotations. \textbf{Detailed lists of natural concepts that reveal those abstract concepts need not be provided.}
Section~\ref{sec:finetune_unseen} evaluates a simple identity-specific instance of this strategy; we leave broader comparison of fine-tuning and ensemble strategies across targets for future work.

\subsection{Contrastive Privacy Audit Workflow}
\label{sec:calc_privacy}

A contrastive privacy audit should proceed iteratively: (1) define $\gamma$, $\widehat{\mathbb{X}}$, and a natural $c$ whose closure plausibly contains $\gamma$; (2) document why this premise should hold and how it might fail; (3) construct a representative proxy set for the reference domain and choose $\mathcal{D}$, using an ensemble for sensitive deployments; (4) sanitize both sets, run the test, inspect its resolution and decisive witness, and reassess closure; and (5) if needed, broaden $c$ and resanitize, improve the proxy or $\mathcal{D}$, or withhold implicated audited renderings, and rerun.
Resolution 0 means only that no violation was detected for that audited set against that proxy under those mechanisms, not that privacy is unconditional. A positive resolution is the largest violation observed across all trials. Every trial satisfies the weak test inequality at this value, and the strict test passes at any larger value; the trial attaining it is returned as the decisive witness.

\begin{algorithm}[H]
\caption{Contrastive Privacy Test for Cosine Distance
\label{alg:clip_priv}}
\begin{algorithmic}
\Require audited set $\widehat{\mathbb{X}}$
\Require reference proxy $\widetilde{\mathbb{X}} \subseteq \mathbb{X}(\gamma\cap c)$
\State $S,T \gets [\,]$ \Comment{matrices with row vectors}
\For{$x \in \widehat{\mathbb{X}}$}
    \State append $\mathcal{E}(\mathcal{X}_c(x))$ to $S$
\EndFor
\For{$y \in \widetilde{\mathbb{X}}$}
    \State append $\mathcal{E}(y) - \mathcal{E}(\mathcal{X}_c(y))$ to $T$
\EndFor
\State $(a_{\max},i^*,j^*) \gets \texttt{maxWithIndex}(ST^\top)$
\State $\delta^* \gets \max(0,a_{\max})$
\State \Return $\delta^*,(i^*,j^*)$
\end{algorithmic}
\end{algorithm}

\noindent\textbf{Efficient Calculation.} Let $m = |\widehat{\mathbb{X}}|$, $n = |\widetilde{\mathbb{X}}|$, and $d$ be the embedding dimension. Algorithm~\ref{alg:clip_priv} stores sanitized audited embeddings as rows of $S$ and reference differences $\mathcal{E}(y)-\mathcal{E}(\mathcal{X}_c(y))$ as rows of $T$. Entry $(i,j)$ of $ST^\top$ is the signed difference for trial $(x_i,y_j)$ in Inequality~\ref{eq:clip_test}; rows of $T$ are not normalized. The algorithm returns boundary resolution $\delta^* = \max(0,a_{\max})$ and a decisive witness. Every trial satisfies the weak inequality at $\delta^*$, and the strict test passes for every $\delta > \delta^*$; if $a_{\max} < 0$, it also passes at 0.

In practice, rows of $S$ can be processed in batches. This uses standard matrix multiplication, avoids storing the full $m \times n$ matrix of signed differences, and retains a trial with the largest difference as a decisive witness. After embedding the sanitized audited renderings and the original and sanitized references, the exact test performs $O(mnd)$ arithmetic operations and uses $O((m+n)d)$ working memory. Computing the returned resolution and a decisive witness requires processing every trial.

For much larger audited or proxy sets, approximate inner-product search, such as HNSW~\cite{malkov2018efficient}, may serve as a preliminary screening step that finds trials with large signed differences without evaluating every trial. An HNSW index is built over the reference rows $t$ of $T$, and each sanitized audited embedding $s$ is used as a query. The exact difference $s \cdot t$ is then recomputed for the top-$k$ candidates returned by the index. The largest difference found yields a lower bound on the exact resolution and a candidate witness; increasing $k$ or the search depth improves recall at additional cost. Because HNSW may omit the true maximum, its approximate resolution may be too small, and the candidate is not known to be decisive. Absent a certified error bound, the exact test is required to certify the returned resolution.

\section{Privatization Experiments} 
\label{sec:experiments}

\newcommand{\suppDicaprioGallery}{%
\begin{strip}
\section{Supplementary Experiment Figures}
\label{sec:supp_figures}
\captionsetup{type=figure}
\begin{subfigure}[t]{0.32\textwidth}
\includegraphics[width=0.99\linewidth]{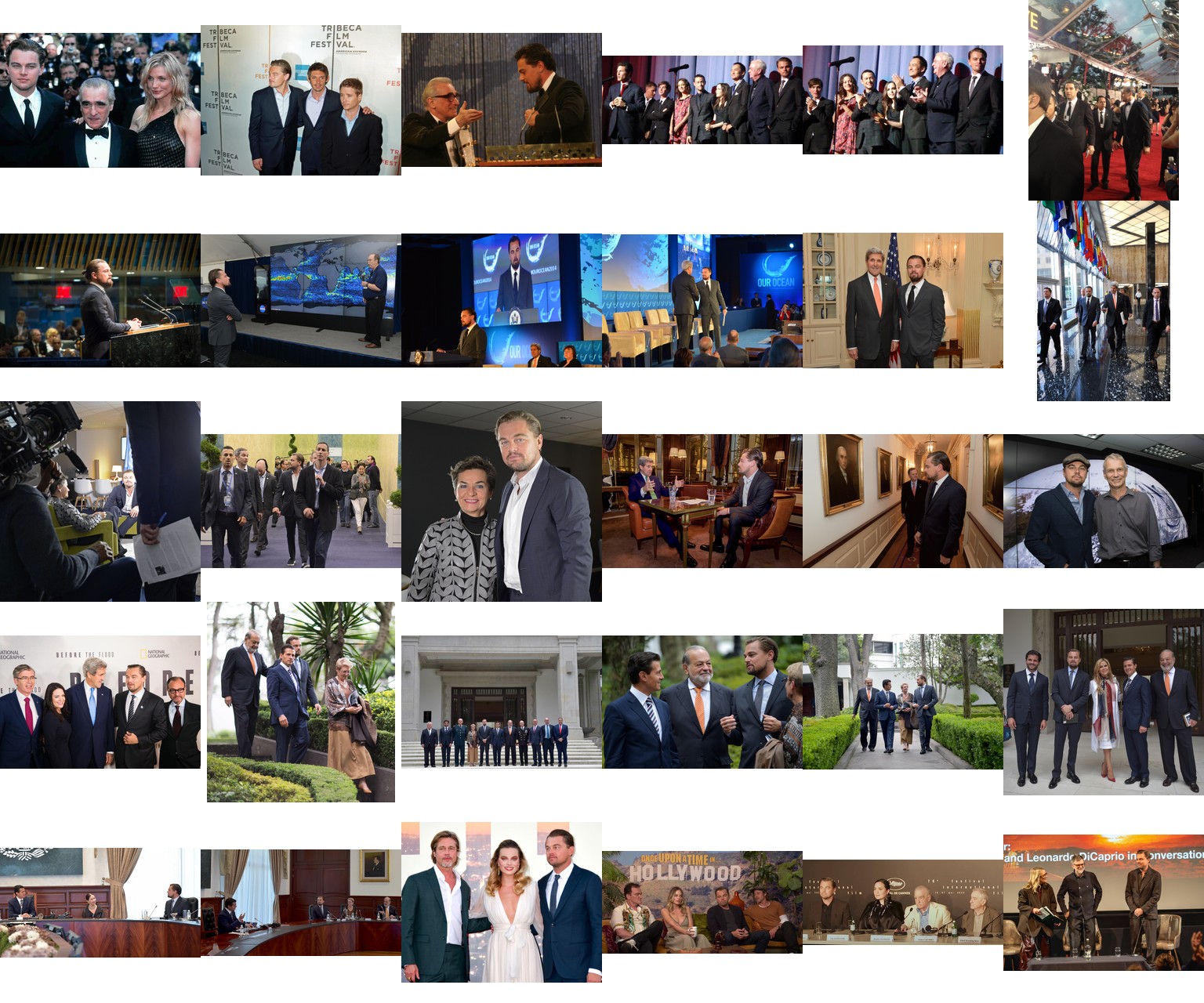} 
\phantomcaption\label{fig:dicapiro_sum_orig}
\end{subfigure}%
\hfill
\begin{subfigure}[t]{0.32\textwidth}
\includegraphics[width=0.99\linewidth]{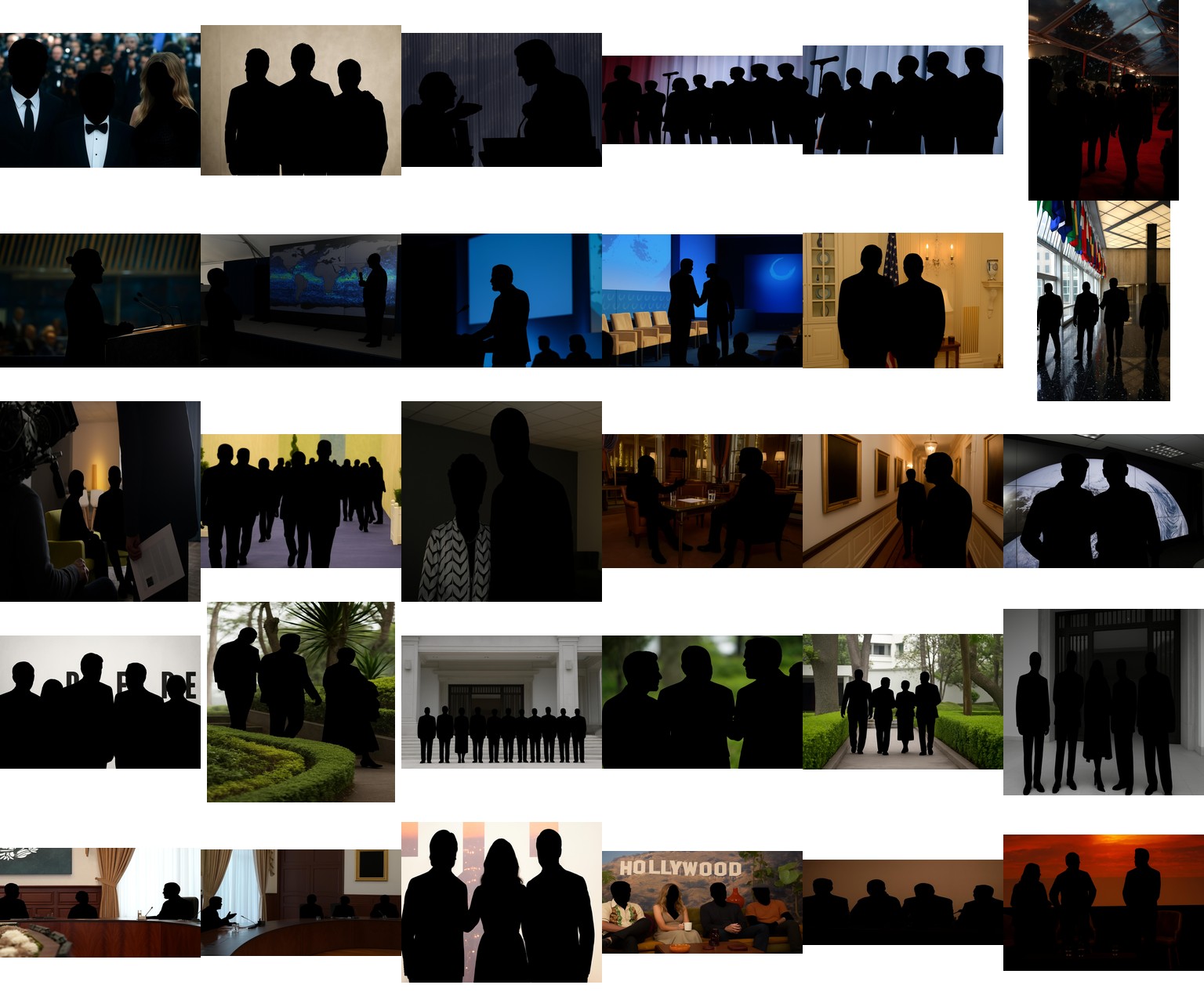}
\phantomcaption\label{fig:dicaprio_igpt1m_gem31p}
\end{subfigure}
\hfill
\begin{subfigure}[t]{0.32\textwidth}
\includegraphics[width=0.99\linewidth]{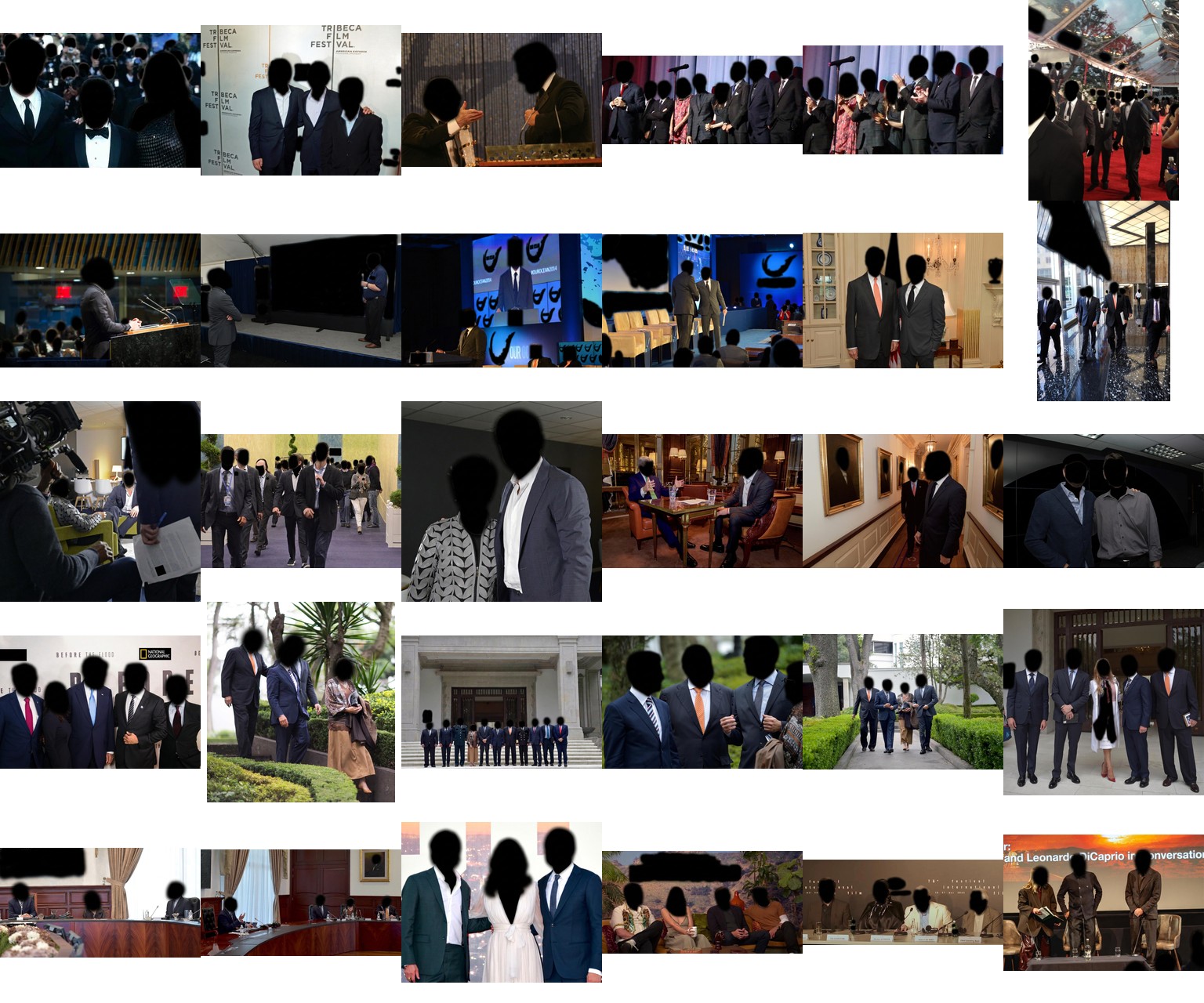}
\phantomcaption\label{fig:dicaprio_man_sam3}
\end{subfigure}
\caption{Thirty images capturing Leonardo DiCaprio (left), sanitized by \texttt{iGPT1m/GEM31p} (center; privacy resolution 0.02 and utility 0.42) and \texttt{iGEM31f/Manual} (right; privacy resolution 0 and utility 0.61). Both mechanisms sanitize \texttt{the identity of the celebrities} and use \texttt{EVA} as $\mathcal{D}$.}
\Description{Three grids show the same 30 DiCaprio images before sanitization and after two sanitization mechanisms. The center grid contains extensive black redaction; the right grid contains more localized redaction.}
\label{fig:priv_dicaprio}
\end{strip}
}

\begin{figure}[t]
\centering  
\small
\setlength{\aboverulesep}{1pt}
\setlength{\belowrulesep}{1pt}
\resizebox{0.9\linewidth}{!}{
\begin{tabular}{@{} l @{\hspace{2ex}} @{}c@{} @{}c@{} @{}c@{} @{}c@{} @{}c@{} }
\toprule
\multicolumn{1}{p{1.3cm}}{\scriptsize\scshape Redaction model}
& \multicolumn{5}{c}{\sc Concept generation model} \\
\cmidrule(lr){1-1}
\cmidrule(lr){2-6}
& \thead{\texttt{GPT54}}
& \thead{\texttt{GEM31p}} 
& \thead{\texttt{OPS46}}
& \thead{\texttt{Manual}}
& \thead{None} \\

\texttt{iGPT15}
& \diagcell{\fpeval{round(min(0.2025,0.2570,0.2003), 2)}}{\fpeval{round(0.3671, 2)}}
& \diagcell{\fpeval{round(min(0.2578, 0.2393, 0.2251), 2)}}{\fpeval{round(0.3787, 2)}}
& \diagcell{\fpeval{round(min(0.2379, 0.1932, 0.1357), 2)}}{\fpeval{round(0.2993, 2)}}
& \failDiagCell{\fpeval{round(min(0.0422, 0.0924, 0), 2)}}{\fpeval{round(0.5310, 2)}}
& \failDiagCell{\fpeval{round(min(0.0394, 0.0936, 0.1316), 2)}}{\fpeval{round(0.4642, 2)}} \\

\texttt{iGPT1m}
& \diagcell{\fpeval{round(min(0.0558, 0.1985, 0.1716), 2)}}{\fpeval{round(0.3861, 2)}}
& \diagcell{\fpeval{round(min(0.0176, 0.1617, 0.1297), 2)}}{\fpeval{round(0.4161, 2)}}
& \diagcell{\fpeval{round(min(0.3799, 0.3672, 0.3828), 2)}}{\fpeval{round(0.4267, 2)}}
& \diagcell{\fpeval{round(min(0.0814, 0.2004, 0.1155), 2)}}{\fpeval{round(0.3866, 2)}}
& \diagcell{\fpeval{round(min(0.2437, 0.3195, 0.2079), 2)}}{\fpeval{round(0.3690, 2)}} \\

\texttt{iGEM3p}
& \diagcell{\fpeval{round(min(0.3211, 0.3940, 0.3959), 2)}}{\fpeval{round(0.2892, 2)}}
& \diagcell{\fpeval{round(min(0.3500, 0.3151, 0.3139), 2)}}{\fpeval{round(0.3794, 2)}}
& \diagcell{\fpeval{round(min(0.3324, 0.2544, 0.3794), 2)}}{\fpeval{round(0.3330, 2)}}
& \diagcell{\fpeval{round(min(0.4946, 0.4480, 0.5029), 2)}}{\fpeval{round(0.4995, 2)}}
& \diagcell{\fpeval{round(min(0.5285, 0.2296, 0.1840), 2)}}{\fpeval{round(0.6587, 2)}} \\

\texttt{iGEM31f}
& \diagcell{\fpeval{round(min(0.2389, 0.3384, 0.2760), 2)}}{\fpeval{round(0.4043, 2)}}
& \diagcell{\fpeval{round(min(0.2854, 0.2279, 0.4095), 2)}}{\fpeval{round(0.3989, 2)}}
& \diagcell{\fpeval{round(min(0.3488, 0.2276, 0.4750), 2)}}{\fpeval{round(0.3134, 2)}}
& \paretoDiagCell{\fpeval{round(min(0, 0, 0.0046), 2)}}{\fpeval{round(0.6123, 2)}}
& \paretoDiagCell{\fpeval{round(min(0.1536, 0.2638, 0.0771), 2)}}{\fpeval{round(0.6702, 2)}} \\

\texttt{FLX2p}
& \diagcell{\fpeval{round(min(0.2974, 0.2474, 0.3648), 2)}}{\fpeval{round(0.6892, 2)}}
& \diagcell{\fpeval{round(min(0.3486, 0.2763, 0.1920), 2)}}{\fpeval{round(0.6649, 2)}}
& \diagcell{\fpeval{round(min(0.2177, 0.2552, 0.2886), 2)}}{\fpeval{round(0.6953, 2)}}
& \diagcell{\fpeval{round(min(0.2178, 0.2756, 0.1766), 2)}}{\fpeval{round(0.7338, 2)}}
& \paretoDiagCell{\fpeval{round(min(0.2931, 0.1557, 0.3792), 2)}}{\fpeval{round(0.7461, 2)}} \\

\texttt{FLX2d}
& \diagcell{\fpeval{round(min(0.1492, 0.1818, 0.1071), 2)}}{\fpeval{round(0.5718, 2)}}
& \diagcell{\fpeval{round(min(0.1110, 0.1393, 0.1389), 2)}}{\fpeval{round(0.6277, 2)}}
& \diagcell{\fpeval{round(min(0.3461, 0.2386, 0.3145), 2)}}{\fpeval{round(0.6731, 2)}}
& \diagcell{\fpeval{round(min(0.2231, 0.2319, 0.2641), 2)}}{\fpeval{round(0.7085, 2)}}
& \paretoDiagCell{\fpeval{round(min(0.2054, 0.2760, 0.2746), 2)}}{\fpeval{round(0.7725, 2)}} \\

\texttt{SAM3}
& \diagcell{0.09}{0.12}
& \diagcell{0.11}{0.26}
& \diagcell{0.12}{0.12}
& \diagcell{0}{0.58}
& \diagcellNA \\

\bottomrule
\end{tabular}
}
\caption{Privacy resolution under \texttt{EVA} and utility for 30 images capturing Leonardo DiCaprio sanitized of \texttt{the identity of the celebrities}. Diagonal cells: privacy resolution (red) and utility (blue); lighter color is better. Black-framed cells are Pareto optimal. Red-framed cells failed to sanitize one image due to model safeguards.}
\Description{A seven-by-five matrix compares redaction and concept-generation models. Each diagonal cell displays privacy resolution in red and utility in blue; borders mark Pareto-optimal cells and safeguard failures.}
\label{fig:dicaprio_model_compare}
\end{figure}

The primary purpose of this section is to demonstrate that contrastive privacy can surface and rank trial violations in practical case studies. Secondarily, we show how various models perform on a limited number of privatization tasks, how privatization performance scales with data size, and how identity-specific fine-tuning affects a distance mechanism's sensitivity. We leave for future work the task of systematic evaluation of the privatization capabilities of these models. (The frontier-model sanitization experiments cost roughly \$1,200 USD in aggregate.)

Contrastive privacy (Definition~\ref{def:privacy_det}) is quite general. In principle, it can be instantiated across modalities including text, images, video, and audio and with forms of sanitization including redaction, replacement, and inpainting. An operational instantiation requires concepts whose sanitization is intended to privatize the abstract target and a distance mechanism conforming to Definition~\ref{def:distance}. Such mechanisms are most mature for text and images, which are the modalities that we chose to evaluate.

Table~\ref{tab:mechanisms} summarizes the models underlying our chosen mechanisms; we provide details of their use later in this section. Our image and Reddit experiments sanitize by redaction, replacing image pixels with black and text characters with boxes; the commercial PII tools use their configured replacement outputs.

All experiments use the symmetric case $\widehat{\mathbb{X}} = \widetilde{\mathbb{X}}$: every rendering captures a direct target-related property in $\gamma \cap c$, as detailed below, and we argue that each collection is a reasonable proxy set for the reference domain.

\subsection{Personal Privacy}
\label{sec:personal_privacy}

We began by computing model-relative resolutions for collections of sanitized images with respect to a given abstract privacy concept. For this task we used the \texttt{EVA} CLIP model for the distance mechanism $\mathcal{D}$ as described in Algorithm~\ref{alg:clip_priv}. We performed sanitization in two phases, denoted as \texttt{REDACT/CONCEPT}. In the concept phase we created a list of natural concepts whose sanitization was believed sufficient to privatize the abstract concept, either manually (\texttt{Manual}) or using one of the image-to-text (i2t) models in Table~\ref{tab:mechanisms}. The former involved an author iteratively guessing concepts whose obfuscation would lead to a lower observed resolution and refining the set based on the resulting privacy resolution. The latter passed each image in the collection to the i2t model, prompting it to generate a list of concepts to be redacted. The responses were compiled into a master list used for redaction.\footnote{Our goal was not to evaluate the ability of humans in general to generate concepts. Instead, we wanted a basic comparison against the models.} In the redaction phase, we prompted i2i sanitization mechanisms from Table~\ref{tab:mechanisms} to redact the list of natural concepts derived in the concept phase. For example, \texttt{iGPT1m/GEM31p} means that \texttt{GEM31p} generated natural concepts from the experimental image set, and \texttt{iGPT1m} redacted those concepts from the same set. If redaction is prompted to remove the abstract concept directly, e.g., \texttt{obfuscate celebrities from the image}, then the concept phase is not performed.

\textbf{Methodology.} Our first image-related task was to use the privacy test to detect failures in model sanitizations of images, each selected to show Leonardo DiCaprio's face. Because the sanitization concepts targeted faces directly or through closely related properties, we assumed that their semantic closure under each configured $\mathcal{D}$ contained DiCaprio's identity. This premise could fail for identifying text or context that $\mathcal{D}$ does not connect to those concepts.
When i2i models alone were used for sanitization we gave them Prompt~\ref{prompt:i2i_abs_con} (in Appendix~\ref{sec:prompts}) with \texttt{CONCEPT} being \texttt{the identity of the celebrities}. For two-stage sanitization, i2t models first returned natural concept lists in response to Prompt~\ref{prompt:i2t_con} with \texttt{CONCEPT} being \texttt{the identity of the celebrities}. In Appendix~\ref{sec:concepts} outputs for \texttt{GPT54}, \texttt{GEM31p}, and \texttt{OPS46} are listed as Concepts~\ref{con:dicaprio_gpt-5.4},~\ref{con:dicaprio_gemini-3.1-pro}, and~\ref{con:dicaprio_opus-4.6}, respectively. Concepts~\ref{concepts:dicaprio_manual}, for \texttt{Manual}, were derived manually. Each i2i mechanism then performed their sanitization in response to Prompt~\ref{prompt:i2i_nat_con}, where \texttt{OBJECTS} were returned by the i2t model.

\textbf{Results.} Figure~\ref{fig:priv_dicaprio} in Appendix~\ref{sec:supp_figures} shows a collection of 30 images depicting Leonardo DiCaprio before (\ref{fig:dicapiro_sum_orig}) and after sanitization via two different mechanisms $\mathcal{X}_c$: \texttt{iGPT1m/GEM31p} (\ref{fig:dicaprio_igpt1m_gem31p}) and \texttt{iGEM31f/Manual} (\ref{fig:dicaprio_man_sam3}). Under \texttt{EVA} on this proxy set, Algorithm~\ref{alg:clip_priv} returned resolution 0 for \texttt{iGEM31f/Manual} but a nonzero resolution for \texttt{iGPT1m/GEM31p}. The image in the lower-right corner of Figure~\ref{fig:dicaprio_man_sam3} highlights a blind spot of \texttt{EVA}: text behind the obfuscated faces clearly spells out ``Leonardo DiCaprio.'' The retained name makes this a resolution-0 false pass. Because \texttt{Manual} refined $c$ using \texttt{EVA}'s resolution on this same proxy set, it could become tailored to \texttt{EVA}'s sensitivities and blind spots, achieving resolution 0 without removing all identifying information. Possible remedies are to add \texttt{text} to $c$ or to add OCR followed by the text test to the heterogeneous ensemble in Section~\ref{sec:distance_robustness}.

\newcommand{\suppMcDonaldsMatrix}{%
\begin{figure}[t]
\centering  
\small
\setlength{\aboverulesep}{1pt}
\setlength{\belowrulesep}{1pt}
\resizebox{0.9\linewidth}{!}{
\begin{tabular}{@{} l @{\hspace{2ex}} @{}c@{} @{}c@{} @{}c@{} @{}c@{} @{}c@{} }
\toprule
\multicolumn{1}{p{1.3cm}}{\scriptsize\scshape Redaction model}
& \multicolumn{5}{c}{\sc Concept generation model} \\
\cmidrule(lr){1-1}
\cmidrule(lr){2-6}
& \thead{\texttt{GPT54}}
& \thead{\texttt{GEM31p}} 
& \thead{\texttt{OPS46}}
& \thead{\texttt{Manual}}
& \thead{None} \\

\texttt{iGPT15}
& \diagcell{\fpeval{round(min(0.3291, 0.3272, 0.3065), 2)}}{\fpeval{round(0.5092, 2)}}
& \diagcell{\fpeval{round(min(0.2555, 0.2453, 0.2628), 2)}}{\fpeval{round(0.4368, 2)}}
& \paretoDiagCell{\fpeval{round(min(0.2240, 0.1637, 0.3049), 2)}}{\fpeval{round(0.5411, 2)}}
& \paretoDiagCell{\fpeval{round(min(0.1654, 0.3708, 0.3664), 2)}}{\fpeval{round(0.7132, 2)}}
& \diagcell{\fpeval{round(min(0.2721, 0.2853, 0.2725), 2)}}{\fpeval{round(0.7132, 2)}} \\

\texttt{iGPT1m}
& \diagcell{\fpeval{round(min(0.4271, 0.3090, 0.3042), 2)}}{\fpeval{round(0.6729, 2)}}
& \diagcell{\fpeval{round(min(0.3178, 0.5418, 0.3859), 2)}}{\fpeval{round(0.7006, 2)}}
& \diagcell{\fpeval{round(min(0.2935, 0.3465, 0.3799), 2)}}{\fpeval{round(0.7072, 2)}}
& \diagcell{\fpeval{round(min(0.2909, 0.4208, 0.4275), 2)}}{\fpeval{round(0.6629, 2)}}
& \diagcell{\fpeval{round(min(0.5590, 0.5225, 0.6106), 2)}}{\fpeval{round(0.5968, 2)}} \\

\texttt{iGEM3p}
& \diagcell{\fpeval{round(min(0.5692, 0.6058, 0.5280), 2)}}{\fpeval{round(0.6299, 2)}}
& \diagcell{\fpeval{round(min(0.4491, 0.5705, 0.6053), 2)}}{\fpeval{round(0.6165, 2)}}
& \diagcell{\fpeval{round(min(0.5621, 0.5410, 0.5317), 2)}}{\fpeval{round(0.6113, 2)}}
& \diagcell{\fpeval{round(min(0.6250, 0.6777, 0.5360), 2)}}{\fpeval{round(0.5540, 2)}}
& \diagcell{\fpeval{round(min(0.3127, 0.5454, 0.5603), 2)}}{\fpeval{round(0.7746, 2)}} \\

\texttt{iGEM31f}
& \diagcell{\fpeval{round(min(0.4064, 0.8024, 0.5421), 2)}}{\fpeval{round(0.6309, 2)}}
& \diagcell{\fpeval{round(min(0.5514, 0.4582, 0.5464), 2)}}{\fpeval{round(0.5604, 2)}}
& \diagcell{\fpeval{round(min(0.4794, 0.4734, 0.3942), 2)}}{\fpeval{round(0.5497, 2)}}
& \diagcell{\fpeval{round(min(0.3979, 0.3104, 0.5142), 2)}}{\fpeval{round(0.5910, 2)}}
& \paretoDiagCell{\fpeval{round(min(0.2838, 0.2568, 0.2093), 2)}}{\fpeval{round(0.7418, 2)}} \\

\texttt{FLX2p}
& \diagcell{\fpeval{round(min(0.3021, 0.4489, 0.2781), 2)}}{\fpeval{round(0.7783, 2)}}
& \diagcell{\fpeval{round(min(0.3102, 0.3397, 0.2979), 2)}}{\fpeval{round(0.7468, 2)}}
& \diagcell{\fpeval{round(min(0.3282, 0.5107, 0.3181), 2)}}{\fpeval{round(0.7781, 2)}}
& \paretoDiagCell{\fpeval{round(min(0.2530, 0.3608, 0.3001), 2)}}{\fpeval{round(0.7940, 2)}}
& \diagcell{\fpeval{round(min(0.3074, 0.3294, 0.3158), 2)}}{\fpeval{round(0.8040, 2)}} \\

\texttt{FLX2d}
& \diagcell{\fpeval{round(min(0.3330, 0.5634, 0.4006), 2)}}{\fpeval{round(0.7309, 2)}}
& \diagcell{\fpeval{round(min(0.3235, 0.2825, 0.2664), 2)}}{\fpeval{round(0.6702, 2)}}
& \diagcell{\fpeval{round(min(0.2786, 0.3194, 0.3270), 2)}}{\fpeval{round(0.6910, 2)}}
& \diagcell{\fpeval{round(min(0.2237, 0.2809, 0.2621), 2)}}{\fpeval{round(0.6903, 2)}}
& \paretoDiagCell{\fpeval{round(min(0.3243, 0.3054, 0.3038), 2)}}{\fpeval{round(0.8580, 2)}} \\

\texttt{SAM3}
& \diagcell{\fpeval{round(0.4025,2)}}{\fpeval{round(0.2130,2)}}
& \diagcell{\fpeval{round(0.4099,2)}}{\fpeval{round(0.3864,2)}}
& \diagcell{\fpeval{round(0.2286,2)}}{\fpeval{round(0.3123,2)}}
& \diagcell{\fpeval{round(0.3140,2)}}{\fpeval{round(0.4859,2)}}
& \diagcellNA \\

\bottomrule
\end{tabular}
}
\caption{Privacy-resolution (red) and utility (blue) values for sanitizing \texttt{the identity of the fast food restaurant} from 49 images of McDonald's, summarized in the top panel of Figure~\ref{fig:privacy_performance}. The formatting matches Figure~\ref{fig:dicaprio_model_compare}.}
\Description{A seven-by-five matrix compares redaction and concept-generation models on the McDonald's images. Each diagonal cell displays privacy resolution in red and utility in blue.}
\label{fig:mcdonalds_model_compare}
\end{figure}
}

Figure~\ref{fig:dicaprio_model_compare} shows the privacy--utility tradeoff for all tested SSMs (listed in Table~\ref{tab:mechanisms}). Each cell, except those for \texttt{SAM3}, gives the lowest observed resolution and its associated utility across three independent sanitization runs.\footnote{Given the number of models we tested, these experiments were costly to run even when limited to three runs per combination.}
Column ``None'' indicates that only the redaction phase was executed. 
We report \textit{utility} as the cosine similarity between \texttt{EVA} embeddings of original and sanitized images. (Utility is maximized at 1.0.) This measure indirectly penalizes both concept ambiguity and unnecessary redaction.

Mechanism \texttt{iGEM31f/Manual} returns resolution 0 under \texttt{EVA} on this proxy set and has relatively high utility of 0.61.
A variety of redaction-phase-only mechanisms are Pareto dominant at higher utility values. Generally, we see no clear advantage for models with more parameters or from later generations. Moreover, we see that the same set of concepts (consistent within columns) leads to drastically different privacy resolutions and utility across redaction models.

As discussed in Section~\ref{sec:application_to_clip}, Figure~\ref{fig:dicaprio_privacy_failures} shows a trial satisfied at $\delta=0$ (\ref{fig:dicaprio_success}) and a witness with positive violation (\ref{fig:dicaprio_failure}) under the distance mechanism \texttt{EVA}. The primary difference between the two obfuscations is that all faces are sanitized in the resolution-0 example, whereas the face of a woman (likely an aide to the then-president of Mexico, Enrique Peña Nieto) remains unsanitized in the image on the left in \ref{fig:dicaprio_failure}. This is a direct sanitization failure relative to the prompt, which indicated that \emph{all} celebrities should be obfuscated. Here, \texttt{FLX2d} appears ambivalent about the celebrity status of the woman, obfuscating her face in one image but not another. This is also an example of the concept ambiguity penalty at work (see Section~\ref{sec:conc_ambig}). Although we care only about privatizing the abstract concept of Leonardo DiCaprio, we broadly requested properties that describe celebrities in general. As a result, the properties associated with this woman were swept into the semantic closure of the obfuscated concepts $c$ because $c$ includes all faces (see Concepts~\ref{con:dicaprio_gemini-3.1-pro}). For \ref{fig:dicaprio_failure}, the presence of the woman's face in $\mathcal{X}_c(x)$, coupled with its absence from $\mathcal{X}_c(y)$, allows $\mathcal{D}(\mathcal{X}_c(x), y)$ to be less than $\mathcal{D}(\mathcal{X}_c(x), \mathcal{X}_c(y))$, giving this trial a positive violation.

Figure~\ref{fig:dicaprio_failure_compare} in Appendix~\ref{sec:privacy_failures} shows a decisive witness for each sanitization listed in Table~\ref{tab:mechanisms}. Most are due again to the concept ambiguity penalty. In multiple cases, the sanitization mechanism returned an image for $\mathcal{X}_c(y)$ that is nearly entirely black, which sweeps all properties in the image into concept $c$. This makes it very difficult for images $\mathcal{X}_c(x)$ and $\mathcal{X}_c(y)$ to be closer than $\mathcal{X}_c(x)$ and $y$. The same scenario occurs when mechanisms return black-and-white silhouettes for $\mathcal{X}_c(y)$. The decisive witness in cell \texttt{FLX2d/GEM31p} is more interesting. Obfuscation fails to remove Margot Robbie and a movie poster for the film \emph{Once Upon a Time... in Hollywood}, in which DiCaprio starred. The violation ostensibly arises from a semantic link between those instances and DiCaprio being interviewed in the image on the right. Another interesting decisive witness appears in cell \texttt{iGPT15/GPT54}. Here, the only prominent feature remaining in the image on the left is an ocean map. The semantic match with the image on the right is ostensibly due to the association of both Leonardo DiCaprio and John Kerry with the Our Ocean Conference.

\begin{table}[t]
\caption{Agreement among distance mechanisms on 94 matched DiCaprio sanitizations. Spearman's $\rho_s$ compares resolution ranks; other columns count outcomes whose unrounded resolution is exactly 0.}
\label{tab:dicaprio_distance_agreement}
\centering
\footnotesize
\setlength{\tabcolsep}{3pt}
\renewcommand{\arraystretch}{1.05}
\begin{tabular*}{\columnwidth}{@{\extracolsep{\fill}}lrrrr@{}}
\toprule
\textbf{Pair $(A,B)$} & $\boldsymbol{\rho_s}$ & \textbf{Both 0} & \textbf{$A$ only 0} & \textbf{$B$ only 0} \\
\midrule
\texttt{EVA}, \texttt{SIGLIP2} & 0.72 & 4 & 0 & 7 \\
\texttt{EVA}, \texttt{METACLIP2} & 0.35 & 0 & 4 & 2 \\
\texttt{SIGLIP2}, \texttt{METACLIP2} & 0.57 & 0 & 11 & 2 \\
\bottomrule
\end{tabular*}
\end{table}

\noindent\textbf{Distance-Mechanism Ensemble.} To measure dependence on $\mathcal{D}$, we held each DiCaprio sanitization fixed and reran the test with \path{google/siglip2-giant-opt-patch16-384}~\cite{tschannen2025siglip2} and \path{facebook/metaclip-2-worldwide-giant-378}~\cite{chuang2025metaclip2}, denoted \texttt{SIGLIP2} and \texttt{METACLIP2}. Together with \texttt{EVA}, all three mechanisms evaluated the 94 individual sanitized collections underlying Figure~\ref{fig:dicaprio_model_compare}. We analyze these matched runs directly because independently selecting each mechanism's minimum could select different sanitized collections. Since resolution scales are mechanism-specific, Table~\ref{tab:dicaprio_distance_agreement} compares ranks and resolution-0 outcomes rather than averaging resolutions. Appendix Table~\ref{tab:dicaprio_distance_raw} reports the underlying resolutions.

The pairwise rank correlations are uniformly positive: 0.72 and 0.57 indicate moderate-to-strong agreement for two model pairs, while 0.35 indicates weaker but still positive agreement for the third. Thus, the mechanisms broadly concur about relative sanitization performance despite their different resolution scales. At the strict resolution-0 boundary, however, \texttt{EVA}, \texttt{SIGLIP2}, and \texttt{METACLIP2} return 0 for 4, 11, and 2 sanitizations, respectively, and none returns 0 under all three. Even there, the models can nearly agree in magnitude: for the four sanitizations where \texttt{EVA} returns 0 but \texttt{METACLIP2} does not, the latter's resolutions range only from 0.0101 to 0.0536. Nevertheless, each single-model resolution-0 result is paired with a positive violation under another mechanism. These results show broad overall agreement, while the boundary differences suggest that an ensemble may help identify model-specific blind spots; shared blind spots and corpus limitations remain.

\subsubsection{Fine-Tuning for an Unseen Identity}
\label{sec:finetune_unseen}

\textbf{Methodology.} We used \texttt{Nano Banana 2}~\cite{google2026nanobanana2} to create two 100-image galleries, each intended to depict a consistent synthetic child identity. Each gallery comprised ten series of ten images in different scenarios; prompts and earlier image references encouraged consistency of the child, as well as objects associated with the child, within and across series. Using the released \path{finetune_clip.py} script, we fine-tuned \path{google/siglip2-giant-opt-patch16-384}~\cite{tschannen2025siglip2} with the labels \texttt{a picture of subject1} and \texttt{a picture of subject2}. Each epoch comprised 100 two-example batches, one example per identity; all 200 images therefore appeared once per epoch. Ten epochs yielded 1,000 optimizer steps and ten exposures per image.

No model parameters were frozen: both the vision and text encoders were updated using SigLIP 2's sigmoid pair loss and AdamW (learning rate $10^{-5}$, weight decay $0.1$, $\beta=(0.9,0.98)$, and $\epsilon=10^{-6}$). We retained the script defaults of no learning-rate scheduler or warmup, global gradient-norm clipping at $1.0$, \texttt{fp32} model weights with \texttt{bf16} autocasting on CUDA, and seed 0. We denote the epoch-10 distance mechanism by \texttt{SIGLIP2-FT10}, distinguishing it from the base \texttt{SIGLIP2} mechanism above.

For the first 100-image gallery, which was also included in fine-tuning, \texttt{SAM3}~\cite{carion2025sam} sanitized $c_1=\{\texttt{head},\texttt{skin}\}$. Every image showed the child's face, so $c_1$ targeted a direct cue to the privacy target, the child's identity. We treat the closure premise as plausible for \texttt{SIGLIP2-FT10} after identity-specific training, but it is deliberately uncertain for the base \texttt{SIGLIP2} and could fail for cues that a checkpoint does not connect to $c_1$, such as clothing or recurring backgrounds. Both checkpoints served as $\mathcal{D}$ in Algorithm~\ref{alg:clip_priv}.
For this experiment, \texttt{SAM3} used \texttt{threshold} $=0.2$, \texttt{blur} $=20$, \texttt{dilate} $=25$, \texttt{min-coverage} $=10^{-5}$, and \texttt{max-coverage} $=1.0$, with convex-hull postprocessing disabled.

\textbf{Results.} With $c_1$, \texttt{SIGLIP2} returned resolution 0. Holding the sanitized images fixed and replacing only $\mathcal{D}$ with \texttt{SIGLIP2-FT10} returned 0.1480, exposing target-related signal missed by the base mechanism. After we added \texttt{clothing} to form $c_2=c_1\cup\{\texttt{clothing}\}$ and resanitized the same original images, \texttt{SIGLIP2-FT10} again returned resolution 0. The broader sanitization of $c_2$ may have removed identity-correlated information learned by the adapted distance mechanism, including clothing and body shape.

Using the same gallery for adaptation and symmetrically as both the audited set and the proxy reference set is appropriate for this corpus-specific diagnostic: after adaptation, the distance mechanism is fixed, and the test asks which target-related signals represented in that corpus remain after sanitization. The fine-tuned model exposed residual properties that the base distance mechanism missed. The result remains conditional on how well this proxy reference set represents the reference domain; omitted reference comparisons remain untested.

\subsection{Brand Privacy}

\newcommand{\suppAvengersMatrix}{%
\begin{figure}[t]
\centering  

\setlength{\aboverulesep}{0pt}
\setlength{\belowrulesep}{0pt}

\begin{minipage}{0.45\textwidth}
    \centering
    \resizebox{\linewidth}{!}{
        \begin{tabular}{@{}c@{} @{}c@{} @{}c@{} @{}c@{} @{}c@{}}
        \toprule
        \paretoDiagcellLabeled{\fpeval{round(min(0.2090, 0.1934, 0.2090), 2)}}{\fpeval{round(0.8131, 2)}}{\texttt{GPT54}} &
        \diagcellLabeled{\fpeval{round(min(0.3535, 0.3535, 0.3438), 2)}}{\fpeval{round(0.7362, 2)}}{\texttt{GPT54m}} & 
        \diagcellLabeled{\fpeval{round(min(0.2197, 0.2188, 0.2188), 2)}}{\fpeval{round(0.7960, 2)}}{\texttt{GPT52}} & 
        \paretoDiagcellLabeled{\fpeval{round(min(0.2676, 0.1914, 0.1836), 2)}}{\fpeval{round(0.7668, 2)}}{\texttt{GPT4o}} & 
        \paretoDiagcellLabeled{\fpeval{round(min(0.1699, 0.1777, 0.2031), 2)}}{\fpeval{round(0.7397, 2)}}{\texttt{GPT4om}} \\
        
        \diagcellLabeled{\fpeval{round(min(0.2090, 0.2266, 0.2168), 2)}}{\fpeval{round(0.7932, 2)}}{\texttt{GEM31p}} &
        \diagcellLabeled{\fpeval{round(min(0.3359, 0.3301, 0.2402), 2)}}{\fpeval{round((0.6671)/1, 2)}}{\texttt{GEM3f}} &
        \diagcellLabeled{\fpeval{round(min(0.2031, 0.1992, 0.1934), 2)}}{\fpeval{round(0.7643, 2)}}{\texttt{GEM25p}} &
        \diagcellLabeled{\fpeval{round(min(0.2227, 0.3086, 0.4502), 2)}}{\fpeval{round(0.8035, 2)}}{\texttt{GEM25f}} &
        \paretoDiagcellLabeled{\fpeval{round(min(0.2207, 0.2188, 0.2188), 2)}}{\fpeval{round(0.8204, 2)}}{\texttt{GEM2f}} \\
        
        \diagcellLabeled{\fpeval{round(min(0.2090, 0.2051, 0.2051), 2)}}{\fpeval{round(0.7920, 2)}}{OPS46} & 
        \diagcellLabeled{\fpeval{round(min(0.2188, 0.2217, 0.2217), 2)}}{\fpeval{round(0.7977, 2)}}{OPS41} & 
        \diagcellLabeled{\fpeval{round(min(0.2266, 0.2285, 0.2373), 2)}}{\fpeval{round(0.8455, 2)}}{SON46} &
        \paretoDiagcellLabeled{\fpeval{round(min(0.2256, 0.2266, 0.2266), 2)}}{\fpeval{round(0.8569, 2)}}{HAK45} & 
        \diagcellLabeled{\fpeval{round(min(0.2168, 0.2168, 0.2168), 2)}}{\fpeval{round(0.7881, 2)}}{HAK35} \\
        
        \bottomrule
        \end{tabular}
    }
\end{minipage}

\caption{Privacy-resolution (red) and utility (blue) values for 49 Reddit comments discussing an \emph{Avengers} movie, summarized in the bottom panel of Figure~\ref{fig:privacy_performance}. The formatting matches Figure~\ref{fig:dicaprio_model_compare}.}
\Description{A three-by-five matrix compares text sanitization models on the Avengers comments. Each diagonal cell displays privacy resolution in red and utility in blue.}
\label{fig:diag_visual_labeled_sidebyside}
\end{figure}
}

\textbf{Methodology.} We next turn our attention to privatizing the concept of \texttt{the identity of the fast food restaurant} from 49 images of McDonald's, each selected to show the Golden Arches logo. Because the sanitization concepts targeted the logo directly or through closely related properties, we assumed that their semantic closure under \texttt{EVA} contained this brand-identity target. This premise could fail for non-logo brand cues that \texttt{EVA} does not connect to those concepts. We used the target string as the \texttt{CONCEPT} in Prompts~\ref{prompt:i2i_abs_con} and~\ref{prompt:i2t_con}. In Appendix~\ref{sec:concepts} outputs for
concept phase models \texttt{GPT54}, \texttt{GEM31p}, and \texttt{OPS46} are listed as 
Concepts~\ref{con:mcdonalds_gpt-5.4},~\ref{con:mcdonalds_gemini-3.1-pro}, and~\ref{con:mcdonalds_opus-4.6}, respectively. 
Their counterparts for a subset of nine images are listed as Concepts~\ref{con:mcdonalds_sample_gpt-5.4},~\ref{con:mcdonalds_sample_gemini-3.1-pro-preview}, and~\ref{con:mcdonalds_sample_claude-opus-4.6}, respectively. 
Concepts~\ref{con:mcdonalds_manual} were used for mechanism \texttt{Manual} in both the nine and 49 image collections.

\textbf{Results.} This privacy problem is challenging because there are many concepts semantically linked to a given brand. This difficulty is illustrated first with a subset of nine images, shown in Figure~\ref{fig:mcdonalds_orig} (see Appendix~\ref{sec:privatizations}). Figures~\ref{fig:mcdonalds_openai_obfs} and~\ref{fig:mcdonalds_sam3_obfs} show the results of sanitization using mechanisms \texttt{iGPT15/GPT54} and \texttt{SAM3/Manual}, respectively. Under \texttt{EVA} on this proxy set, the latter returned resolution 0, while the former returned 0.05; a decisive witness contained a silhouetted ``M'' on a window.

The top panel of Figure~\ref{fig:privacy_performance} summarizes privacy resolution and utility (the \texttt{EVA} cosine similarity of images before and after sanitization) for each sanitization mechanism on the set of 49 images; Figure~\ref{fig:mcdonalds_model_compare} in Appendix~\ref{sec:supp_figures} gives the complete matrix.
Mechanisms incorporating \texttt{iGPT15} for the redaction phase generally have the lowest resolution, while those incorporating \texttt{FLX2p} generally have the highest utility.
Mechanism \texttt{iGEM31f}, also a smaller model, has one of the lowest resolutions, making parameter size a poor predictor of performance. 
Figure~\ref{fig:priv_mcdonalds_49} in Appendix~\ref{sec:privatizations} shows the original images (left) and the sanitizations made by \texttt{iGPT15}.

Figure~\ref{fig:mcdonalds_failure_compare} in Appendix~\ref{sec:privacy_failures} shows a decisive witness for each SSM. As was the case for the DiCaprio dataset, totally blacked-out images often appeared in decisive witnesses and produced large violations due to the ambiguity penalty. The framework also detected clear sanitization failures. For example, \texttt{FLX2p/OPS46} almost entirely fails to obfuscate the carton of French fries, \texttt{iGPT1m/GEM31p} fails to obfuscate any food items, and \texttt{iGEM3p} leaves most of the exterior of a McDonald's restaurant unobfuscated.

\subsection{Text Privatization}

\textbf{Methodology.} Contrastive privacy can also be applied to text renderings. We analyzed 49 Reddit comments (originally collected~\cite{baumgartner2020pushshift} in April 2019) that capture the \emph{Avengers} movie franchise, each selected to contain the string \texttt{Avengers}. Because sanitization targeted this franchise name directly or through related concepts, we assumed that their semantic closure under \texttt{QWEN3} contained the franchise-identity target. This premise could fail for oblique plot details or lore that \texttt{QWEN3} does not connect to those concepts. We chose a relatively old dataset to increase the likelihood that the topics covered, and perhaps even the data themselves, were part of the training data used by our chosen distance mechanism, \texttt{QWEN3}. We tested 15 different LLM-based SSMs, summarized in the bottom panel of Figure~\ref{fig:privacy_performance} and shown individually in Figure~\ref{fig:diag_visual_labeled_sidebyside} in Appendix~\ref{sec:supp_figures}, and one open-vocabulary entity-recognition system, \texttt{GLN2}. The LLM-based mechanisms sanitized the passages in response to Prompt~\ref{prompt:avengers_gpt5.2}, while \texttt{GLN2} sanitized the manually chosen Concepts~\ref{con:avengers_manual}, where quoted items indicate entity removal (i.e., the exact word was removed). In all cases, we measured utility in terms of the cosine similarity between passages before and after sanitization when embedded with the \texttt{QWEN3} model.

\textbf{Results.} Under \texttt{QWEN3} on the nine-comment proxy set, \texttt{GLN2} returned resolution 0 and a utility of 0.75. No other evaluated mechanism returned resolution 0; \texttt{GEM25p} was closest, at 0.04 with a utility of 0.78. On the full set of 49 comments, \texttt{GLN2} returned resolution 0.18 with a utility of 0.77. The bottom panel of Figure~\ref{fig:privacy_performance} shows the utility--privacy tradeoff for the LLM-based mechanisms. Utility remained fairly high relative to the nine-comment sample, but observed resolutions were high across all mechanisms.

\begin{figure}[t]
\centering
\footnotesize 
\setlength{\tabcolsep}{4pt} 
\renewcommand{\arraystretch}{1.3} 

\begin{tabularx}{\linewidth}{@{} p{0.65\linewidth} | X @{}}
\toprule
\textcolor{red}{$\mathcal{X}_c(x)$} \hfill \textit{(ejujfgx)} & \textcolor{red}{$y$} \hfill \textit{(ejunc9k)} \\
\midrule
That explains a lot. The Phase 1s felt new. Part of that is certainly due to the sudden resurgence of the superhero movie, but some of the later films felt like they could have recaptured that and just didn’t. \par\addvspace{0.4em}
I generally enjoy the more recent films, but they feel cookie cutter. I’ll watch them, but I rarely get excited about seeing them like I did for \obf{2.5em} \obf{1.5em} or the original \obf{4em}.
&
Watchmen, Dark Knight, Avengers \\
\midrule
\textcolor{blue}{$\mathcal{X}_c(x)$} & \textcolor{blue}{$\mathcal{X}_c(y)$} \\
\midrule
That explains a lot. The Phase 1s felt new. Part of that is certainly due to the sudden resurgence of the superhero movie, but some of the later films felt like they could have recaptured that and just didn’t. \par\addvspace{0.4em}
I generally enjoy the more recent films, but they feel cookie cutter. I’ll watch them, but I rarely get excited about seeing them like I did for \obf{2.5em} \obf{1.5em} or the original \obf{4em}.
&
\obf{0.6em}, \obf{2.5em} \obf{3em}, \obf{4.5em} \\
\bottomrule
\end{tabularx}
\caption{A decisive witness for \texttt{GPT54} under \texttt{QWEN3} on the 49-comment \emph{Avengers} dataset.}
\Description{A two-column comparison shows a sanitized Avengers comment and a short comparison passage, followed by their respective sanitizations. The comment retains contextual references to superhero movies despite redacting explicit names.}
\label{fig:avengers_GPT54_failure}
\end{figure}

Figure~\ref{fig:avengers_GPT54_failure} shows a decisive witness for \texttt{GPT54} on the 49-comment dataset. In passage $x$ (left), the mechanism removed ``Iron Man'' and ``Avengers'' but retained ``superhero movie,'' which forms a strong semantic connection to $y$: ``Watchmen, Dark Knight, Avengers.'' This trial, with similar sanitizations, was also a decisive witness for \texttt{GPT4o}, \texttt{HAK45}, and \texttt{HAK35}. A witness exhibiting the second-largest violation for \texttt{GPT54} matched against the same passage $y$, with $\mathcal{X}_c(x)$ having sanitized phrases such as ``comic,'' ``single giant crossover with six or seven heroes,'' and ``films grossing around 200 million,'' which strongly suggest \emph{the Avengers}. This trial was a decisive witness for \texttt{GPT4om}. Every SSM on the Pareto frontier had one of these two trials as a decisive witness. This suggests that mechanisms for text tend to perform more uniformly across model families.

\subsubsection{Off-the-Shelf Text Sanitization}
\label{sec:commercial_text}

\textbf{Methodology.} From the Enron corpus~\cite{klimt2004enron}, we retained 1,841 emails containing \texttt{Ken Lay} and at most 500 words after stripping header metadata where possible. We treated the four commercial systems in Table~\ref{tab:mechanisms} as SSMs, $\mathcal{X}_c$. Because their standard PII categories target person names, we assumed the semantic closure of $c$ under each SSM, according to distance mechanism $\mathcal{D}$, contained Ken Lay's identity; this may fail for contextual facts the model does not connect to his name. We used each system off the shelf, without fine-tuning or custom recognizers, and selected its most inclusive standard acceptance threshold. With such settings, Privacy Filter sanitized all spans assigned to its eight privacy categories; Comprehend all returned PII entities; Presidio all candidates from enabled recognizers for the selected language; and Cloud DLP all findings rated \texttt{VERY\_UNLIKELY} or higher. We applied the operational test to the sanitized emails with respect to Ken Lay's identity, using \texttt{QWEN3-4B} as $\mathcal{D}$.

\begin{table}[t]
\caption{Model-relative resolution and case-insensitive \texttt{Ken Lay} matches among 1,841 sanitized emails; Share is the match rate.}
\label{tab:enron_name_baseline}
\centering
\footnotesize
\setlength{\tabcolsep}{3pt}
\renewcommand{\arraystretch}{1.05}
\begin{tabular*}{\columnwidth}{@{\extracolsep{\fill}}lrrr@{}}
\toprule
\textbf{Mechanism} & \textbf{Resolution} & \textbf{Matches} & \textbf{Share} \\
\midrule
Privacy Filter & 0.3047 & 225 & 12.22\% \\
Comprehend & 0.1953 & 9 & 0.49\% \\
Presidio & 0.1914 & 17 & 0.92\% \\
Cloud DLP & 0.1426 & 0 & 0.00\% \\
\bottomrule
\end{tabular*}
\end{table}

\textbf{Results.} All four mechanisms returned nonzero resolution. Table~\ref{tab:enron_name_baseline} compares contrastive privacy with a case-insensitive exact-name baseline: sanitized emails still containing \texttt{Ken Lay}. The baseline detects explicit identifiers but not identifying context. Decisive witnesses for Privacy Filter, Comprehend, and Presidio retained the full string. For Comprehend and Presidio, these flagrant failures occurred in only 9 (0.49\%) and 17 (0.92\%) emails, respectively, yet Algorithm~\ref{alg:clip_priv} selected them as decisive. Cloud DLP had no exact match, but its decisive witness preserved a sentence stating that the addressee had sold ``\$101 million worth of Enron stock'' while urging employees to buy it. Thus, the test surfaced rare explicit-name leaks and contextual disclosure, which the exact-name baseline necessarily misses.

\section{Resolution Analysis}

Privacy engineers may need to release renderings sanitized with a mechanism for which the test returns a nonzero resolution under the chosen distance mechanism. In such cases, it is natural to ask how to interpret the observed resolution.

Familiar semantic comparisons calibrate unitless $\delta$. For each concept in Table~\ref{tab:resolution_analysis}, we used 49 or 50 images, or 50 one-sentence descriptions from \texttt{GEM31p}, blacked out the concept at threshold 0.4, and retained renderings with at least 0.1\% coverage. For each retained pair $(u,v) \in \mathbb{X}(c) \times \mathbb{X}(d)$, $b(u,v) = \min\{\mathcal{D}(\mathcal{X}_c(u), v), \mathcal{D}(u, \mathcal{X}_d(v))\} - \mathcal{D}(u,v)$ is the shared strict upper bound on $\delta$ for both Definition~\ref{def:sem_conn_det} inequalities; Table~\ref{tab:resolution_analysis} reports cross-product percentiles. Images used \texttt{SAM3} with \texttt{EVA}; text used \texttt{GLN2} with \texttt{QWEN3}.

Let $\delta$ denote the configured test's resolution. If we assume that $\mathcal{D}$ fully captures the adversary's semantic sensitivity, accepting an audit result at resolution $\delta$ requires treating violations up to $\delta$ as undetectable by the adversary. At a resolution just above 0.246, for example, this corresponds to an adversary that would fail to register roughly 95\% of the empirical \texttt{EVA} \texttt{dog}--\texttt{cat} comparisons as connected, despite their familiar relationship.

\begin{table}[t]
    \caption{Percentiles of empirical pairwise boundaries for semantic connectedness.}
    \label{tab:resolution_analysis}
    \centering
    \small
    \setlength{\tabcolsep}{2pt}
    \renewcommand{\arraystretch}{1.05}
    \begin{tabular*}{\columnwidth}{@{\extracolsep{\fill}}lrrrrrr@{}}
        \toprule
        & \multicolumn{3}{c}{\textbf{Image} (\texttt{EVA})} & \multicolumn{3}{c}{\textbf{Text} (\texttt{QWEN3})} \\
        \cmidrule(lr){2-4} \cmidrule(lr){5-7}
        \textbf{Concepts} & \textbf{90\%} & \textbf{95\%} & \textbf{99\%} & \textbf{90\%} & \textbf{95\%} & \textbf{99\%} \\
        \midrule
        desk vs.\ paperclip & 0.035 & 0.057 & 0.092 & 0 & 0.003 & 0.017 \\
        car vs.\ boat & 0.076 & 0.101 & 0.143 & 0 & 0 & 0.015 \\
        dog vs.\ cat & 0.217 & 0.246 & 0.308 & 0 & 0.007 & 0.027 \\
        orange vs.\ lemon & 0.272 & 0.341 & 0.487 & 0.053 & 0.063 & 0.085 \\
        \bottomrule
    \end{tabular*}
\end{table}

\section{Related Work}

Section~\ref{sec:experiments} compares frontier and off-the-shelf sanitization mechanisms across image and text; we situate this evaluation among the following approaches.

\textbf{Sanitization Systems.} We evaluate AWS Comprehend, Microsoft Presidio, Google Cloud Data Loss Prevention, and Privacy Filter from OpenAI in Section~\ref{sec:commercial_text}. Imago Obscura~\cite{monteiro2025imago} cannot be directly tested by our approach, and Unsafe2\-Safe~\cite{dinh2026unsafe2safe} is not publicly available. VisShield~\cite{chen2025vision}, which filters Protected Health Information, and GUIGuard~\cite{wang2026guiguard}, which detects and locally sanitizes sensitive regions before screenshots are sent to remote agents, require domain-specific distance mechanisms, so we defer their evaluation.

\textbf{Measuring Sanitization Privacy.} Abdulaziz et al.~\cite{abdulaziz2025evaluation} evaluate AI identity obfuscation over 18 manually labeled attributes, finding residual risks and no uniformly best method. Patwari et al.~\cite{patwari2024perceptanon} train a neural network with human annotations to detect obfuscation failures. Both depend on selected attributes and manual labels; our approach avoids per-item labels but inherits the distance mechanism's semantic blind spots. Using SPriV, which measures the proportion of private tokens leaked, Garza et al.~\cite{garza2025prvl} find that instruction-tuned LLMs can preserve utility while leaking little private information. Calculating SPriV requires hand-labeled ground truth. Like our approach, Asiri et al.~\cite{asiri2025spadr} use a notion of connectivity to detect semantic links. But their entity graphs capture only a fraction of the information encoded by modern embedding models.

Other methods below do not sanitize concept lists from unstructured data and so cannot be evaluated directly by our approach. 

\textbf{Generative Models.} Early context-preserving methods required bounding boxes and focused mainly on faces~\cite{sun2018natural,sun2018hybrid,hukkelaas2019deepprivacy,maximov2020ciagan}, with some providing privacy guarantees~\cite{li2021differentially}; later methods blend\\ anonymized objects into the original background~\cite{hukkelaas2023deepprivacy2,barattin2023attribute,malm2023rad,zwick2024context,kung2025face,ertan2025beyond,chang2025privacy}.

\textbf{Privacy from Face Recognition.} Ilia et al.~\cite{ilia2015face} introduced face-level access control. Other systems perturb facial features to prevent detection~\cite{Chandrasekaran2021FaceOff,Rajabi2021Impracticality,Le2025DiffPrivate}; Yan et al.~\cite{yan2024coder} obtain high-utility face privacy using differential privacy (DP), while Shoshitaishvili et al.~\cite{shoshitaishvili2015portrait} show that face detection can reveal human relationships.

\textbf{Privatizing Text.} LLMs can reproduce training passages~\cite{carlini2021extracting,he2025privacyxray}, even after DP training~\cite{pang2025reconstruction}. Pr$\epsilon\epsilon$mpt~\cite{roychowdhury2026preempt} sanitizes prompts with token-level DP, while Pham et al.~\cite{pham2025can} find that LLMs poorly obfuscate named entities. Xin et al.~\cite{xin2025false} quantify re-identification risk using subjective privacy-loss measures, such as LLM-assessed severity of linked auxiliary information; Paudel~\cite{paudel2026sanitization} shows that traditional measures can be overly optimistic for LLM sanitization.

\textbf{Differential Privacy.} DP provides formal guarantees under stated assumptions, whereas our operational test is a model- and corpus-relative diagnostic; it is also complementary and can sanitize segmented image concepts. Zhao and Chen~\cite{zhao2022survey} distinguish pixel-level, latent-vector, and metric DP. Maris et al.~\cite{maris2025differential} extend $k$-anonymity and DP to image re-identification, Pittaluga et al.~\cite{pittaluga2023ldp} privatize entire images, and Xue et al.~\cite{xue2021dp} heuristically select and perturb identifying facial features. Persuasive Privacy~\cite{bon2026persuasive} generalizes forms of DP to measure deterministic systems semantically in terms of adversarial knowledge, but not linguistically semantic features. Other alternatives shuffle pixels~\cite{li2024visualmixer} or alter faces to be perceptually similar yet unlinkable~\citep{chen2021perceptual}. Li et al.~\cite{li2013membership} model the background knowledge needed for re-identification as a distribution; our approach instead represents such knowledge with a trained model.

\textbf{Privacy of Video and Audio.} Video work privatizes temporal blocks~\cite{li2023stprivacy,aslam2026pixels}, obfuscates faces across frames~\cite{rosberg2023fiva} or in real time~\cite{huh2025novel}, and uses GANs to preserve utility; audio work obfuscates recordings~\cite{perero2022x,lim2022overo,ghosh2024improving,bibbo2025speech}.

\section{Conclusion}
We introduced \textit{contrastive privacy}, a formal definition and algorithm for auditing sanitized renderings with respect to an abstract privacy target $\gamma$, producing model-relative resolutions and decisive witnesses. The auditor specifies $\gamma$ but need not enumerate its potentially open-ended constituent properties or compute $\mathcal{X}_\gamma$. This trades enumeration for explicit application-specific conditions: semantic closure and diverse representation formally, and adequate proxy coverage and mechanism fidelity operationally. These conditions are often plausible when $c$ includes a direct cue, such as a face, name, or logo, and references capture target-related properties. Embeddings provide a natural implementation but are not required. Under the formal assumptions, passing the idealized full-reference test at a resolution where semantic closure holds makes every sanitized rendering in $\widehat{\mathbb{X}}$ private with respect to $\gamma$; a finite proxy-set pass is only a model- and corpus-relative diagnostic. Across image, text, and Enron evaluations, witnesses revealed retained identifiers and identifying context, including after name removal. Identity-specific fine-tuning exposed signals missed by the base distance mechanism and guided broader sanitization. Three distance mechanisms produced broadly correlated rankings across 94 DiCaprio sanitizations but differed at resolution 0, supporting heterogeneous ensembles for cautious assessments while leaving shared blind spots and adversarial SSMs to future work.

\appendix

\begin{ethics}
A stakeholder-based ethics analysis did not identify significant ethical concerns. We considered the research team, society at large, potential users of our system, and people depicted in images or discussed in texts used in our evaluations. This analysis was conducted by the authors rather than an independent analyst. The work did not require, and did not receive, review by an institutional review board or another external ethics panel.

\begin{itemize}
    \item \textbf{Respect for Persons.} Our evaluations use images depicting real people and passages written by real people or discussing them. We used only images obtained from Wikimedia Commons and publicly available text corpora, and we are unaware of efforts by the relevant individuals to retract those materials.
    \item \textbf{Potential Misuse and Bias.} An operator could rely on our test to claim that a dataset is private while concealing violations that it identifies. Sanitization can also bias images or text by selectively retaining content, whether unintentionally or through a malicious operator's deliberate removal of more content than necessary. We cannot prevent these forms of misuse.
    \item \textbf{Respect for Law and the Public Interest.} Our work may help government officials and other entities evaluate compliance with legal requirements for sanitizing data before release.
    \item \textbf{Beneficence.} Our work seeks to improve privacy; this benefit must be weighed against the potential harms described above. Overall, we judge the expected benefits to outweigh these risks.
\end{itemize}

The primary goal of our work is to \emph{measure} how mechanisms sanitize private information from data. A central risk is that the diagnostic could be deployed and relied upon beyond its supported assumptions. We therefore state its limitations throughout the paper.
\end{ethics}

\begin{openscience}
Artifacts for the image, Reddit, and Enron experiments and the core analysis are available in a public repository at \url{https://github.com/umass-forensics/contrastive-privacy}:
\begin{itemize}
    \item The image, Reddit, and processed Enron data are in the top-level \path{data} folder; Wikimedia Commons source pages and licenses are in \path{image_licenses.csv}.
    \item Scripts for the core analysis are in \path{src/contrastive_privacy/scripts}.
    \item The fine-tuning script is \path{src/contrastive_privacy/scripts/finetune_clip.py}.
    \item Commands for the released experiments are in \path{experiments.sh}.
    \item Detailed setup and execution instructions, as well as an explanation of the parameters, are in \path{README.md}.
\end{itemize}
\end{openscience}

\begin{ai}
As described in Section~\ref{sec:experiments}, AI-based tools were used as part of the research methodology to generate and sanitize images and text and to generate candidate sanitization concepts. Figures~\ref{fig:dicaprio_privacy_failures}, \ref{fig:avengers_GPT54_failure}, \ref{fig:priv_dicaprio}, \ref{fig:not_priv_mcdonalds}, \ref{fig:priv_mcdonalds_49}, \ref{fig:dicaprio_failure_compare}, and~\ref{fig:mcdonalds_failure_compare} display their sanitization outputs. \texttt{Nano Banana 2} generated the synthetic child galleries, Appendix~\ref{app:sc} displays separate images generated with \texttt{GPT-4o}, and \texttt{GEM31p} generated the synthetic descriptions used in Table~\ref{tab:resolution_analysis}. Generative AI-based tools were also used to develop code implementing the experimental tools and workflows used throughout the paper, revise prose, improve flow, and correct language. The authors assume responsibility for the accuracy and integrity of all AI-assisted output used in this work.
\end{ai}

\bibliographystyle{abbrvnat}
\bibliography{references}

\clearpage

\section{Examples of Semantically Connected Concepts}
\label{app:sc}
  
We use CLIP ViT-L/14 to illustrate why apples and oranges may be considered semantically connected, whereas oranges and iPhones are not. Here we remove the objects instead of blacking out pixels. All images in this section were generated with \texttt{GPT-4o}. In Example~1, the displayed pair provides evidence that apples and oranges are semantically connected, but the same inequalities would need to hold for every pair in $\mathbb{X}(c)\times\mathbb{X}(d)$ to establish semantic connectedness under Definition~\ref{def:sem_conn_det}. By contrast, the single pair in Example~2 suffices to disprove semantic connectedness because it violates one of the required inequalities.

\noindent \textbf{Example 1. Evidence That $c$=Oranges and $d$=Apples Are Connected.}

\noindent    
\begin{center}
\resizebox{0.99\linewidth}{!}{
\begin{tabular}{rrrl}
\toprule$x$=\includegraphics[width=1.15cm]{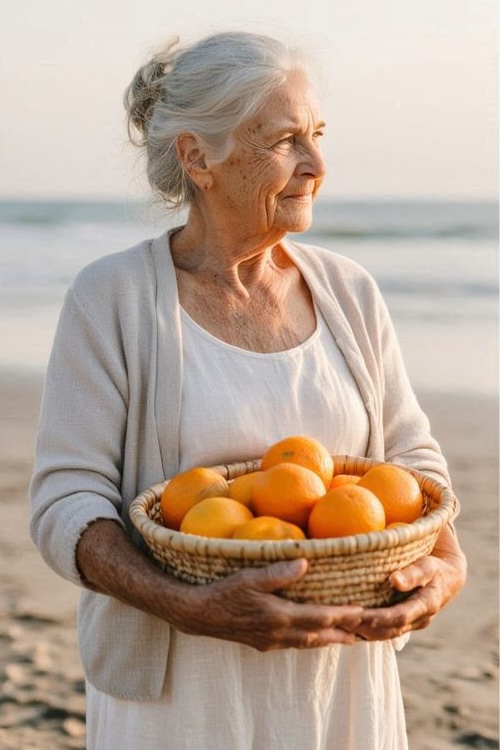} &
$y$= \includegraphics[width=1.3cm]{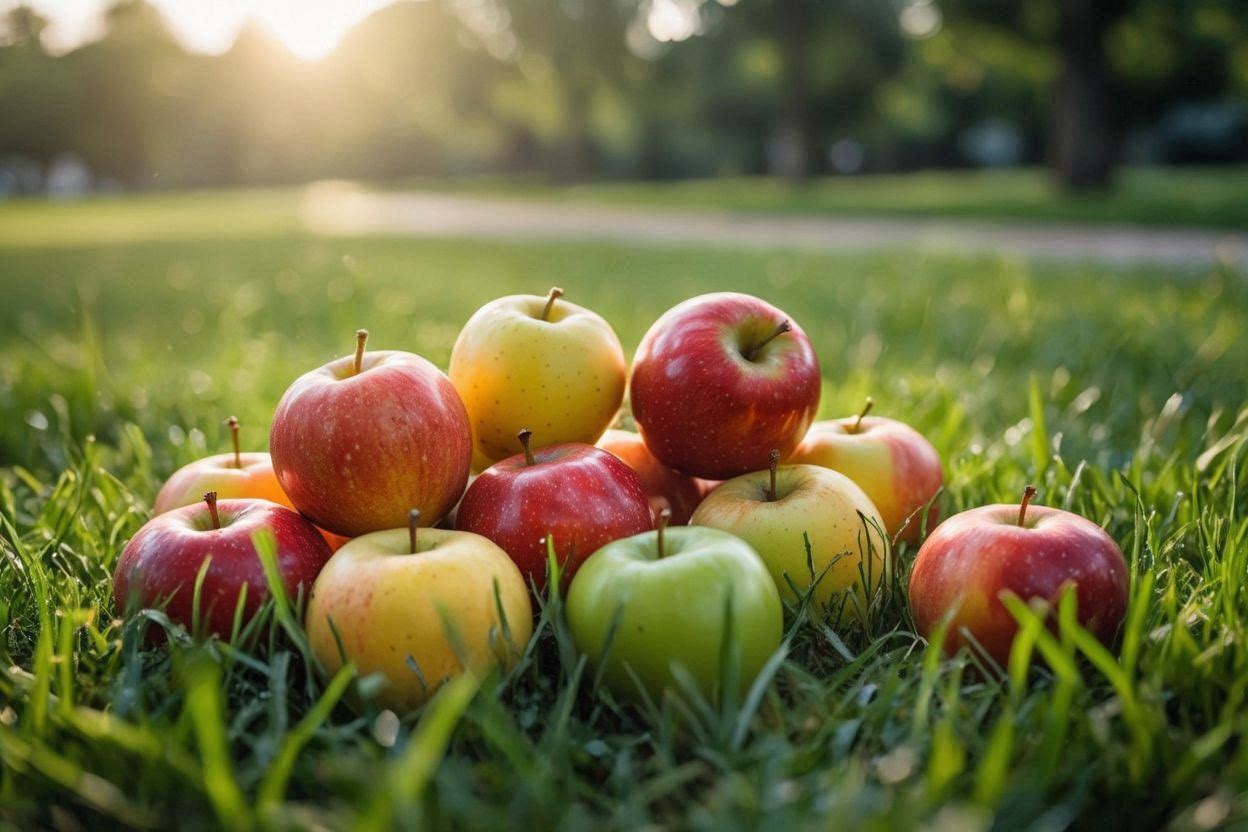} &
$\mathcal{D}(x, y)=$&
0.33\\\midrule

$\mathcal{X}_{c}(x)$=\includegraphics[width=1.15cm]{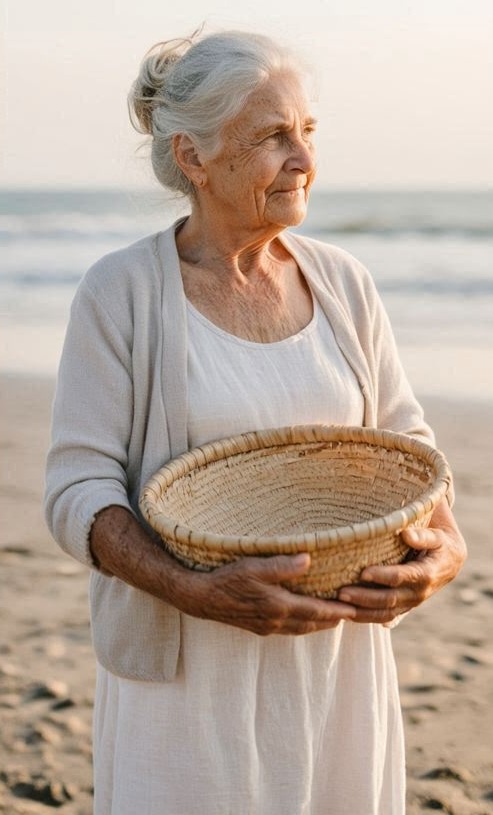} &
$y$= \includegraphics[width=1.3cm]{images/apples_with.jpg} &
$\mathcal{D}(\mathcal{X}_{c}(x), y)=$ &
0.41\\\midrule

$x$=\includegraphics[width=1.15cm]{images/oranges_with.jpg} &
$\mathcal{X}_{d}(y)$= \includegraphics[width=1.3cm]{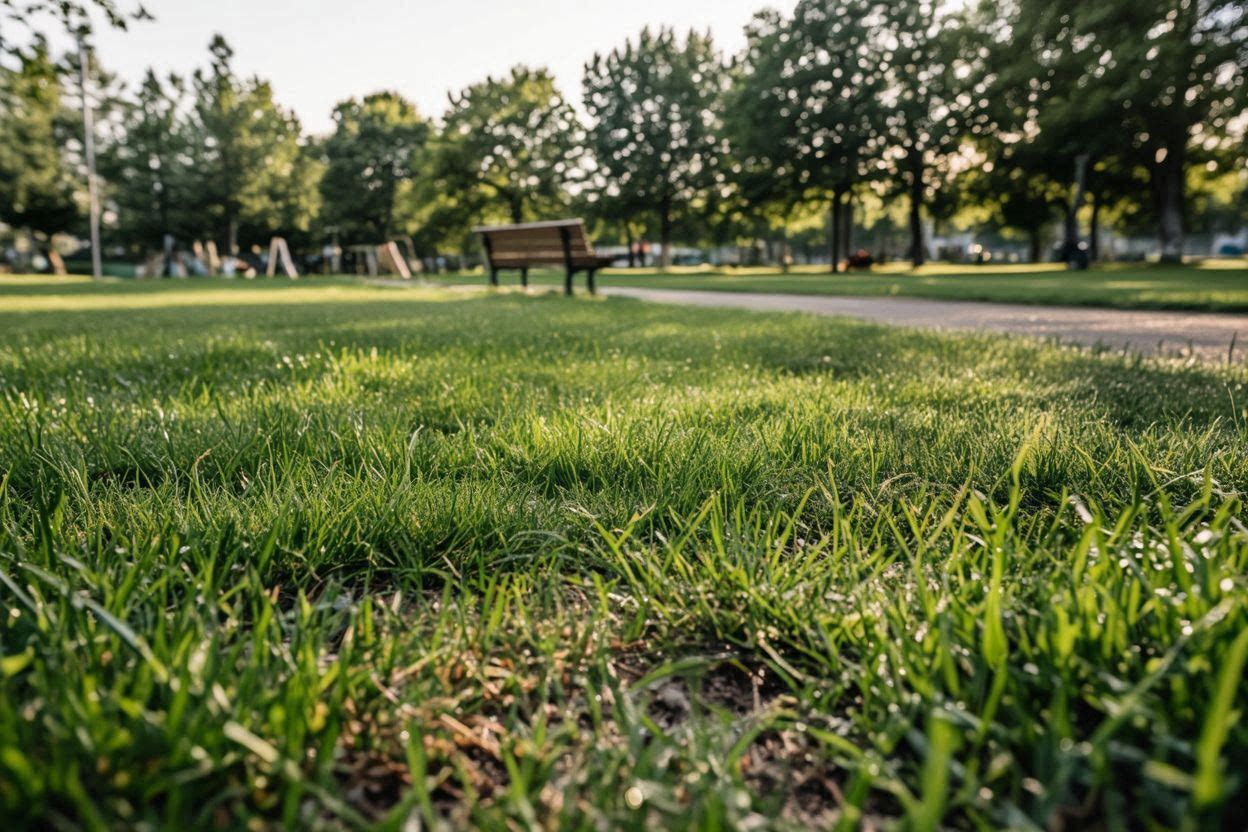}&
$\mathcal{D}(x, \mathcal{X}_{d}(y))=$ &
0.42\\\bottomrule
\end{tabular}
}
\end{center}

\smallskip
\noindent For this pair, both inequalities required for semantic connectedness hold for an appropriate $\delta \geq 0$:
\begin{align*}
\mathcal{D}(x,y) + \delta < \mathcal{D}(\mathcal{X}_{c}(x), y) &\Rightarrow 0.33 + \delta < 0.41 \\
\mathcal{D}(x,y) + \delta < \mathcal{D}(x, \mathcal{X}_{d}(y)) &\Rightarrow 0.33 + \delta < 0.42
\end{align*}
The equations hold for any $0 \leq \delta < 0.08$.

\smallskip
\noindent\textbf{Example 2. $c$=Oranges and $d$=iPhones Are Not Connected.}

\noindent
\begin{center}
\resizebox{0.99\linewidth}{!}{
\begin{tabular}{rrrl}
\toprule$x$=\includegraphics[width=1.15cm]{images/oranges_with.jpg} &
$y$= \includegraphics[width=1.15cm]{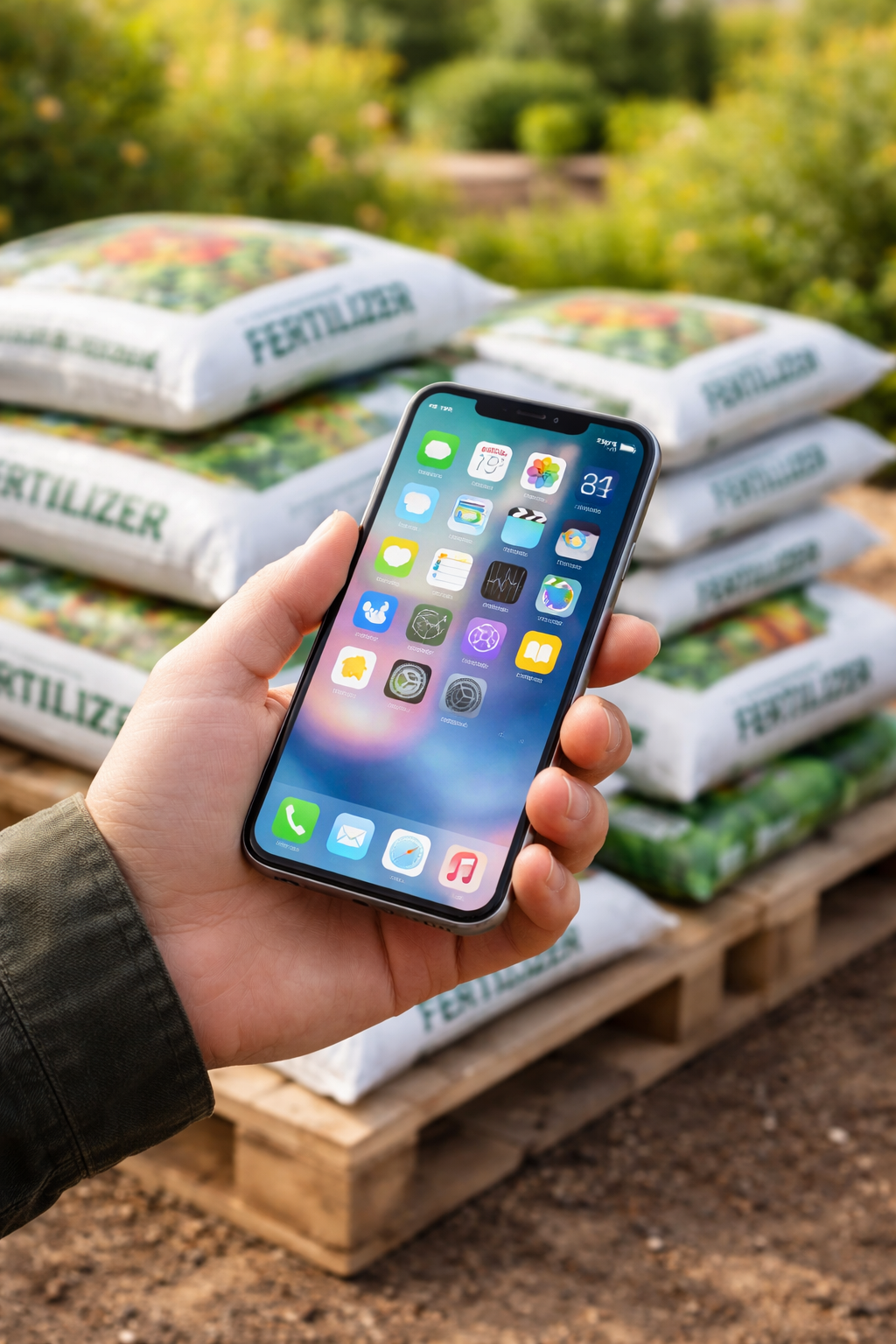} &
$\mathcal{D}(x, y)=$&
0.42\\\midrule

$\mathcal{X}_{c}(x)$=\includegraphics[width=1.15cm]{images/oranges_without.jpg} &
$y$= \includegraphics[width=1.15cm]{images/iphone_with.png} &
$\mathcal{D}(\mathcal{X}_{c}(x), y)=$ &
0.49\\\midrule

$x$=\includegraphics[width=1.15cm]{images/oranges_with.jpg} &
$\mathcal{X}_{d}(y)$= \includegraphics[width=1.15cm]{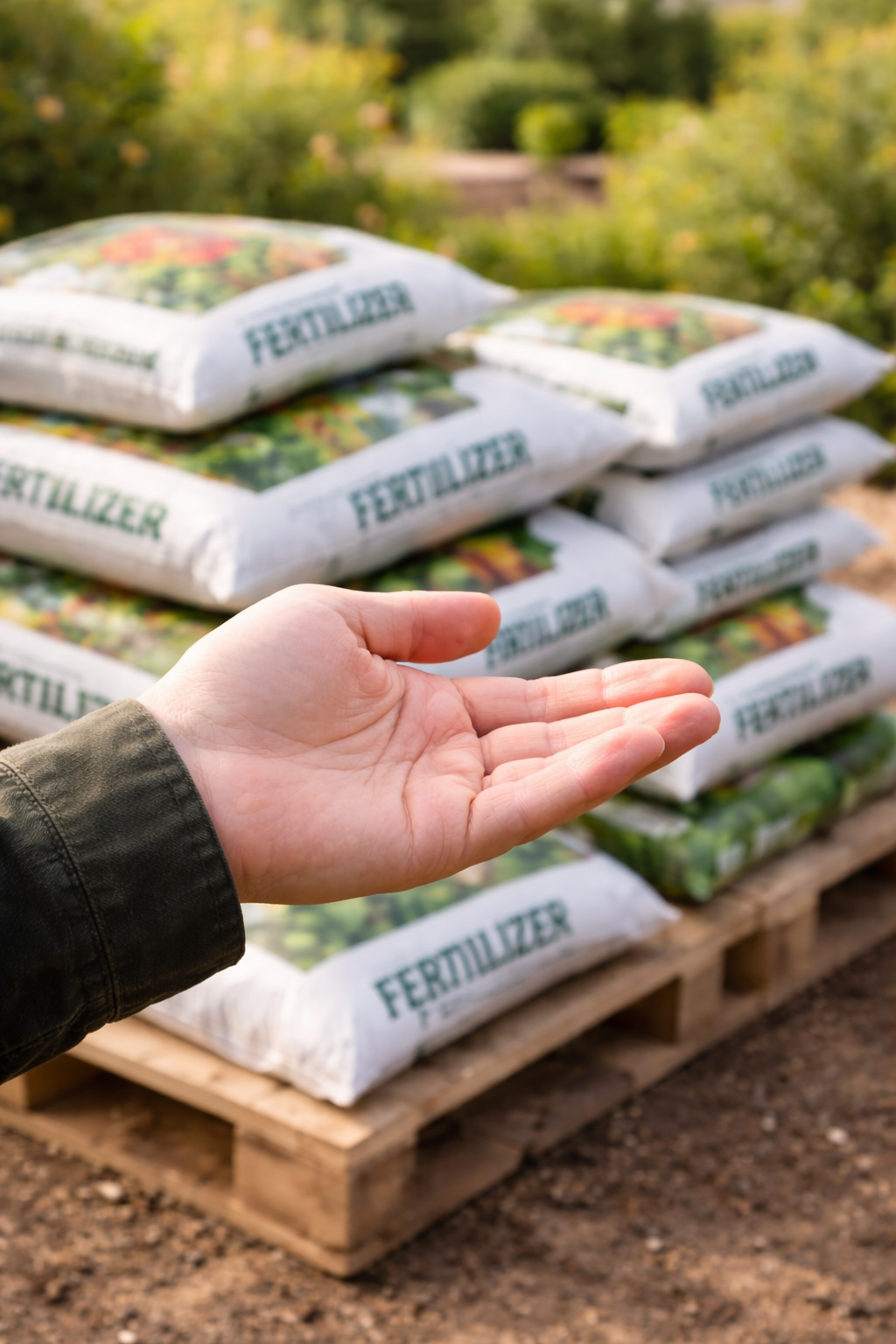}&
$\mathcal{D}(x, \mathcal{X}_{d}(y))=$ &
0.37\\\bottomrule
\end{tabular}
}
\end{center}

\smallskip
\noindent To disprove semantic connectedness, it is enough to show that no $\delta \geq 0$ makes both Eqs.~\ref{eq:conn1} and~\ref{eq:conn2} hold for this pair.
\begin{align*}
\mathcal{D}(x,y) + \delta < \mathcal{D}(\mathcal{X}_{c}(x), y) &\Rightarrow 0.42 + \delta < 0.49 \\
\mathcal{D}(x,y) + \delta < \mathcal{D}(x, \mathcal{X}_{d}(y)) &\Rightarrow 0.42 + \delta < 0.37
\end{align*}
No such $\delta \geq 0$ exists.

\FloatBarrier
\section{Mechanisms}
\label{sec:mechanisms}

\begin{table}[h]
\caption{
Distance and Sanitization Mechanisms.
Unless otherwise noted, \texttt{SAM3} sets the following parameters: threshold = 0.2, blur = 5, and dilate = 25. \texttt{GLN2} sets threshold = 0.05. Dashes indicate that the underlying model is unspecified.}
\label{tab:mechanisms}
\centering
\renewcommand{\arraystretch}{1}
\setlength{\tabcolsep}{2pt}
\footnotesize
\begin{adjustbox}{max width=\columnwidth}
\begin{tabular}{c l l l}
\toprule
\textbf{Mechanism} & \textbf{Name} & \textbf{Model} & \textbf{Modality} \\
\midrule
  & \texttt{EVA}  & EVA-CLIP-18B  & images \\
  \textit{Distance} & \texttt{SIGLIP2} & SigLIP 2 Giant (base) & images \\
  $(\mathcal{D})$ & \texttt{SIGLIP2-FT10} & SigLIP 2 Giant (epoch 10) & images \\
  & \texttt{METACLIP2} & MetaCLIP 2 Giant & images \\
  & \texttt{QWEN3}  & Qwen3-Embedding-8B  & text \\
  & \texttt{QWEN3-4B}  & Qwen3-Embedding-4B  & text \\
\midrule

  & \texttt{GPT54} & GPT 5.4  & i2t, t2t  \\
  & \texttt{GPT54m} & GPT 5.4 Mini  & t2t  \\
  & \texttt{GPT52} & GPT 5.2  & t2t  \\
  & \texttt{GPT4o} & GPT 4o  & t2t  \\
  & \texttt{GPT4om} & GPT 4o Mini  & t2t  \\
  & \texttt{GEM31p} & Gemini 3.1 Pro & i2t, t2t \\
  & \texttt{GEM3f} & Gemini 3 Flash & t2t \\
  & \texttt{GEM25p} & Gemini 2.5 Pro & t2t \\
  & \texttt{GEM25f} & Gemini 2.5 Flash & t2t \\
  & \texttt{GEM2f} & Gemini 2 Flash & t2t \\
  & \texttt{OPS46} & Opus 4.6 & i2t, t2t  \\
  & \texttt{OPS41} & Opus 4.1 & t2t  \\
  \textit{Sanitization} & \texttt{SON46} & Sonnet 4.6 & t2t  \\
  $(\mathcal{X}_c)$ & \texttt{HAK45} & Haiku 4.5 & t2t  \\
  & \texttt{HAK35} & Haiku 3.5 & t2t  \\

 & \texttt{iGPT15} & GPT Image 1.5 & i2i \\
 & \texttt{iGPT1m} & GPT Image 1 Mini & i2i \\
 & \texttt{iGEM3p} & Gemini 3 Pro Image & i2i \\
  & \texttt{iGEM31f} & Gemini 3.1 Flash Image & i2i \\
  & \texttt{FLX2p} & FLUX 2 Pro & i2i \\
  & \texttt{FLX2d} & FLUX 2 Dev & i2i \\
  & \texttt{SAM3} & SAM3 & i2i \\
  & \texttt{GLN2} & GLiNER2 base v1 & t2t \\
  & \texttt{Privacy Filter} & -- & t2t \\
  & \texttt{AWS Comprehend} & -- & t2t \\
  & \texttt{Microsoft Presidio} & -- & t2t \\
  & \texttt{Google Cloud DLP} & -- & t2t \\
\bottomrule
\end{tabular}
\end{adjustbox}
\end{table}

\FloatBarrier

\section{Obfuscation Concepts}
\label{sec:concepts}

\begin{concept}
\label{concepts:dicaprio_manual}
    face, flags, globe hair, logo, map, neck 
\end{concept}

\begin{concept}
\label{con:dicaprio_gpt-5.4}
accessories, age appearance, beard, body build, body shape, body silhouette, bodyguard, celebrity face, celebrity name, clothing, distinctive clothing, distinctive features, distinctive profile, event backdrop, event signage, event text, event title, eye color, face, faces, facial features, facial hair, festival logo, festival text, formal attire, formal clothing, formal suit, glasses, group composition, group lineup, hair color, hairstyle, hairstyles, headwear, jewelry, lapel pin, logo, military insignia, name badge, name plaque, name tag, person, person name, portrait, pose, profile, profile view, red carpet, red carpet photo, red carpet setting, service medals, side profile, smile, sponsor logo, suit, text, tie
\end{concept}

\begin{concept}
\label{con:dicaprio_gemini-3.1-pro}
award, badge, face, facial hair, glasses, head, id badge, logo, medals, military uniform, name tag, nameplate, person, plaque, sign, text
\end{concept}

\begin{concept}
\label{con:dicaprio_opus-4.6}
backdrop text, badge, beard, body shape, brand logo, building facade, building plaque, building sign, bust sculpture, clothing, clothing style, credential badge, document, emblem, event backdrop, event banner, event logo, event name, event signage, event text, eyeglasses, face, facial features, facial hair, festival name, flag, flag pin, goatee, government building, government emblem, government seal, hairstyle, hashtag, lanyard, lapel pin, large screen, military uniform, movie title, name placard, name plaque, name tag, name text, national flags, naval uniform, official seal, organization text, plaque, podium, portrait painting, premiere background, projected text, red carpet, seal, sponsor logo, sponsor text, stage backdrop
\end{concept}

\begin{concept}
\label{con:mcdonalds_gpt-5.2}
    advertisement banner, architectural design, barcode, billboard ad, brand colors, brand icon, brand icons, brand mark, brand mascots, brand name, brand signage, brand slogan, brand symbol, brand text, branded bags, branded colors, branded cups, branded icons, branded kiosk, branded packaging, building facade, building facade text, building logo, cafe sub-brand, character image, counter decals, counter text, cup branding, cup design, cup logo, delivery box branding, distinctive arches, distinctive color scheme, distinctive colors, distinctive icon, door stickers, drink cup, drink cup logo, drink cups, drink machine branding, drive-thru sign, employee name tag, employee uniform logo, flag text, food ads, food container, food images, food packaging, food wrapper, food wrappers, fry carton, golden arches, hat logo, hours sign, interior branding colors, interior decor, item names, logo, map graphic, mascot characters, mascot costume, mascot icon, mascot statue, menu advertisement, menu board, menu boards, menu branding, menu flyer, menu images, menu item name, menu item text, menu items, menu posters, menu prices, menu screen, menu screens, menu sign, menu text, on-screen text, order screen, packaging, packaging design, packaging logo, packaging text, paper bag, paper bags, phone number, play area sign, prices, printed pattern, product branding, product icons, product name, product names, product packaging, product photos, promotional banner text, promotional poster, promotional posters, promotional screens, promotional signage, promotional text, receipt, receipt header, receipt kiosk, receipt text, restaurant decor, restaurant logo, restaurant name, restaurant sign, restaurant signage, restaurant text, roof sign, sauce packet, self-order kiosk branding, shirt logo, signage, signature burger, signature colors, slogan text, staff uniform, staff uniform logo, store name text, store number, store sign, store signage, storefront sign, storefront signage, storefront text, sub-brand logo, tagline, tagline text, tray liner, tray liners, uniform, uniform branding, uniform logo, uniforms, wall sign, wall signage, wall text, window decal, window decals, window signage, window text, wrapper design, wrapper text
\end{concept}

\begin{concept}
\label{con:mcdonalds_gpt-5.4}
    awning text, beverage machine branding, billboard, brand colors, brand name, brand pattern, brand sign, brand slogan, brand symbol, brand text, branded bag, branded colors, branded flag, branded packaging, building signage, burger icon, burger image, cafe branding, cafe sign, character makeup, chinese text, counter signage, coupon, cup design, delivery scooter branding, digital screens, door text, drink cup, drink cup logo, drink cups, drive-thru sign, employee uniform, flag text, food advertisement, food container, food icons, food item name, food packaging, food poster, food text, food wrapper, fries icon, fries image, fry carton, golden arches, hat logo, hours sign, interior decor, item names, kiosk screen, location text, logo, map, mascot, mascot costume, mccafe sign, mccafé sign, meal names, menu board, menu boards, menu branding, menu flyer, menu item name, menu item names, menu item text, menu items, menu photos, menu placemat, menu poster, menu screen, menu sign, menu text, name tag, order display, packaging, packaging design, paper bag, paper bags, payment terminal, pickup counter, poster, price text, product name, promotional poster, promotional sign, promotional text, receipt, receipt text, restaurant branding, restaurant colors, restaurant entrance, restaurant facade, restaurant interior, restaurant name, restaurant sign, restaurant signage, roadside sign, sauce packet, self-order kiosk, sign, slogan, store directory, store locations, store number, store sign, storefront, storefront branding, storefront logo, storefront sign, storefront signage, storefront text, text, tray liner, uniform, uniform emblem, uniform logo, wall graphics, wall sign, wall signage, wall text, website text, welcome message, window decal, window decals, window graphics, window posters, window text, wrapper, wrapper text
\end{concept}

\begin{concept}
\label{con:mcdonalds_gemini-3.1-pro}
    advertisement, architectural feature, billboard, brand colors, brand name, brand pattern, brand symbol, branded packaging, food packaging, logo, mascot, mascot costume, menu board, menu items, menu text, name tag, packaging, packaging design, product name, promotional text, receipt, restaurant sign, sign, slogan, store sign, storefront sign, text, uniform, uniform logo, wall graphics
\end{concept}

\begin{concept}
\label{con:mcdonalds_opus-4.6}
    arches symbol, architectural arches, awning, banner, beverage brand, beverage dispenser logo, brand color, brand color scheme, brand mascot, brand name, brand name text, brand sign, brand symbol, brand text, branded bag, branded cap, branded clothing, branded cup, branded decor, branded lanyard, branded packaging, branded shirt, branded text, branded uniform, branded wrapper, building design, building sign, building signage, burger box, burger wrapper, cafe sign, character figure, chicken box, chinese text, clown costume, clown figure, color scheme, coupon, cup design, digital signage, directional signage, drink cup, drink dispenser, drive thru sign, drive-through sign, drive-thru sign, employee uniform, flag, food imagery, food wrapper, fry container, golden arches, hat logo, illuminated sign, interior decor, letter emblem, logo, mascot costume, mascot statue, menu board, menu display, menu item names, menu item text, menu items, name badge, operating hours sign, order display, order screen, order screens, packaging, packaging text, play area sign, playground structure, pole sign, poster, price signage, product name, promotional banner, promotional poster, promotional posters, promotional sign, promotional signage, receipt, red wig, roadside sign, sandwich wrapper, sauce packet, self-order kiosk, shirt logo, shop address, sign, signage text, slogan, storefront sign, storefront text, striped sleeves, text, text signage, trademark symbol, tray liner, uniform, uniform logo, wall text, website text, welcome text, window decal, wrapper text, yellow jumpsuit
\end{concept}

\begin{concept}
\label{con:mcdonalds_sample_gpt-5.4}
    awning branding, billboard, brand character, brand colors, brand icon, brand logo, brand name, brand sign, brand text, combo meal names, decor pattern, drink cup, drive-thru sign, food box, food packaging, french fry carton, fry carton, golden arches, hat logo, logo, mascot costume, mascot makeup, menu board, menu item names, menu poster, menu text, packaging, packaging logo, playground branding, prices, promotional display, promotional poster, receipt, restaurant interior, restaurant sign, restaurant signage, restaurant text, slogan, storefront sign, storefront text, text, uniform, uniform logo, wall signage, window decal, wrapper text
\end{concept}

\begin{concept}
\label{con:mcdonalds_sample_gemini-3.1-pro-preview}
    architectural feature, brand name, branded packaging, logo, mascot costume, menu board, menu item name, menu item names, name tag, promotional, restaurant sign, sign, slogan, store sign, uniform logo
\end{concept}

\begin{concept}
\label{con:mcdonalds_sample_claude-opus-4.6}
    architectural arches, awning, banner, beverage brand, brand color scheme, brand mascot, brand name, brand symbol, brand text, branded cap, branded cup, branded packaging, branded shirt, branded uniform, building design, building sign, burger box, chicken box, clown costume, color scheme, drink cup, drink dispenser, drive-through sign, food wrapper, fry container, golden arches, hat logo, interior decor, logo, mascot costume, menu board, menu boards, menu item names, name badge, packaging text, playground structure, pole sign, product name, promotional banner, promotional posters, promotional signage, receipt, red wig, sign, slogan, storefront sign, striped sleeves, tray liner, uniform logo, window decal, yellow jumpsuit
\end{concept}

\begin{concept}
\label{con:mcdonalds_manual}
    arches, clown, drink, face, food, logo, roof, sign, vent
\end{concept}

\begin{concept}
\label{con:avengers_manual}
    comic, movie, publisher, author, superhero, abbreviation, archetype, year, ``remake'', ``blockbuster'', ``battle'', ``first'', ``powerful'', ``defeated'', ``movies'', ``die''
\end{concept}

\FloatBarrier
\section{Obfuscation Prompts}
\label{sec:prompts}

\begin{prompt}
\label{prompt:i2i_nat_con}
    Black out all regions containing: \texttt{OBJECTS}. The blacked-out regions must be solid black. Do not alter any other pixels. When the blacked-out object can be recognized by silhouette alone, blur the edges to ensure it cannot be recognized.    
\end{prompt}

\begin{prompt}
\label{prompt:i2i_abs_con}
    Black out every region that can reveal \texttt{CONCEPT}. The blacked-out regions must be solid black. Do not alter any other pixels. When the blacked-out object can be recognized by silhouette alone, blur the edges to ensure it cannot be recognized.
\end{prompt}

\begin{prompt}
\label{prompt:i2t_con}
    Analyze this image and identify specific concepts (objects, text, visual elements) that should be obfuscated to preserve the privacy of the following target concept:\\ \texttt{CONCEPT}. Please provide a list of general, reusable concepts that appear in this image and would reveal information about the target concept if left visible. Use general, broad concepts whenever possible (e.g., "logo" instead of "McDonald's logo", "sign" instead of "Joe's Restaurant sign", "menu" instead of "specific menu item names"). Return ONLY a JSON array of concept strings, one per line, with no additional text or explanation. Each concept should be a single word or short phrase (2-3 words maximum). Focus on concepts that are directly related to revealing the target concept. Example format: ["logo", "sign", "menu", "brand name", "text"] If no relevant concepts are found, return an empty array: []
\end{prompt}

\begin{prompt}
\label{prompt:avengers_gpt5.2}
    Please redact the following passage by replacing with a single block character (\rule{0.6em}{0.7em}, Unicode U+2588) every word or phrase that is consistent with this privacy concept: \texttt{anything that can identify the movie discussed in this passage}. Rules: (1) Replace only the sensitive tokens with \rule{0.6em}{0.7em}, one character per original character so that the length and layout of the passage are preserved. (2) Do not redact words that do not match the concept. (3) Return only the redacted passage, no explanation. 
\end{prompt}

\clearpage
\suppDicaprioGallery
\suppMcDonaldsMatrix
\suppAvengersMatrix
\FloatBarrier

\clearpage
\twocolumn[{
\begin{minipage}{\textwidth}
\section{Privatizations}
\label{sec:privatizations}
\captionsetup{type=figure}
    \centering
    \begin{subfigure}[t]{.3\textwidth}
        \centering
        \includegraphics[width=\linewidth]{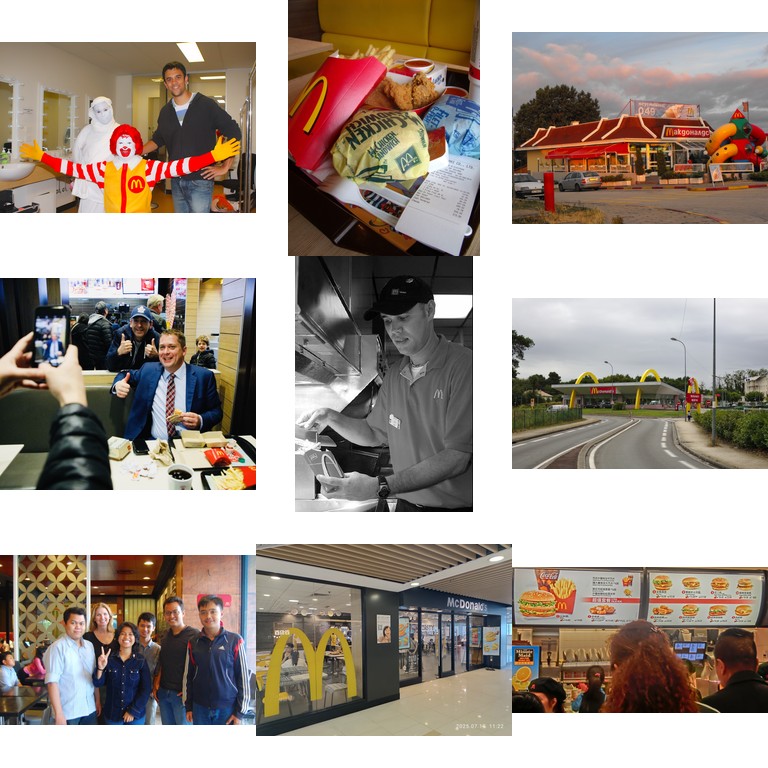} 
        \caption{Unaltered images capturing McDonald's.}
        \label{fig:mcdonalds_orig}
    \end{subfigure} %
    \hspace{0.03\textwidth}%
    \begin{subfigure}[t]{.3\textwidth}
        \centering
        \includegraphics[width=\linewidth]{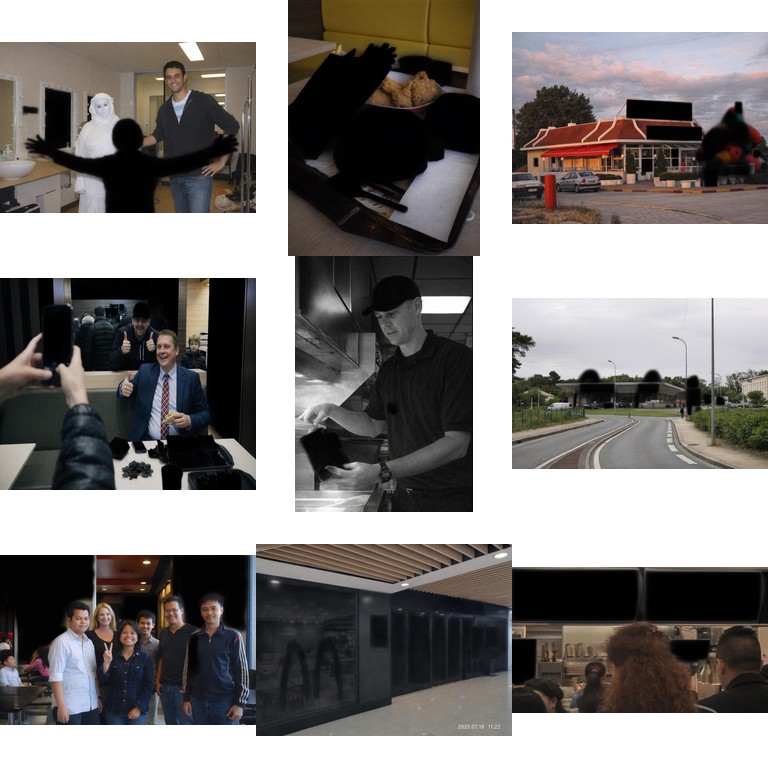}
        \caption{Images obfuscated by \texttt{iGPT15/GPT54}; the test returns resolution 0.05 under \texttt{EVA} and utility 0.63.}
        \label{fig:mcdonalds_openai_obfs}
    \end{subfigure}%
    \hspace{0.03\textwidth}%
    \begin{subfigure}[t]{.3\textwidth}
        \centering
        \includegraphics[width=\linewidth]{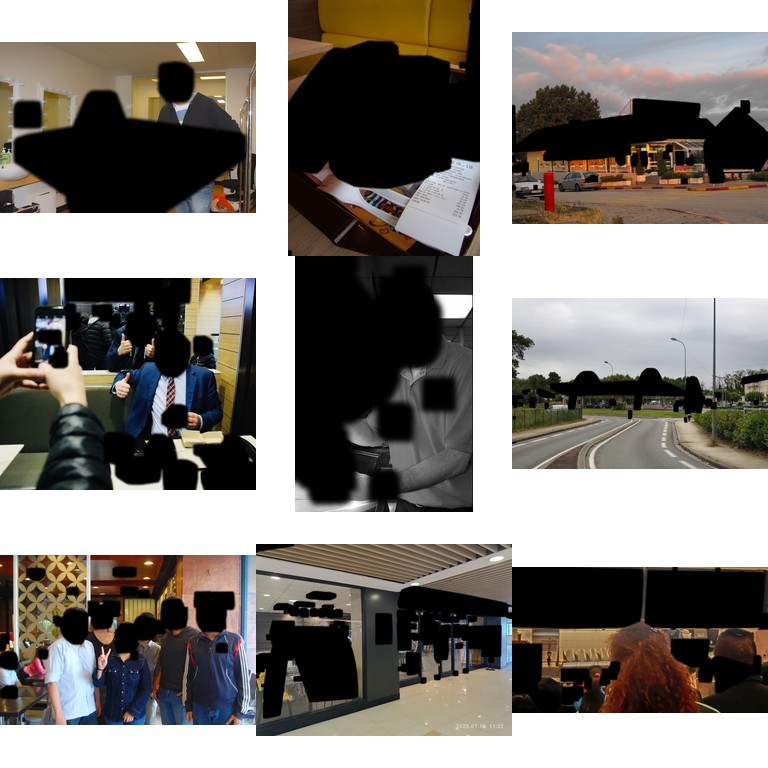}
        \caption{Images obfuscated by \texttt{SAM3/Manual}; the test returns resolution 0 under \texttt{EVA} and utility 0.46.}
        \label{fig:mcdonalds_sam3_obfs}
    \end{subfigure}
    
    \caption{An attempt to sanitize the concept \texttt{the identity of the fast food restaurant} from a collection of nine images capturing the brand McDonald's.}
    \Description{Three nine-image grids show McDonald's images before sanitization and after two mechanisms. The sanitizations obscure different combinations of logos, signs, food, and restaurant features.}
    \label{fig:not_priv_mcdonalds}

\par\medskip
    \centering
    \begin{subfigure}[t]{0.48\textwidth}
        \includegraphics[width=0.9\linewidth, height=6cm]{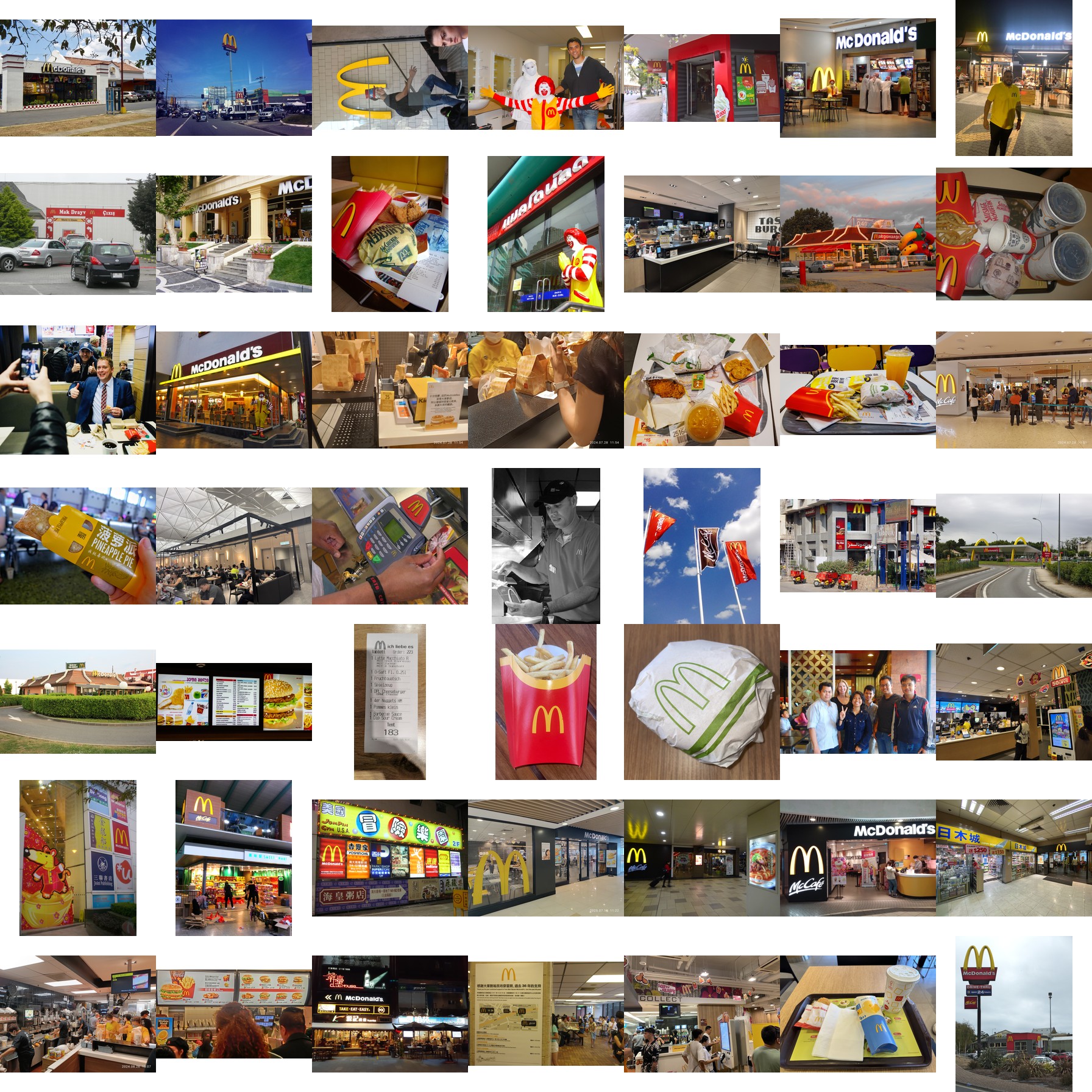} 
        \caption{Unaltered images capturing the McDonald's brand.}
        \label{fig:mcd_49_sum_orig}
    \end{subfigure}%
    \hspace{0.04\textwidth}%
    \begin{subfigure}[t]{0.48\textwidth}
        \includegraphics[width=0.9\linewidth, height=6cm]{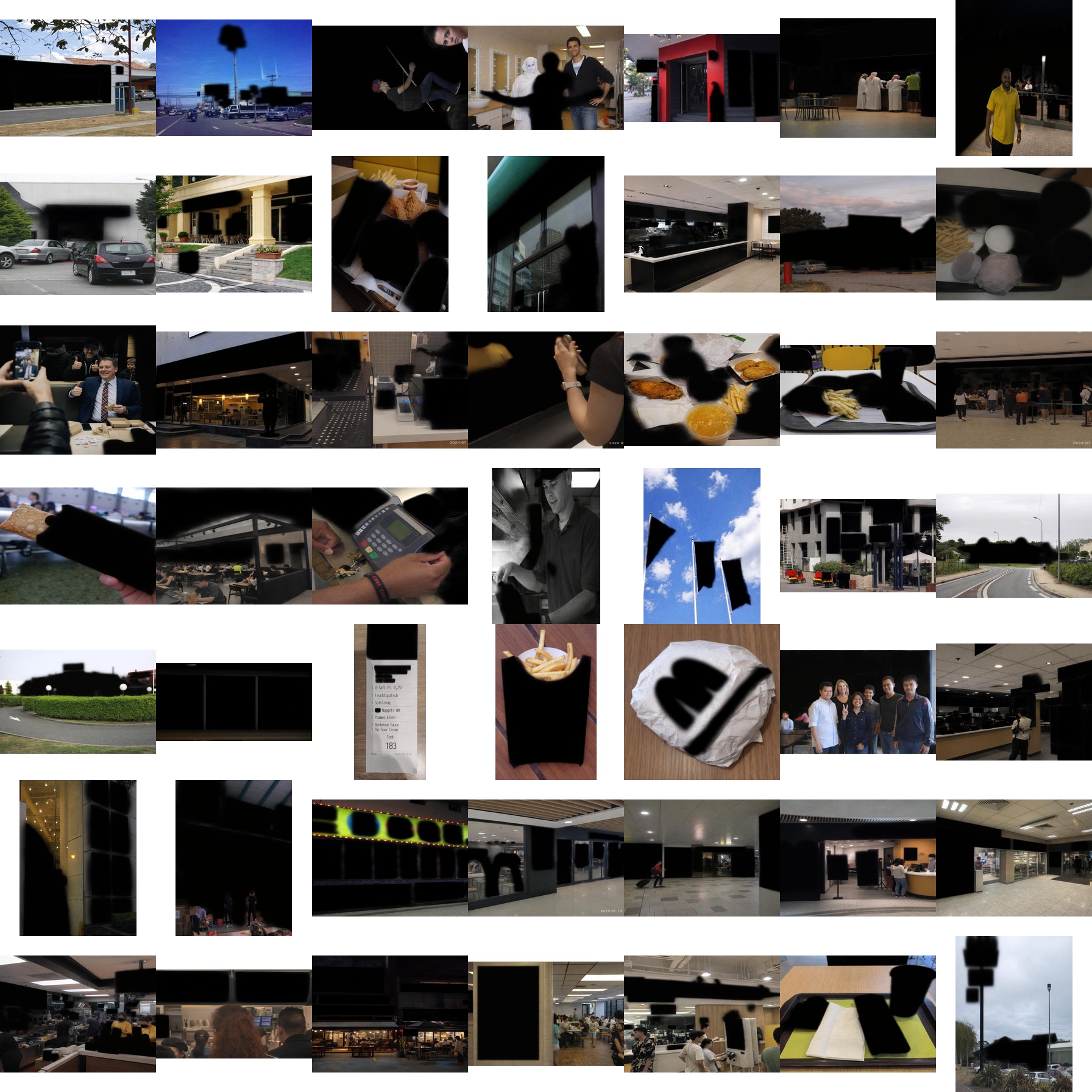}
        \caption{Images with Concept~\ref{con:mcdonalds_gpt-5.2} obfuscated by \texttt{iGPT15}; the test returns resolution 0.17 under \texttt{EVA}.}
        \label{fig:mcd_49_sum_obfs}
    \end{subfigure}
    
    \caption{A sanitized collection of 49 images depicting McDonald's. Resolutions are calculated using the \texttt{EVA} distance mechanism.}
    \Description{Two grids show 49 McDonald's images before sanitization and after an image model obfuscates concepts generated by a language model.}
    \label{fig:priv_mcdonalds_49}
\end{minipage}
}]

\clearpage
\twocolumn[{
\begin{minipage}{\textwidth}
\section{Decisive Witnesses}
\label{sec:privacy_failures}
\captionsetup{type=figure}
\centering  
\small

\setlength{\aboverulesep}{1pt}
\setlength{\belowrulesep}{1pt}
\setlength{\fboxrule}{1.5pt}

\resizebox{0.99\textwidth}{!}{
\begin{tabular}{@{} l @{\hspace{3ex}} c c c c c @{}}
\toprule
Model 
& \thead{\texttt{GPT54}}
& \thead{\texttt{GEM31p}} 
& \thead{\texttt{OPS46}}
& \thead{\texttt{Manual}}
& \thead{None} \\
\midrule
\addlinespace[1.5ex]

\texttt{iGPT15}
& \includegraphics[width=0.17\textwidth, valign=m]{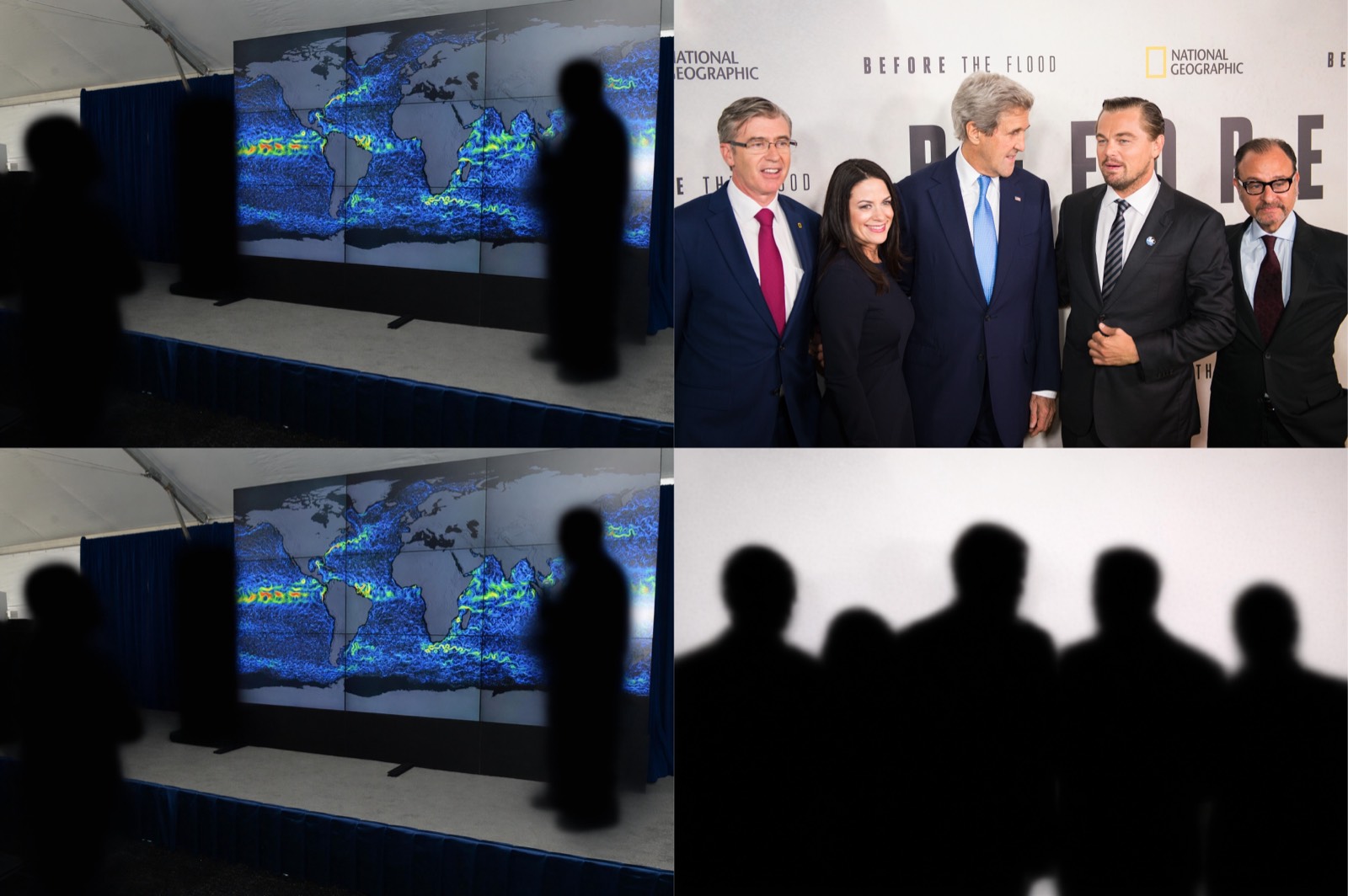}
& \includegraphics[width=0.17\textwidth, valign=m]{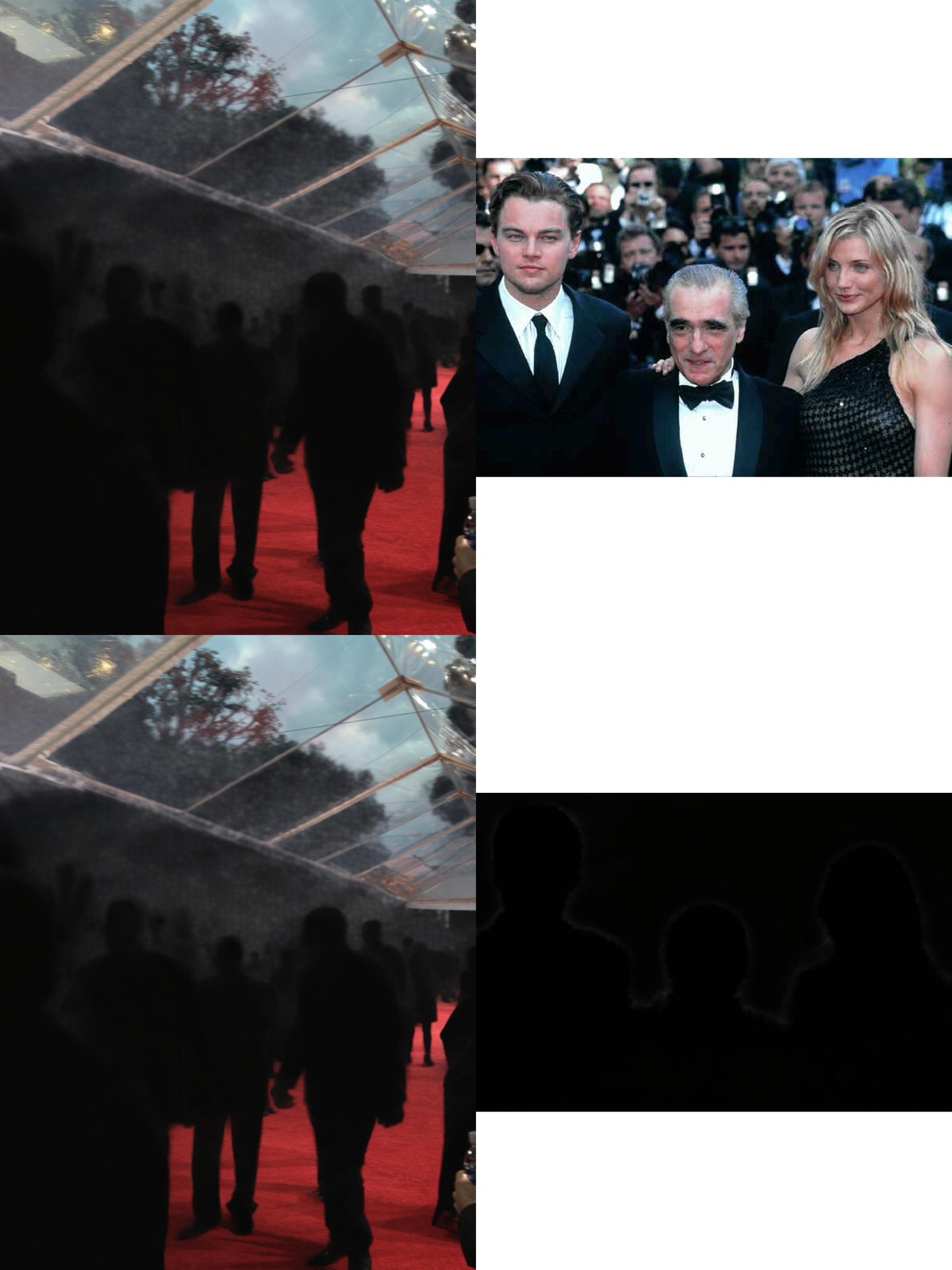}
& \includegraphics[width=0.17\textwidth, valign=m]{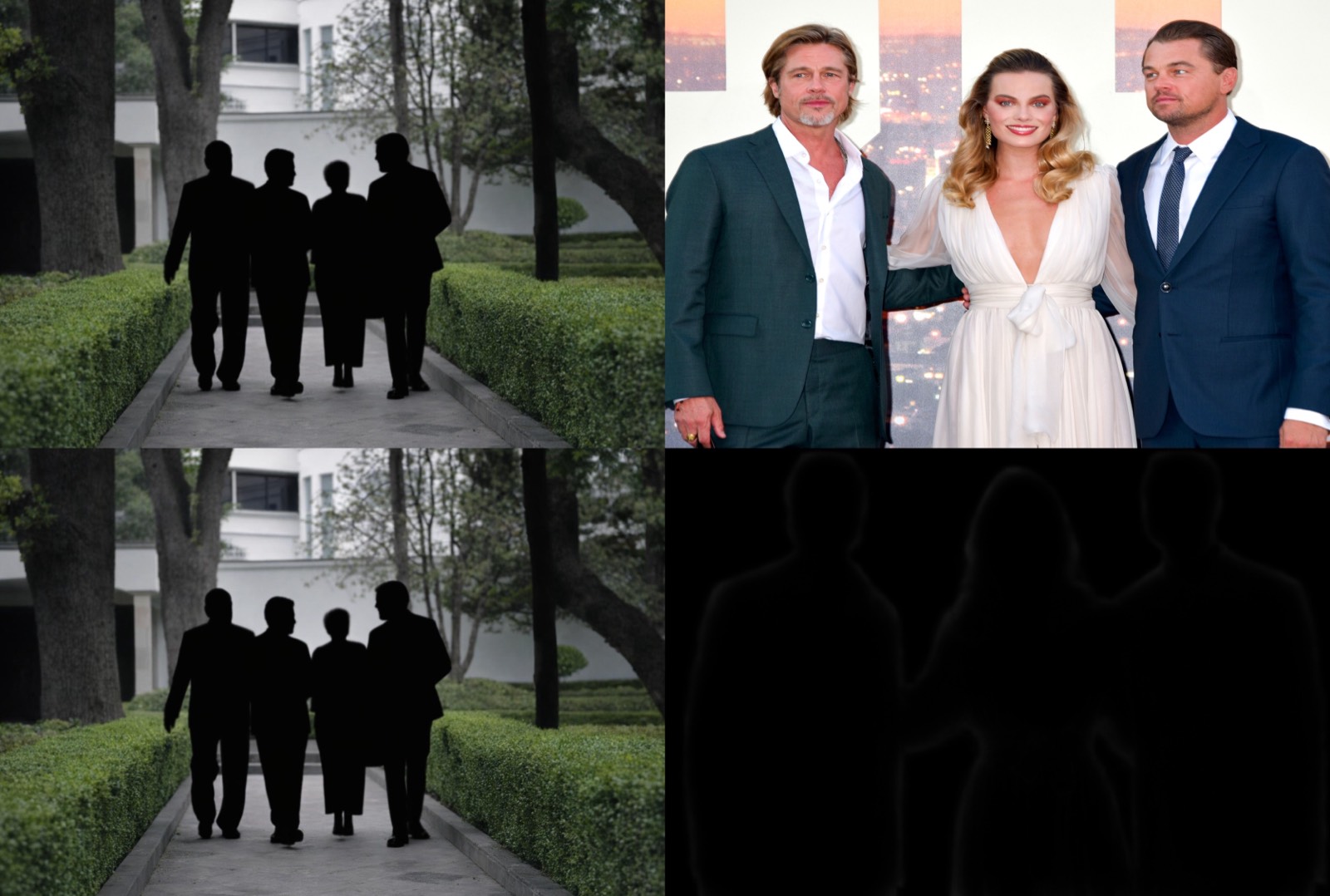}
& 
& \includegraphics[width=0.17\textwidth, valign=m]{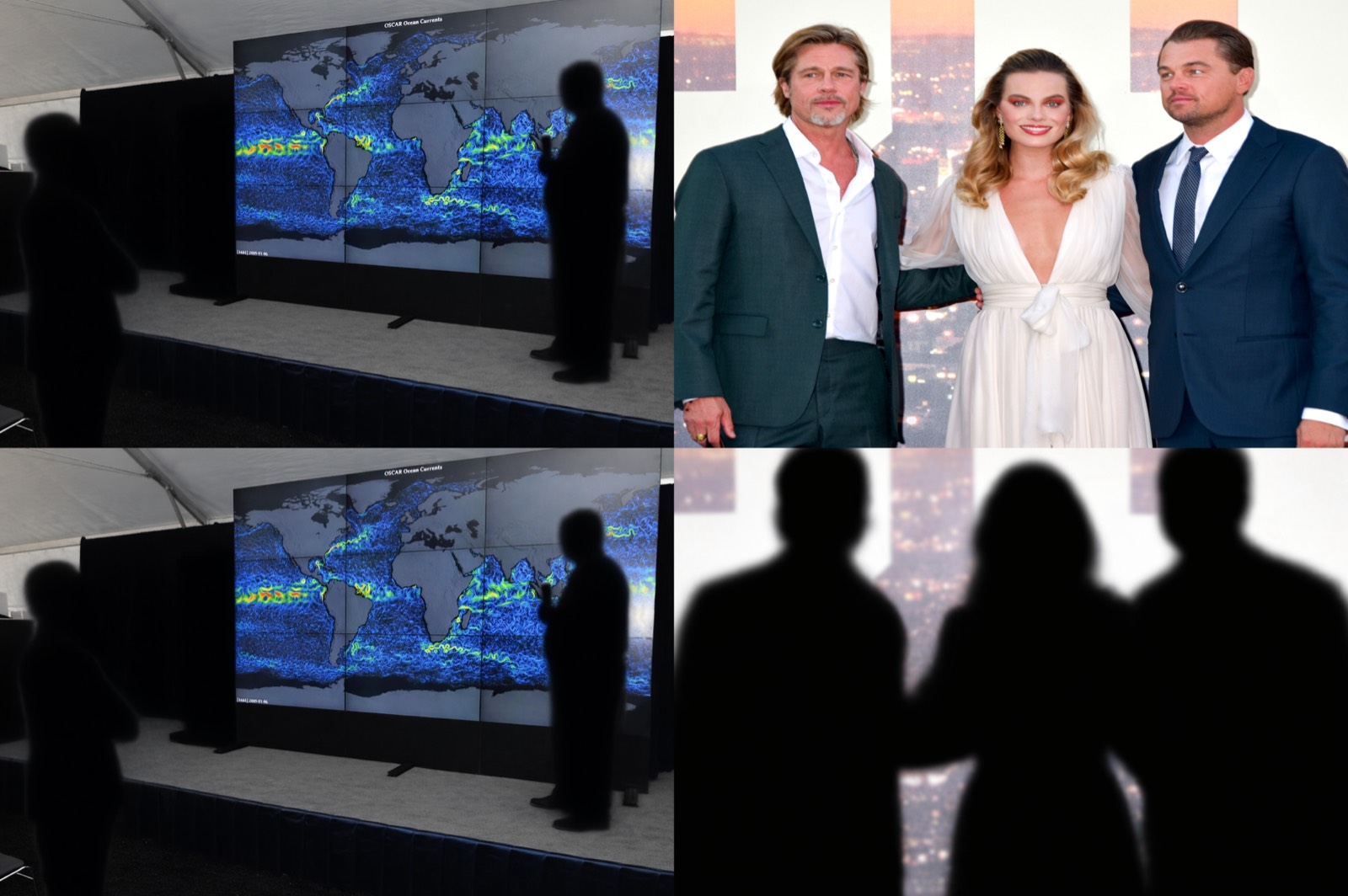} \\[1.5ex] 

\texttt{iGPT1m}
& \includegraphics[width=0.17\textwidth, valign=m]{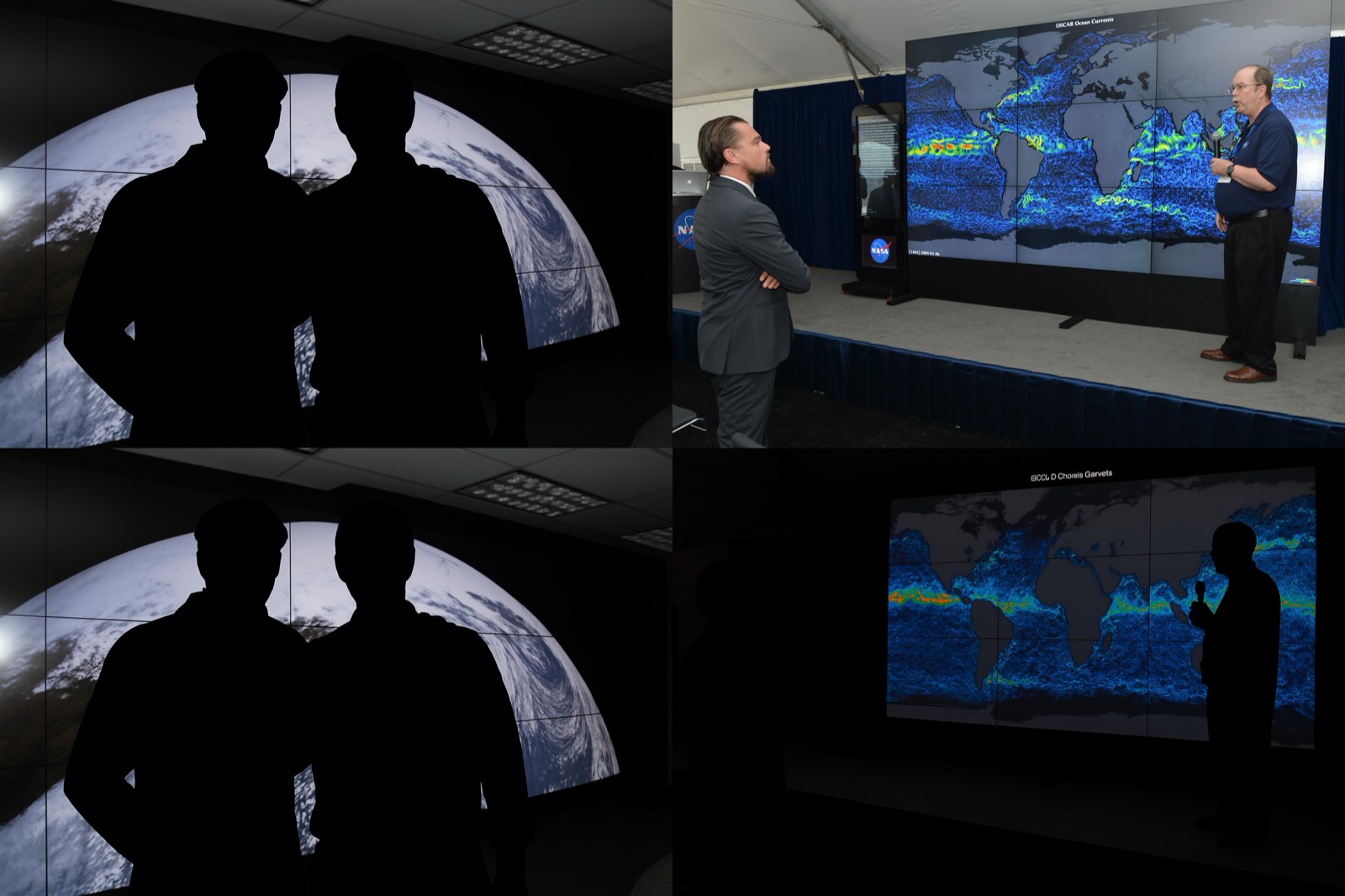}
& \includegraphics[width=0.17\textwidth, valign=m]{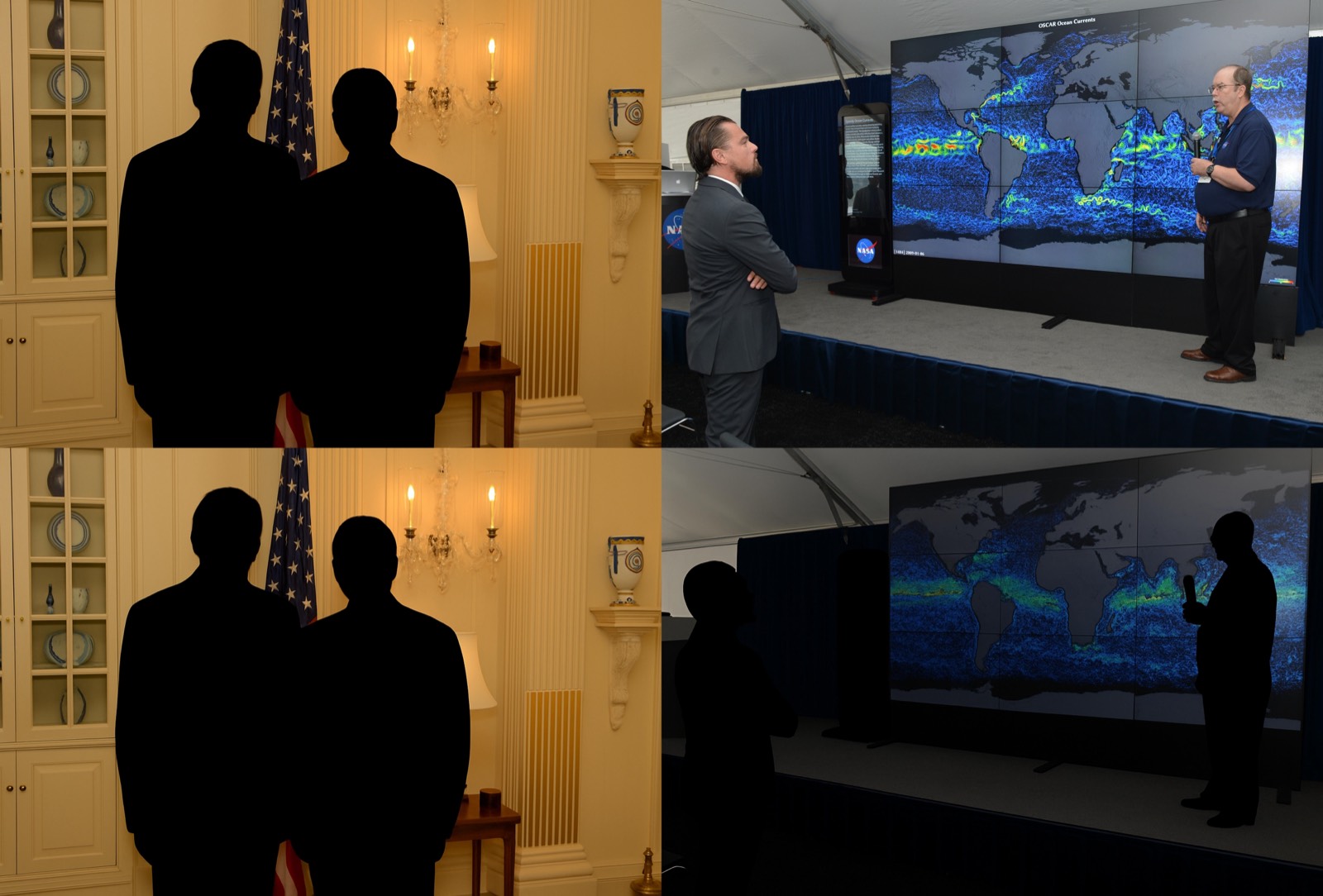}
& \includegraphics[width=0.17\textwidth, valign=m]{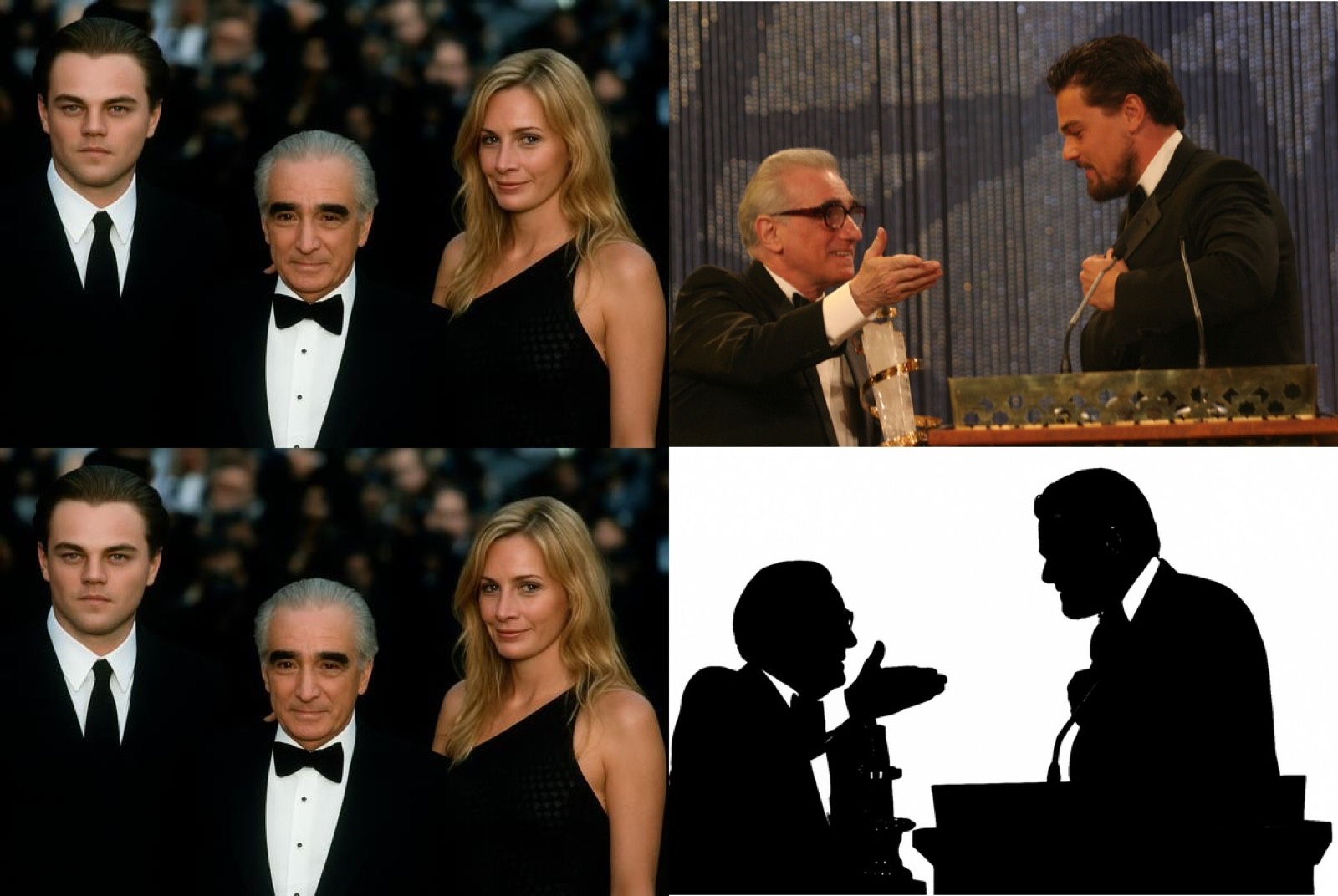}
& \includegraphics[width=0.17\textwidth, valign=m]{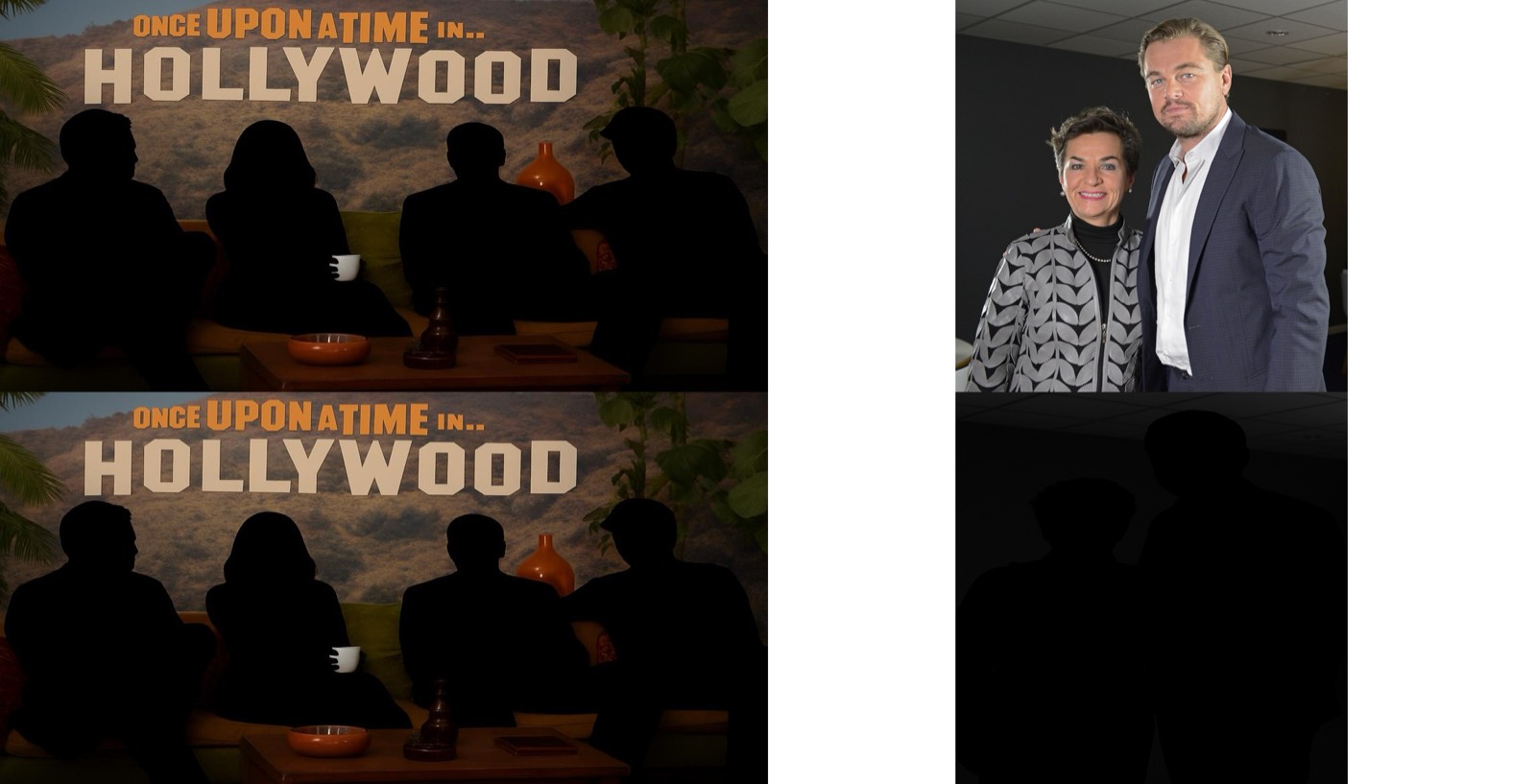}
& \includegraphics[width=0.17\textwidth, valign=m]{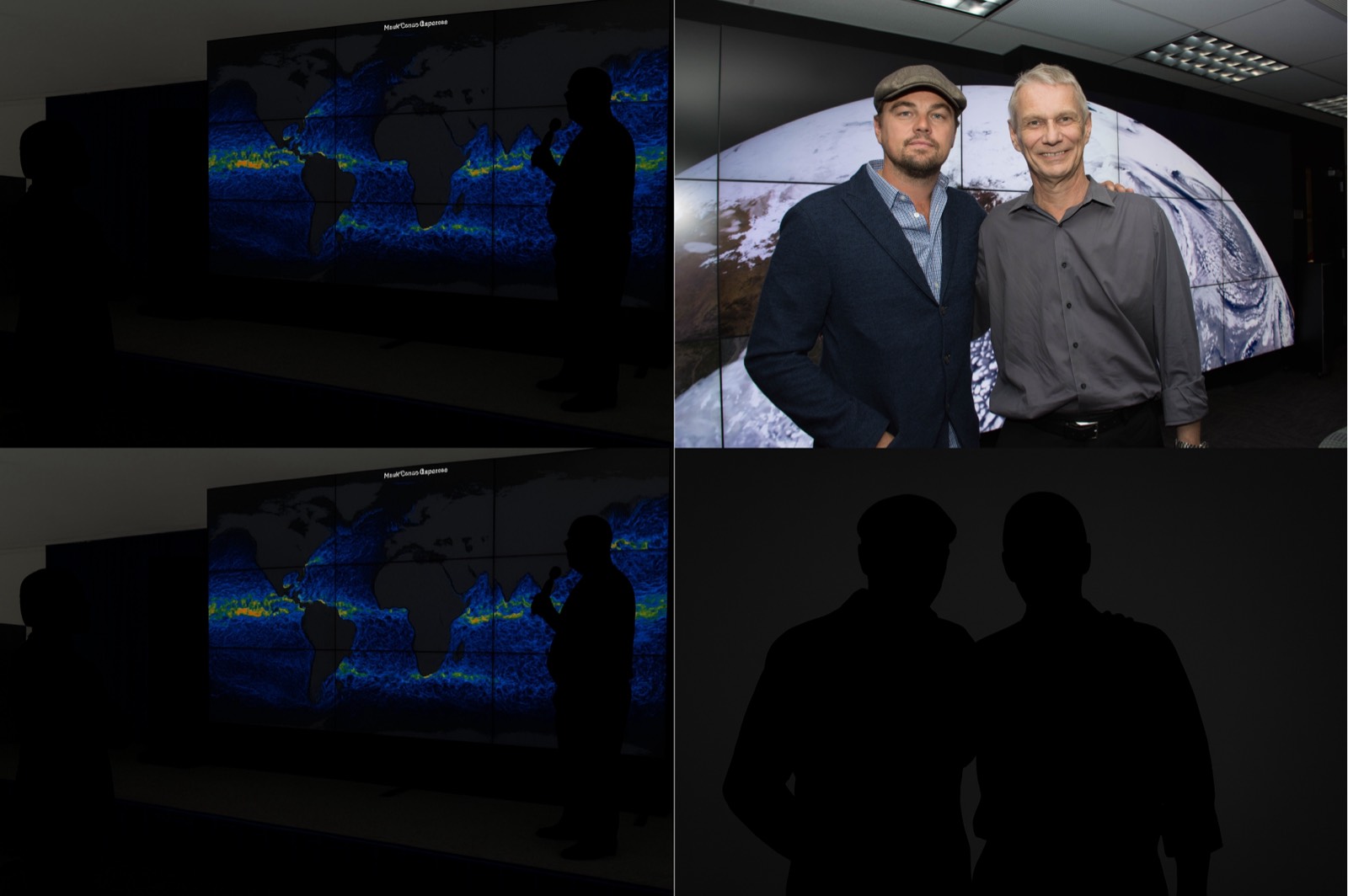} \\[1.5ex]

\texttt{iGEM3p}
& \includegraphics[width=0.17\textwidth, valign=m]{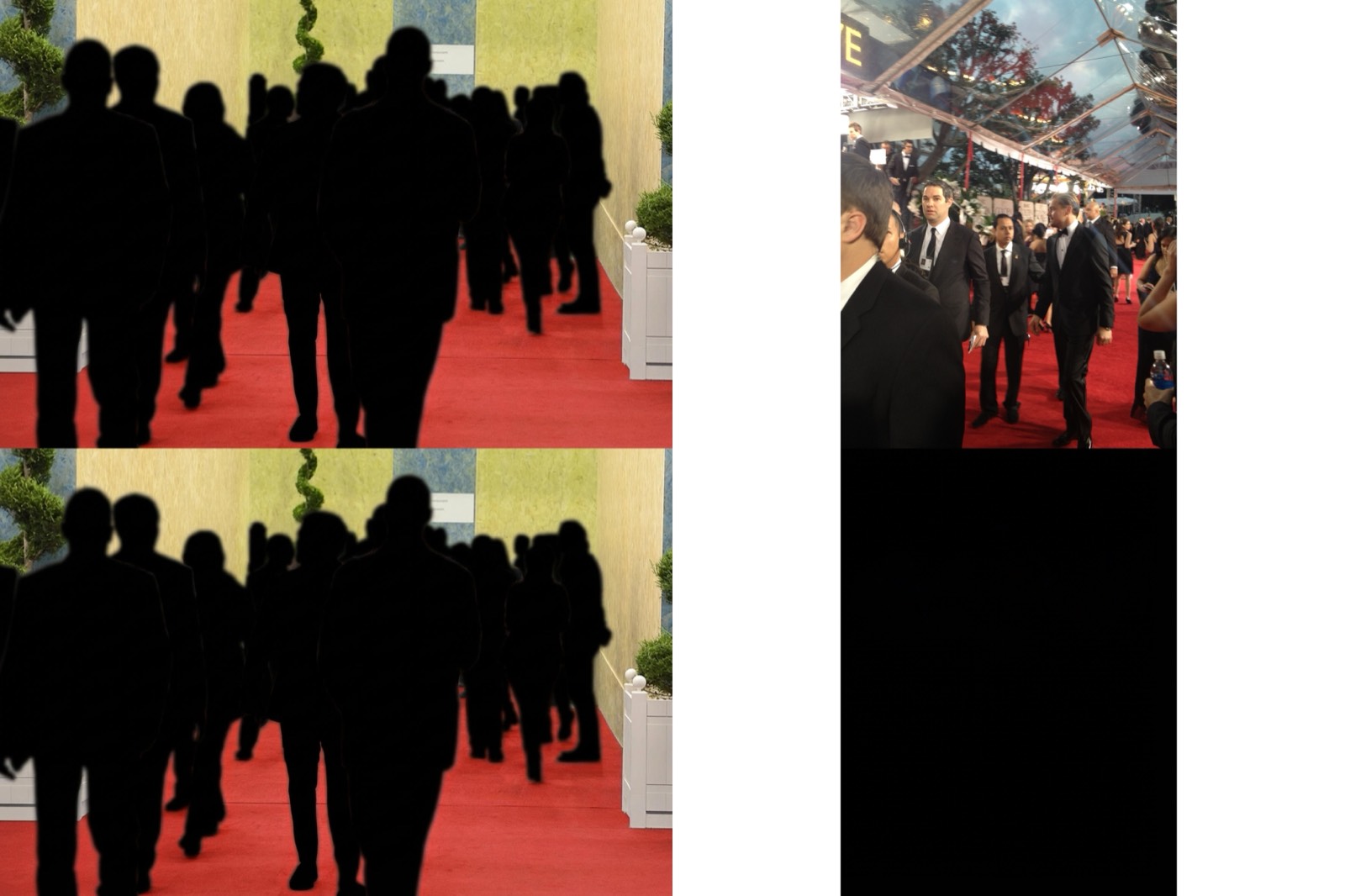}
& \includegraphics[width=0.17\textwidth, valign=m]{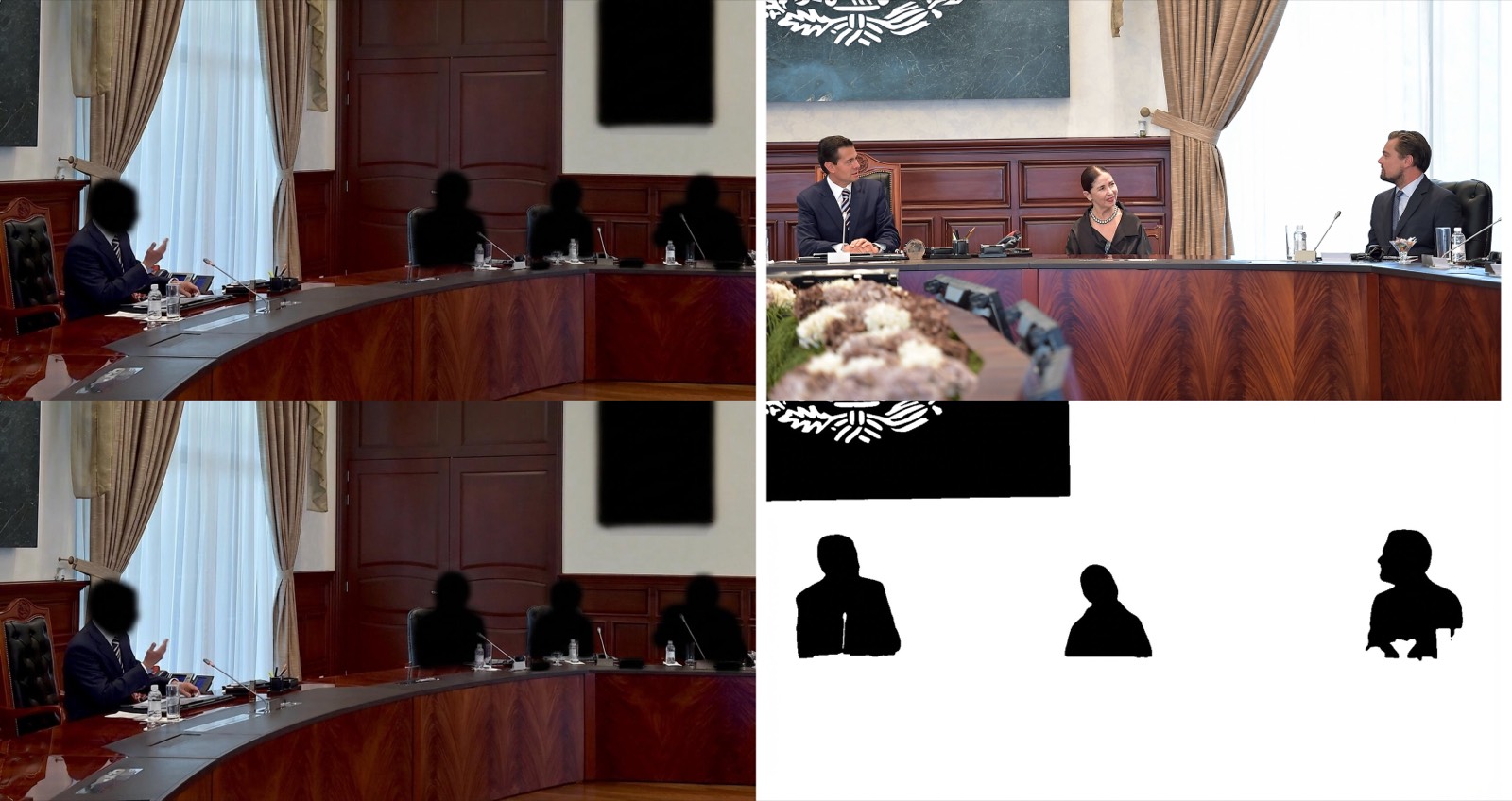}
& \includegraphics[width=0.17\textwidth, valign=m]{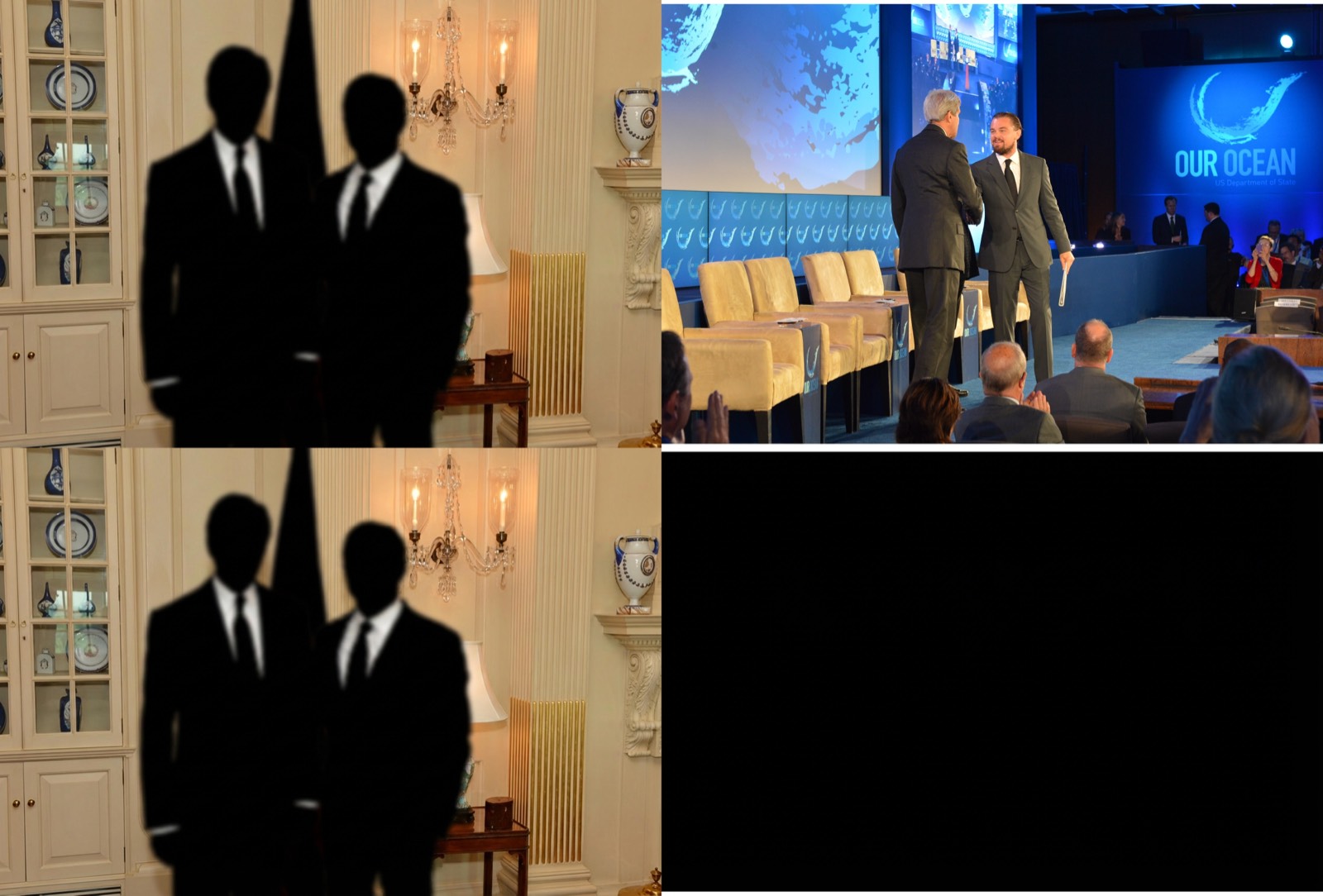}
& \includegraphics[width=0.17\textwidth, valign=m]{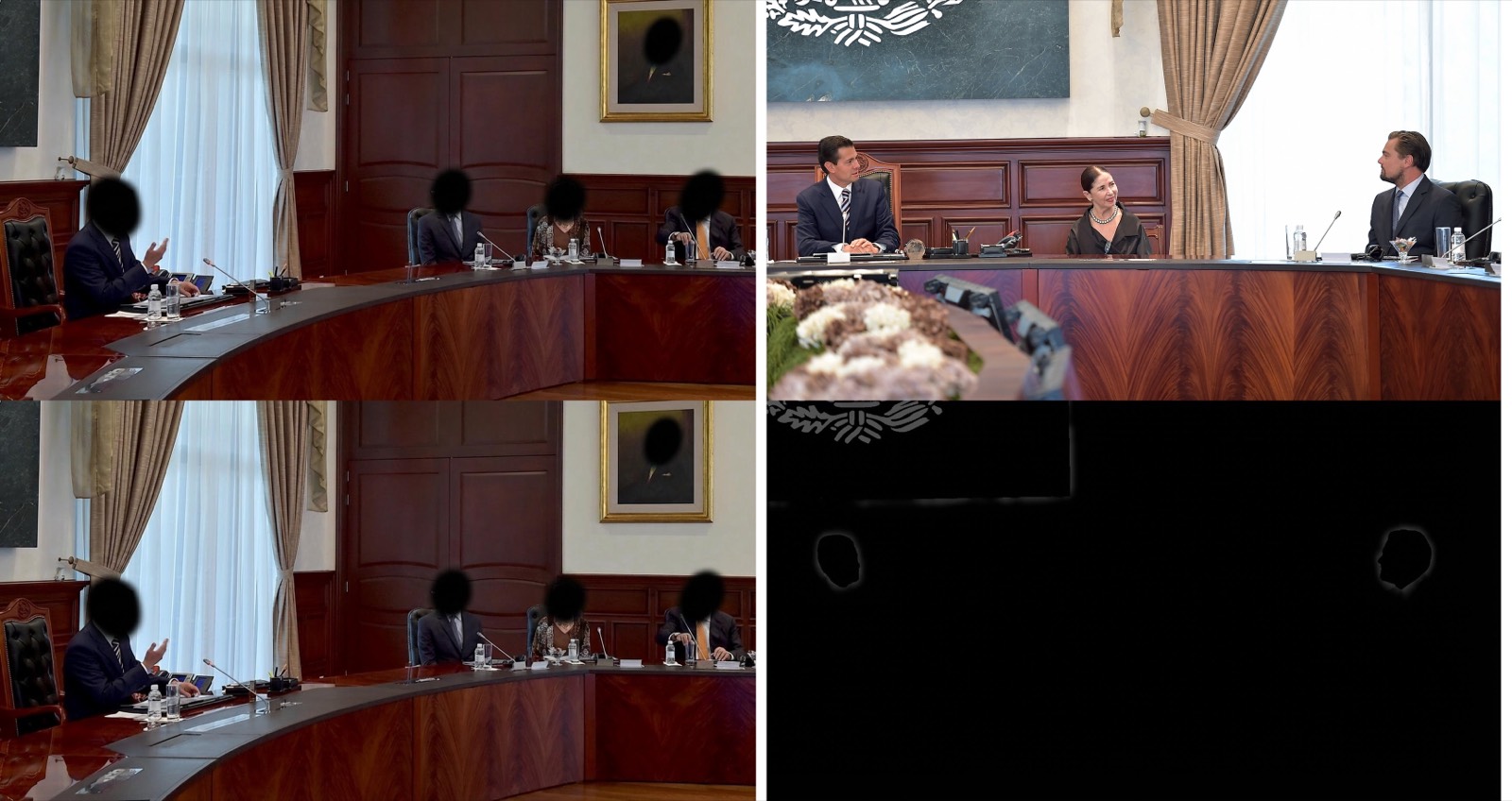}
& \includegraphics[width=0.17\textwidth, valign=m]{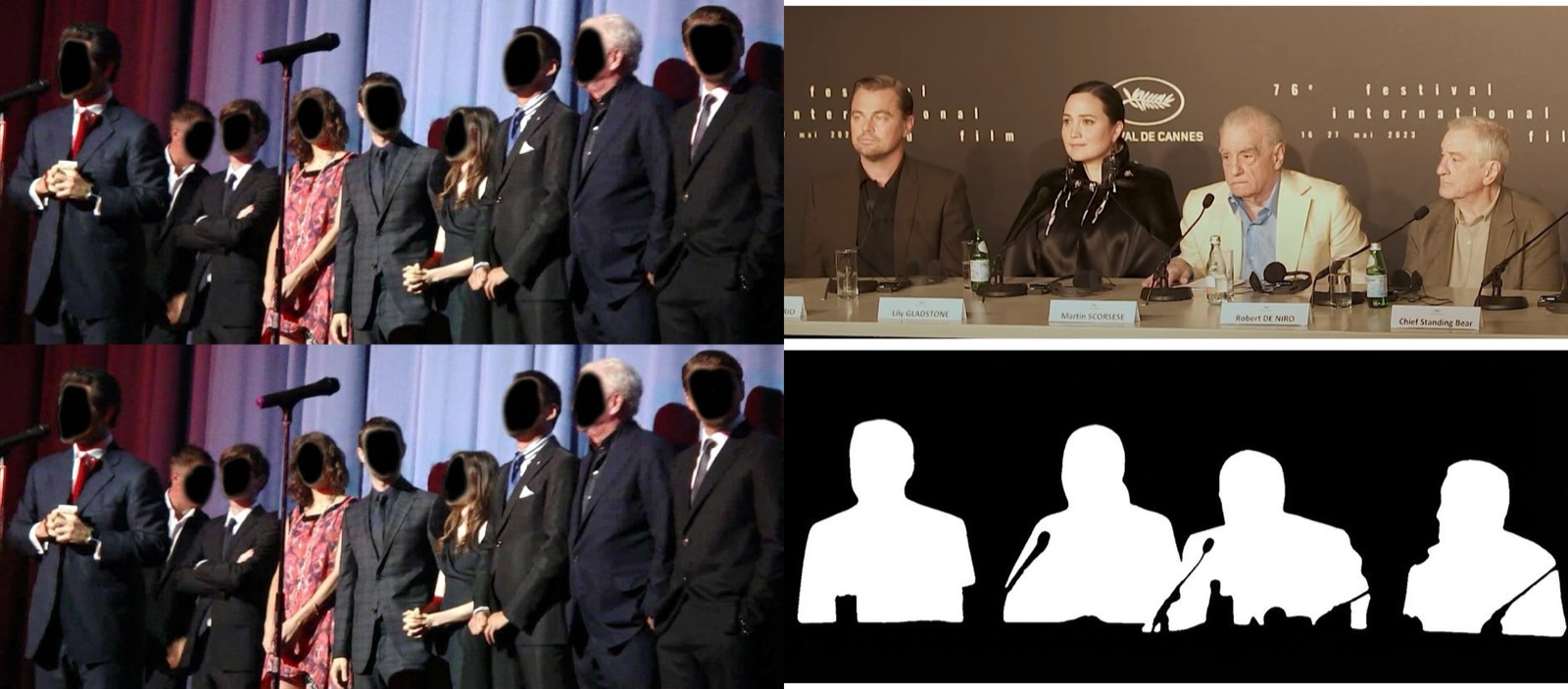} \\[1.5ex]

\texttt{iGEM31f}
& \includegraphics[width=0.17\textwidth, valign=m]{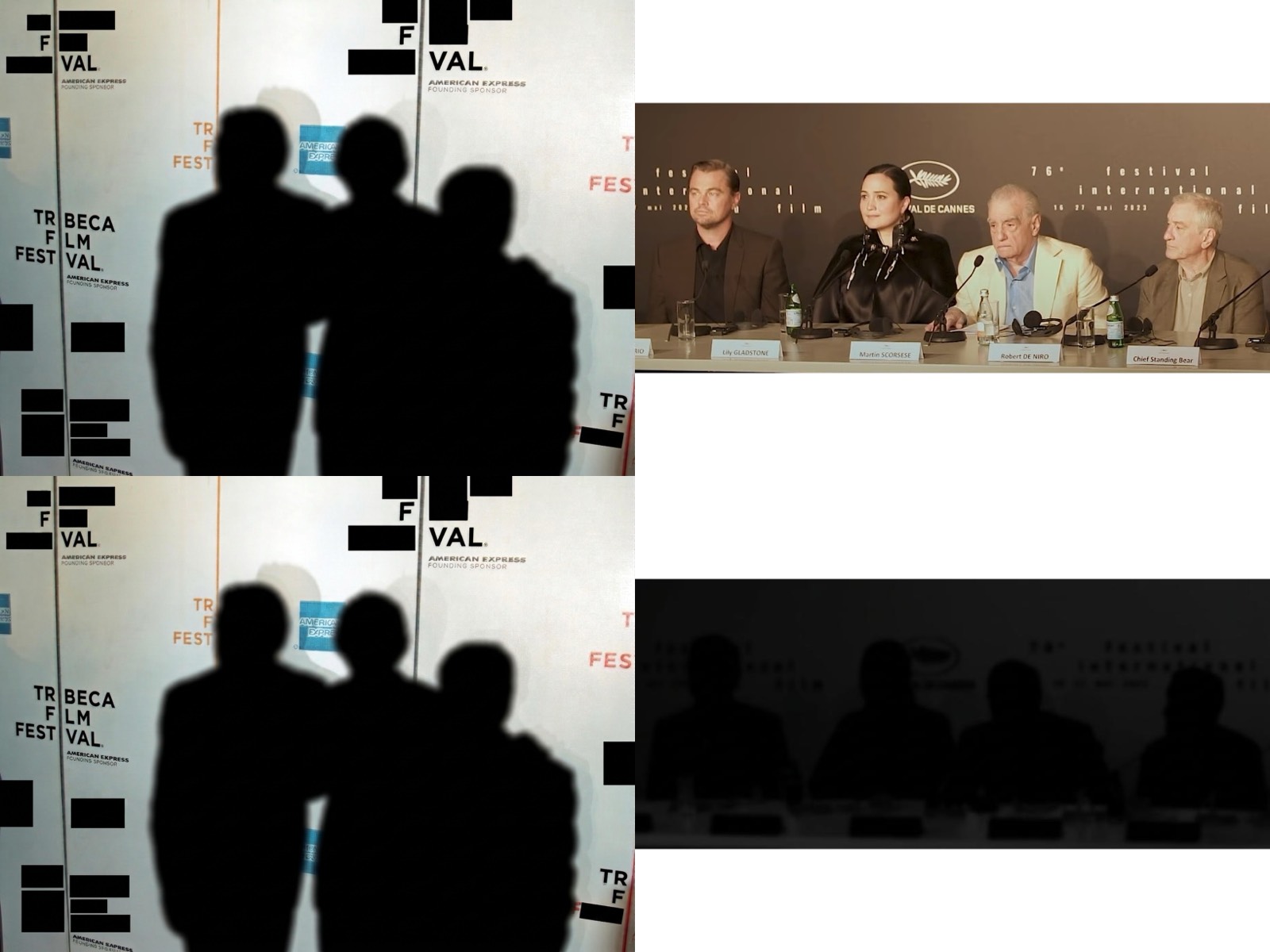}
& \includegraphics[width=0.17\textwidth, valign=m]{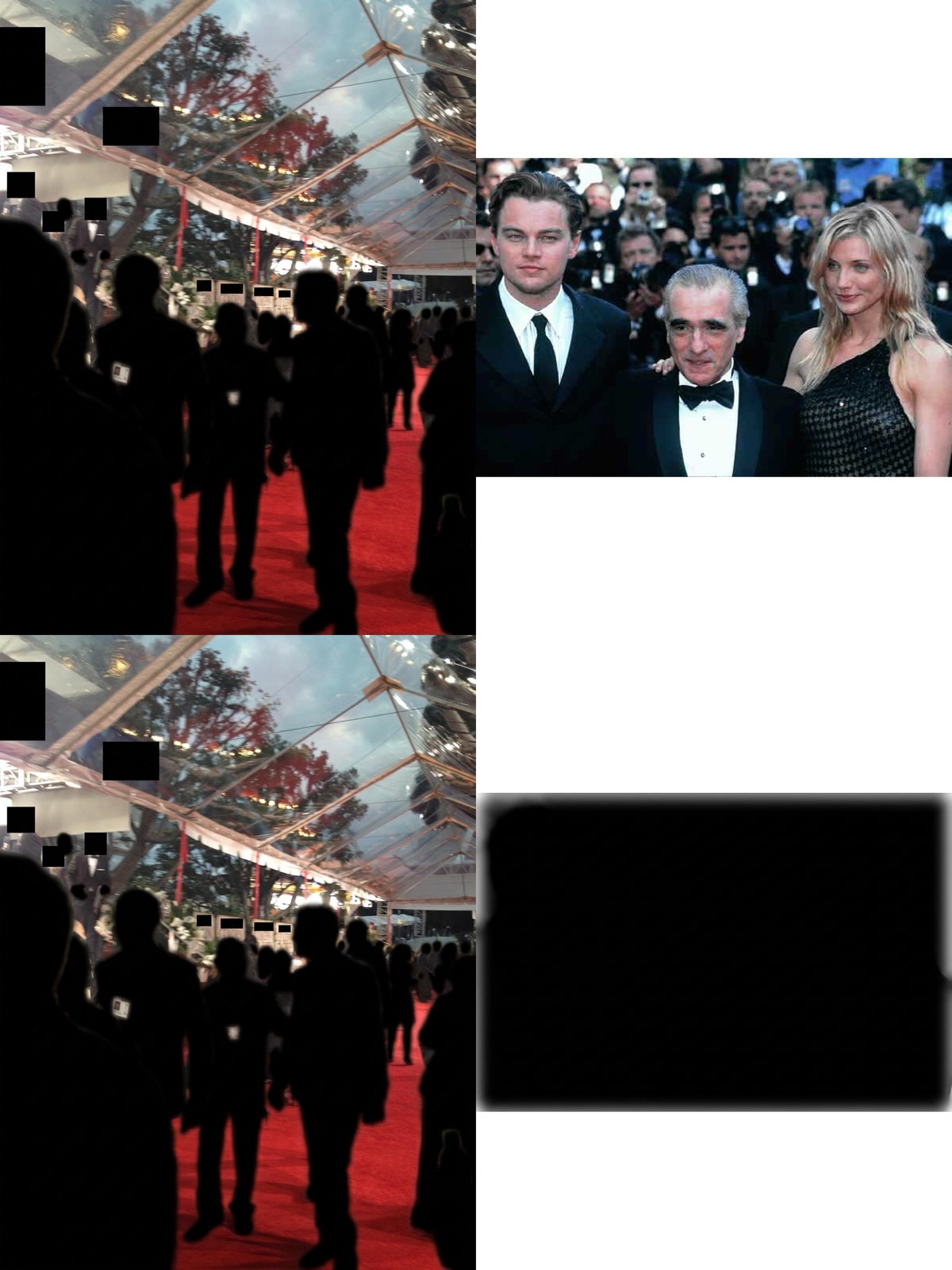}
& \includegraphics[width=0.17\textwidth, valign=m]{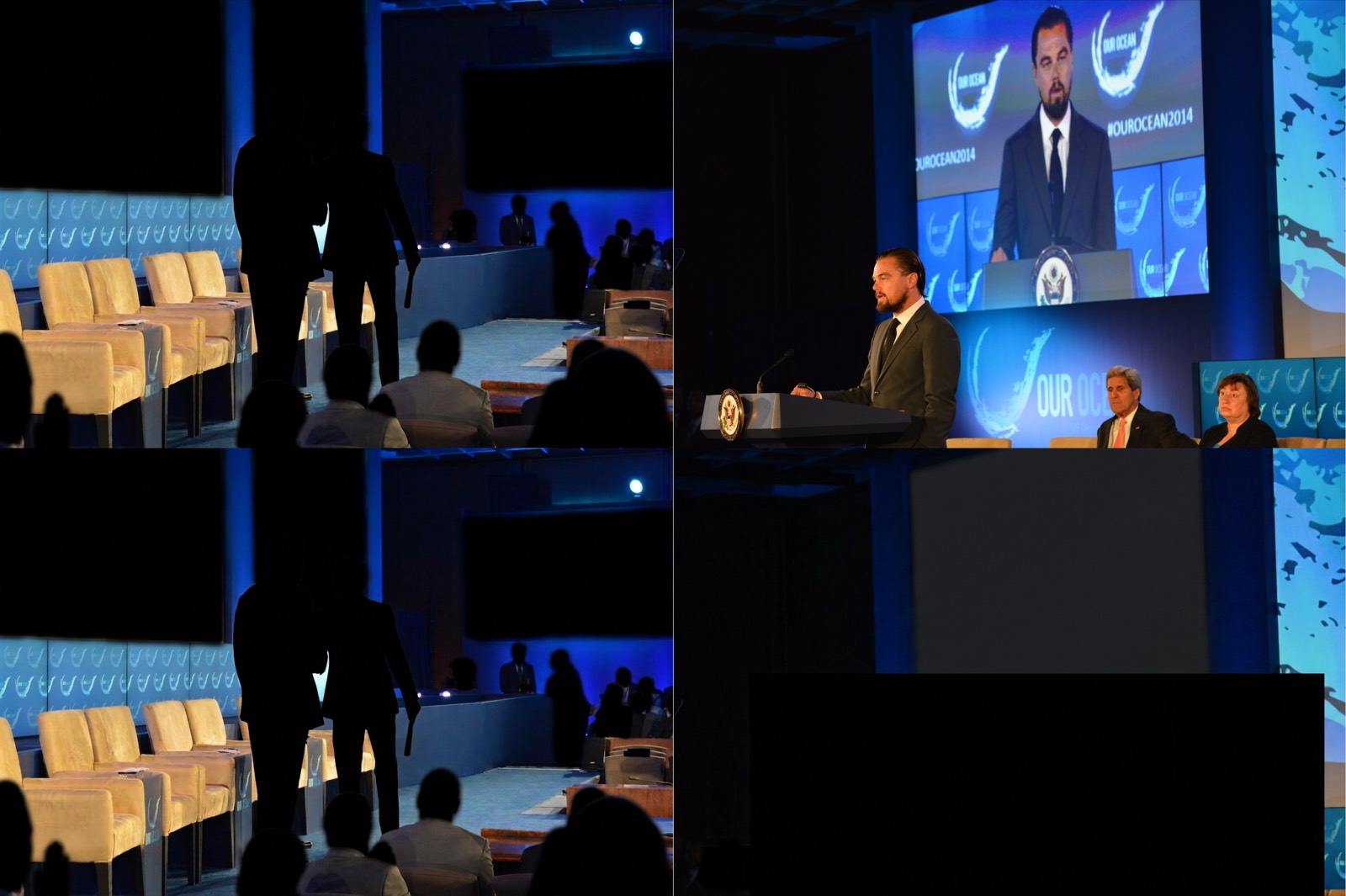}
& 
& \includegraphics[width=0.17\textwidth, valign=m]{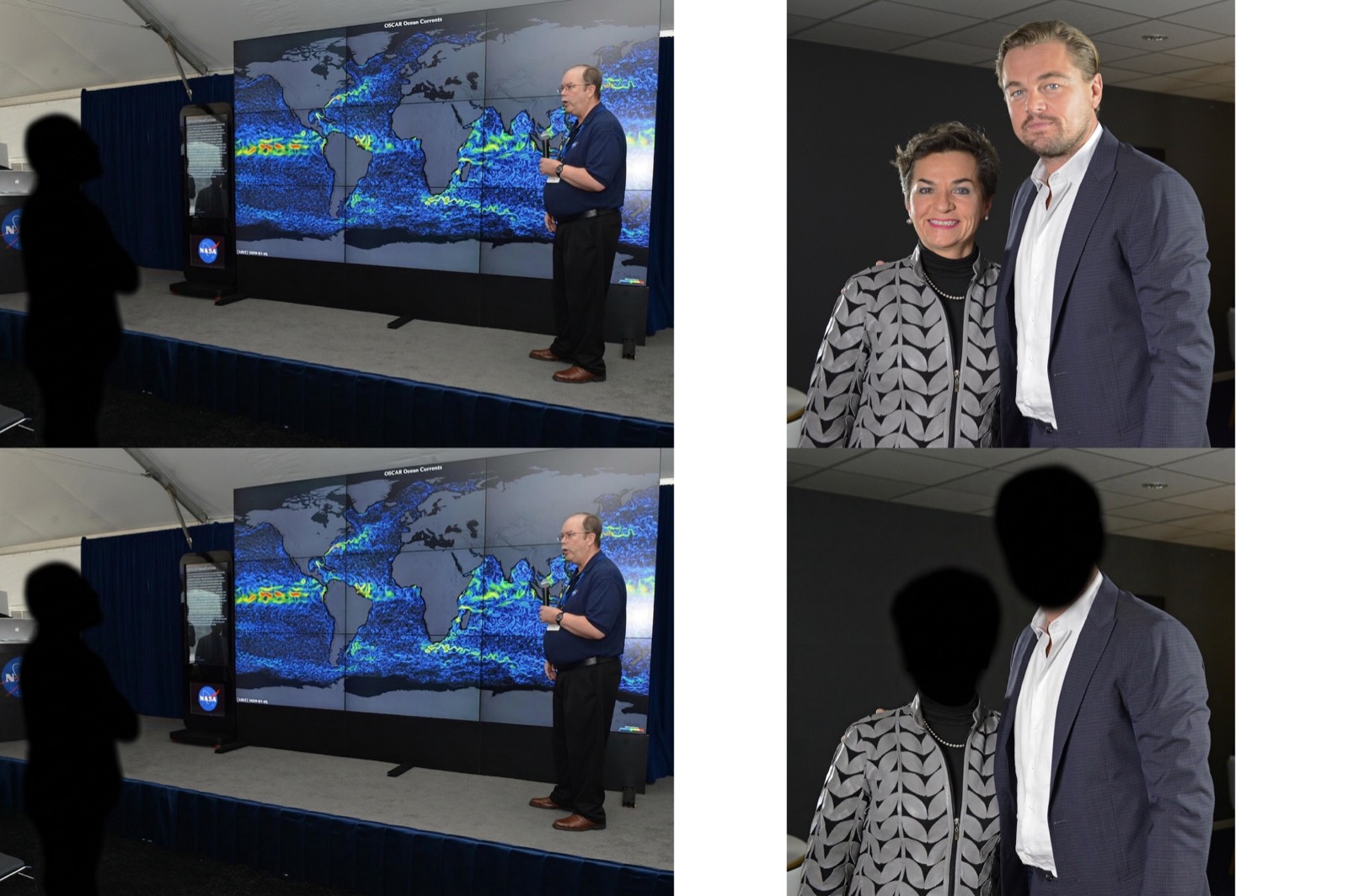} \\[1.5ex]

\texttt{FLX2p}
& \includegraphics[width=0.17\textwidth, valign=m]{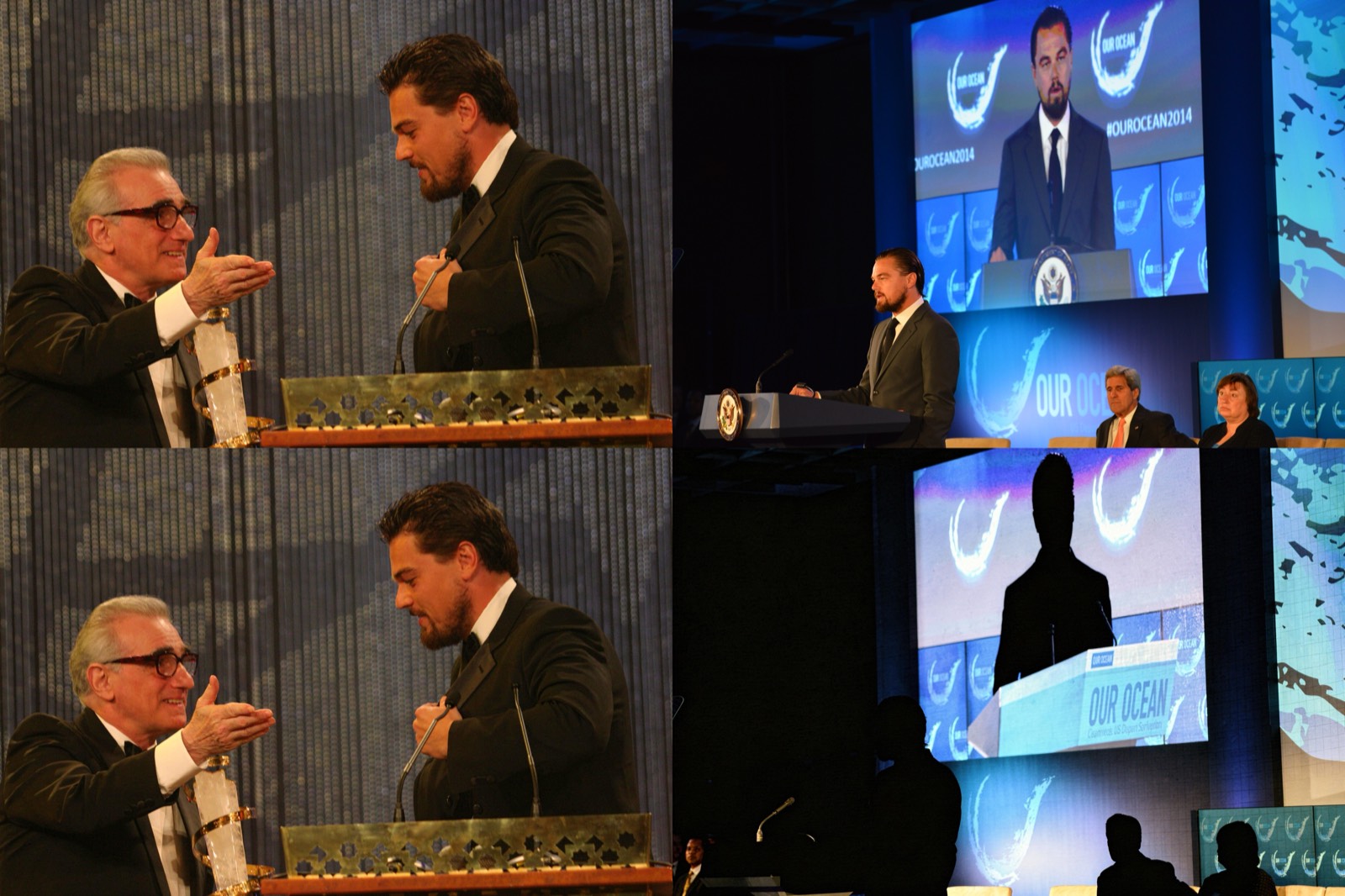}
& \includegraphics[width=0.17\textwidth, valign=m]{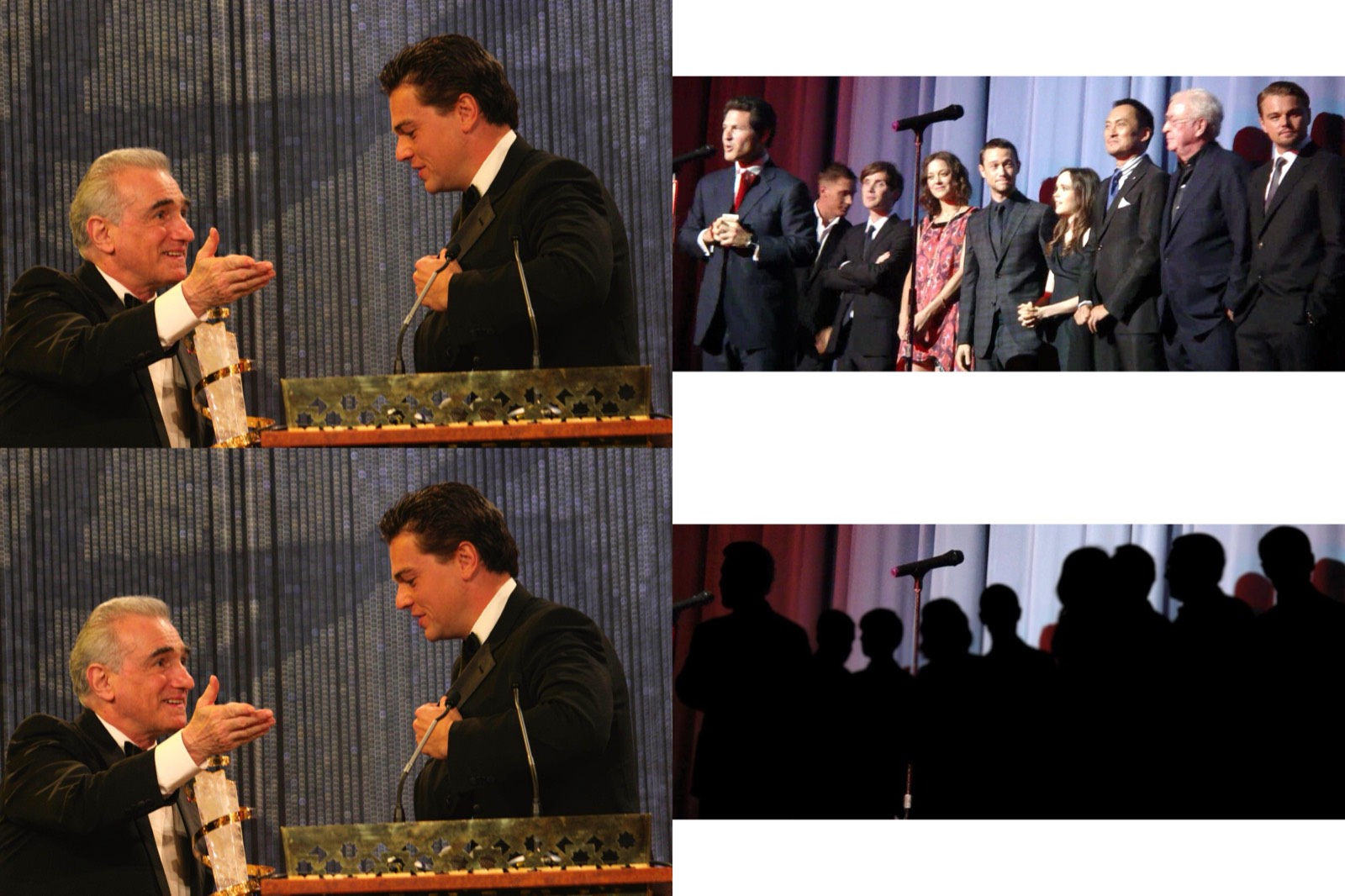}
& \includegraphics[width=0.17\textwidth, valign=m]{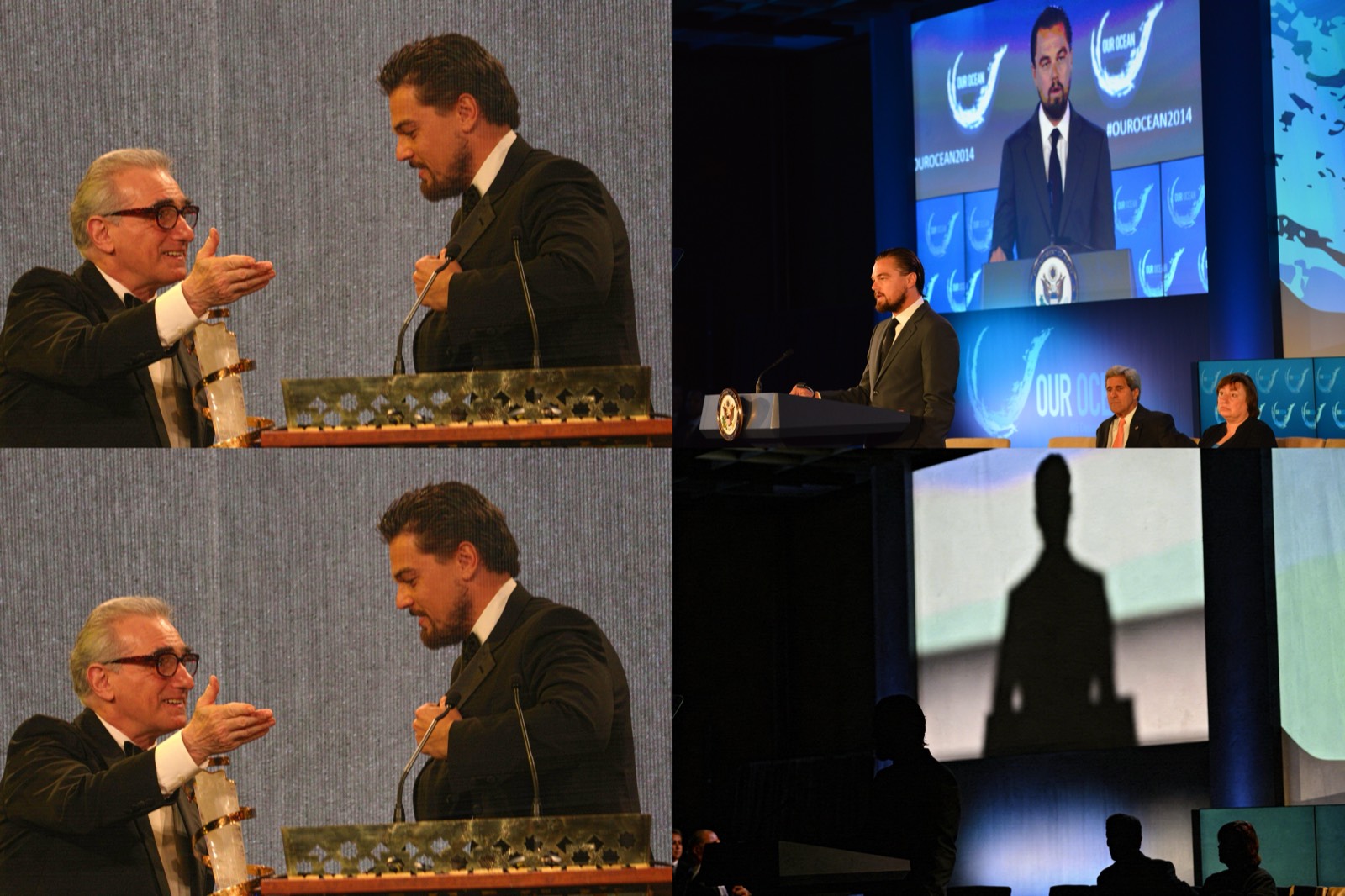}
& \includegraphics[width=0.17\textwidth, valign=m]{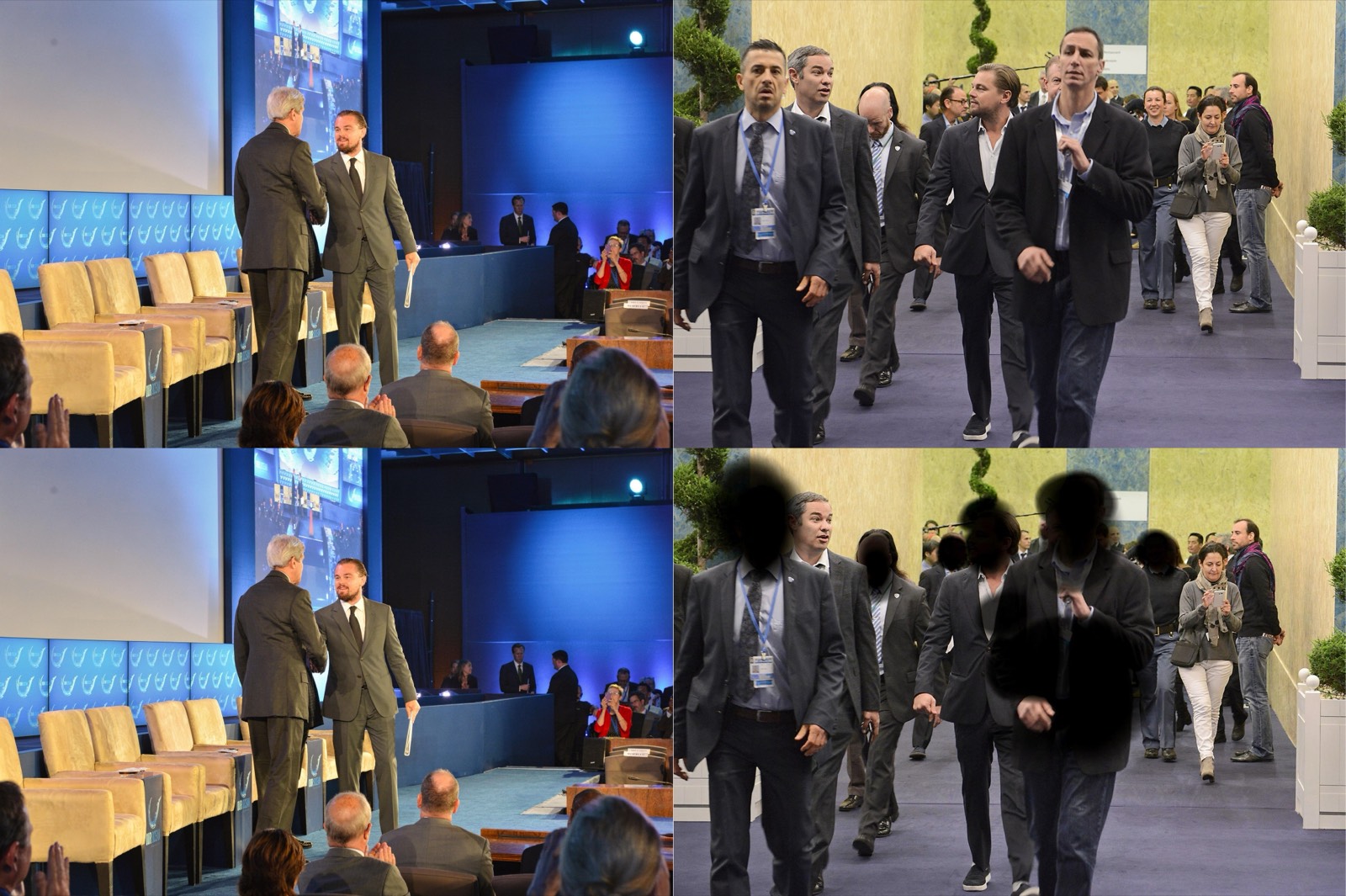}
& \includegraphics[width=0.17\textwidth, valign=m]{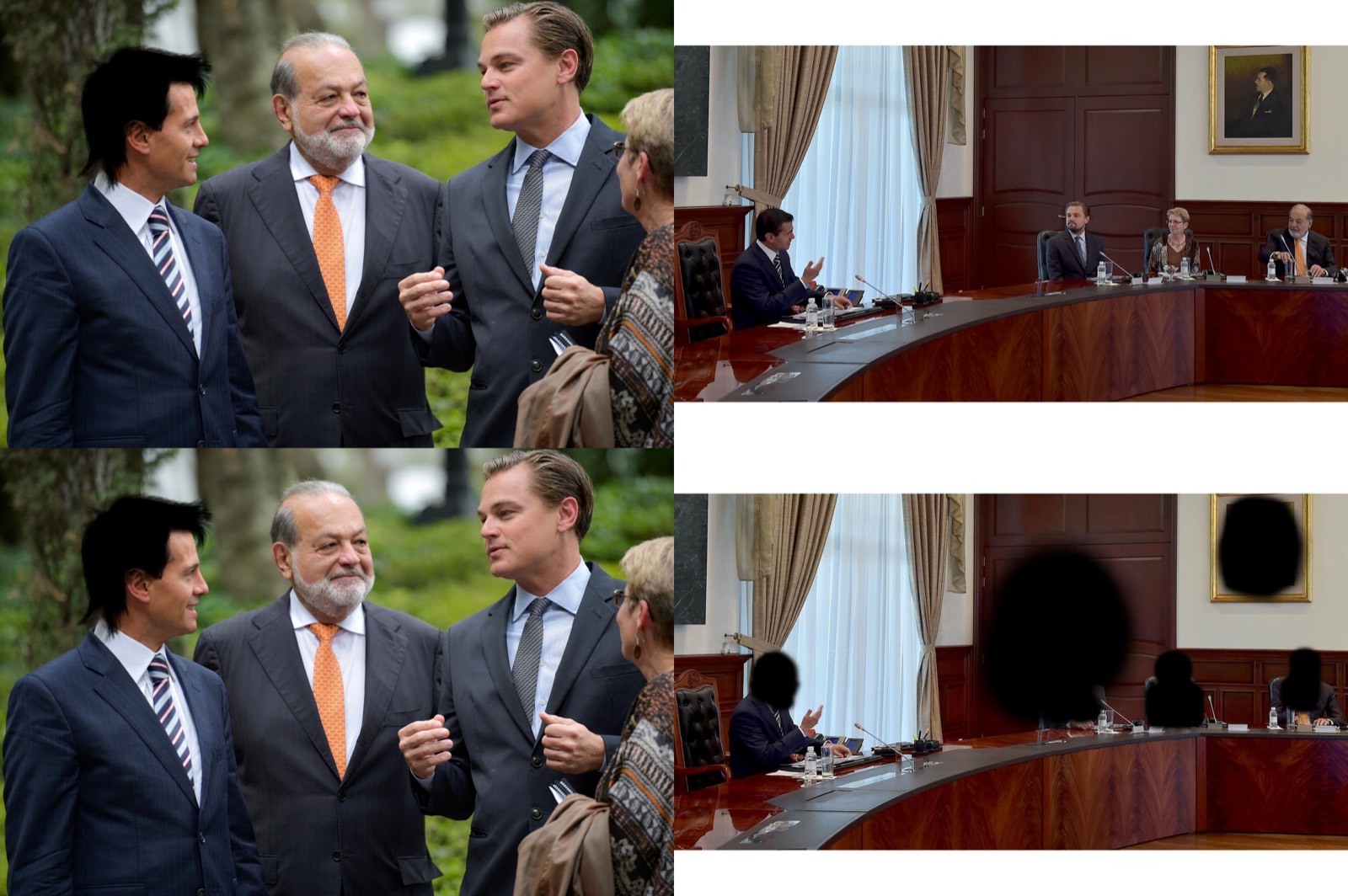} \\[1.5ex]

\texttt{FLX2d}
& \includegraphics[width=0.17\textwidth, valign=m]{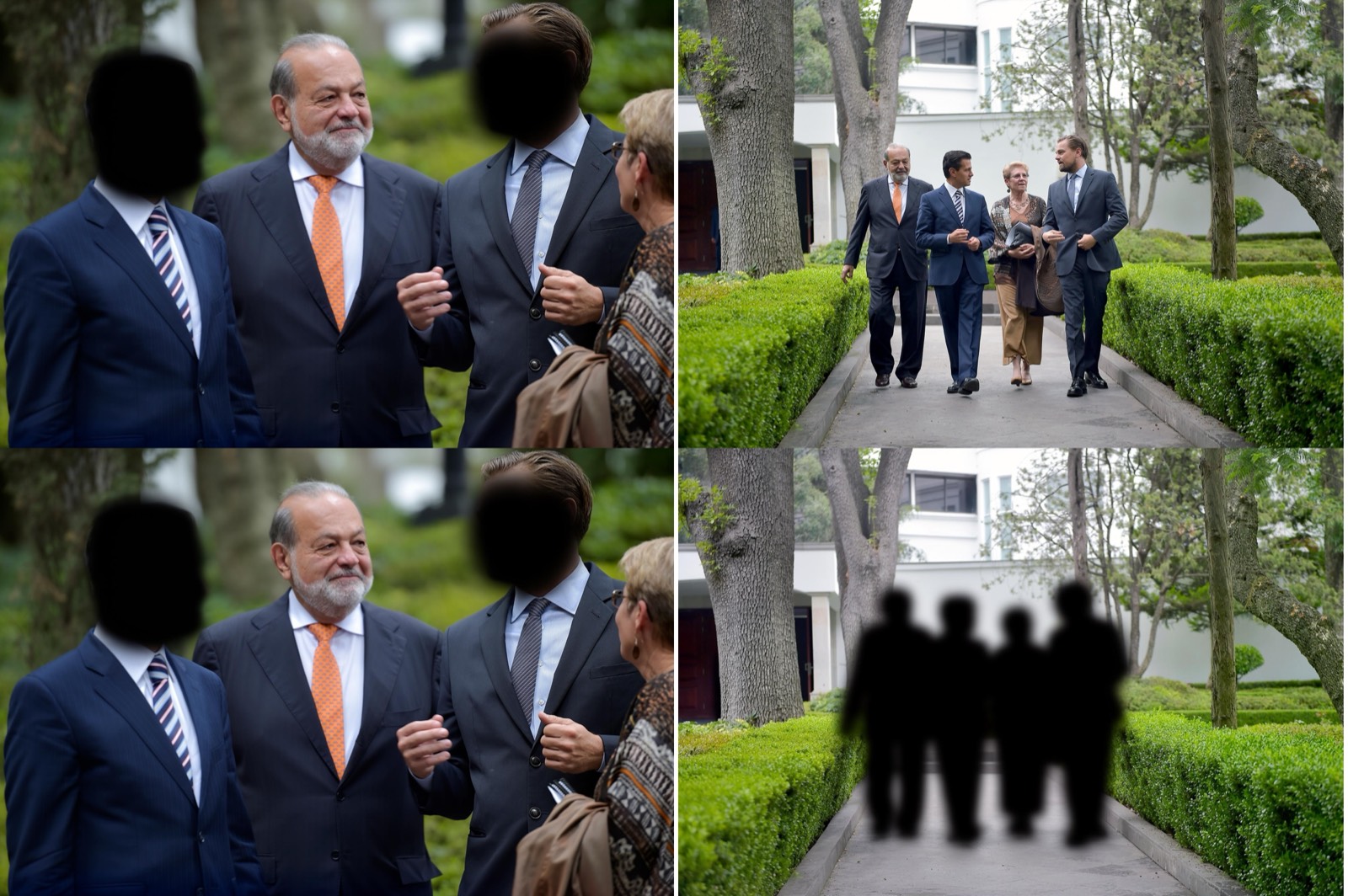}
& \includegraphics[width=0.17\textwidth, valign=m]{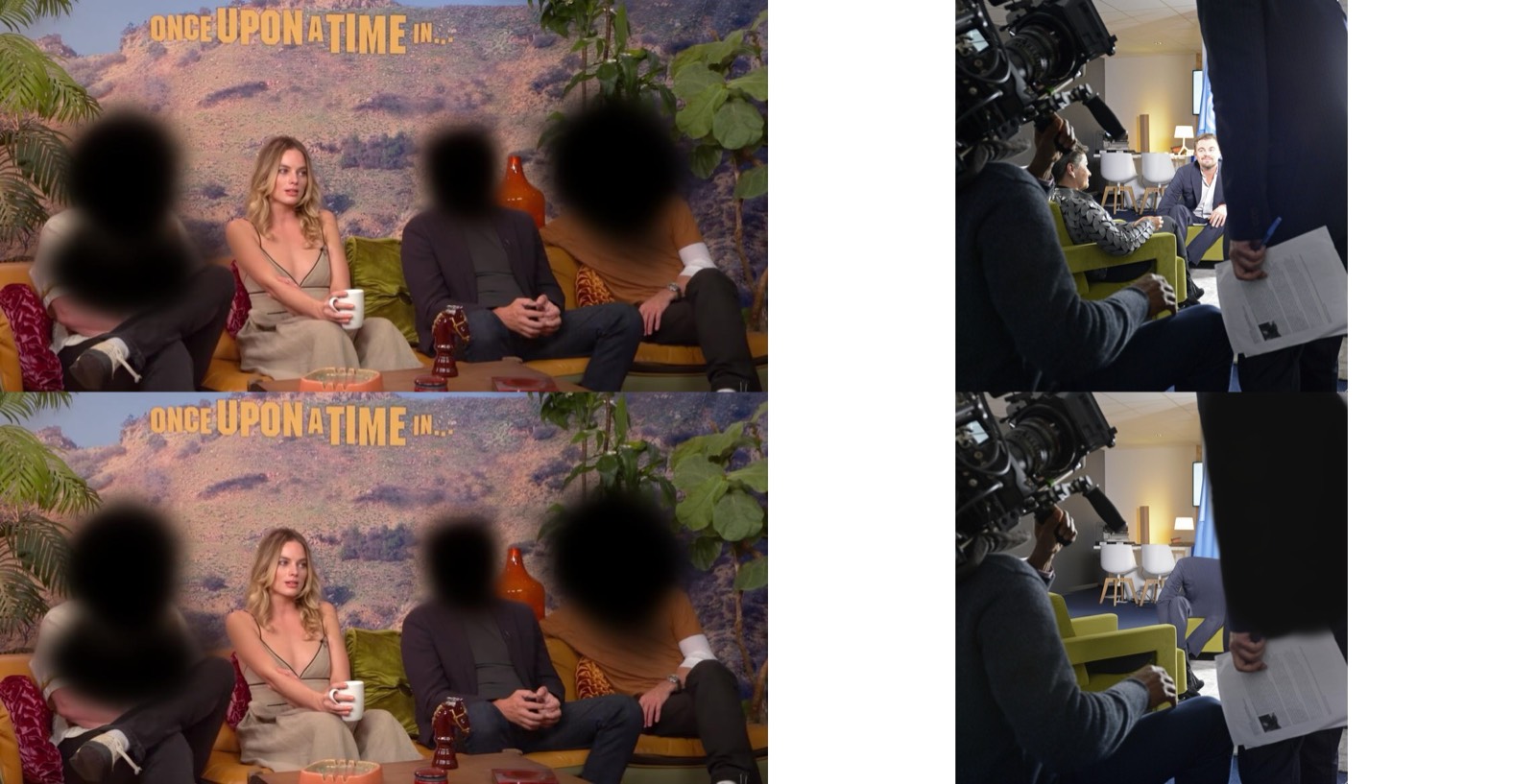}
& \includegraphics[width=0.17\textwidth, valign=m]{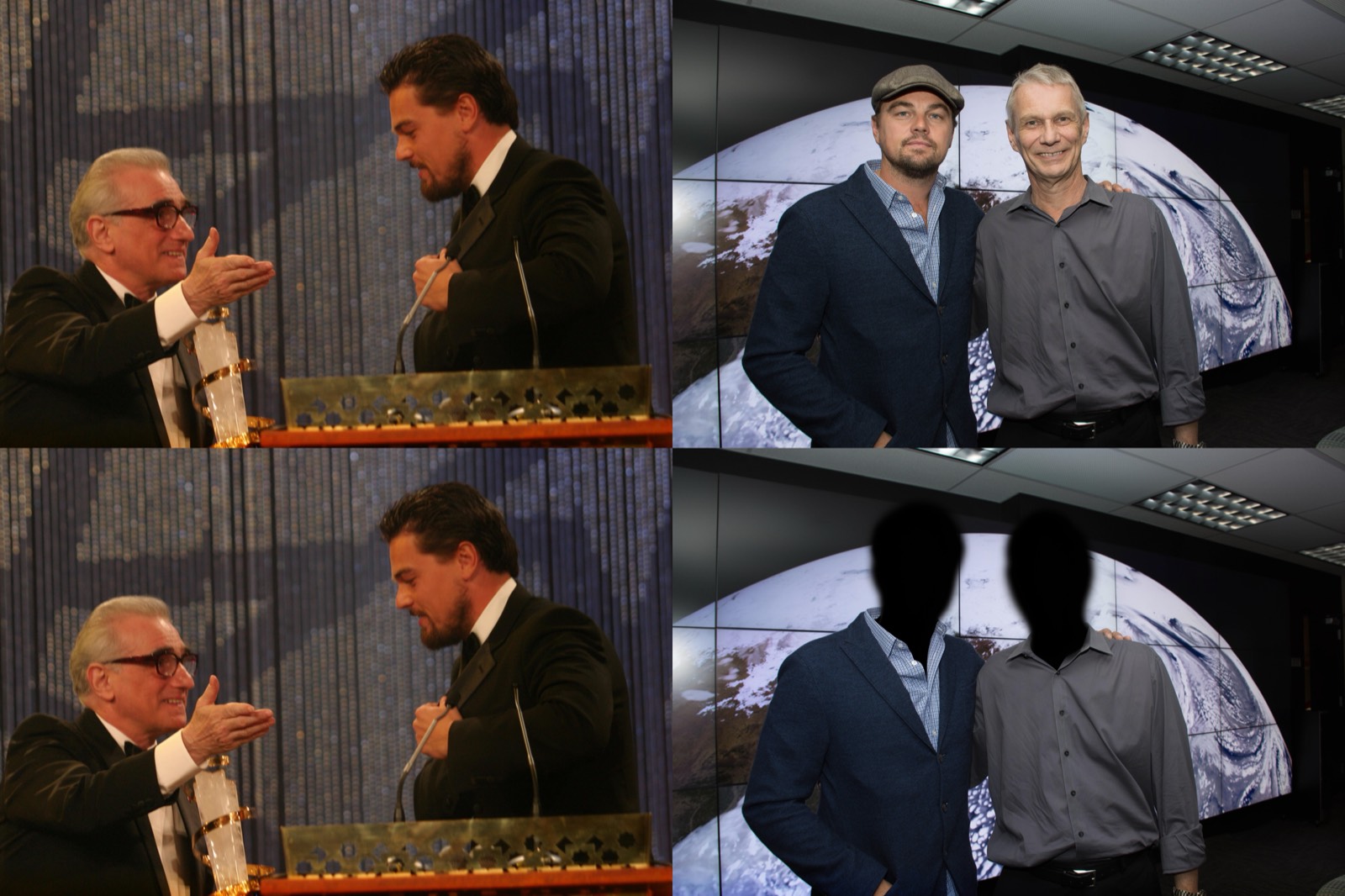}
& \includegraphics[width=0.17\textwidth, valign=m]{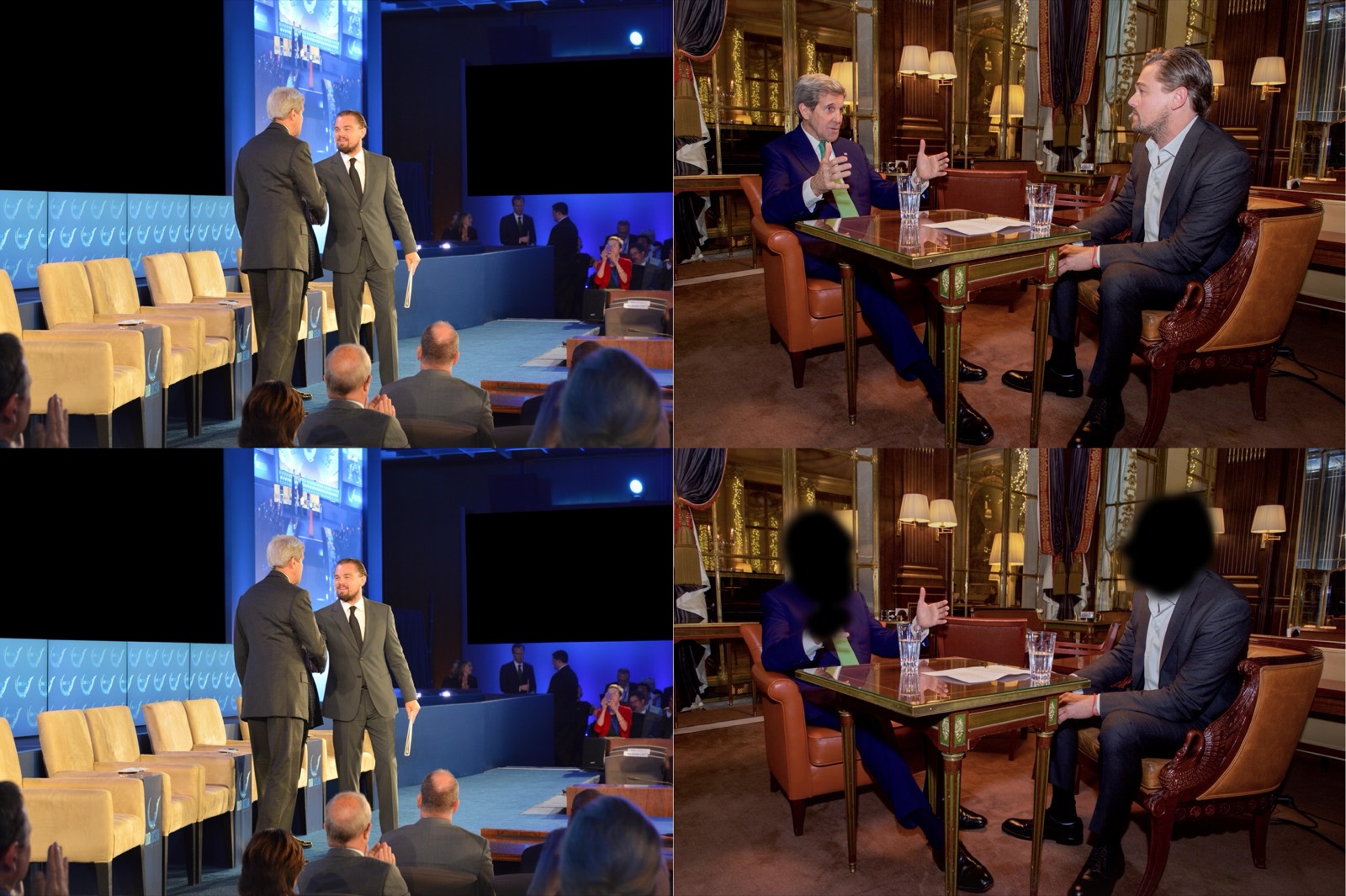}
& \includegraphics[width=0.17\textwidth, valign=m]{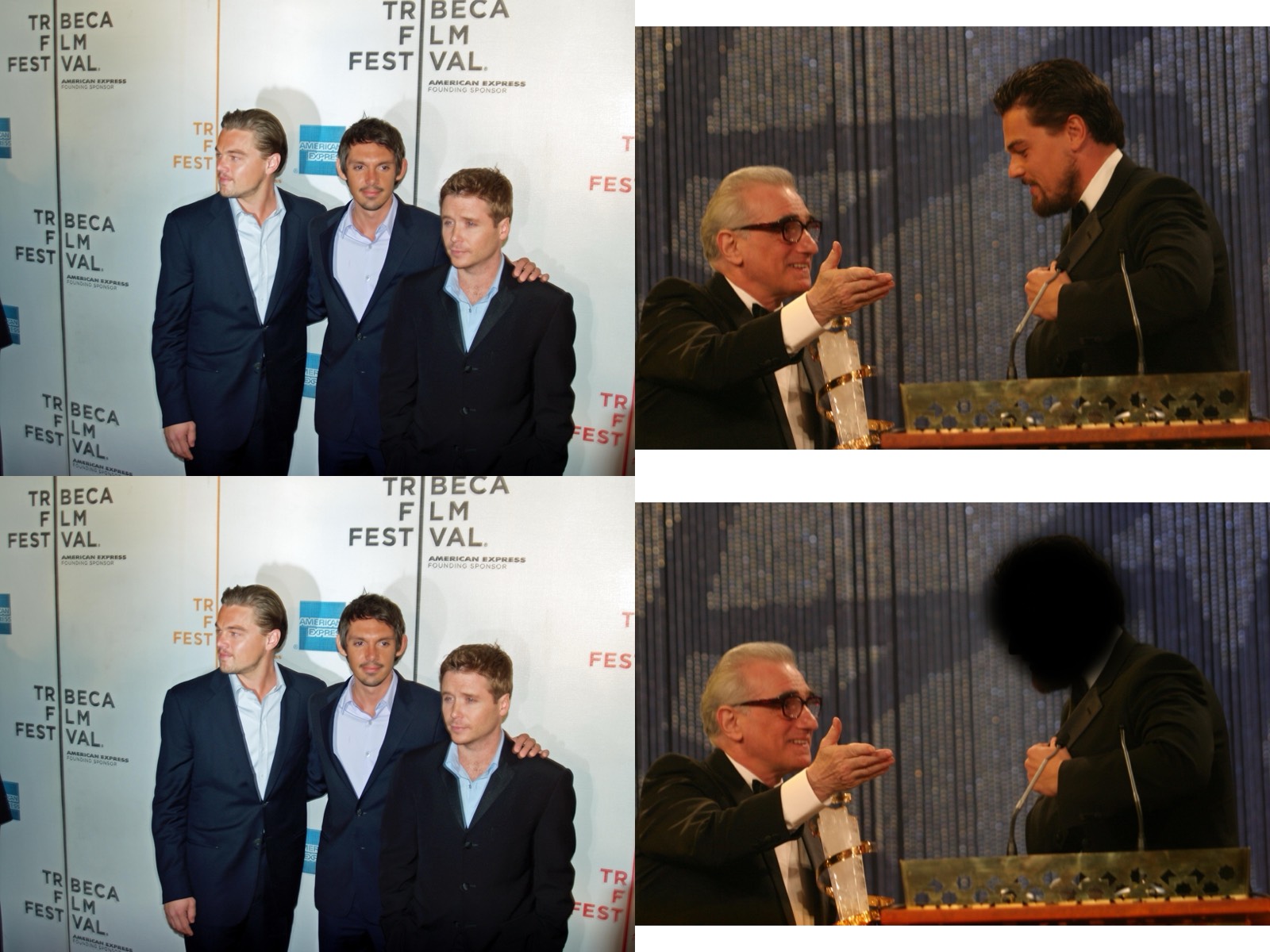} \\[1.5ex]

\texttt{SAM3}
& \includegraphics[width=0.17\textwidth, valign=m]{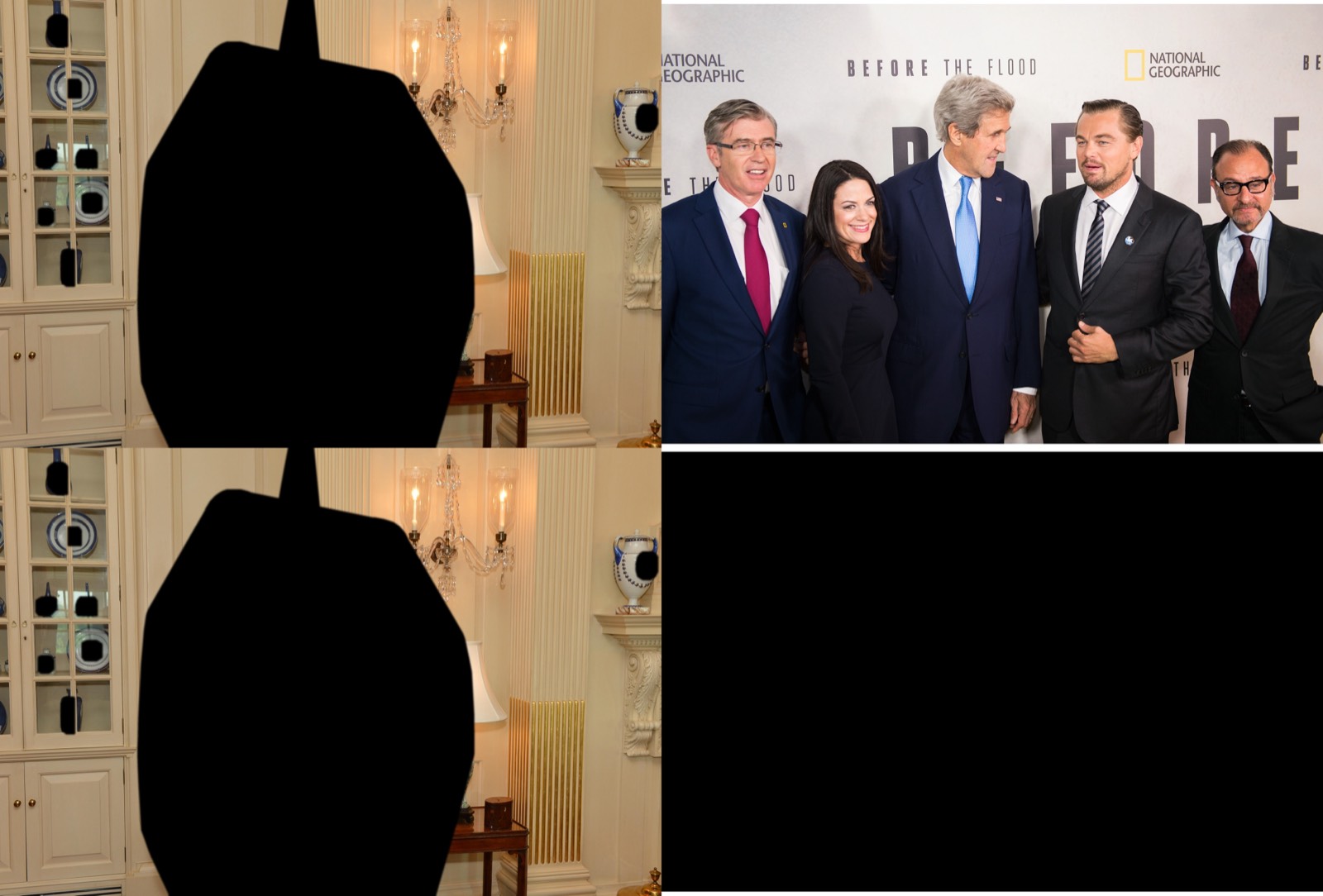}
& \includegraphics[width=0.17\textwidth, valign=m]{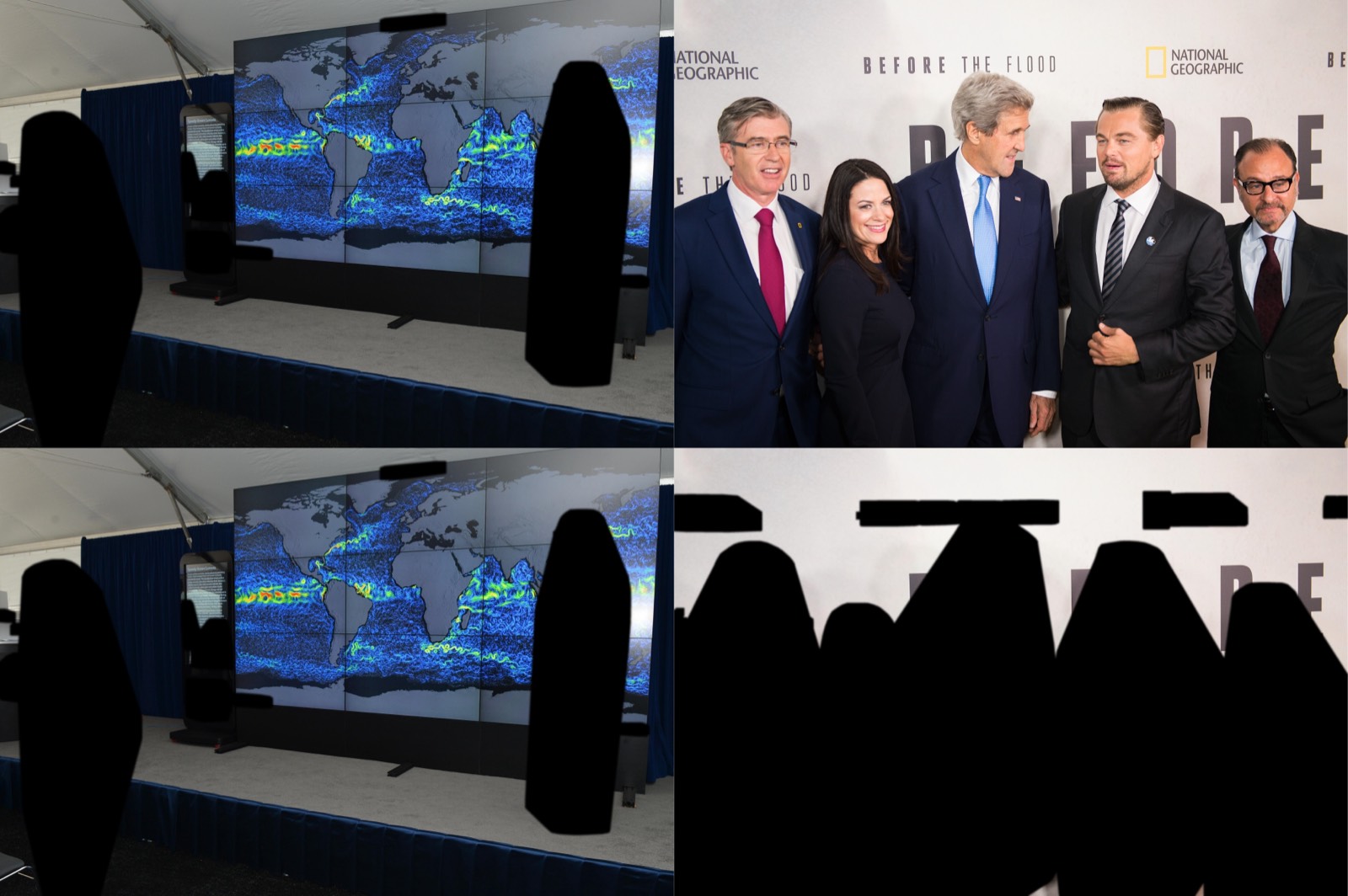}
& \includegraphics[width=0.17\textwidth, valign=m]{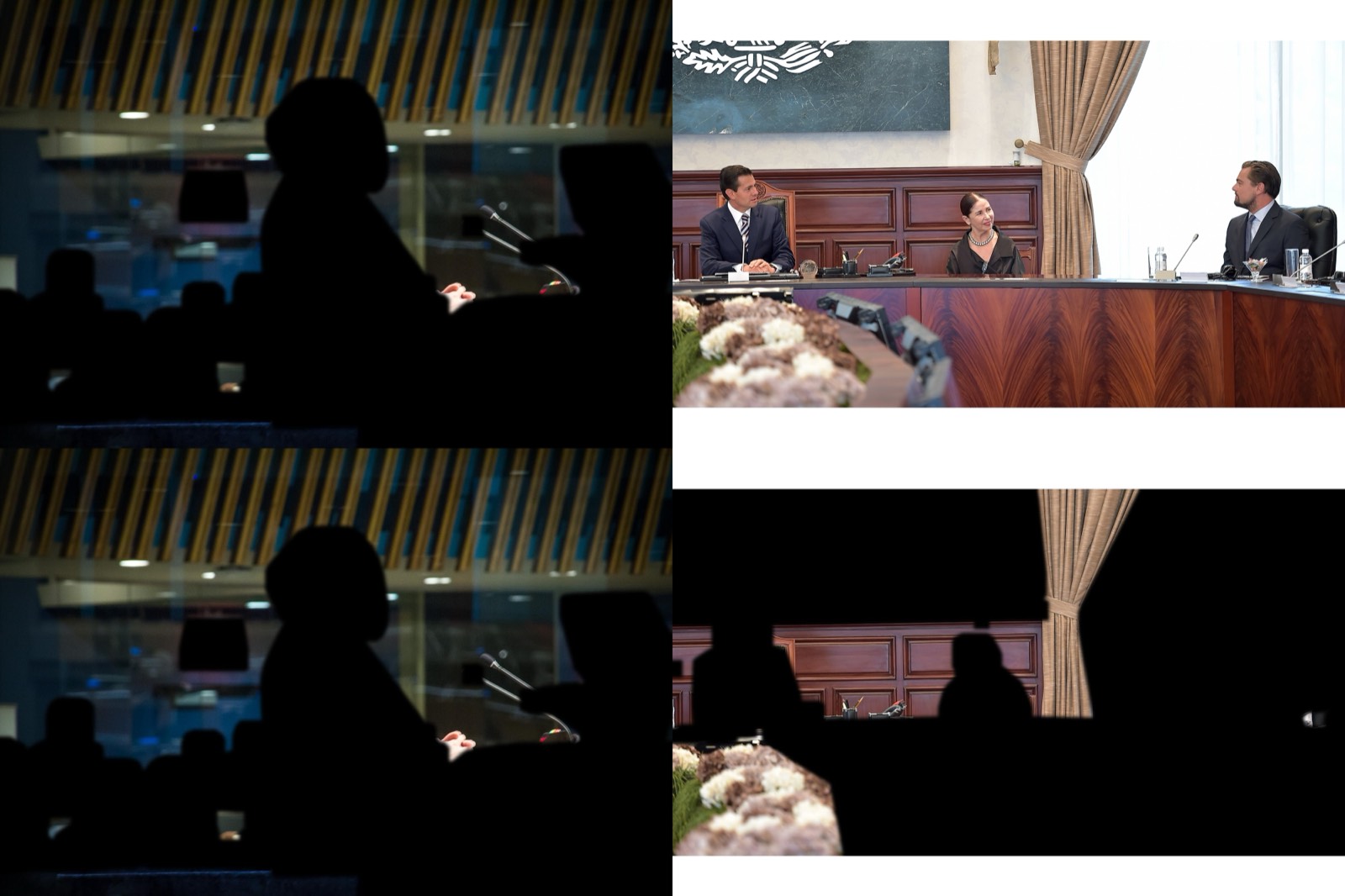}
& 
& \\
\addlinespace[1ex] 

\bottomrule
\end{tabular}
}
\caption{Decisive witnesses for 30 images capturing Leonardo DiCaprio and sanitizing the abstract concept \texttt{the identity of the celebrities}. Empty cells indicate that the test returned resolution 0 under \texttt{EVA}.}
\Description{A matrix of image pairs gives the decisive witness for each DiCaprio sanitization mechanism. Rows are redaction models, columns are concept-generation models, and empty cells indicate resolution zero.}
\label{fig:dicaprio_failure_compare}
\end{minipage}
}]
\clearpage

\twocolumn[{
\begin{minipage}{\textwidth}
\captionsetup{type=figure}
\centering  
\small

\setlength{\aboverulesep}{1pt}
\setlength{\belowrulesep}{1pt}
\setlength{\fboxrule}{1.5pt} 

\resizebox{0.99\textwidth}{!}{
\begin{tabular}{@{} l @{\hspace{3ex}} c c c c c @{}}
\toprule
Model 
& \thead{\texttt{GPT54}}
& \thead{\texttt{GEM31p}} 
& \thead{\texttt{OPS46}}
& \thead{\texttt{Manual}}
& \thead{None} \\
\midrule
\addlinespace[1.5ex] 

\texttt{iGPT15}
& \includegraphics[width=0.17\textwidth, valign=m]{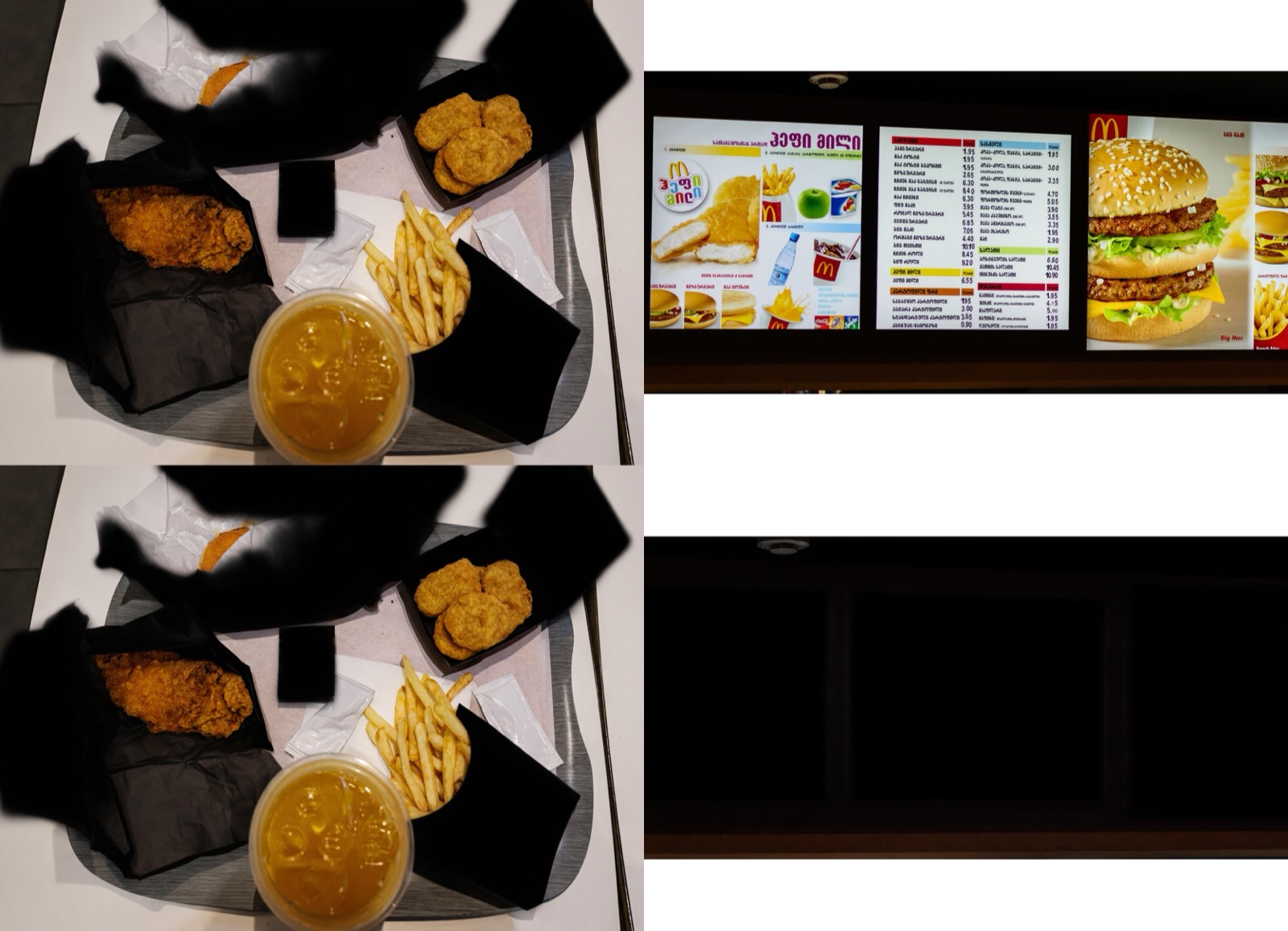}
& \includegraphics[width=0.17\textwidth, valign=m]{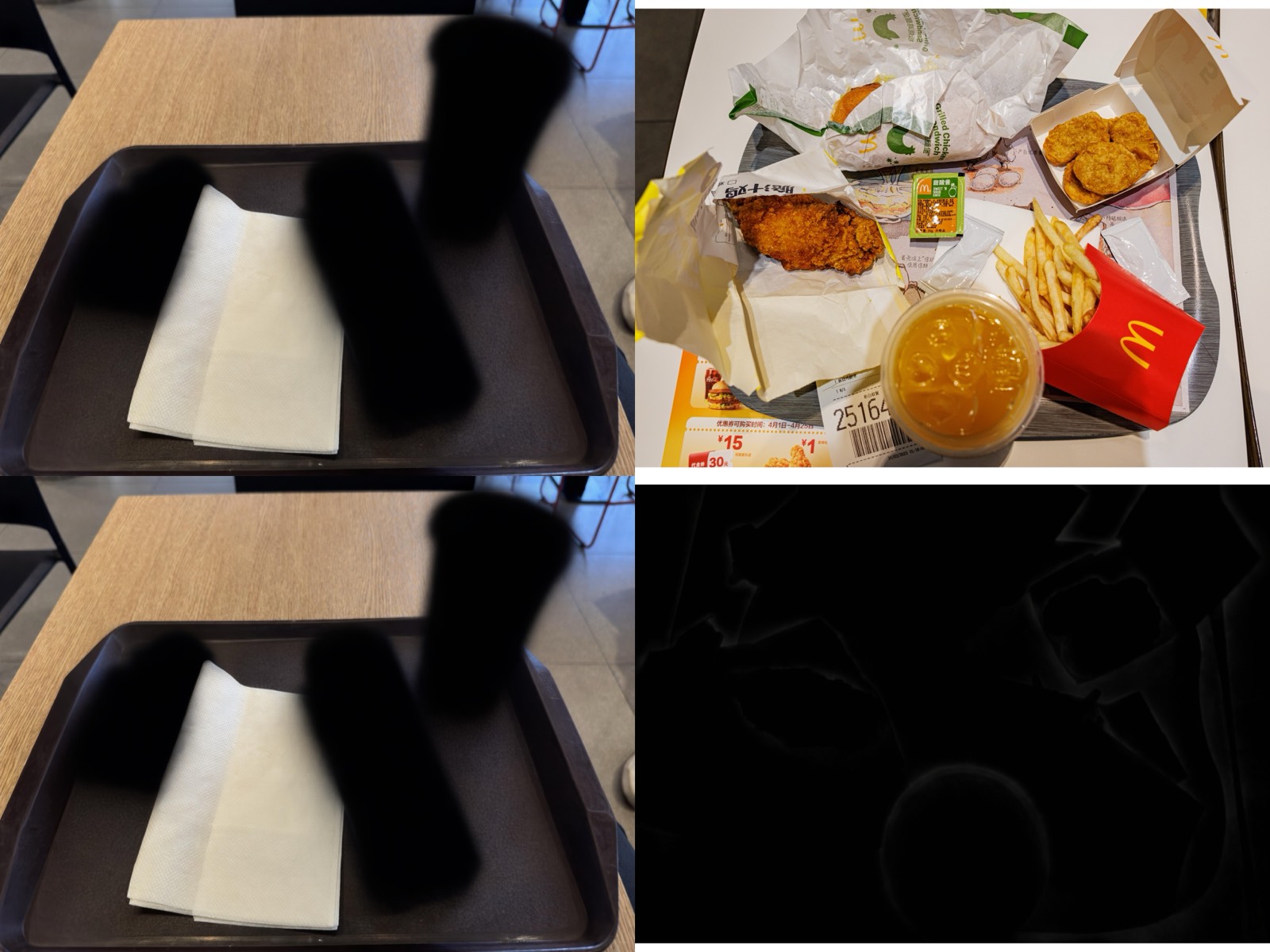}
& \includegraphics[width=0.17\textwidth, valign=m]{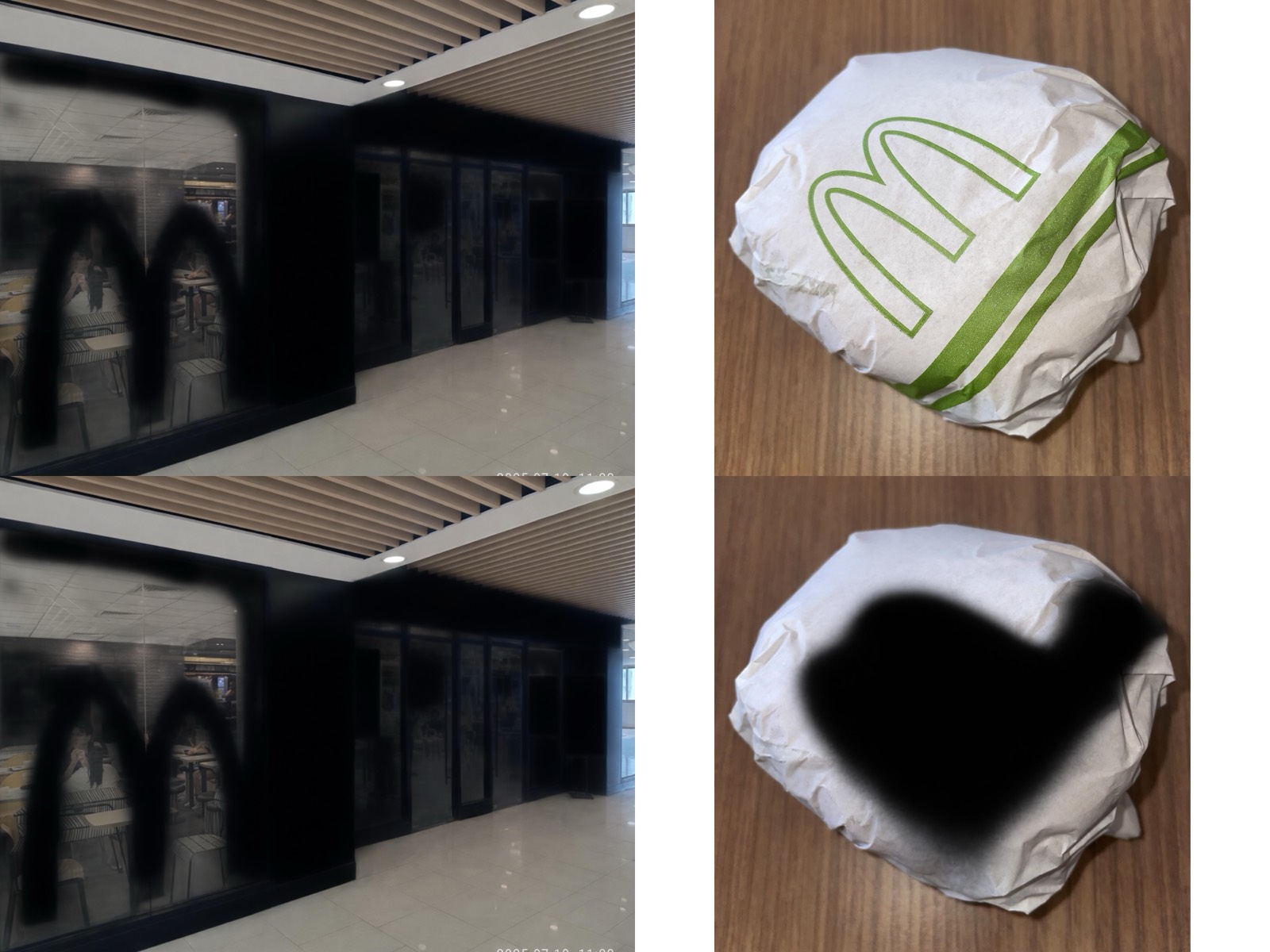}
& \includegraphics[width=0.17\textwidth, valign=m]{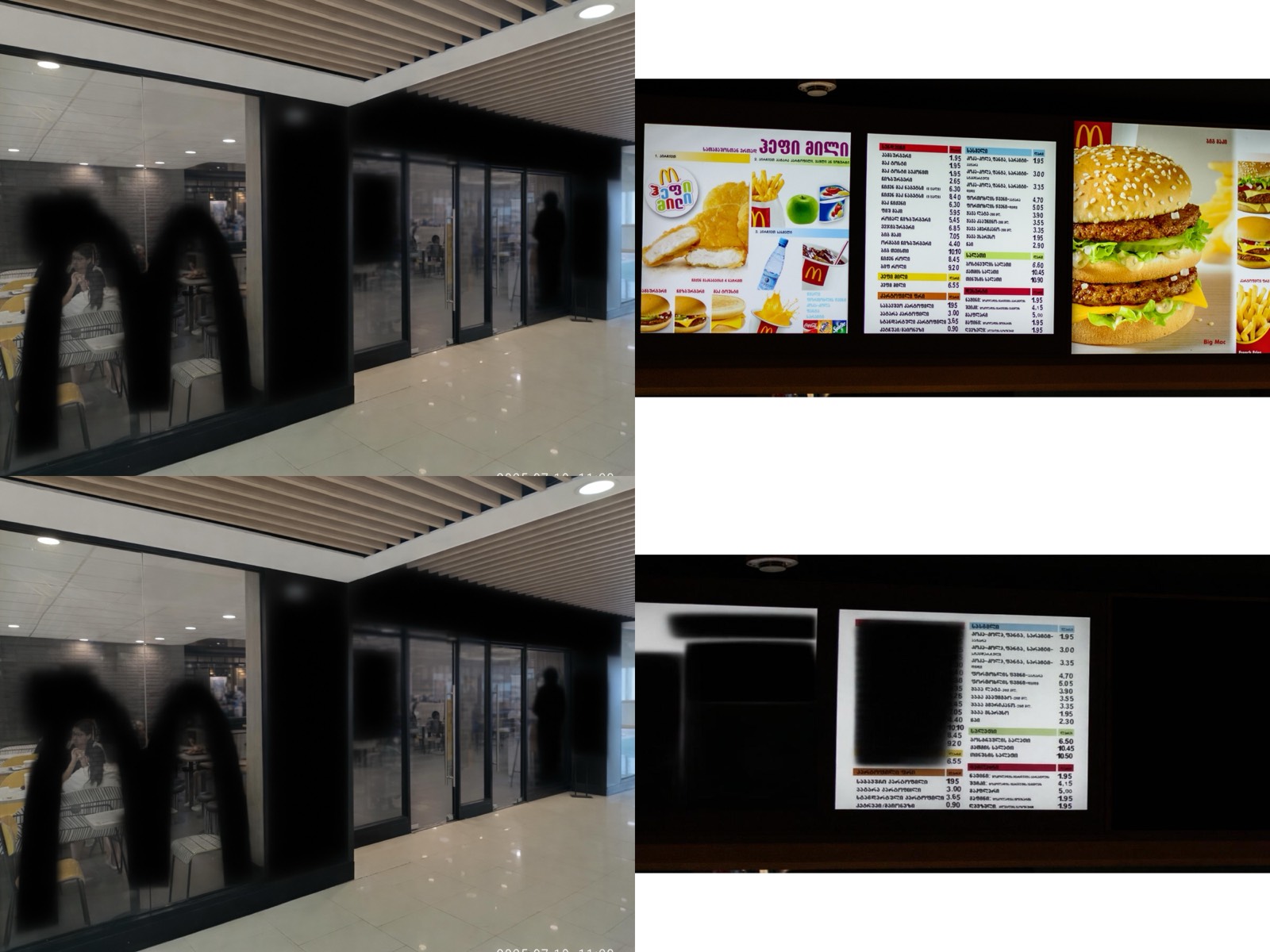}
& \includegraphics[width=0.17\textwidth, valign=m]{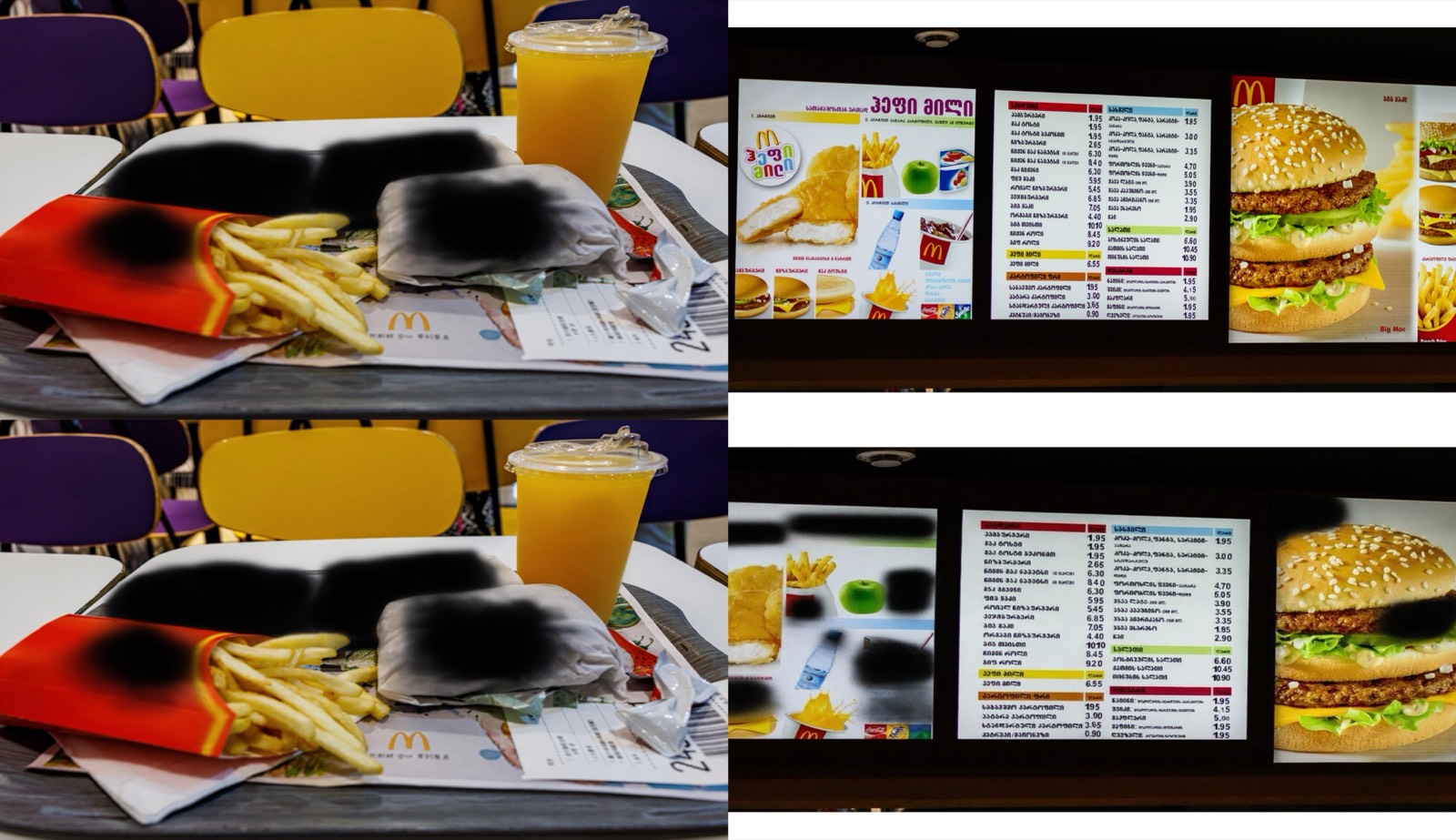} \\[1.5ex] 

\texttt{iGPT1m}
& \includegraphics[width=0.17\textwidth, valign=m]{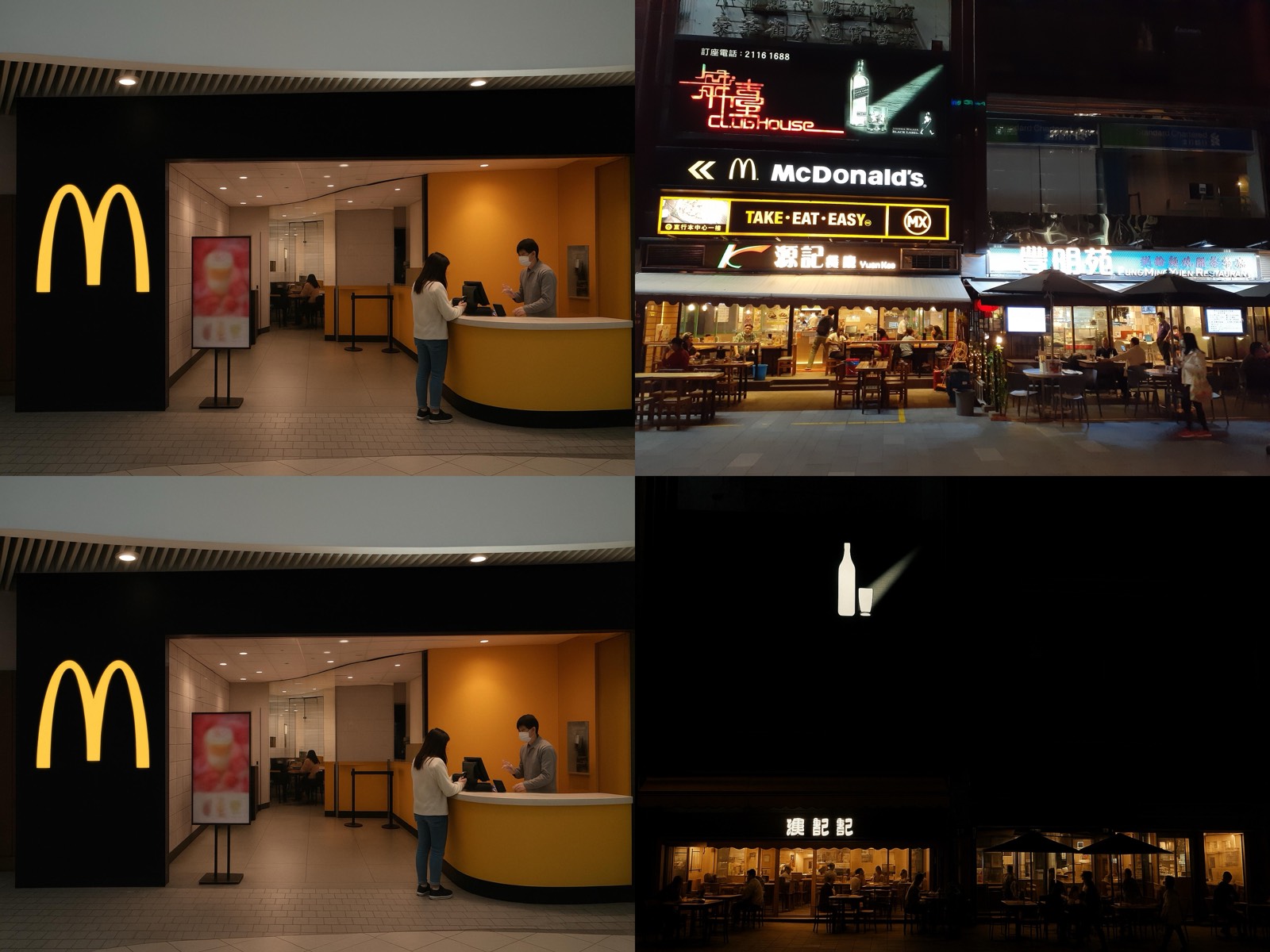}
& \includegraphics[width=0.17\textwidth, valign=m]{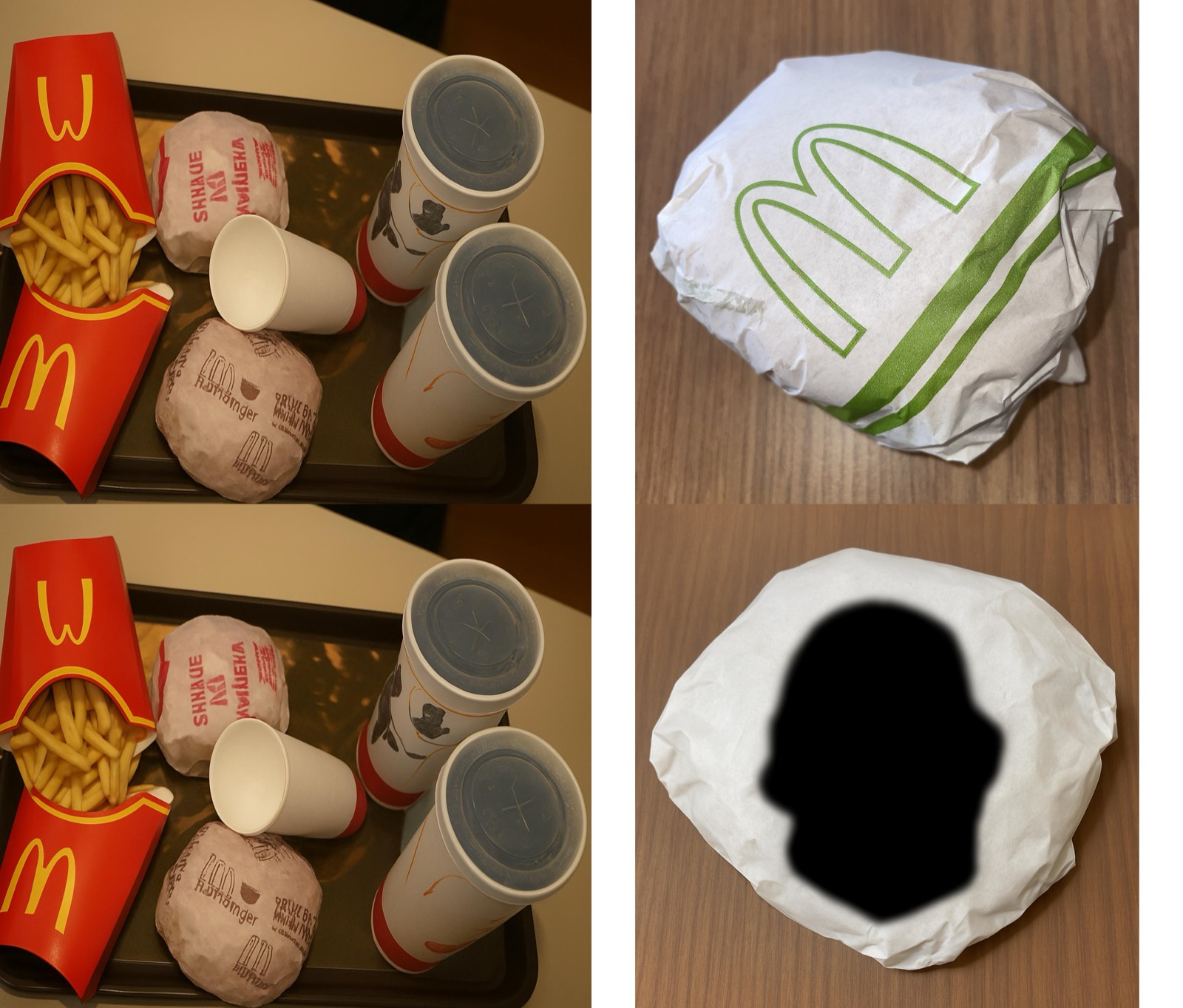}
& \includegraphics[width=0.17\textwidth, valign=m]{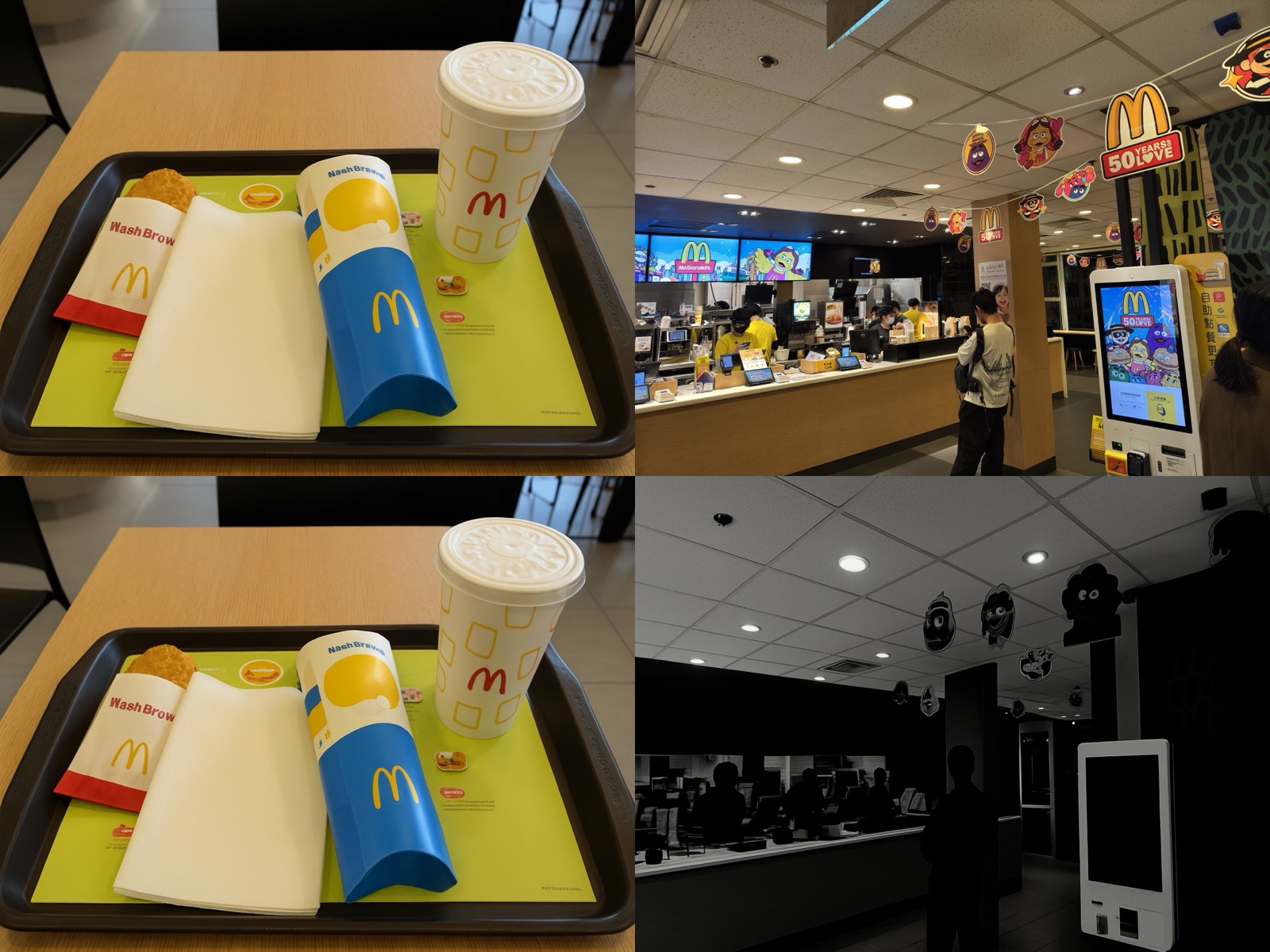}
& \includegraphics[width=0.17\textwidth, valign=m]{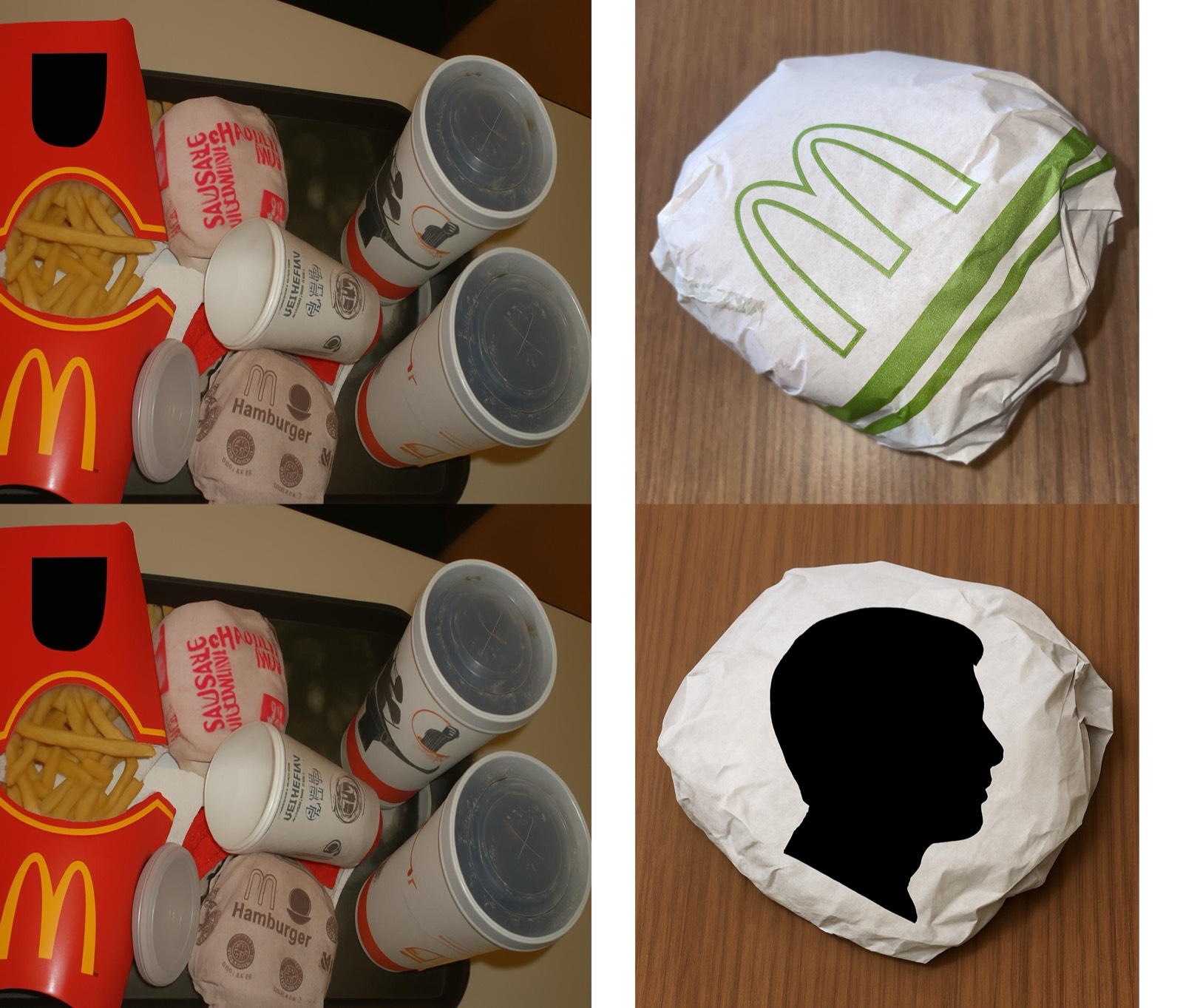}
& \includegraphics[width=0.17\textwidth, valign=m]{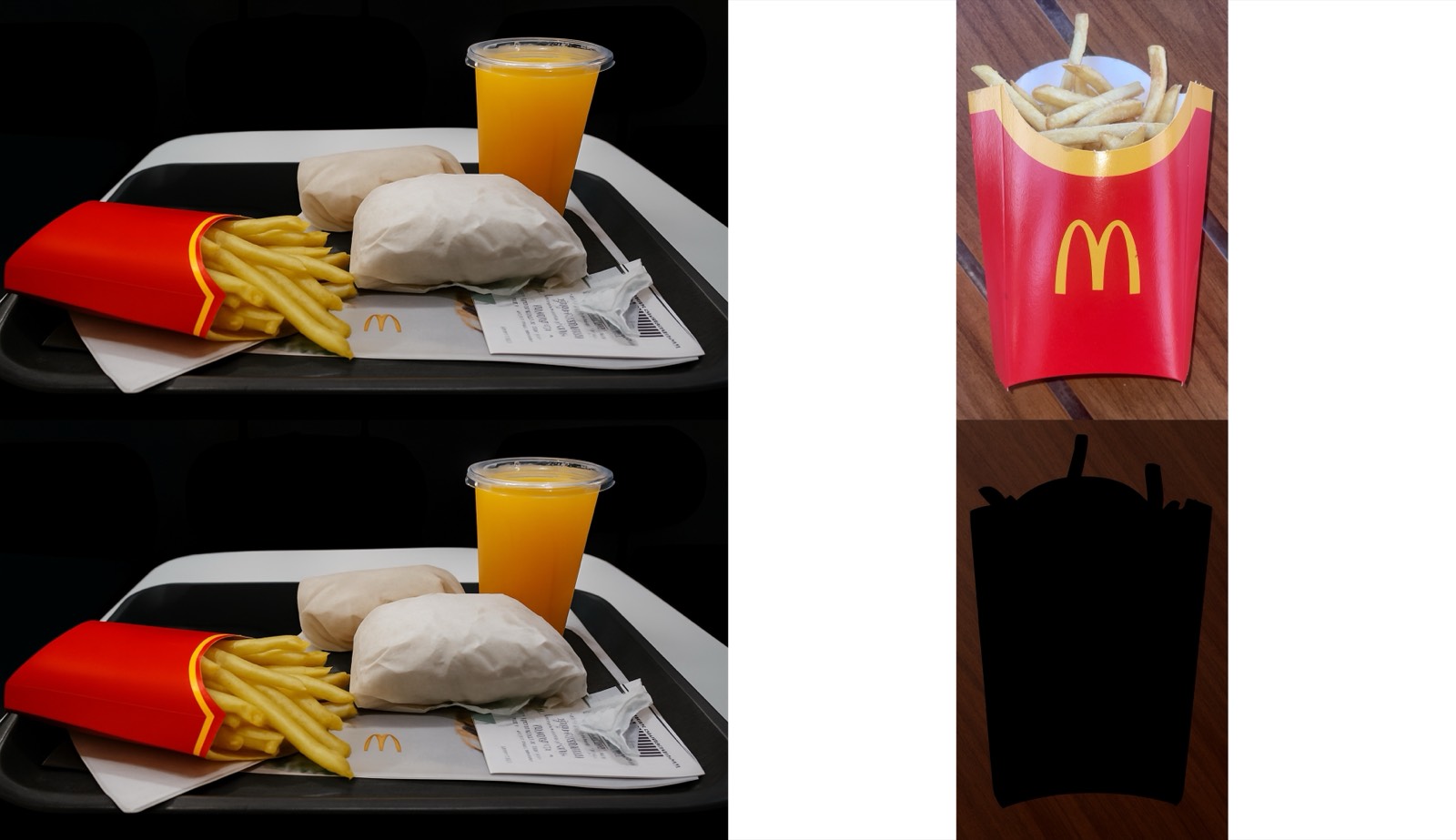} \\[1.5ex]

\texttt{iGEM3p}
& \includegraphics[width=0.17\textwidth, valign=m]{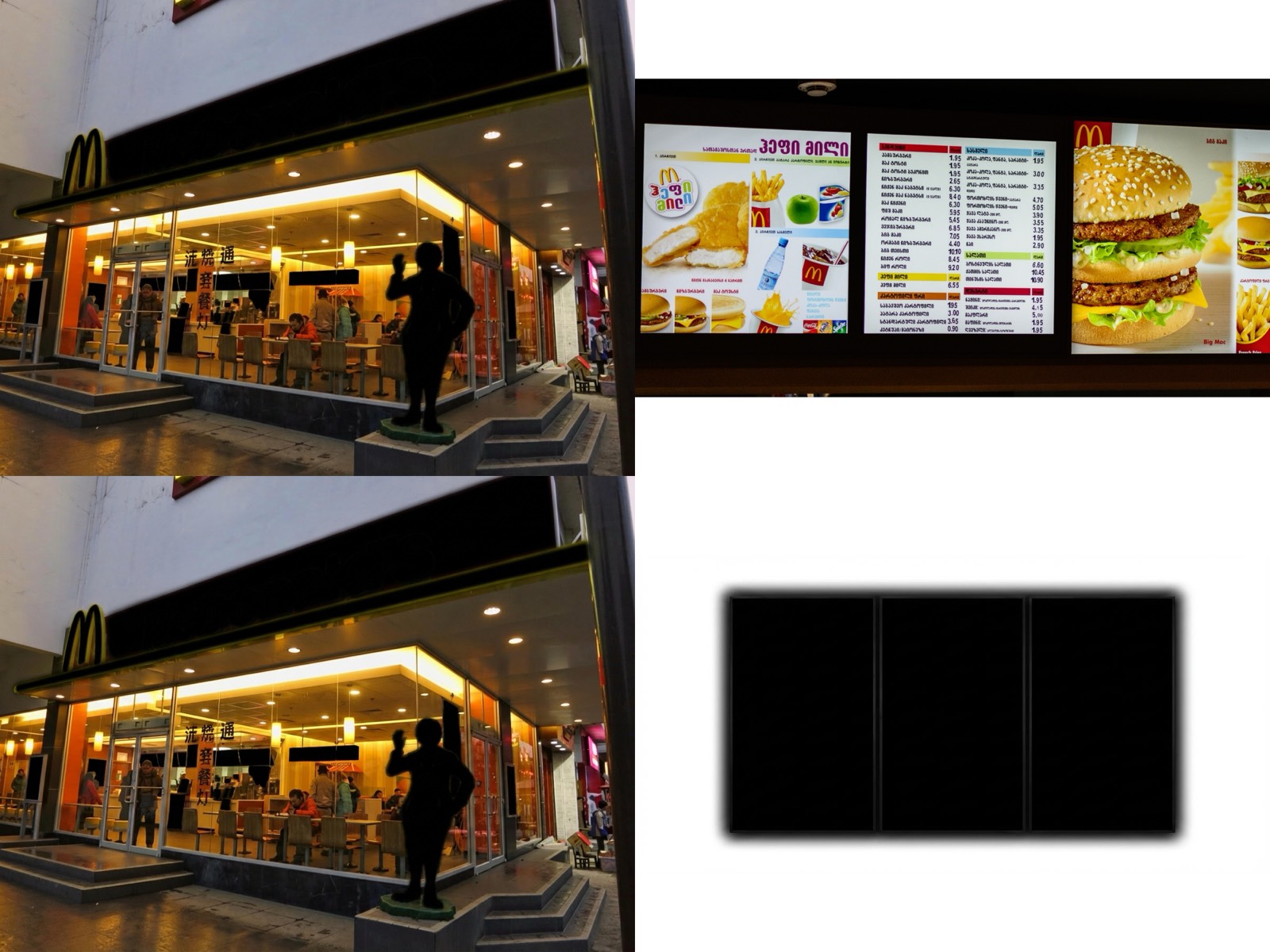}
& \includegraphics[width=0.17\textwidth, valign=m]{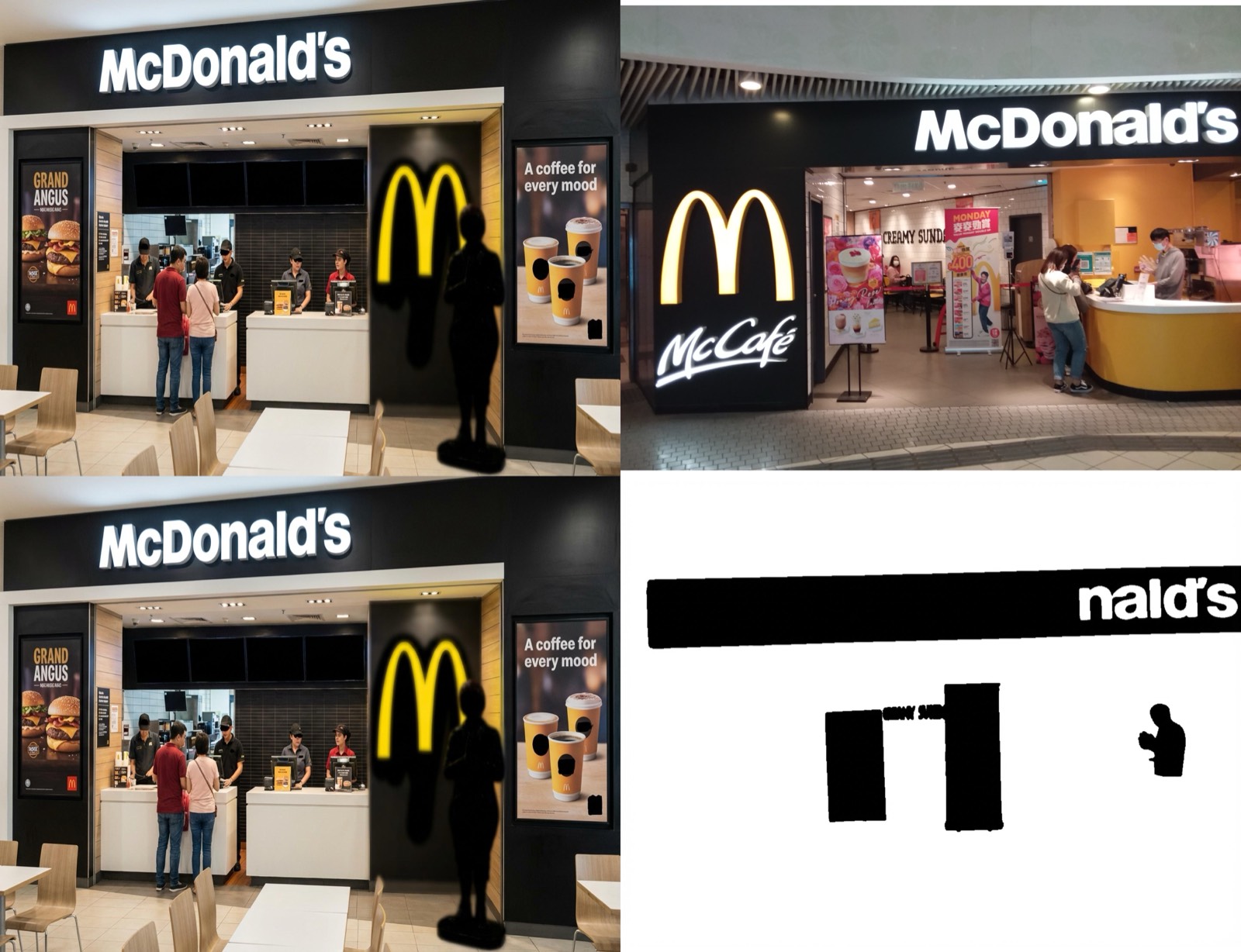}
& \includegraphics[width=0.17\textwidth, valign=m]{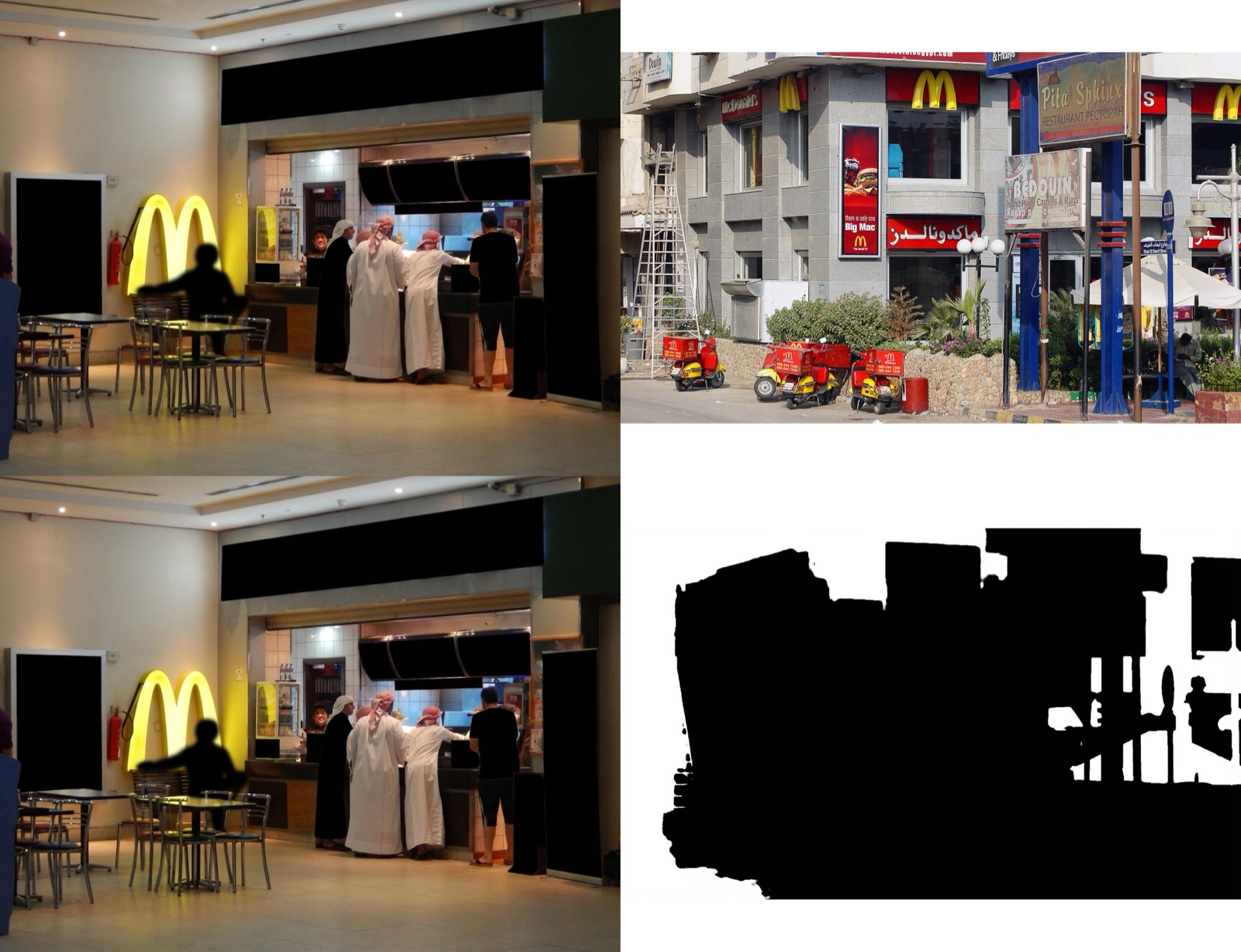}
& \includegraphics[width=0.17\textwidth, valign=m]{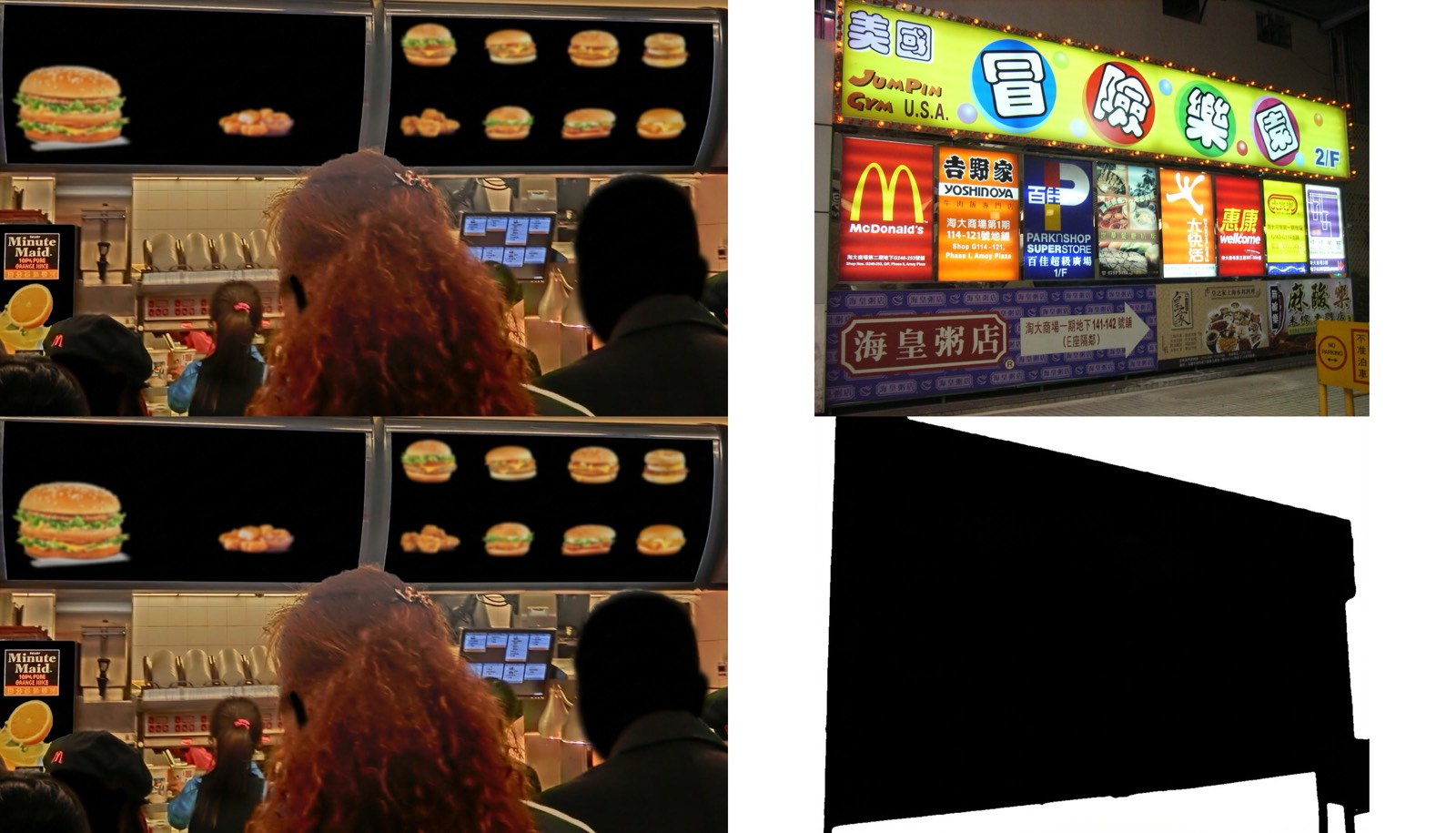}
& \includegraphics[width=0.17\textwidth, valign=m]{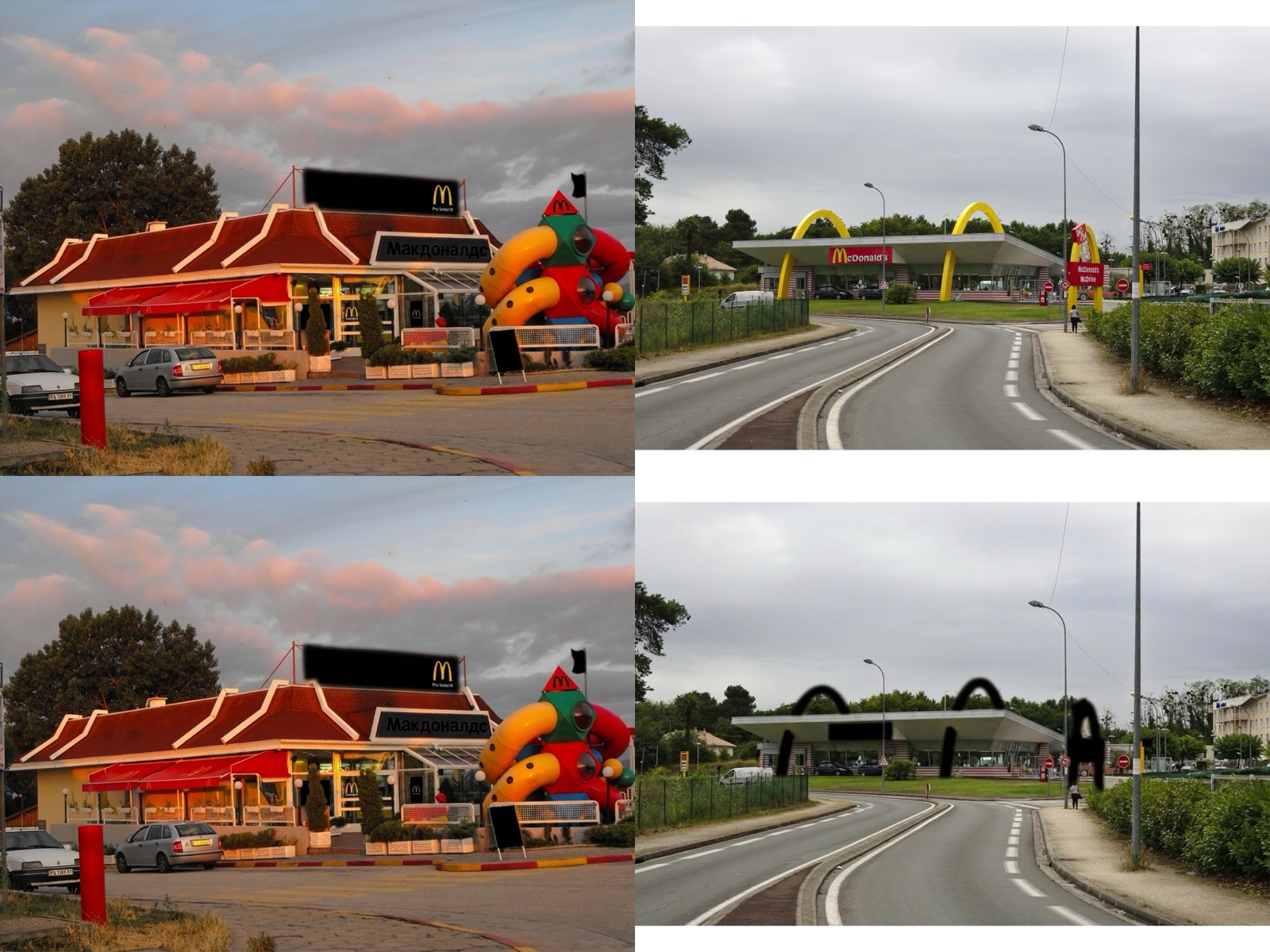} \\[1.5ex]

\texttt{iGEM31f}
& \includegraphics[width=0.17\textwidth, valign=m]{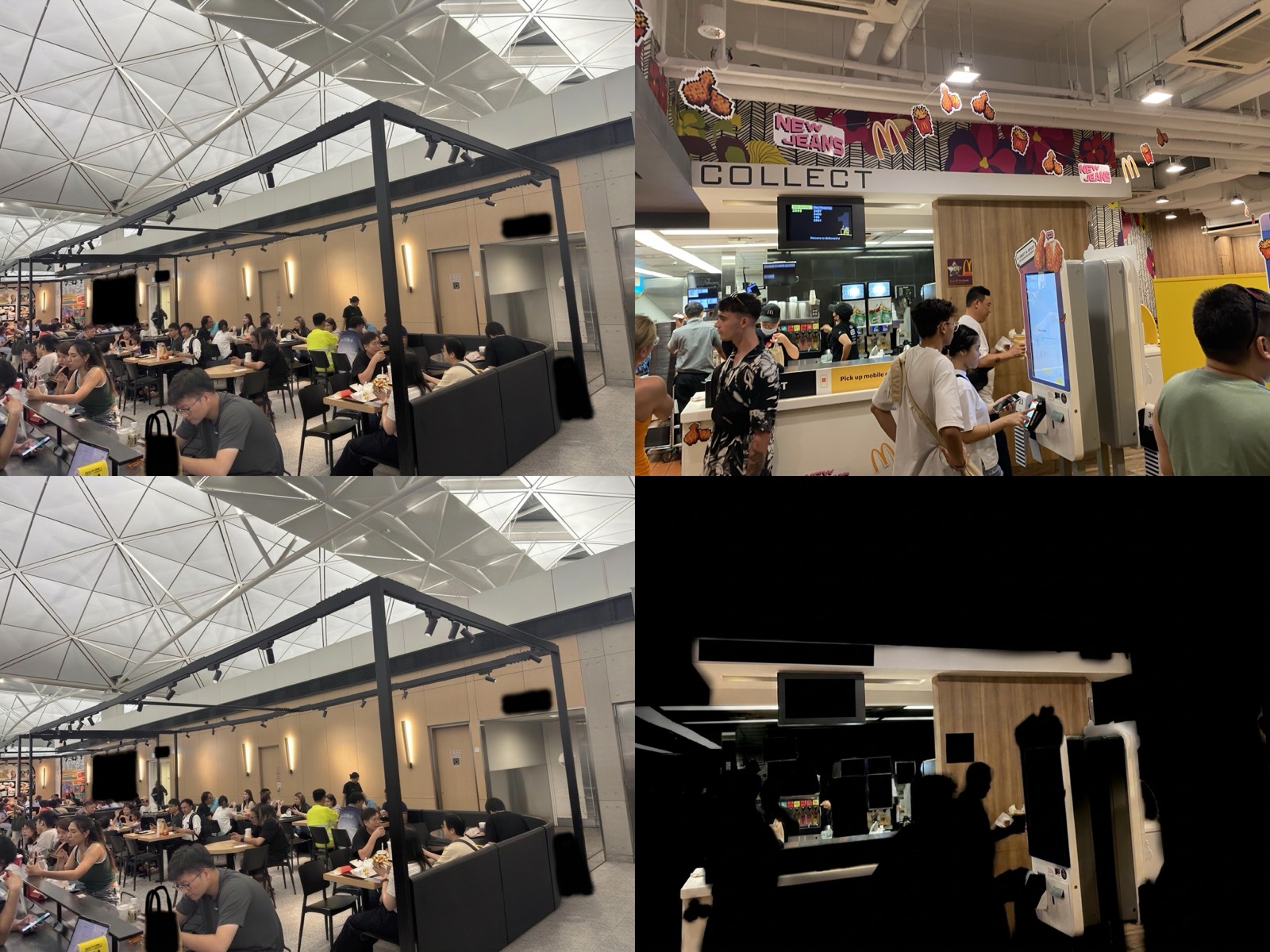}
& \includegraphics[width=0.17\textwidth, valign=m]{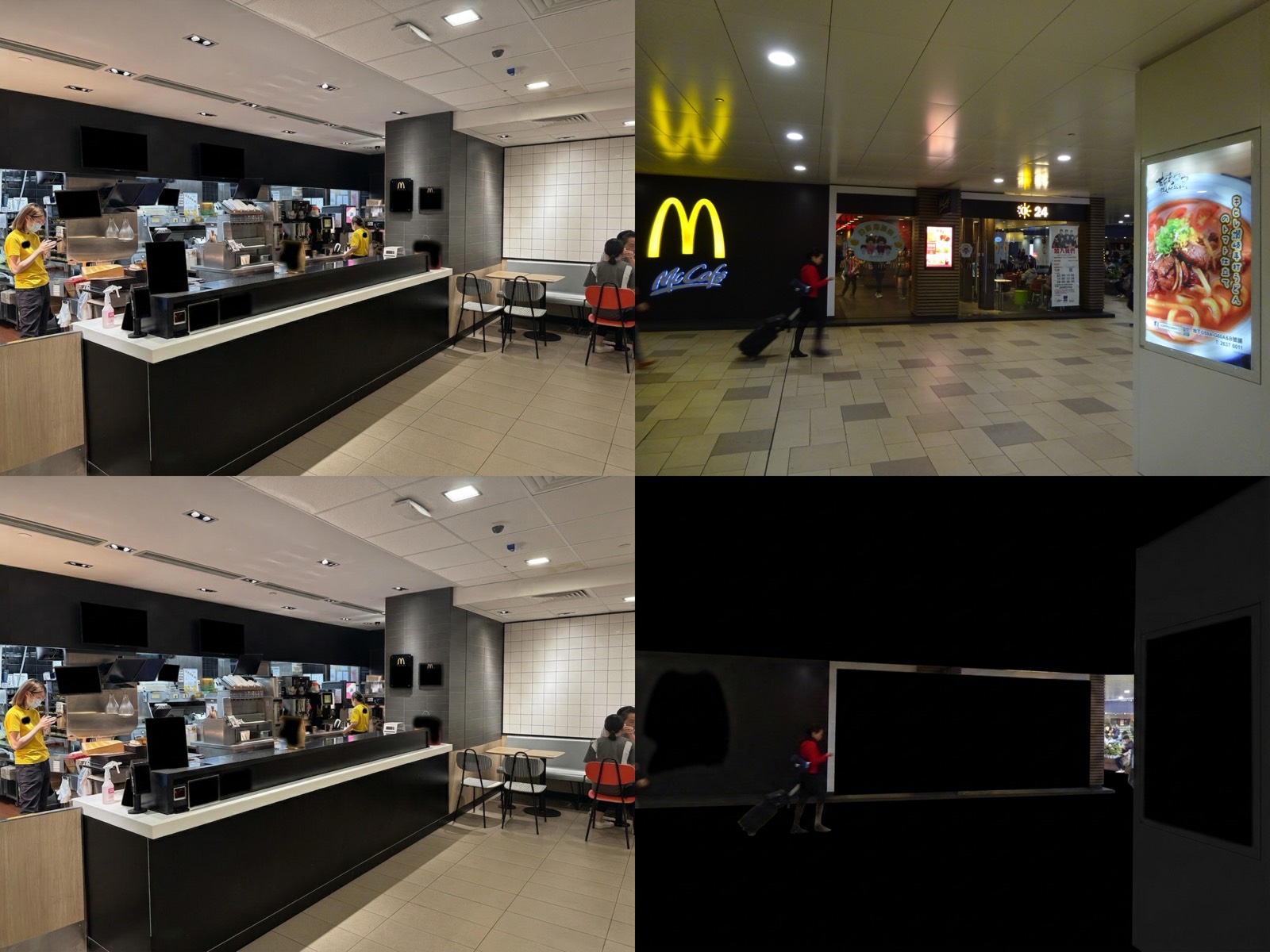}
& \includegraphics[width=0.17\textwidth, valign=m]{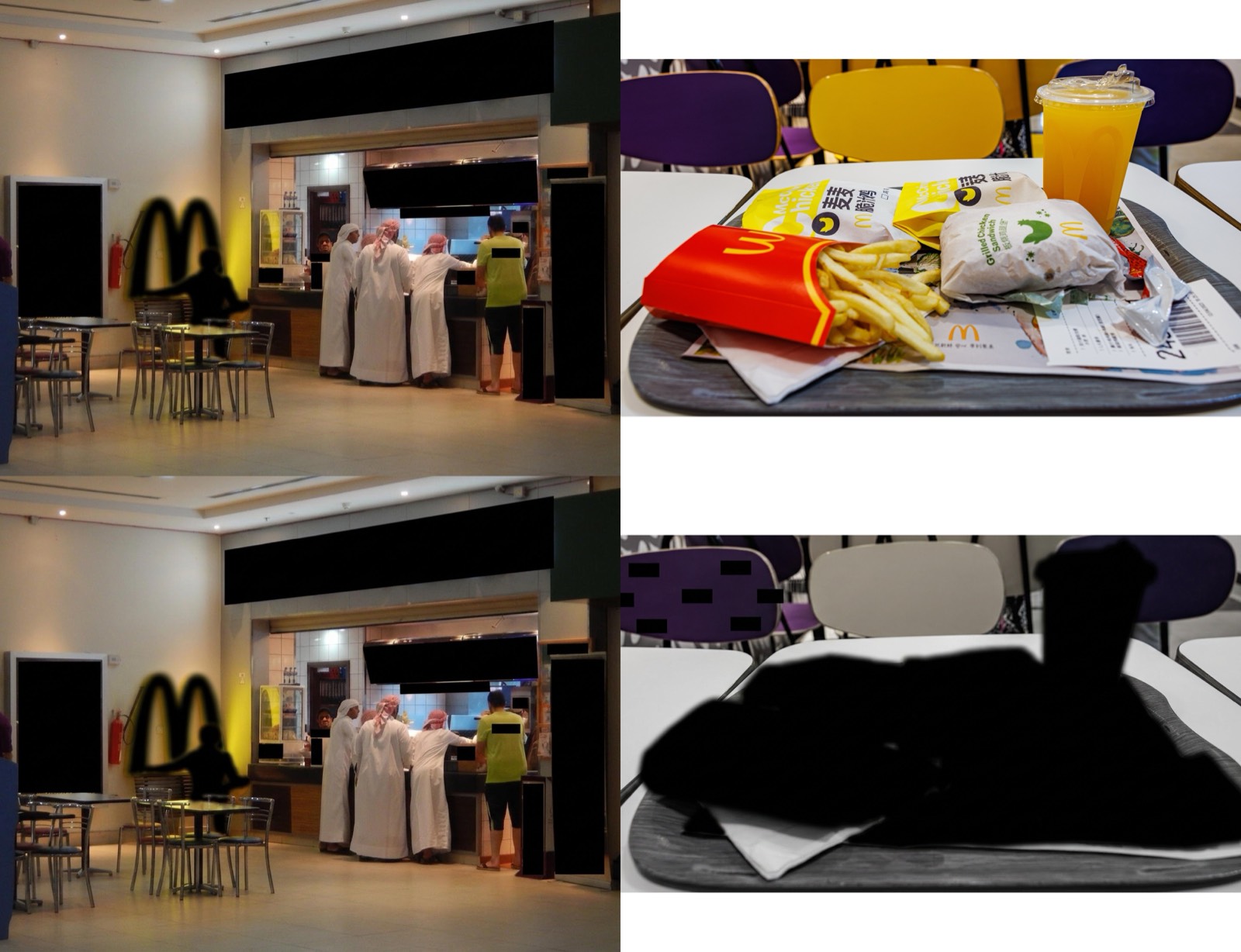}
& \includegraphics[width=0.17\textwidth, valign=m]{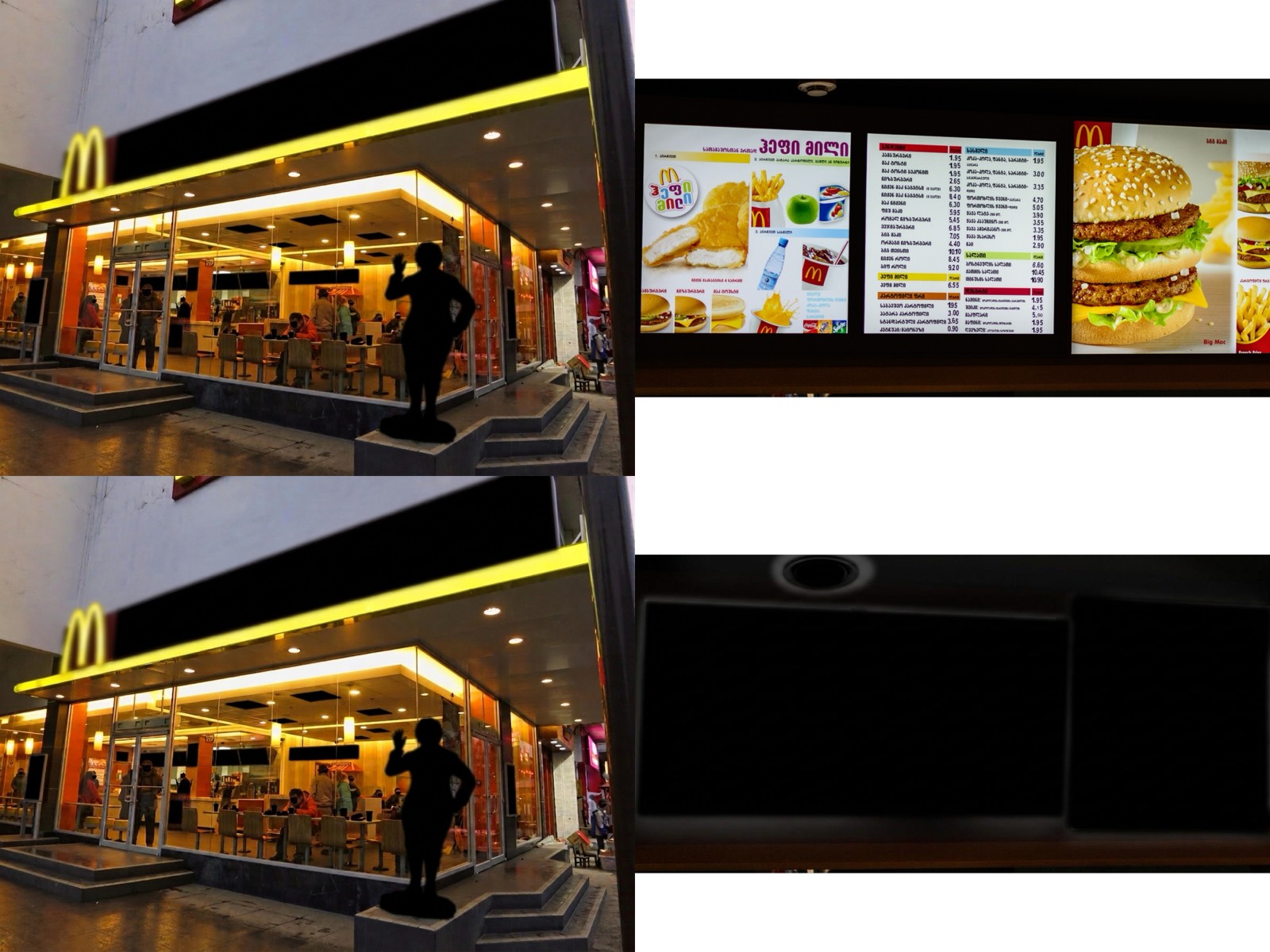}
& \includegraphics[width=0.17\textwidth, valign=m]{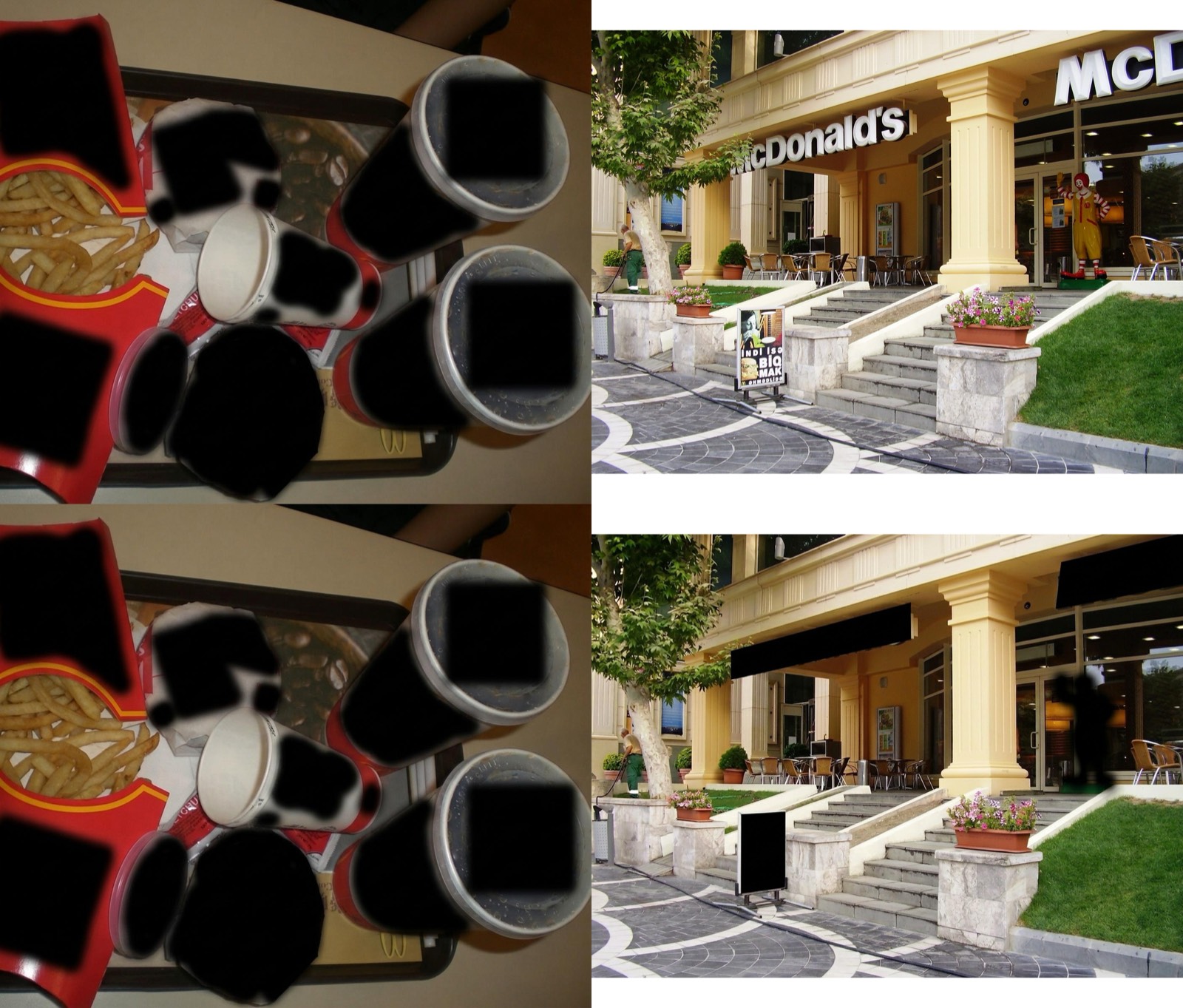} \\[1.5ex]

\texttt{FLX2p}
& \includegraphics[width=0.17\textwidth, valign=m]{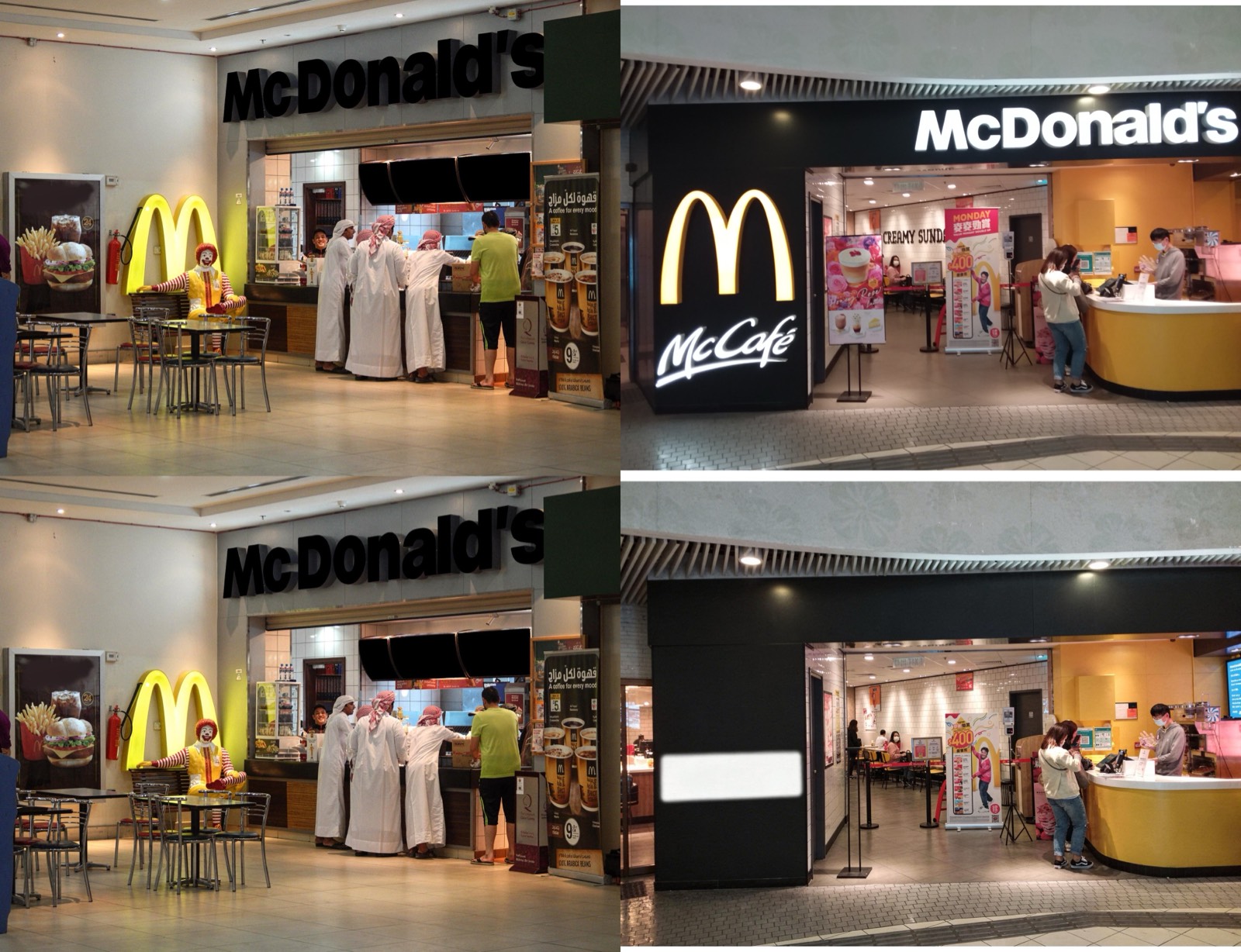}
& \includegraphics[width=0.17\textwidth, valign=m]{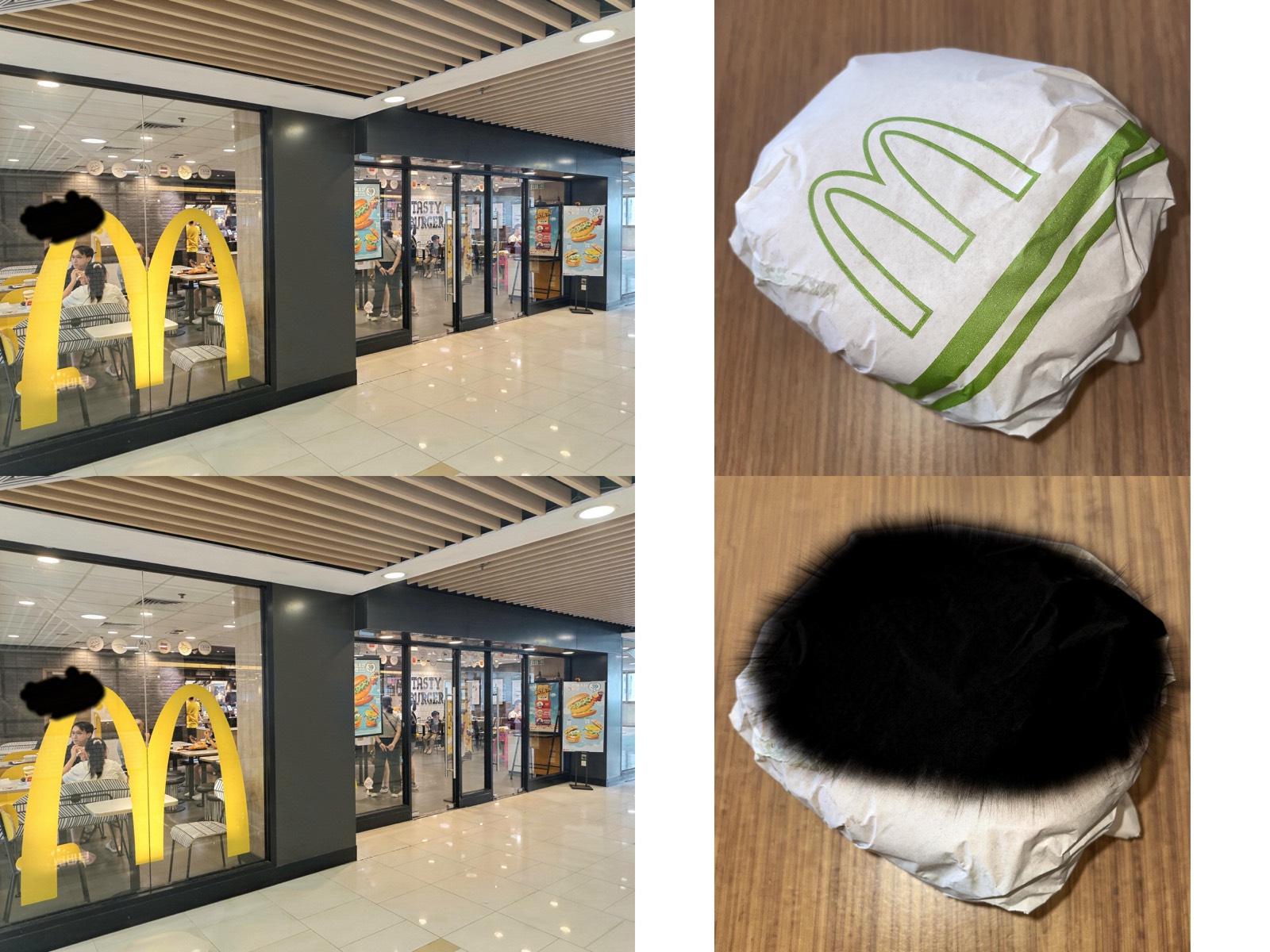}
& \includegraphics[width=0.17\textwidth, valign=m]{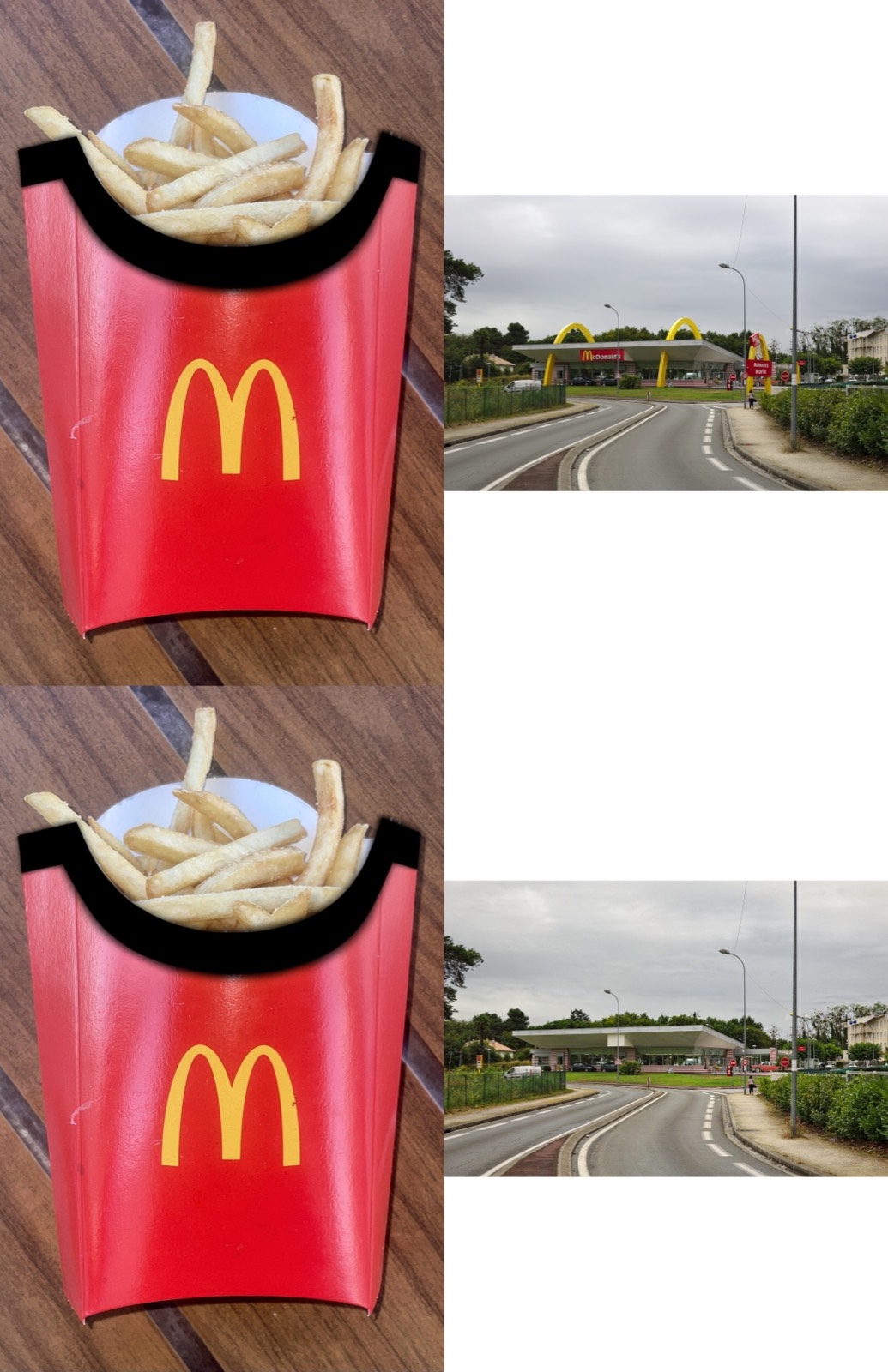}
& \includegraphics[width=0.17\textwidth, valign=m]{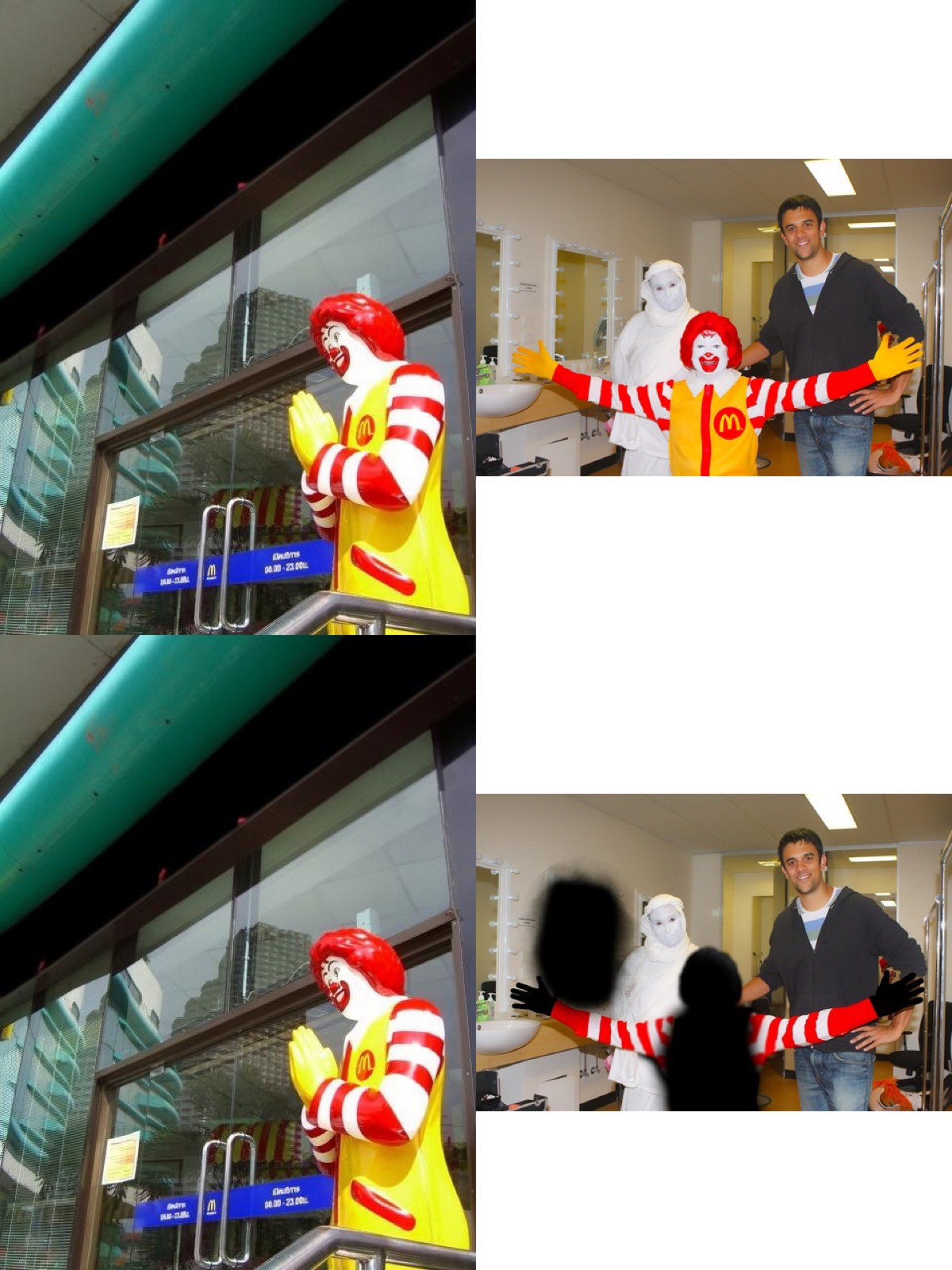}
& \includegraphics[width=0.17\textwidth, valign=m]{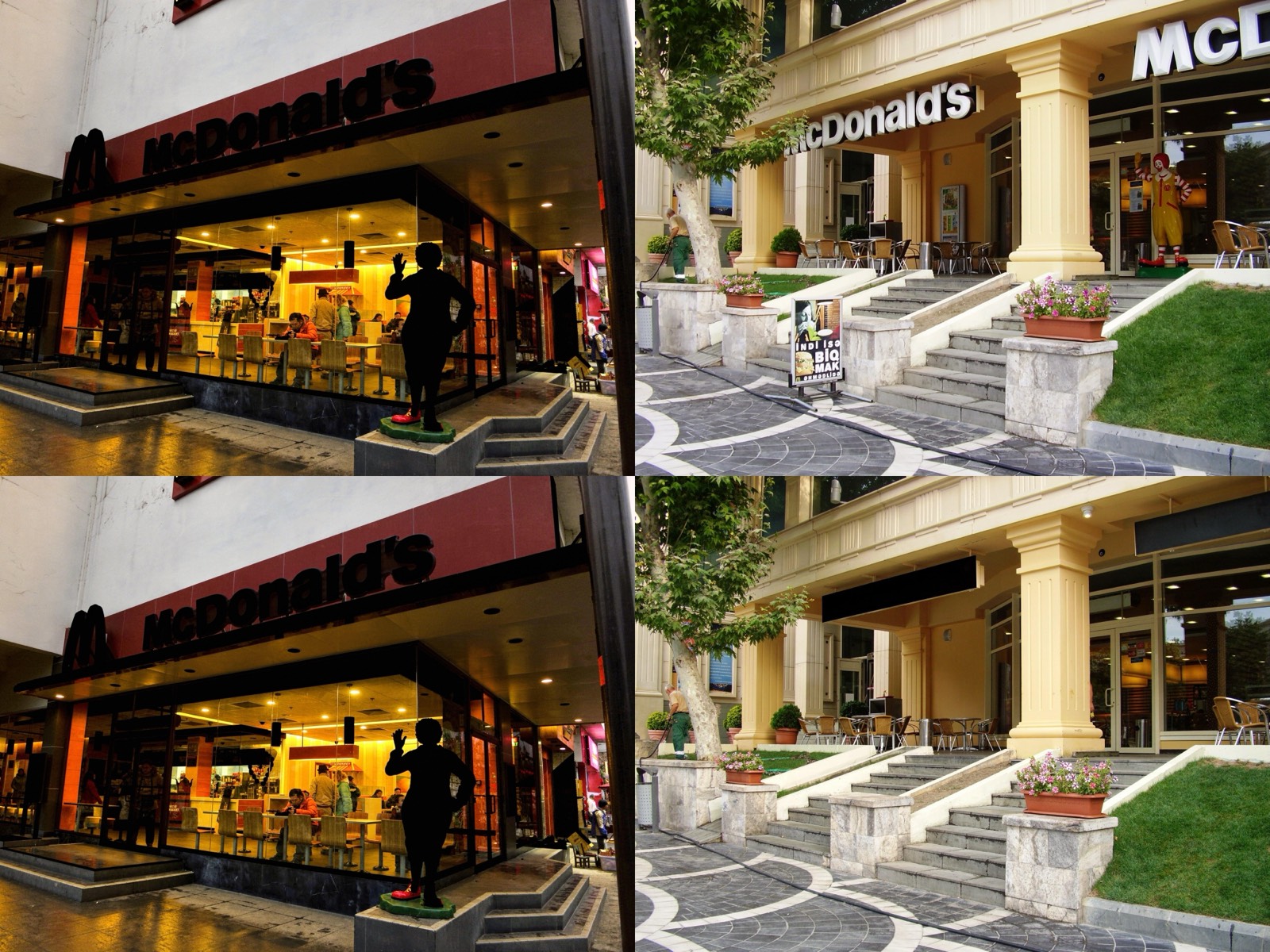} \\[1.5ex]

\texttt{FLX2d}
& \includegraphics[width=0.17\textwidth, valign=m]{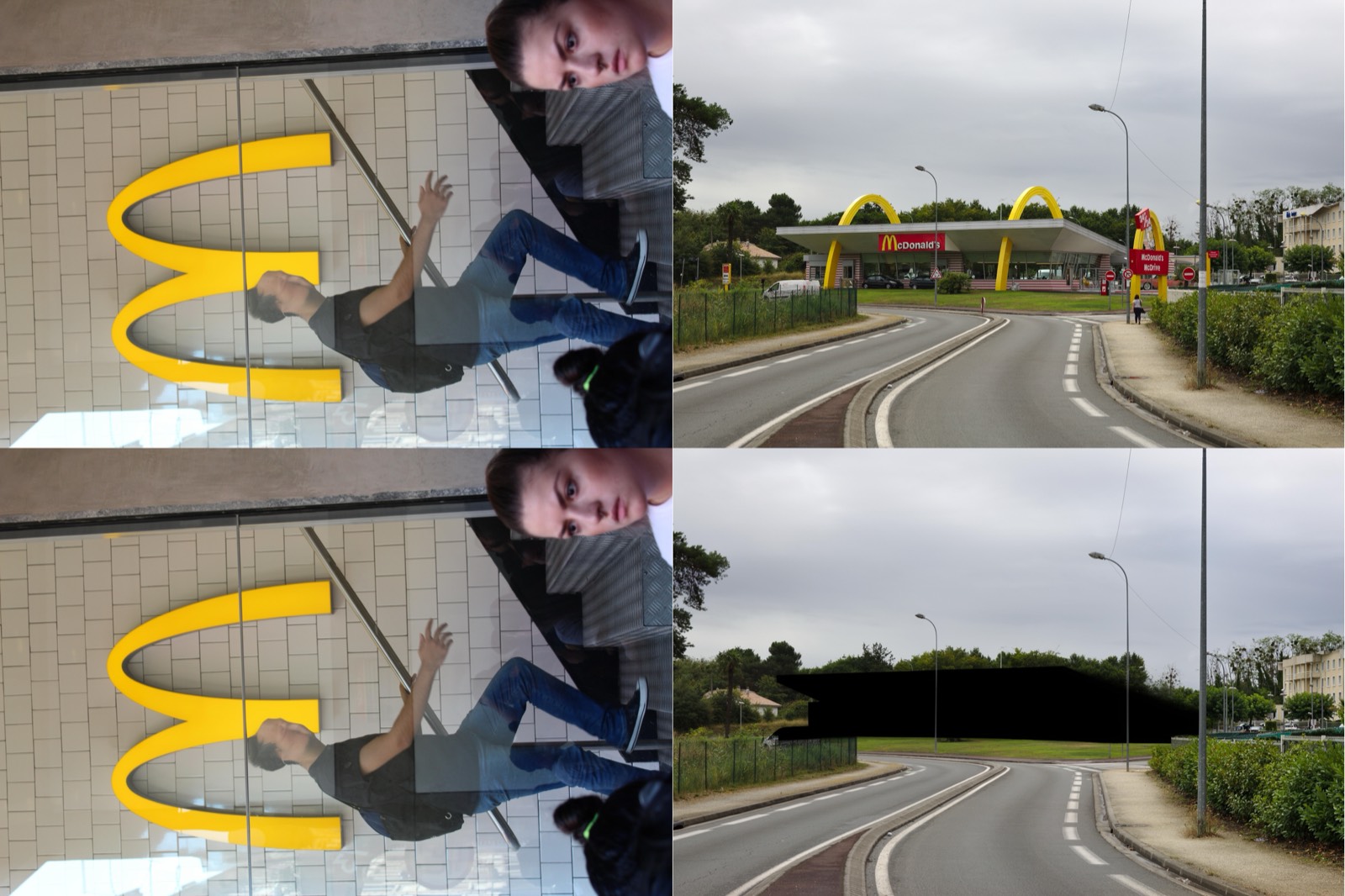}
& \includegraphics[width=0.17\textwidth, valign=m]{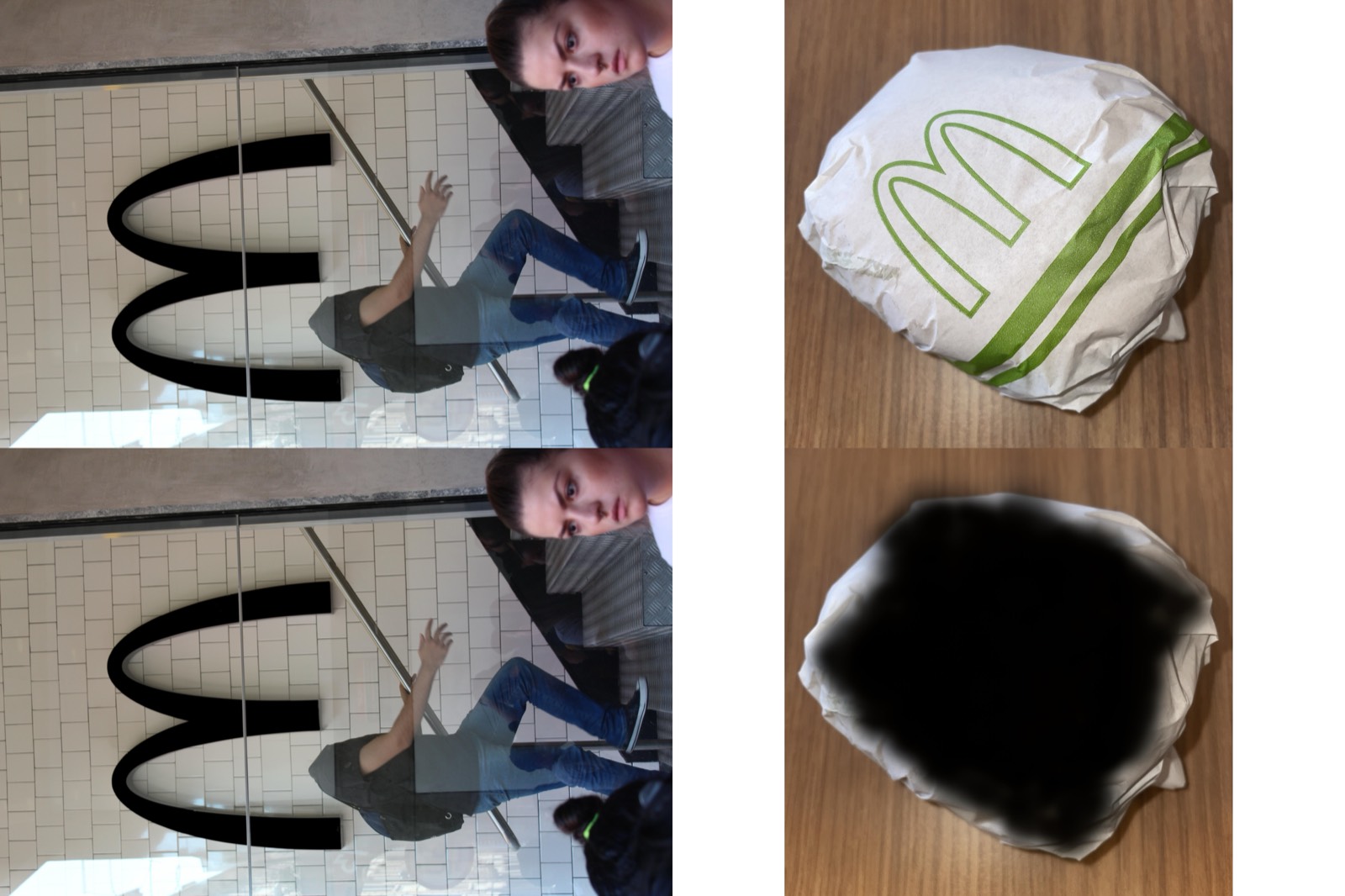}
& \includegraphics[width=0.17\textwidth, valign=m]{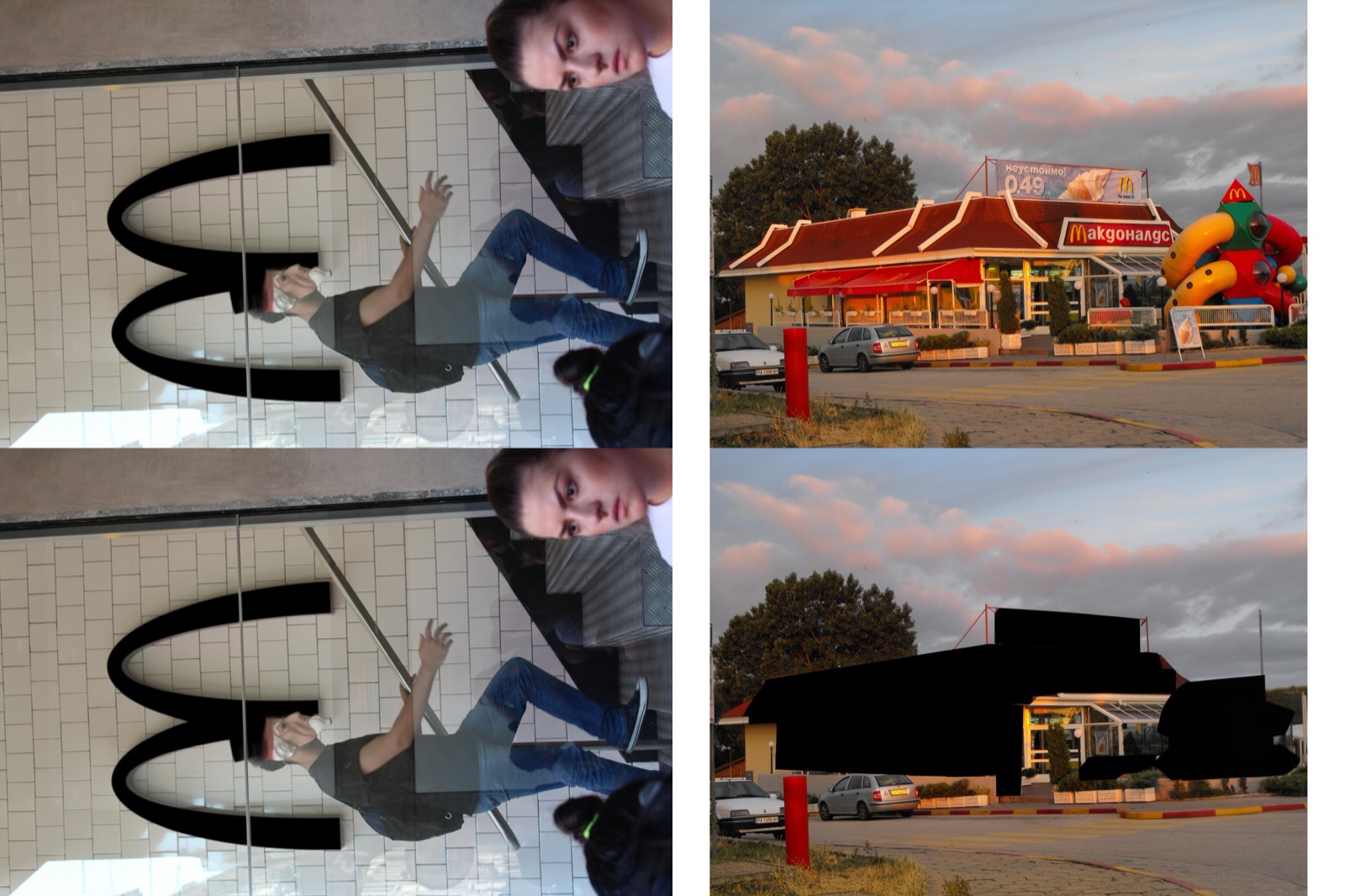}
& \includegraphics[width=0.17\textwidth, valign=m]{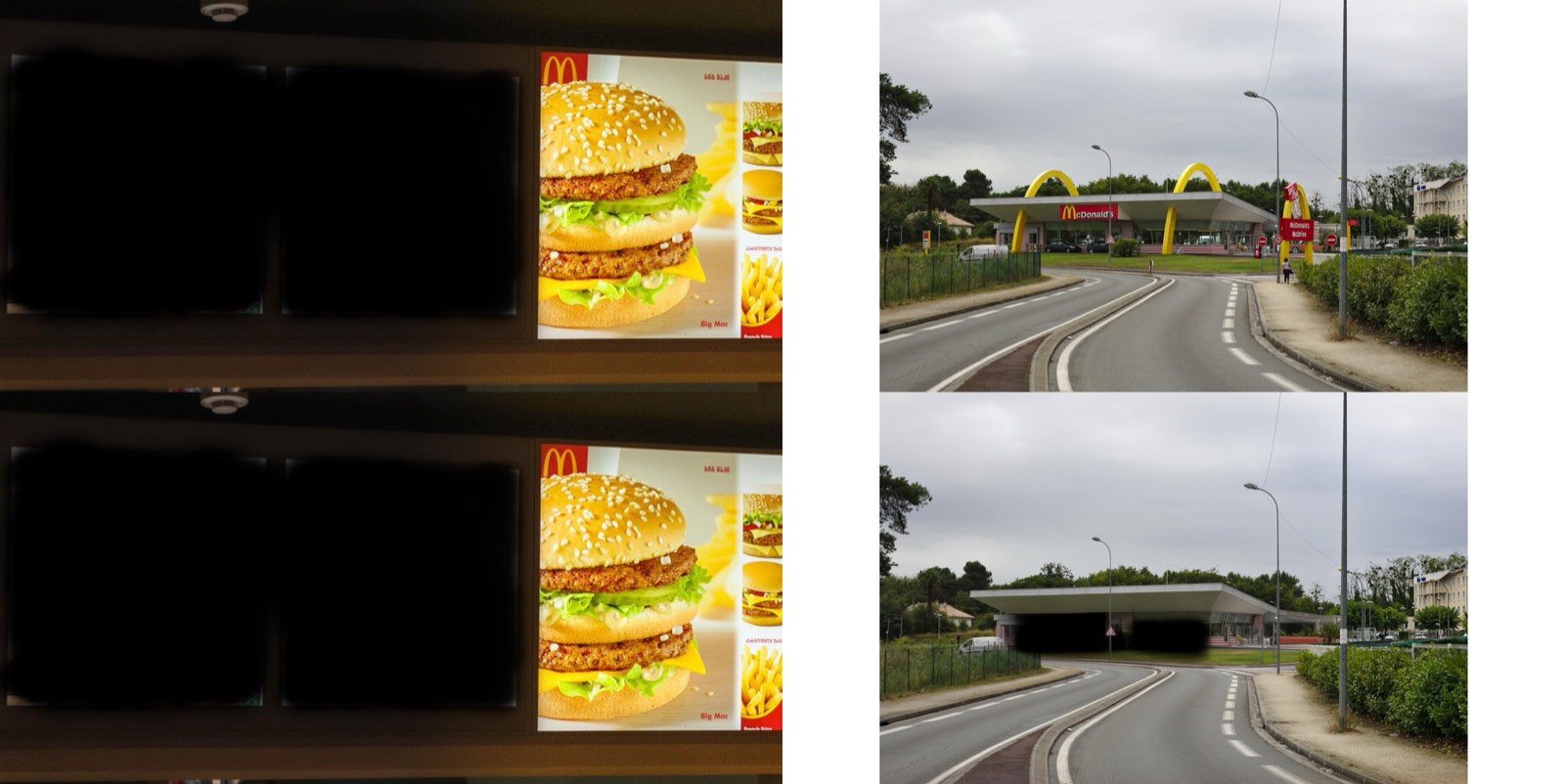}
& \includegraphics[width=0.17\textwidth, valign=m]{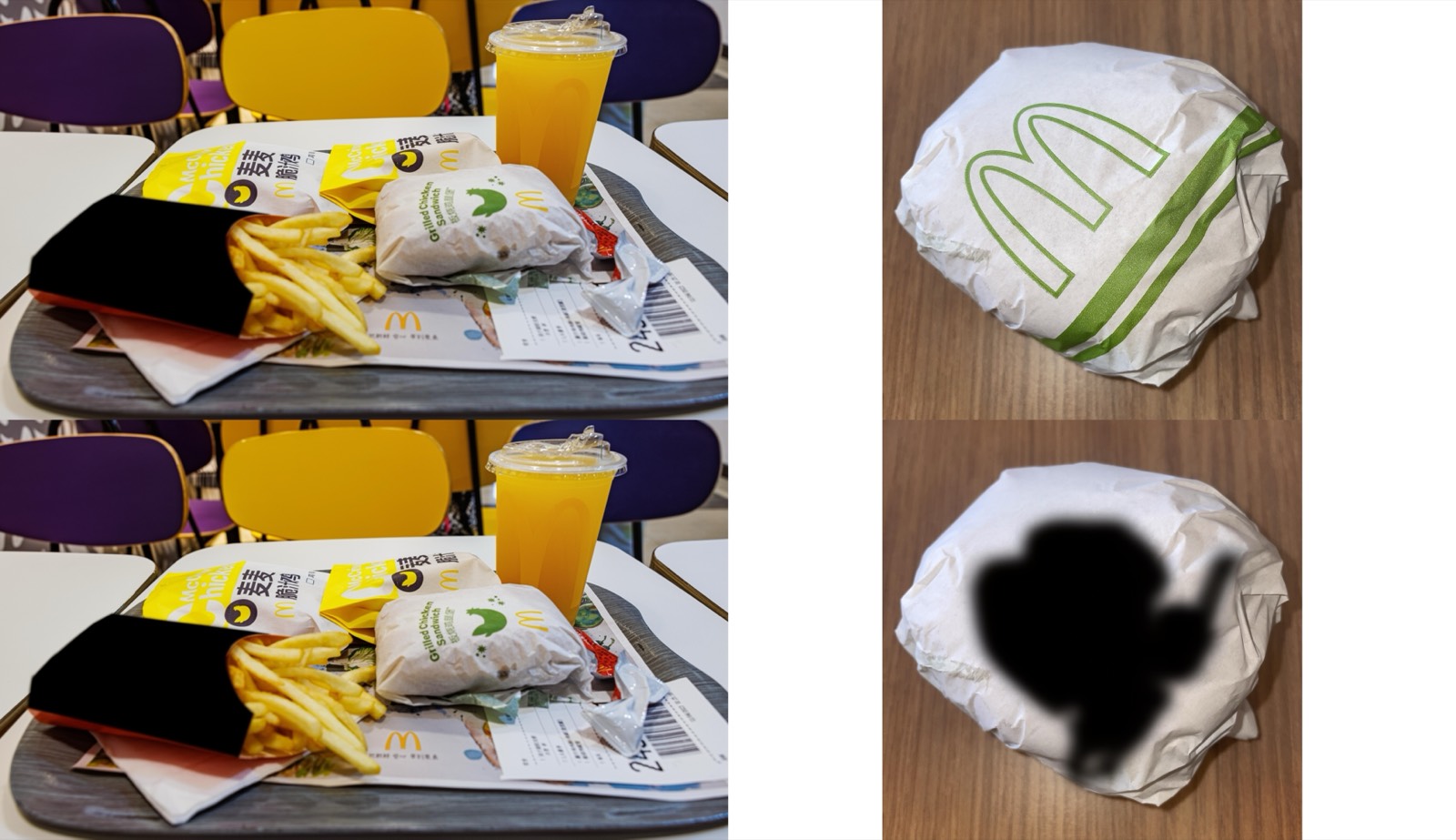} \\[1.5ex]

\texttt{SAM3}
& \includegraphics[width=0.17\textwidth, valign=m]{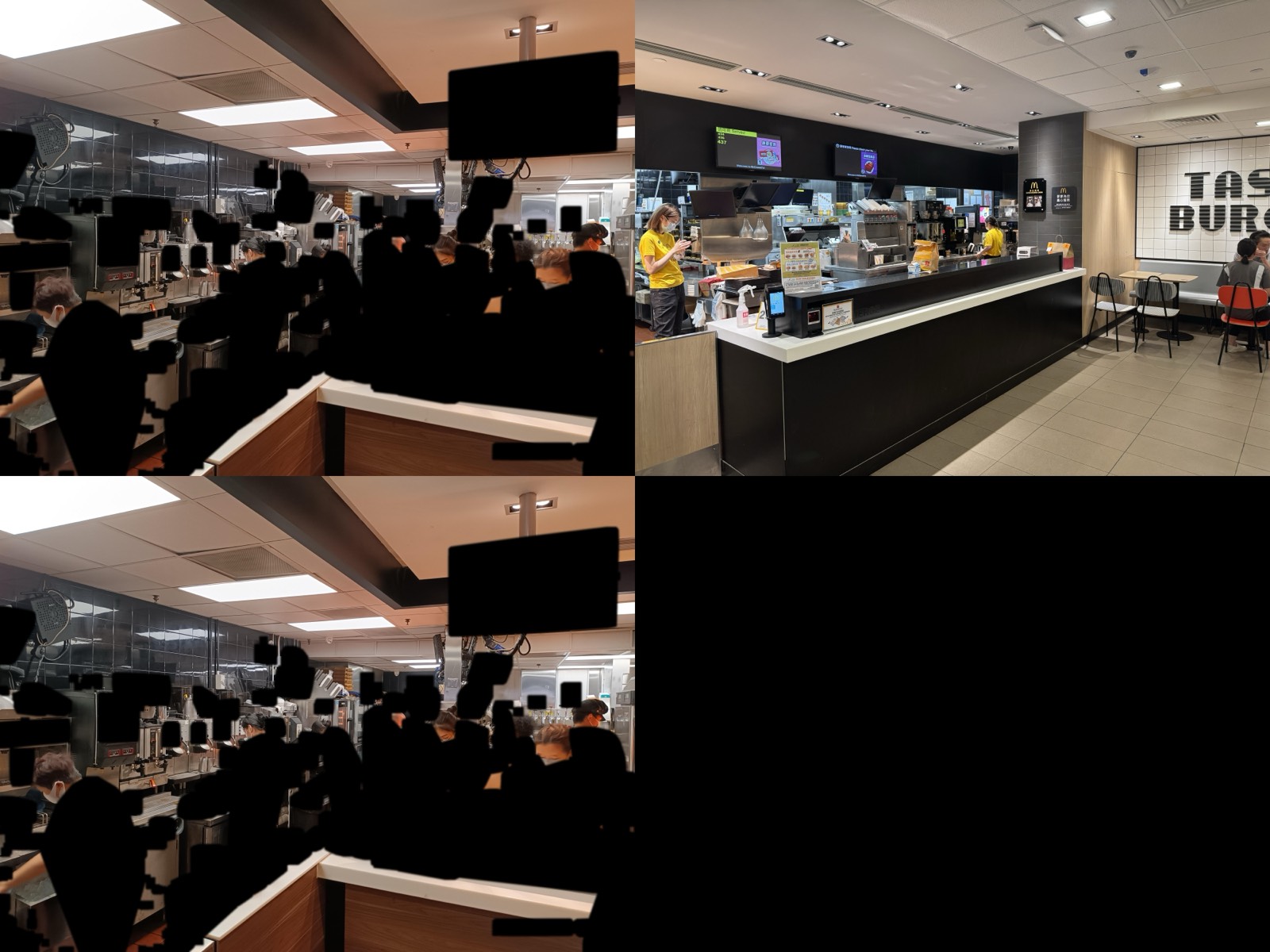}
& \includegraphics[width=0.17\textwidth, valign=m]{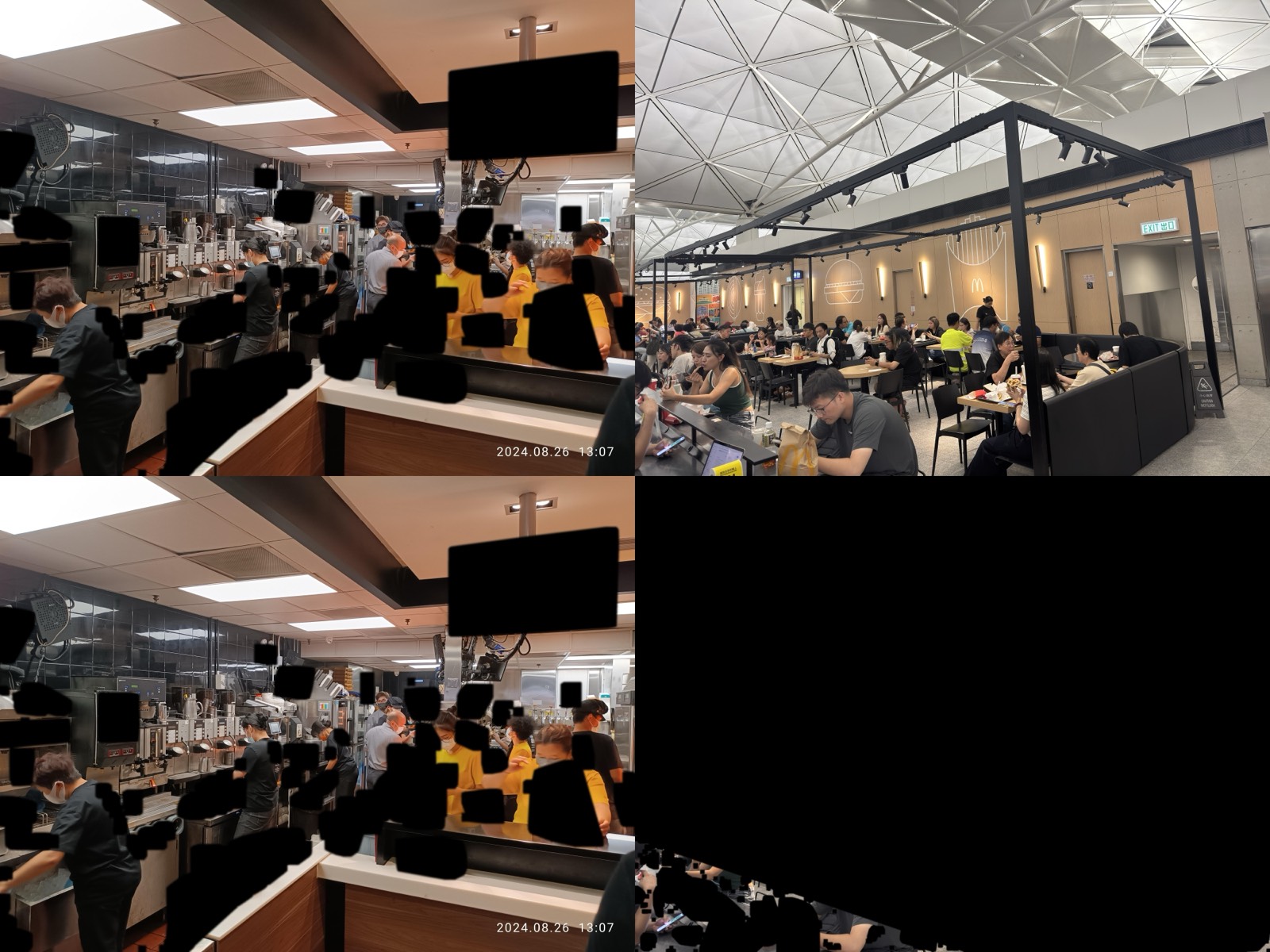}
& \includegraphics[width=0.17\textwidth, valign=m]{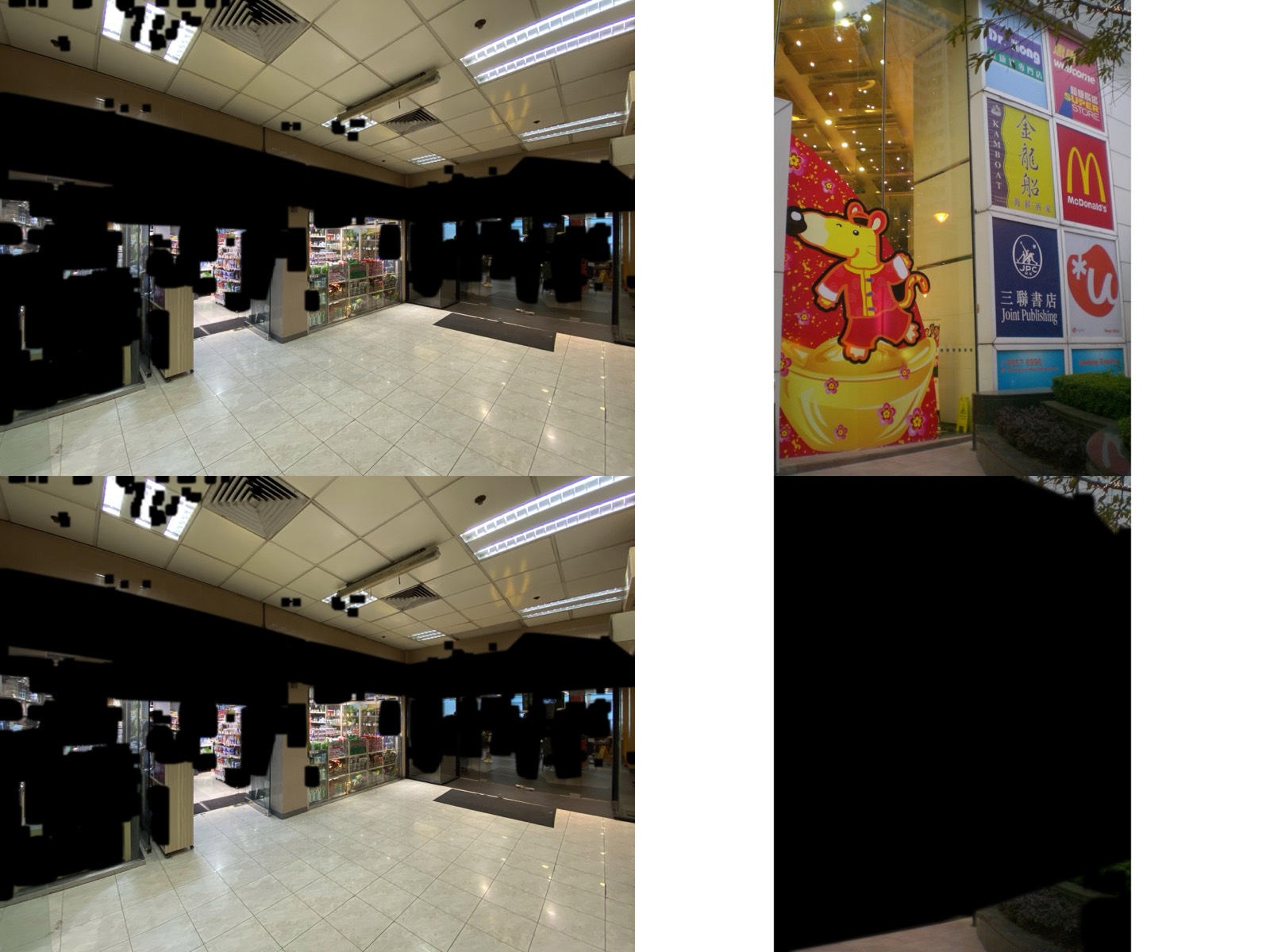}
& \includegraphics[width=0.17\textwidth, valign=m]{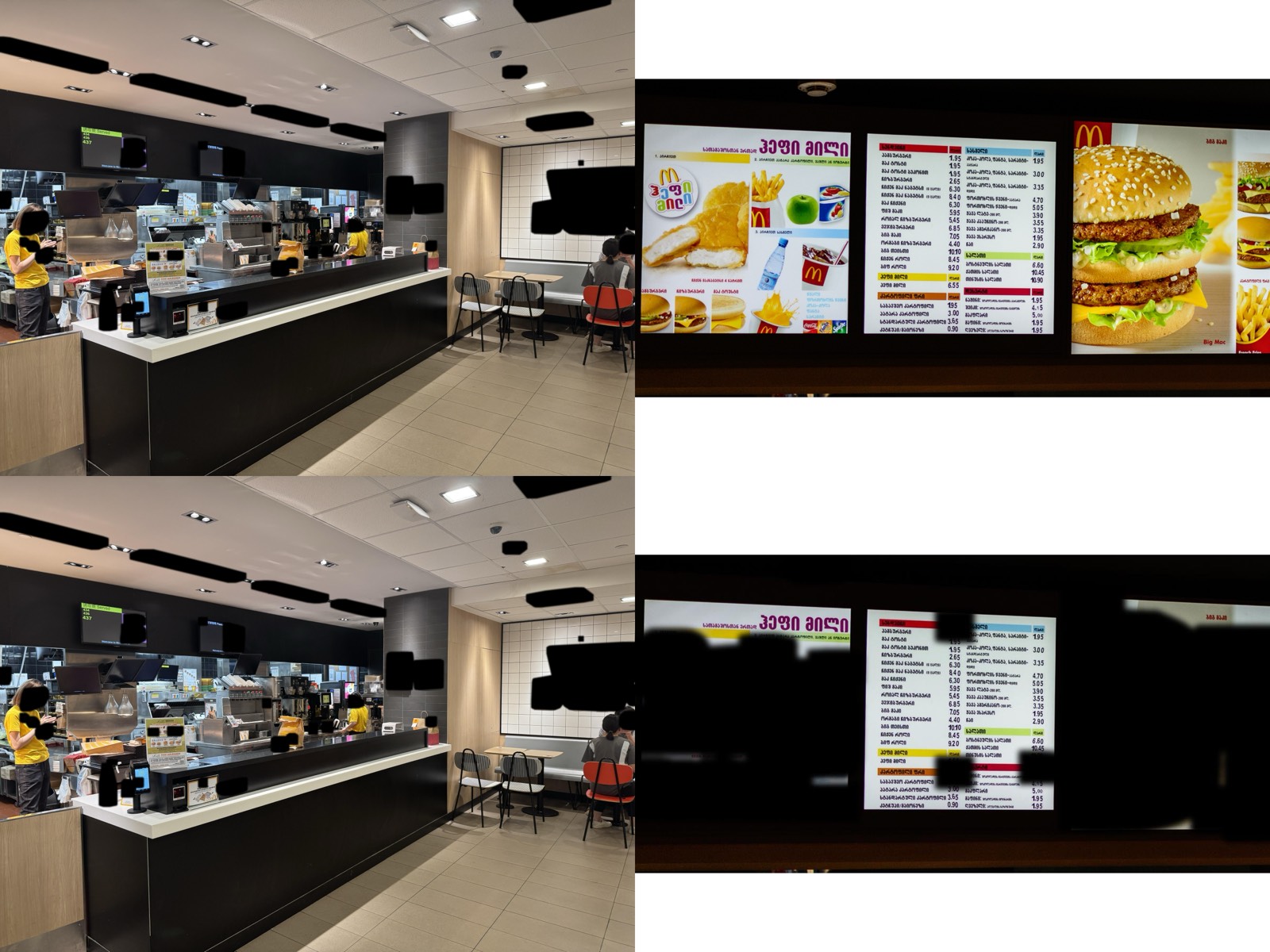}
& \\
\addlinespace[1ex] 

\bottomrule
\end{tabular}
}
\caption{Decisive witnesses for 49 images capturing the McDonald's brand and sanitizing the abstract concept \texttt{the identity of the fast food restaurant}.}
\Description{A matrix of image pairs gives the decisive witness for each McDonald's sanitization mechanism. Rows are redaction models, and columns are concept-generation models.}
\label{fig:mcdonalds_failure_compare}
\end{minipage}
}]

\clearpage
\twocolumn[{
\begin{minipage}{\textwidth}
\section{Raw Distance-Mechanism Results}
\captionsetup{type=table}
\caption{Run-level privacy resolutions for the 94 matched DiCaprio sanitizations summarized in Table~\ref{tab:dicaprio_distance_agreement}. Each three-value entry lists matched sanitization runs in order across distance mechanisms; because they are deterministic, \texttt{SAM3} configurations have only one run. Resolution magnitudes are mechanism-specific. Positive resolutions use the same two-decimal, ties-to-even rounding as Figure~\ref{fig:dicaprio_model_compare}. Exact zeros are shown as \texttt{0}; \texttt{0.00} denotes a small positive resolution rounded to zero.}
\label{tab:dicaprio_distance_raw}
\centering
\footnotesize
\setlength{\tabcolsep}{3pt}
\renewcommand{\arraystretch}{1}
\begin{tabular*}{\textwidth}{@{\extracolsep{\fill}}llccc@{}}
\toprule
\textbf{Redaction} & \textbf{Concept source} & \texttt{EVA} & \texttt{SIGLIP2} & \texttt{METACLIP2} \\
\midrule
\texttt{iGPT15} & \texttt{GPT54}   & \texttt{0.20, 0.26, 0.20} & \texttt{0.01, 0, 0}       & \texttt{0.00, 0.00, 0.05} \\
                   & \texttt{GEM31p}  & \texttt{0.26, 0.24, 0.23} & \texttt{0.03, 0, 0.00}    & \texttt{0.03, 0.02, 0.06} \\
                   & \texttt{OPS46}   & \texttt{0.24, 0.19, 0.14} & \texttt{0.03, 0.02, 0.02} & \texttt{0, 0, 0.05} \\
                   & \texttt{Manual}  & \texttt{0.04, 0.09, 0}      & \texttt{0.03, 0.10, 0}    & \texttt{0.11, 0.11, 0.05} \\
                   & \texttt{None}    & \texttt{0.04, 0.09, 0.13} & \texttt{0.07, 0.10, 0}    & \texttt{0.12, 0.11, 0.05} \\
\addlinespace[0.5ex]
\texttt{iGPT1m} & \texttt{GPT54}   & \texttt{0.06, 0.20, 0.17} & \texttt{0.05, 0.04, 0.05} & \texttt{0.07, 0.02, 0.03} \\
                   & \texttt{GEM31p}  & \texttt{0.02, 0.16, 0.13} & \texttt{0.04, 0.10, 0.05} & \texttt{0.03, 0.17, 0.04} \\
                   & \texttt{OPS46}   & \texttt{0.38, 0.37, 0.38} & \texttt{0, 0.19, 0.20}    & \texttt{0.27, 0.32, 0.24} \\
                   & \texttt{Manual}  & \texttt{0.08, 0.20, 0.12} & \texttt{0.02, 0.03, 0.09} & \texttt{0.04, 0.05, 0.04} \\
                   & \texttt{None}    & \texttt{0.24, 0.32, 0.21} & \texttt{0.05, 0.08, 0.08} & \texttt{0.09, 0.03, 0.05} \\
\addlinespace[0.5ex]
\texttt{iGEM3p} & \texttt{GPT54}   & \texttt{0.32, 0.39, 0.40} & \texttt{0.17, 0.28, 0.21} & \texttt{0.29, 0.15, 0.16} \\
                   & \texttt{GEM31p}  & \texttt{0.35, 0.32, 0.31} & \texttt{0.26, 0.15, 0.14} & \texttt{0.23, 0.05, 0.11} \\
                   & \texttt{OPS46}   & \texttt{0.33, 0.25, 0.38} & \texttt{0.28, 0.14, 0.18} & \texttt{0.18, 0.10, 0.12} \\
                   & \texttt{Manual}  & \texttt{0.49, 0.45, 0.50} & \texttt{0.35, 0.23, 0.32} & \texttt{0.11, 0.19, 0.15} \\
                   & \texttt{None}    & \texttt{0.53, 0.23, 0.18} & \texttt{0.27, 0.13, 0.15} & \texttt{0.20, 0.15, 0.30} \\
\addlinespace[0.5ex]
\texttt{iGEM31f} & \texttt{GPT54}  & \texttt{0.24, 0.34, 0.28} & \texttt{0, 0.30, 0.13}    & \texttt{0.03, 0.18, 0.03} \\
                    & \texttt{GEM31p} & \texttt{0.29, 0.23, 0.41} & \texttt{0.23, 0.03, 0.22} & \texttt{0.08, 0.04, 0.31} \\
                    & \texttt{OPS46}  & \texttt{0.35, 0.23, 0.48} & \texttt{0.22, 0.03, 0.24} & \texttt{0.13, 0.02, 0.08} \\
                    & \texttt{Manual} & \texttt{0, 0, 0.00}           & \texttt{0, 0, 0}         & \texttt{0.01, 0.03, 0.01} \\
                    & \texttt{None}   & \texttt{0.15, 0.26, 0.08} & \texttt{0.09, 0.13, 0.03} & \texttt{0.10, 0.24, 0.08} \\
\addlinespace[0.5ex]
\texttt{FLX2p} & \texttt{GPT54}    & \texttt{0.30, 0.25, 0.36} & \texttt{0.13, 0.15, 0.16} & \texttt{0.29, 0.26, 0.25} \\
                  & \texttt{GEM31p}   & \texttt{0.35, 0.28, 0.19} & \texttt{0.16, 0.10, 0.12} & \texttt{0.30, 0.27, 0.29} \\
                  & \texttt{OPS46}    & \texttt{0.22, 0.26, 0.29} & \texttt{0.15, 0.14, 0.11} & \texttt{0.24, 0.23, 0.28} \\
                  & \texttt{Manual}   & \texttt{0.22, 0.28, 0.18} & \texttt{0.16, 0.12, 0.14} & \texttt{0.29, 0.33, 0.23} \\
                  & \texttt{None}     & \texttt{0.29, 0.16, 0.38} & \texttt{0.12, 0.07, 0.17} & \texttt{0.25, 0.31, 0.25} \\
\addlinespace[0.5ex]
\texttt{FLX2d} & \texttt{GPT54}    & \texttt{0.15, 0.18, 0.11} & \texttt{0.12, 0.07, 0.05} & \texttt{0.28, 0.06, 0.16} \\
                  & \texttt{GEM31p}   & \texttt{0.11, 0.14, 0.14} & \texttt{0.08, 0.10, 0.08} & \texttt{0.17, 0.15, 0.07} \\
                  & \texttt{OPS46}    & \texttt{0.35, 0.24, 0.31} & \texttt{0.15, 0.17, 0.18} & \texttt{0.34, 0.24, 0.34} \\
                  & \texttt{Manual}   & \texttt{0.22, 0.23, 0.26} & \texttt{0.14, 0.08, 0.14} & \texttt{0.30, 0.19, 0.21} \\
                  & \texttt{None}     & \texttt{0.21, 0.28, 0.27} & \texttt{0.14, 0.15, 0.18} & \texttt{0.24, 0.22, 0.26} \\
\addlinespace[0.5ex]
\texttt{SAM3} & \texttt{GPT54}      & \texttt{0.09}                   & \texttt{0.05}                & \texttt{0.27} \\
                & \texttt{GEM31p}     & \texttt{0.11}                   & \texttt{0.02}                & \texttt{0.31} \\
                & \texttt{OPS46}      & \texttt{0.12}                   & \texttt{0.08}                & \texttt{0.32} \\
                & \texttt{Manual}     & \texttt{0}                      & \texttt{0}                   & \texttt{0.01} \\
\bottomrule
\end{tabular*}
\end{minipage}
}]

\end{document}